\documentclass[11pt,oneside]{book}

\usepackage[margin=1.1in]{geometry}
\usepackage[T1]{fontenc}
\usepackage[utf8]{inputenc}
\usepackage{lmodern}
\usepackage{microtype}
\usepackage{amsmath,amssymb,amsthm,mathtools,bm}
\usepackage{booktabs}
\usepackage{enumitem}
\usepackage{xcolor}
\usepackage{graphicx}
\usepackage{tikz}
\usetikzlibrary{calc,decorations.pathreplacing,arrows.meta,positioning}

\usepackage{hyperref}
\usepackage[nameinlink,noabbrev]{cleveref}
\usepackage[backend=biber,style=numeric,sorting=nyt,maxnames=8]{biblatex}
\hypersetup{
  pdftitle={A Quantum Path for Partial Differential Equations},
  pdfauthor={Xiantao Li},
  pdfsubject={Quantum algorithms for partial differential equations},
  pdfkeywords={quantum algorithms, partial differential equations, block encoding, QSVT, Hamiltonian simulation, finite differences, finite elements}
}

\hypersetup{
  colorlinks=true,
  linkcolor=blue!50!black,
  citecolor=blue!50!black,
  urlcolor=blue!50!black
}

\newcommand{\DiagonalPipelineGraphic}{%
\resizebox{16.0cm}{!}{%
\begin{tikzpicture}[
    line cap=round, line join=round, >=Latex,
    every node/.style={font=\small},
    flow/.style={gray!62, line width=0.8pt, ->, shorten >= 4pt, shorten <= 4pt},
    box/.style={
        draw=gray!62, fill=white, rounded corners=5pt, line width=0.75pt,
        align=center, inner sep=15pt, text=gray!56
    },
    mathtext/.style={gray!56},
    faint/.style={gray!32,line width=0.52pt}
]

    \node (mesh) at (0,0) {
        \begin{tikzpicture}[scale=1.1] % Scaled up the mesh slightly to match bigger boxes
            \coordinate (a1) at (0.00,0.15);  \coordinate (a2) at (0.55,0.88);
            \coordinate (a3) at (0.18,1.85);  \coordinate (a4) at (1.18,0.22);
            \coordinate (a5) at (1.05,1.08);  \coordinate (a6) at (0.98,2.00);
            \coordinate (a7) at (2.05,0.12);  \coordinate (a8) at (1.98,0.98);
            \coordinate (a9) at (1.88,1.92);  \coordinate (a10) at (2.95,0.34);
            \coordinate (a11) at (2.92,1.30); \coordinate (a12) at (2.70,2.08);

            \draw[faint] (a1)--(a4)--(a7)--(a10)--(a11)--(a12)--(a9)--(a6)--(a3)--cycle;
            \draw[faint] (a1)--(a2)--(a3);    \draw[faint] (a1)--(a4)--(a2);
            \draw[faint] (a2)--(a4)--(a5);    \draw[faint] (a3)--(a2)--(a5);
            \draw[faint] (a3)--(a5)--(a6);    \draw[faint] (a4)--(a7)--(a8);
            \draw[faint] (a4)--(a8)--(a5);    \draw[faint] (a5)--(a8)--(a9);
            \draw[faint] (a5)--(a9)--(a6);    \draw[faint] (a7)--(a10)--(a8);
            \draw[faint] (a8)--(a10)--(a11);  \draw[faint] (a8)--(a11)--(a9);
            \draw[faint] (a9)--(a11)--(a12);

            \foreach \p in {a1,a2,a3,a4,a5,a6,a7,a8,a9,a10,a11,a12}
                \fill[gray!38] (\p) circle (1.05pt);
        \end{tikzpicture}
    };

    \node[box, below right=0.45cm and 0.25cm of mesh] (UA) {
        {\normalsize $U_{A_h}$}\\[8pt]
        $\displaystyle \left(\langle 0^a|\otimes I\right)\, U_{A_h}\, \left(|0^a\rangle\otimes I\right) = A_h/\alpha$
    };

    \node[box, below right=0.45cm and 0.25cm of UA, minimum width=3.6cm, minimum height=2.0cm] (MID) {
        \fontsize{26}{26}\selectfont $\cdots$
    };

    \node[box, below right=0.45cm and 0.25cm of MID] (UP) {
        {\normalsize $U_{p}$}\\[8pt]
        $\displaystyle \left(\langle 0^a|\otimes I\right)\, U_{p}\, \left(|0^a\rangle\otimes I\right) = p(A_h/\alpha)$
    };

    \node[below right=0.45cm and 0.25cm of UP] (MEAS) {
        \begin{tikzpicture}[scale=1.15] % Scaled up to match
            \draw[gray!62,line width=0.62pt] (0,0) rectangle (1.35,0.98);
            \draw[gray!62,line width=0.62pt] (0.18,0.18) -- (0.62,0.48) -- (1.02,0.76);
            \draw[gray!62,line width=0.62pt] (0.62,0.48) -- (0.62,0.14);
            \draw[gray!62,line width=0.62pt] (0.18,0.18) -- (0.18,0.80);
            \draw[gray!62,line width=0.62pt] (0.18,0.18) -- (1.10,0.18);
        \end{tikzpicture}
    };

    \node[mathtext, right=0.25cm of MEAS] (RESULT) {
        $\displaystyle \frac{p(A_h/\alpha)\,|\psi\rangle} {\|p(A_h/\alpha)\,|\psi\rangle\|}$
    };

    \draw[flow] (mesh) -- (UA);
    \draw[flow] (UA)   -- (MID);
    \draw[flow] (MID)  -- (UP);
    \draw[flow] (UP)   -- (MEAS);

\end{tikzpicture}%
}%
}

\newcommand{\BookCover}{%
\begin{titlepage}
    \thispagestyle{empty}
    \begin{tikzpicture}[remember picture,overlay]
        \fill[white] (current page.south west) rectangle (current page.north east);

        \begin{scope}[shift={($(current page.north east)+(-4.8cm,-2.8cm)$)},opacity=0.28]
            \foreach \i in {-3,-2.5,...,3}{
                \draw[gray!24,line width=0.35pt] plot[smooth,domain=-3.0:3.0,samples=70] (\x,{0.45*sin(55*\x)+0.10*\i*\i-0.12*\x*\i});
            }
            \foreach \j in {-3,-2.5,...,3}{
                \draw[gray!18,line width=0.30pt] plot[smooth,domain=-3.0:3.0,samples=70] ({\x},{0.45*sin(55*\j)+0.10*\x*\x-0.12*\x*\j});
            }
        \end{scope}

        \foreach \i in {0,...,13}{
            \foreach \j in {0,...,16}{
                \fill[gray!35,opacity=0.10] ($(current page.north east)+(-0.50cm-0.16cm*\i,-0.55cm-0.16cm*\j)$) circle (0.45pt);
            }
        }

        \begin{scope}[shift={(current page.south west)},opacity=0.50]
            \foreach \i in {0,...,24}{
                \draw[gray!14,line width=0.25pt] plot[smooth,domain=0:17.0,samples=95] ({\x-0.8},{0.65+0.055*\i+0.23*sin(25*\x+9*\i)});
            }
            \foreach \i in {0,...,38}{
                \foreach \j in {0,...,8}{
                    \fill[gray!28,opacity=0.13] ({0.2+0.35*\i},{0.25+0.21*\j+0.10*sin(20*\i)}) circle (0.32pt);
                }
            }
        \end{scope}

        \node[anchor=north west] at ($(current page.north west)+(1.25cm,-1.50cm)$) {\DiagonalPipelineGraphic};

    \end{tikzpicture}

    \vspace*{0.33\paperheight} 
    \begin{center}
        {\bfseries\fontsize{31}{37}\selectfont A Quantum Path to\\[0.16em] Partial Differential Equations\par}

        \vspace{0.72cm}

        {\color{gray!60} \rule{3.2cm}{0.35pt}\, \raisebox{-0.12em}{$\diamond$}\, \rule{3.2cm}{0.35pt}\par}

        \vspace{0.72cm}

        {\Large Xiantao Li\par}

        \vspace{0.42cm}

        {\large July 9, 2026\par}
    \end{center}
\end{titlepage}%
}
\numberwithin{equation}{chapter}

\newtheorem{definition}{Definition}[chapter]

\newtheorem{lemma}[definition]{Lemma}
\newtheorem{proposition}[definition]{Proposition}
\newtheorem{example}[definition]{Example}
\newtheorem{remark}[definition]{Remark}

\newcommand{\C}{\mathbb{C}}
\newcommand{\R}{\mathbb{R}}

\newcommand{\ket}[1]{\left|#1\right\rangle}
\newcommand{\bra}[1]{\left\langle#1\right|}
\newcommand{\braket}[2]{\left\langle #1\middle|#2\right\rangle}
\newcommand{\ip}[2]{\left\langle #1,#2\right\rangle}
\newcommand{\norm}[1]{\left\|#1\right\|}
\newcommand{\abs}[1]{\left|#1\right|}

\newcommand{\polylog}{\operatorname{polylog}}

\newcommand{\Real}{\operatorname{Re}}
\newcommand{\Imag}{\operatorname{Im}}

\newcommand{\eps}{\varepsilon}
\newcommand{\ot}{\otimes}

\newtheorem{exercise}{Exercise}[chapter]

\title{A Quantum Path to Partial Differential Equations}
\author{Xiantao Li}
\date{\today}

\begin{document}

\BookCover

\frontmatter
\hypersetup{pageanchor=false}
\hypersetup{pageanchor=true}

\tableofcontents
% Body-only LaTeX file.  It is intended to be incorporated by \include{...}
% into a larger book manuscript.  No documentclass, packages, theorem
% declarations, or macro preamble are included here.

\chapter*{Preface}
\addcontentsline{toc}{chapter}{Preface}

Quantum computing has long been motivated by problems in quantum physics. Over the past two decades, this promise has become increasingly concrete in quantum chemistry and materials science, where Hamiltonian simulation and related algorithms provide a natural connection between physical models and quantum computation. More recently, general scientific computing problems---especially differential equations, including partial differential equations---have emerged as a promising application area outside quantum physics for fault-tolerant quantum algorithms
\cite{ChildsLiuOstrander2021,BabbushBerryKothariSommaWiebe2023}.

This book grew out of a series of talks and lectures I have given on this subject. It is written primarily for readers who are familiar with PDEs and their numerical solution and who would like to begin working on the quantum side. It should be equally useful to readers who already have experience with quantum algorithms and would like a numerically grounded introduction to PDE applications,  including, specifically,  the discretizations behind the matrices, the conditioning, normalization, and measurement questions that determine whether a quantum PDE algorithm is useful in an end-to-end manner.

There are already several excellent introductions to quantum computing and quantum algorithms. Nielsen and Chuang provide the standard comprehensive reference
\cite{NielsenChuang2010}. The lecture notes of Childs and de Wolf give broad and complementary accounts of quantum algorithms and quantum computation
\cite{ChildsQuantumAlgorithmsNotes,deWolf2019LectureNotes}. The lecture notes of Lin and Wiebe on quantum algorithms for scientific computation are particularly close in spirit to the present book: they develop block encoding, qubitization, quantum signal processing, quantum singular value transformation, Hamiltonian simulation, and quantum algorithms for linear systems and differential equations from the viewpoint of scientific computation
\cite{LinWiebe2026QASC}. 

For readers who learn best by building circuits and running code, Qiskit and IBM Quantum Learning provide an intuitive route from gates and circuits to small algorithmic examples
\cite{JavadiAbhariEtAl2024Qiskit,IBMQuantumLearning2026}. PennyLane offers a complementary software-oriented entry point, with tutorials illustrating block encodings from matrix access oracles, linear-combination-of-unitaries constructions, and QSVT in practice
\cite{PennyLaneBlockEncodingDemo,PennyLaneLCUBlockEncodingDemo,PennyLaneQSVTPracticeDemo}. More recently, toolboxes have begun to target the block-encoding path directly. For example, Unitaria provides a Python interface for quantum linear algebra through block encodings, allowing users to compose operations such as addition, multiplication, tensor products, and QSVT at the level of matrix-like objects
\cite{DeimlHuettenhoferMoscoKottmannPeterseim2026Unitaria}. Qrisp similarly treats block encodings as high-level programming abstractions for quantum linear algebra
\cite{PetricZander2026QrispBlockEncoding}. These tools are still early, and their interfaces will certainly evolve, but they are an important sign that block encoding is becoming not only a theorem language, but also a programming interface.

The purpose of this book is narrower than these general references and software platforms. Its particular emphasis is on PDEs and on the interaction between quantum algorithms and the classical numerical structures that arise from their discretization. The hope is not to replace the broader quantum-computing literature, but to give readers a concrete path from finite differences, finite elements, and operator semigroups to quantum primitives, including block encodings, quantum singular value transformation, Hamiltonian simulation, and measurement.

This emphasis is timely. The search for useful quantum applications is now a major theme not only in academic research, but also in industry-facing quantum algorithm development. The recent perspective paper \emph{The Grand Challenge of Quantum Applications}, written by researchers at Google Quantum AI, stresses the need to identify concrete problem instances and to connect abstract algorithmic advantages to real use cases
\cite{BabbushKingEtAl2025GrandChallenge}. Differential equations and PDE-like simulation tasks are increasingly visible in this landscape. For example, IBM's Qiskit ecosystem now includes application-oriented PDE solver functions, such as QUICK-PDE developed by ColibriTD
\cite{IBMQuickPDE2025}; PsiQuantum-affiliated work revisited Carleman-based algorithms for nonlinear dynamics
\cite{JenningsKorzekwaLostaglioSornborgerSubasiWang2025Carleman}; and Quantinuum-affiliated work develops QFT-based quantum circuits for PDEs in Fourier space
\cite{LubaschKikuchiWrightMcKeever2025FourierPDE}. Outside the block-encoding route, PDE-oriented software platforms are also beginning to appear; for instance, SJTU's UnitaryLab emphasizes quantum scientific computing for differential equations through Schr\"odingerization-based workflows
\cite{SJTUUnitaryLab2025}. These examples do not settle the question of practical advantage, but they show that differential equations are becoming a serious meeting point between quantum algorithms, scientific computing, and emerging software stacks.

In seminars and conferences, the question I hear most often from numerical analysts and scientific computing researchers is a practical one:
\emph{What are the basic building blocks?}
Classical scientific computing has a well-developed answer. Libraries such as BLAS and LAPACK provide a foundation of linear algebraic operations---matrix products, factorizations, linear solves, and eigenvalue routines---on which discretizations, time integrators, and application codes are built. Quantum scientific computing does not yet have a direct counterpart of comparable maturity, but block encoding offers a useful organizing principle. 

A block encoding represents a matrix as a distinguished block of a larger unitary operator. Once a discretized differential operator has been placed in this form, a relatively small collection of primitives---quantum singular value transformation, Hamiltonian simulation, linear combinations of unitaries, amplitude amplification, postselection, and measurement---can be composed into algorithms for elliptic, parabolic, and hyperbolic equations. This is not the only possible route to quantum PDE algorithms. It is, however, a flexible and mathematically transparent path, and it allows several different classes of equations to be treated within a common language. This is also the origin of the book's title: it presents \emph{a} quantum path to PDEs, not \emph{the} quantum path. Variational methods, continuous-variable approaches, analog simulation, tensor networks, and problem-specific quantum constructions all deserve attention. The path emphasized here is the block-encoding path because it is one of the clearest ways to see how classical discretizations, matrix functions, and quantum circuits fit together.

At the same time, a PDE is not simply a large linear system waiting to be solved. PDE models come with function spaces, weak formulations, boundary conditions, stability principles, conservation laws, regularity assumptions, discretization errors, and mesh-dependent conditioning. These features are routine in numerical analysis, but they are not always visible when PDE problems are formulated in the quantum computing literature. Conversely, amplitude encoding, block encoding, postselection, quantum signal processing, and quantum singular value transformation are not normally part of the numerical PDE curriculum.

The aim of this book is therefore to build a bridge between the two subjects. Whenever possible, the classical and quantum formulations are placed side by side. The reader can follow the full computational pipeline:
\[
\begin{aligned}
 &\text{continuous PDE}
\longrightarrow
\text{discretization}
\longrightarrow
\text{matrix or semidiscrete evolution} \\
& \quad \longrightarrow
\text{quantum encoding}
\longrightarrow
\text{quantum transformation}
\longrightarrow
\text{quantity of interest}.
\end{aligned}
\]
The discretization error, quantum algorithmic error, state-preparation cost, success probability, and measurement cost are all parts of this pipeline.

A recurring theme is that a quantum algorithm for a PDE is not merely a faster linear algebra routine. A quantum computer generally prepares a normalized state rather than returning every entry of a solution vector. Consequently, the cost of loading the input, the normalization of the solution, the probability of successful postselection, and the number of measurements needed to extract an observable can be as important as the complexity of the central matrix transformation. An apparent speedup in one stage need not translate automatically into an end-to-end advantage. Keeping these issues visible is essential for a meaningful comparison with classical algorithms.

The organization of the book reflects these goals. Chapter~\ref{chap:basic-elements} introduces the basic quantum primitives in the language of numerical linear algebra. Readers already familiar with quantum algorithms may skip much of this chapter or use it as a reference; readers coming from numerical PDEs should find its examples---Laplacians, grid functions, and stencils---familiar objects presented in a new vocabulary. Chapters~\ref{chap:elliptic-quantum}, \ref{chap:hyperbolic-quantum}, and \ref{chap:parabolic-quantum} treat canonical elliptic, hyperbolic, and parabolic model problems. Each chapter begins with familiar finite difference or finite element discretizations and then asks how the resulting matrices or semidiscrete evolutions can be encoded and processed on a quantum computer. Readers from numerical analysis may skim these classical openings and concentrate on the encoding, transformation, and measurement stages. Meanwhile, readers from quantum computing can read the same sections as a compact review of the discrete structures---stiffness and mass matrices, CFL conditions, mesh-dependent conditioning---that any quantum algorithm for PDEs must respect. Chapter~\ref{chap:nonlinear-quantum} gives a brief first look at nonlinear problems through two representative linearization frameworks. The PDE chapters close with outlook sections describing open research directions, and every chapter ends with exercises, so the book can be used for self-study or as the basis of a one-semester topics course.

On the classical PDE side, Evans's book remains a standard and accessible reference for the analytical theory, including weak formulations, energy estimates, regularity, and the function-space viewpoint for elliptic, parabolic, and hyperbolic equations
\cite{Evans2010PDE}. For readers who already know quantum algorithms and wish to understand the numerical side of PDEs, the books of Larsson and Thom\'ee, Tveito and Winther, and Morton and Mayers provide useful entry points into finite differences, finite elements, stability, and error estimates
\cite{LarssonThomee2009,TveitoWinther2005,MortonMayers2005}. Readers familiar with the finite difference and finite element treatments in these books should find the discretizations used here recognizable. On the quantum side, no prior training in quantum mechanics is assumed; elementary linear algebra is the only essential prerequisite. The necessary notation and circuit concepts are introduced as they are needed.

This book is deliberately not a comprehensive survey, and it does not claim
that quantum computers will solve PDEs faster than classical computers across
the board.
 Such a claim would be premature and, in many settings, misleading. PDE computation includes nonlinear conservation laws, shocks and singularities, multiscale problems, fluid--structure interaction, stochastic dynamics, turbulence, moving interfaces, and many other challenges for which the quantum algorithmic picture remains incomplete.

Several important quantum approaches are also treated only briefly or omitted, including variational quantum algorithms, continuous-variable methods, tensor-network approaches, stochastic differential equations, and detailed fault-tolerant resource estimates. Hardware-level compilation, which depends strongly on the native gate set, qubit connectivity, and error-correction architecture of a particular platform, is likewise beyond the scope of this introduction. These omissions are intentional. The objective is to make one coherent route sufficiently explicit that readers can understand its mechanisms, limitations, and possible extensions.

Quantum algorithms for PDEs are still at an early stage. The subject lies at the intersection of quantum algorithms, numerical analysis, applied mathematics, and scientific computing, and substantial progress will require ideas from all of these areas. My hope is that this book provides a useful entry point: a first map for students, a common vocabulary for researchers from different communities, and an invitation to develop sharper algorithms, stronger analysis, and more realistic applications.

If the book helps a few more researchers cross the boundary between numerical PDEs and quantum computation, it will have served its purpose. The author gratefully acknowledges the support of the National Science
Foundation under Grants No.~DMS-2411120 and DMS-2552687. Suggestions, corrections, and questions are welcome and may be sent to
\href{mailto:Xiantao.Li@psu.edu}{Xiantao.Li@psu.edu}.

\vspace{1.5em}
\noindent
Xiantao Li\\
Department of Mathematics, Pennsylvania State University \\
University Park, Pennsylvania

\mainmatter
\chapter{Basic Quantum Elements for PDE Algorithms}
\label{chap:basic-elements}

% The figures in this chapter are drawn directly with TikZ.
% The main file should load: \usepackage{tikz} and \usetikzlibrary{decorations.pathreplacing}.
\newcommand{\SwapTestCircuit}{%
\begin{tikzpicture}[x=0.9cm,y=0.65cm,line cap=round,line join=round,
    qwire/.style={gray!70,line width=0.45pt},
    gate/.style={draw=gray!70,fill=white,minimum width=0.46cm,minimum height=0.36cm,inner sep=1pt},
    ctrl/.style={circle,fill=gray!70,inner sep=1.25pt},
    meter/.style={draw=gray!70,fill=white,minimum width=0.58cm,minimum height=0.38cm,inner sep=1pt}]
\foreach \y/\lab in {0/{\ket{0}},-1/{\ket{\bm u}},-2/{\ket{\bm v}}}{
  \draw[qwire] (0,\y)--(6.8,\y);
  \node[left=3pt,gray!70] at (0,\y) {$\lab$};
}
\node[gate] at (0.9,0) {$H$};
\node[ctrl] at (2.6,0) {};
\draw[qwire] (2.6,0)--(2.6,-2);
\draw[gray!70,line width=0.5pt] (2.48,-1.12)--(2.72,-0.88);
\draw[gray!70,line width=0.5pt] (2.48,-0.88)--(2.72,-1.12);
\draw[gray!70,line width=0.5pt] (2.48,-2.12)--(2.72,-1.88);
\draw[gray!70,line width=0.5pt] (2.48,-1.88)--(2.72,-2.12);
\node[gate] at (4.1,0) {$H$};
\node[meter] at (5.35,0) {$M$};
\draw[qwire,->] (5.65,0)--(6.55,0);
\end{tikzpicture}%
}

\newcommand{\HadamardTestCircuit}{%
\begin{tikzpicture}[x=0.9cm,y=0.7cm,line cap=round,line join=round,
    qwire/.style={gray!70,line width=0.45pt},
    gate/.style={draw=gray!70,fill=white,minimum width=0.50cm,minimum height=0.38cm,inner sep=1pt},
    biggate/.style={draw=gray!70,fill=white,minimum width=0.75cm,minimum height=0.48cm,inner sep=2pt},
    ctrl/.style={circle,fill=gray!70,inner sep=1.25pt},
    meter/.style={draw=gray!70,fill=white,minimum width=0.58cm,minimum height=0.38cm,inner sep=1pt}]
\foreach \y/\lab in {0/{\ket{0}},-1/{\ket{\psi}}}{
  \draw[qwire] (0,\y)--(6.8,\y);
  \node[left=3pt,gray!70] at (0,\y) {$\lab$};
}
\node[gate] at (0.9,0) {$H$};
\node[ctrl] at (2.4,0) {};
\draw[qwire] (2.4,0)--(2.4,-1);
\node[biggate] at (2.4,-1) {$W$};
\node[gate] at (4.0,0) {$H$};
\node[meter] at (5.2,0) {$M$};
\draw[qwire,->] (5.55,0)--(6.55,0);
\end{tikzpicture}%
}

\newcommand{\JointUnitaryRegisters}{%
\begin{tikzpicture}[x=0.95cm,y=0.78cm,line cap=round,line join=round,
    qwire/.style={gray!70,line width=0.48pt},
    biggate/.style={draw=gray!68,fill=white,rounded corners=1pt,
      minimum width=1.55cm,minimum height=1.24cm,inner sep=2pt},
    lab/.style={gray!62,font=\scriptsize}]
\draw[qwire] (0,0)--(6.7,0);
\draw[qwire] (0,-1)--(6.7,-1);
\node[left=4pt,gray!70] at (0,0) {$\ket{0^a}$};
\node[left=4pt,gray!70] at (0,-1) {$\ket{\psi}$};
\node[lab,left=25pt] at (0,0) {ancilla register};
\node[lab,left=25pt] at (0,-1) {system register};
\node[biggate] at (3.1,-0.5) {$U_{\rm joint}$};
\node[lab,above=5pt] at (3.1,0.16)
  {$U_{\rm joint}:\mathcal H_{\rm a}\otimes\mathcal H_{\rm s}
  \rightarrow\mathcal H_{\rm a}\otimes\mathcal H_{\rm s}$};
\draw[qwire,->] (5.6,0)--(6.55,0);
\draw[qwire,->] (5.6,-1)--(6.55,-1);
\end{tikzpicture}%
}

\newcommand{\LaplacianBlockEncodingCircuit}{%
\begin{tikzpicture}[x=1.0cm,y=0.78cm,line cap=round,line join=round,
    qwire/.style={gray!65,line width=0.52pt},
    gate/.style={draw=gray!60,fill=white,rounded corners=1pt,
      minimum width=0.78cm,minimum height=0.46cm,inner sep=2pt},
    biggate/.style={draw=gray!60,fill=white,rounded corners=1pt,
      minimum width=1.65cm,minimum height=1.20cm,inner sep=3pt},
    lab/.style={gray!55,font=\scriptsize}]
\draw[qwire] (0,0)--(8.4,0);
\draw[qwire] (0,-1)--(8.4,-1);
\node[left=4pt,gray!65] at (0,0) {$\ket{0^2}$};
\node[left=4pt,gray!65] at (0,-1) {$\ket{\psi}$};

\node[gate] at (1.35,0) {$P_L$};
\node[biggate] at (3.95,-0.5) {$\operatorname{SELECT}_L$};
\node[gate] at (6.55,0) {$P_L^\dag$};
\draw[qwire,->] (7.25,0)--(8.2,0);
\draw[qwire,->] (7.25,-1)--(8.2,-1);
\node[lab,anchor=west] at (8.25,-0.5) {$U_L$};

\node[lab] at (1.35,-1.55) {prepare};
\node[lab,align=center] at (3.95,-1.55)
  {controlled stencil\\$I$, $-S$, or $-S^\dag$};
\node[lab] at (6.55,-1.55) {unprepare};
\end{tikzpicture}%
}

\newcommand{\GroverRudolphTree}{%
\begin{tikzpicture}[x=1.0cm,y=0.72cm,line cap=round,line join=round,
    edge/.style={gray!65,line width=0.45pt},
    nodebox/.style={draw=gray!65,fill=white,rounded corners=1pt,inner sep=2pt,font=\scriptsize},
    leaf/.style={gray!60,font=\scriptsize}]
\node[nodebox] (r) at (0,0) {$P_{\varnothing}=1$};
\node[nodebox] (a) at (-2,-1.35) {$P_0$};
\node[nodebox] (b) at (2,-1.35) {$P_1$};
\node[nodebox] (aa) at (-3,-2.7) {$P_{00}$};
\node[nodebox] (ab) at (-1,-2.7) {$P_{01}$};
\node[nodebox] (ba) at (1,-2.7) {$P_{10}$};
\node[nodebox] (bb) at (3,-2.7) {$P_{11}$};
\foreach \u/\v/\lab in {r/a/{0},r/b/{1},a/aa/{0},a/ab/{1},b/ba/{0},b/bb/{1}}{
  \draw[edge] (\u)--node[leaf,fill=white,inner sep=1pt] {$\lab$} (\v);
}
\node[leaf] at (-3,-3.45) {$I_{00}$};
\node[leaf] at (-1,-3.45) {$I_{01}$};
\node[leaf] at (1,-3.45) {$I_{10}$};
\node[leaf] at (3,-3.45) {$I_{11}$};
\draw[edge] (aa)--(-3,-3.18);
\draw[edge] (ab)--(-1,-3.18);
\draw[edge] (ba)--(1,-3.18);
\draw[edge] (bb)--(3,-3.18);
\node[leaf,align=center] at (0,0.82) {conditional rotations use\\$P_{b0}/P_b$ and $P_{b1}/P_b$};
\end{tikzpicture}%
}

\newcommand{\QFTFourCircuit}{%
\begin{tikzpicture}[x=0.72cm,y=0.62cm,line cap=round,line join=round,
    qwire/.style={gray!70,line width=0.45pt},
    gate/.style={draw=gray!70,fill=white,minimum width=0.45cm,minimum height=0.36cm,inner sep=1pt,font=\scriptsize},
    ctrl/.style={circle,fill=gray!70,inner sep=1.15pt},
    swapx/.style={gray!70,line width=0.50pt},
    lab/.style={gray!65,font=\scriptsize}]
\foreach \y/\labtxt in {0/{\ket{x_3}},-1/{\ket{x_2}},-2/{\ket{x_1}},-3/{\ket{x_0}}}{
  \draw[qwire] (0,\y)--(13.0,\y);
  \node[left=3pt,gray!70] at (0,\y) {$\labtxt$};
}
\node[gate] at (0.85,0) {$H$};
\foreach \x/\yy/\r in {1.65/-1/2,2.45/-2/3,3.25/-3/4}{
  \node[gate] at (\x,0) {$R_{\r}$};
  \node[ctrl] at (\x,\yy) {};
  \draw[qwire] (\x,\yy)--(\x,0);
}
\node[gate] at (4.30,-1) {$H$};
\foreach \x/\yy/\r in {5.10/-2/2,5.90/-3/3}{
  \node[gate] at (\x,-1) {$R_{\r}$};
  \node[ctrl] at (\x,\yy) {};
  \draw[qwire] (\x,\yy)--(\x,-1);
}
\node[gate] at (6.95,-2) {$H$};
\node[gate] at (7.75,-2) {$R_2$};
\node[ctrl] at (7.75,-3) {};
\draw[qwire] (7.75,-3)--(7.75,-2);
\node[gate] at (8.80,-3) {$H$};
% final bit reversal swaps
\foreach \x/\ya/\yb in {10.35/0/-3,11.25/-1/-2}{
  \draw[qwire] (\x,\ya)--(\x,\yb);
  \foreach \yy in {\ya,\yb}{
    \draw[swapx] (\x-0.13,\yy+0.13)--(\x+0.13,\yy-0.13);
    \draw[swapx] (\x-0.13,\yy-0.13)--(\x+0.13,\yy+0.13);
  }
}
\node[lab] at (6.6,-4.0) {$\operatorname{QFT}_{2^4}$};
\end{tikzpicture}%
}

\newcommand{\StencilDiagram}{%
\begin{tikzpicture}[x=1.0cm,y=0.62cm,line cap=round,line join=round,
    ptnode/.style={circle,draw=gray!65,fill=white,minimum size=0.42cm,inner sep=0pt,font=\scriptsize},
    edge/.style={gray!55,line width=0.45pt},
    coeff/.style={gray!58,font=\scriptsize}]
\node[ptnode] (jm2) at (-2,0) {$j-2$};
\node[ptnode] (jm1) at (-1,0) {$j-1$};
\node[ptnode] (j0)  at (0,0) {$j$};
\node[ptnode] (jp1) at (1,0) {$j+1$};
\node[ptnode] (jp2) at (2,0) {$j+2$};
\draw[edge] (jm2)--(jm1)--(j0)--(jp1)--(jp2);
\node[coeff] at (-1,-0.65) {$-1$};
\node[coeff] at (0,-0.65) {$2$};
\node[coeff] at (1,-0.65) {$-1$};
\draw[edge,->] (0,0.72) -- (0,0.24);
\node[coeff] at (0,1.00) {center};
\node[coeff] at (0,-1.12) {$\displaystyle (L_h\bm u)_j=(2U_j-U_{j-1}-U_{j+1})/h^2$};
\end{tikzpicture}%
}

\newcommand{\FivePointStencilDiagram}{%
\begin{tikzpicture}[x=0.72cm,y=0.72cm,line cap=round,line join=round,
    ptnode/.style={circle,draw=gray!65,fill=white,minimum size=0.48cm,inner sep=0pt,font=\scriptsize},
    edge/.style={gray!55,line width=0.45pt},
    coeff/.style={gray!58,font=\scriptsize}]
\node[ptnode] (c) at (0,0) {$4$};
\node[ptnode] (l) at (-1,0) {$-1$};
\node[ptnode] (r) at (1,0) {$-1$};
\node[ptnode] (u) at (0,1) {$-1$};
\node[ptnode] (d) at (0,-1) {$-1$};
\draw[edge] (l)--(c)--(r);
\draw[edge] (d)--(c)--(u);
\node[coeff] at (0,-1.75) {five-point positive Laplacian stencil};
\end{tikzpicture}%
}

\newcommand{\OneDLaplacianStencil}{%
\begin{tikzpicture}[x=1.0cm,y=0.8cm,line cap=round,line join=round,
    nodept/.style={circle,fill=gray!45,inner sep=1.35pt},
    wire/.style={gray!45,line width=0.5pt},
    coeff/.style={gray!60,font=\footnotesize}]
\foreach \tx in {-2,-1,0,1,2}{
  \node[nodept] at (\tx,0) {};
}
\foreach \tx in {-2,-1,0,1}{
  \draw[wire] (\tx,0)--(\tx+1,0);
}
\node[coeff] at (-1,0.55) {$-1$};
\node[coeff] at (0,0.75) {$2$};
\node[coeff] at (1,0.55) {$-1$};
\node[coeff] at (-1,-0.45) {$j-1$};
\node[coeff] at (0,-0.45) {$j$};
\node[coeff] at (1,-0.45) {$j+1$};
\draw[gray!60,decorate,decoration={brace,mirror,amplitude=4pt}] (-1.15,-0.8)--(1.15,-0.8)
  node[midway,below=5pt,gray!60] {$(-u_{j-1}+2u_j-u_{j+1})/h^2$};
\end{tikzpicture}%
}

\newcommand{\RepeatedPostselectionCircuit}{%
\begin{tikzpicture}[x=0.92cm,y=0.78cm,line cap=round,line join=round,
    qwire/.style={gray!70,line width=0.48pt},
    biggate/.style={draw=gray!68,fill=white,rounded corners=1pt,
      minimum width=0.76cm,minimum height=1.18cm,inner sep=2pt},
    meter/.style={draw=gray!68,fill=white,rounded corners=1pt,
      minimum width=0.82cm,minimum height=0.40cm,inner sep=1pt,font=\scriptsize},
    reset/.style={draw=gray!68,fill=white,rounded corners=1pt,
      minimum width=0.68cm,minimum height=0.38cm,inner sep=1pt,font=\scriptsize},
    lab/.style={gray!62,font=\scriptsize}]
\draw[qwire] (0,0)--(11.8,0);
\draw[qwire] (0,-1)--(11.8,-1);
\node[left=4pt,gray!70] at (0,0) {$\ket{0}$};
\node[left=4pt,gray!70] at (0,-1) {$\ket{\psi}$};
\node[lab,left=18pt] at (0,0) {flag};
\node[lab,left=18pt] at (0,-1) {system};

\node[biggate] at (1.25,-0.5) {$\mathcal A_1$};
\node[meter] at (2.35,0) {$\mathsf{meas}$};
\node[lab] at (2.35,0.55) {accept $0$};
\node[reset] at (3.35,0) {$\mathsf{reset}$};

\node[biggate] at (4.45,-0.5) {$\mathcal A_2$};
\node[meter] at (5.55,0) {$\mathsf{meas}$};
\node[lab] at (5.55,0.55) {accept $0$};
\node[reset] at (6.55,0) {$\mathsf{reset}$};

\node[lab] at (7.35,-0.48) {$\cdots$};

\node[biggate] at (8.55,-0.5) {$\mathcal A_m$};
\node[meter] at (9.65,0) {$\mathsf{meas}$};
\node[lab] at (9.65,0.55) {accept $0$};
\draw[qwire,->] (10.15,-1)--(11.55,-1);
\node[lab,align=center] at (5.65,-1.72)
  {continue only after outcomes $0,0,\ldots,0$; after each success the flag is refreshed to $\ket{0}$};
\end{tikzpicture}%
}

\section{An overview}

This chapter introduces a small collection of quantum primitives that will be used throughout this lecture note. The intended reader is familiar with finite-dimensional linear algebra, matrix factorizations, finite differences, finite elements, and operator discretization, but is not assumed to know quantum mechanics.

The central message is simple. A quantum computer stores a vector in an exponentially large Hilbert space, but only through unitary transformations and measurements. Thus the relevant question is not merely whether a vector or a matrix can be represented, but whether it can be represented in a form that can be manipulated, transformed, and measured efficiently. The primitive that will play the role of a matrix access model is the \emph{block encoding}. Once an operator has been block encoded, the quantum singular value transformation (QSVT) gives a systematic way to apply polynomial functions to that operator. This is the quantum analogue of a familiar numerical-analysis idea: approximate a function of a matrix by a polynomial, then apply the resulting polynomial through repeated matrix-vector-like operations.  Standard background on the circuit model and the postulates of quantum mechanics can be found in Nielsen and Chuang \cite{NielsenChuang2010}; Lin and Wiebe's lecture notes provide a complementary introduction oriented toward scientific computing and matrix algorithms \cite{Lin2022QASC,LinWiebe2026QASC}.

The chapter is organized as follows.  We first introduce quantum states, computational-basis vectors, Pauli strings, tensor-product registers, and a short norm dictionary, followed by measurement, grid-function amplitude encoding, and state preparation.  The transform-and-estimate primitives come next: the quantum Fourier transform, quantum phase estimation, inner-product estimation, postselection, amplitude amplification, and amplitude estimation.  We then develop the matrix access model: block encodings and their calculus, diagonal and sparse-matrix access, QSVT, and explicit block encodings of elementary finite-difference operators.  A short section connects these primitives to PDE algorithms, illustrated by two worked examples---a Poisson solver and a Schr\"odinger evolution built on the same discrete Laplacian---then assembles the pieces, and a final section explains how to read query and gate-count estimates.

The notation introduced here is used consistently in the PDE chapters.  The amplitude encoding of grid vectors in \eqref{eq:amplitude-encoding-grid-vector} and the $L^2$--$\ell_2$ scaling in \eqref{eq:L2-l2-scaling} are the bridge between functions and quantum states.  The QFT definition in \eqref{eq:qft-definition} is the structured-grid spectral transform used later when periodic or tensor-product boundary conditions allow diagonalization. The block-encoding definition \eqref{eq:block-encoding-definition} is the matrix access model used in the elliptic, hyperbolic, and parabolic chapters.  The finite-difference block encodings in Section~\ref{sec:finite-difference-block-encodings}, especially \eqref{eq:1d-positive-laplacian-shift-form} and \eqref{eq:2d-laplacian-shift-form}, should be viewed as the first concrete examples of the operator encodings used later in Chapters~\ref{chap:elliptic-quantum}--\ref{chap:parabolic-quantum}. Throughout the lecture note, $\widetilde O(\cdot)$ suppresses logarithmic factors in dimension, precision, condition-number, and related parameters.

\begin{remark}
The presentation deliberately avoids physical language when possible. The word ``state'' should be read as ``unit-norm vector in a complex vector space.'' The word ``gate'' should be read as ``unitary matrix with a compact circuit representation.'' The word ``measurement'' should be read as ``randomized access to the squared moduli of coordinates in a specified basis.''
\end{remark}

\section{Quantum states as normalized vectors}

Let $n$ be a positive integer and set $N=2^n$. An $n$-qubit pure quantum state is a vector \cite{NielsenChuang2010,LinWiebe2026QASC}
\begin{equation}
  \ket{\psi}=\sum_{j=0}^{N-1}\psi_j\ket{j}\in \C^N,
  \qquad
  \sum_{j=0}^{N-1}|\psi_j|^2=1.
\end{equation}
The vectors $\ket{0},\ldots,\ket{N-1}$ naturally form the \emph{computational basis}. They are ordinary Euclidean basis vectors, but indexed by binary strings. If the integer $j$ has the binary representation
\begin{equation}
  j=(j_{n-1},\ldots,j_1,j_0)_2
  :=j_{n-1}2^{n-1}+\cdots+j_1 2+j_0,
  \qquad j_k\in\{0,1\},
\end{equation}
where the subscript $2$ denotes base two, then
\begin{equation}
  \ket{j}=\ket{j_{n-1}}\ot\cdots\ot\ket{j_1}\ot\ket{j_0}.
\end{equation}
The two basis states of one qubit are
\begin{equation}
  \ket{0}=\begin{bmatrix}1\\0\end{bmatrix},
  \qquad
  \ket{1}=\begin{bmatrix}0\\1\end{bmatrix}.
\end{equation}
Thus a quantum state is not mysterious: it is a unit vector. The difference from classical numerical linear algebra is operational. One does not read the vector entries directly. A measurement of $\ket{\psi}$ in the computational basis returns the integer $j$ with probability $|\psi_j|^2$.

\subsection{Gates as unitary matrices}

A quantum circuit typically applies a product of unitary matrices to a state. Common one-qubit gates include \cite{NielsenChuang2010,LinWiebe2026QASC}
\begin{equation}
  X=\begin{bmatrix}0&1\\1&0\end{bmatrix},
  \qquad
  Z=\begin{bmatrix}1&0\\0&-1\end{bmatrix},
  \qquad
  H=\frac{1}{\sqrt2}\begin{bmatrix}1&1\\1&-1\end{bmatrix}.
\end{equation}

It is often helpful to look at what a gate operation does to some specific states. For example,
the Hadamard gate satisfies
\begin{equation}
  H\ket{0}=\frac{\ket{0}+\ket{1}}{\sqrt2},
  \qquad
  H\ket{1}=\frac{\ket{0}-\ket{1}}{\sqrt2}.
\end{equation}

This can be generalized to multiple qubits.
Applying $H$ to every qubit gives the uniform superposition
\begin{equation}
  H^{\ot n}\ket{0^n}=\frac{1}{\sqrt{N}}\sum_{j=0}^{N-1}\ket{j}.
\end{equation}
This is the quantum analogue of initializing all grid points at once.

\subsection{Pauli matrices and Pauli strings}
\label{subsec:pauli-strings}

For quantum algorithms based on Hamiltonian simulation, it is useful to have a standard basis for matrices on qubits.  The one-qubit Pauli matrices are \cite{NielsenChuang2010}
\begin{equation}
  I=\begin{bmatrix}1&0\\0&1\end{bmatrix},\qquad
  X=\begin{bmatrix}0&1\\1&0\end{bmatrix},\qquad
  Y=\begin{bmatrix}0&-i\\i&0\end{bmatrix},\qquad
  Z=\begin{bmatrix}1&0\\0&-1\end{bmatrix}.
  \label{eq:pauli-matrices}
\end{equation}
They are Hermitian and unitary, and satisfy
\begin{equation}
  X^2=Y^2=Z^2=I,
  \qquad
  XY=iZ,
  \qquad
  YZ=iX,
  \qquad
  ZX=iY.
  \label{eq:pauli-products}
\end{equation}
Consequently,
\begin{equation}
  [X,Y]=2iZ,
  \qquad
  [Y,Z]=2iX,
  \qquad
  [Z,X]=2iY,
  \label{eq:pauli-commutators}
\end{equation}
and the nonidentity Pauli matrices $X,Y,Z$ pairwise anticommute, equivalently
\begin{equation}
  \{X,Y\}=\{Y,Z\}=\{Z,X\}=0.
\end{equation}  A Pauli string on $n$ qubits is a tensor product
\begin{equation}
  P_{\mu}=\sigma_{\mu_1}\ot\cdots\ot\sigma_{\mu_n},
  \qquad
  \sigma_{\mu_j}\in\{I,X,Y,Z\}.
  \label{eq:pauli-string-definition}
\end{equation}
Every Pauli string is Hermitian and unitary.  Two Pauli strings either commute or anticommute.  More precisely, their product differs from the reverse product by a sign determined by the number of tensor positions at which the nonidentity factors anticommute.

The $4^n$ Pauli strings form an orthogonal basis for all $2^n\times2^n$ complex matrices with respect to the Hilbert--Schmidt inner product:
\begin{equation}
  \operatorname{Tr}(P_\mu^\dag P_\nu)=2^n\delta_{\mu\nu}.
  \label{eq:pauli-hilbert-schmidt-orthogonality}
\end{equation}
Thus any matrix $H\in\C^{2^n\times2^n}$ can be expanded as
\begin{equation}
  H=\sum_{\mu\in\{I,X,Y,Z\}^n} h_\mu P_\mu,
  \qquad
  h_\mu=2^{-n}\operatorname{Tr}(P_\mu H).
  \label{eq:pauli-expansion}
\end{equation}
If $H$ is Hermitian, then all coefficients $h_\mu$ are real.  This expansion is conceptually simple, but it is not automatically efficient: a sparse matrix, such as the tridiagonal Laplacian, need not be sparse in the Pauli basis.

\begin{remark}[Commutators and product formulas]
Pauli commutation relations are the algebraic source of many Hamiltonian-simulation estimates.  A unitary evolution of a Pauli string can be implemented directly on the quantum circuit \cite{NielsenChuang2010}.  If $H=A+B$ has two non-commuting Pauli strings, a first-order product formula replaces $e^{-itH}$ by $e^{-itA}e^{-itB}$, and the leading error is controlled by commutators such as $[A,B]$.  When $A$ and $B$ are sums of Pauli strings, their commutators can be evaluated from the pairwise Pauli commutation rules above.
\end{remark}

\subsection{Tensor products and register notation}

If $\ket{\phi}\in\C^{N_1}$ and $\ket{\psi}\in\C^{N_2}$, then their tensor product is a vector in $\C^{N_1N_2}$:
\begin{equation}
  \ket{\phi}\ket{\psi}
  =\ket{\phi}\ot\ket{\psi}.
\end{equation}
For readability we will usually omit the tensor symbol. For example,
\begin{equation}
  \ket{0}\ket{j}=\ket{0}\ot\ket{j}.
\end{equation}
A group of qubits is called a \emph{register}. In numerical terms, a register is simply a tensor factor in the ambient vector space \cite{NielsenChuang2010,LinWiebe2026QASC}.

A recurring quantum-algorithmic construction is to enlarge the space containing a system state $\ket{\psi}\in\mathcal H_{\rm s}$ by appending an ancilla register initialized in a known state, usually $\ket{0^a}\in\mathcal H_{\rm a}$.  One then applies a joint unitary
\begin{equation}
  U_{\rm joint}:\mathcal H_{\rm a}\otimes\mathcal H_{\rm s}
  \longrightarrow
  \mathcal H_{\rm a}\otimes\mathcal H_{\rm s}.
  \label{eq:joint-register-unitary}
\end{equation}
The additional register may store a branch label, a success flag, an eigenvalue estimate, or temporary arithmetic.  Measuring or projecting the ancilla at the end extracts information about the system block.  Block encodings, phase estimation, the Hadamard test, and amplitude amplification are all instances of this general ``enlarge--evolve--read out'' pattern.

\begin{figure}[htbp]
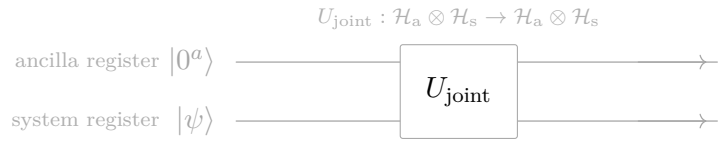

\centering
\JointUnitaryRegisters
\caption{A system register is elevated to the larger space $\mathcal H_{\rm a}\otimes\mathcal H_{\rm s}$ by appending an ancilla register, followed by a unitary that acts jointly on both registers.  Later circuits specialize the joint unitary and the final ancilla readout.}
\label{fig:joint-unitary-registers}
\end{figure}

\section{A short norm dictionary}
\label{sec:norm-dictionary}

Before introducing block encodings, we record several matrix norms that appear repeatedly in complexity estimates.  The operator norm, or spectral norm, is
\begin{equation}
  \norm{A}_2=\max_{\bm x\neq0}\frac{\norm{A\bm x}_2}{\norm{\bm x}_2}.
  \label{eq:spectral-norm}
\end{equation}
For Hermitian $A$, this is the largest absolute eigenvalue.  For a diagonal matrix $D$ with entries $d_0,\ldots,d_{N-1}$,
\begin{equation}
  \norm{D}_2=\max_j |d_j|.
  \label{eq:diagonal-spectral-norm}
\end{equation}
The entrywise maximum norm is also commonly used,
\begin{equation}
  \norm{A}_{\max}=\max_{i,j}|A_{ij}|.
  \label{eq:entrywise-max-norm}
\end{equation}
It is not equivalent to the spectral norm with dimension-independent constants.  For an $N\times N$ matrix,
\begin{equation}
  \norm{A}_{\max}\leq \norm{A}_2\leq N\norm{A}_{\max}.
  \label{eq:max-spectral-equivalence}
\end{equation}
The first inequality follows by testing on coordinate vectors; the second follows from the Frobenius norm.  A sharper and often more useful estimate is
\begin{equation}
  \norm{A}_2\leq \sqrt{\norm{A}_1\norm{A}_\infty},
  \label{eq:norm-1-infty-bound}
\end{equation}
where $\norm{A}_1$ is the maximum column sum and $\norm{A}_\infty$ is the maximum row sum.  In particular, if $A$ has at most $s$ nonzero entries in each row and column, then
\begin{equation}
  \norm{A}_2\leq s\norm{A}_{\max}.
  \label{eq:sparse-max-norm-bound}
\end{equation}

These elementary inequalities matter because a block encoding represents $A/\alpha$, not $A$ itself.  The normalization $\alpha$ must be at least of the order of a norm bound for $A$.  A crude bound, such as $N\norm{A}_{\max}$ for a sparse operator, may be much worse than the natural sparse bound $s\norm{A}_{\max}$ or the true spectral norm.  In PDE discretizations, tracking this normalization is as important as tracking the condition number.

\begin{remark}[Soft-O notation]
We write $\widetilde O(\cdot)$ to suppress factors that are logarithmic in the dimension, precision, condition number, or other explicitly named parameters.  This convention is common in quantum algorithms, but in PDE applications the hidden logarithms and state-preparation costs may still be important.
\end{remark}

\section{Measurement and observables}
\label{sec:measurement-observables}

A measurement is the main point at which quantum notation differs from standard matrix computation. A circuit may prepare or transform a vector in a high-dimensional space, but the final readout is random.  The projective-measurement rule used below is one of the basic postulates of quantum mechanics \cite[Sec.~2.2]{NielsenChuang2010}. The simplest case is measurement in the computational basis: for
\begin{equation}
  \ket{\psi}=\sum_{j=0}^{N-1}\psi_j\ket{j},
\end{equation}
the observed index is $j$ with probability $|\psi_j|^2$.

More generally, a directly measurable scalar quantity is represented by a Hermitian matrix, called an \emph{observable}. Let
\begin{equation}
  O=\sum_{r}\lambda_r\Pi_r
  \label{eq:observable-spectral-decomposition}
\end{equation}
be its spectral decomposition, where the eigenvalues $\lambda_r$ are distinct and $\Pi_r$ is the orthogonal projector onto the corresponding eigenspace. Measuring $O$ in the state $\ket{\psi}$ returns the eigenvalue $\lambda_r$ with probability
\begin{equation}
  p_r = \norm{\Pi_r\ket{\psi}}_2^2
      = \bra{\psi}\Pi_r\ket{\psi}.
  \label{eq:observable-measurement-probability}
\end{equation}
If $p_r>0$, the post-measurement state conditioned on observing $\lambda_r$ is
\begin{equation}
  \frac{\Pi_r\ket{\psi}}{\sqrt{p_r}}.
  \label{eq:post-measurement-state}
\end{equation}
Thus direct measurement samples from the spectral measure of $O$ induced by $\ket{\psi}$; it does not return the expectation value in one shot. The expectation is recovered only statistically:
\begin{equation}
  \mathbb E[\lambda]=\sum_r \lambda_r p_r
  = \bra{\psi}O\ket{\psi}.
  \label{eq:observable-expectation}
\end{equation}
This is often the most useful numerical interpretation: a quantum measurement gives samples of eigenvalues of the observable, with weights determined by the squared projected components of the state.

\begin{example}[Computational-basis measurement as an observable]
Computational-basis measurement is the special case in which
\begin{equation}
  O=\sum_{j=0}^{N-1} j\ket{j}\bra{j}.
\end{equation}
Then $\Pi_j=\ket{j}\bra{j}$ and \eqref{eq:observable-measurement-probability} gives $p_j=|\psi_j|^2$.
\end{example}

\begin{example}[Diagonal observables for grid quantities]
Let $q(x)$ be a real-valued grid observable and define
\begin{equation}
  O_q=\sum_{j=0}^{N-1}q(x_j)\ket{j}\bra{j}.
\end{equation}
Let $\ket{\bm f_h} \propto (f(x_0),\ldots,f(x_{N-1}))^T $ be a function defined on some grid points.
Measuring $O_q$ on $\ket{\bm f_h}$ samples the values $q(x_j)$ with probabilities proportional to $|f(x_j)|^2$. Its expectation is
\begin{equation}
  \bra{\bm f_h}O_q\ket{\bm f_h}
  =\frac{\sum_j q(x_j)|f(x_j)|^2}{\sum_j |f(x_j)|^2},
\end{equation}
which is a normalized weighted average. To recover the corresponding unnormalized continuum integral, one must multiply by the appropriate norm and quadrature factors.
\end{example}

\begin{remark}[Expectation estimation is a sampling problem]
If the eigenvalues of $O$ lie in $[a,b]$, then estimating $\bra{\psi}O\ket{\psi}$ by repeated direct measurements has Monte Carlo error proportional to $(b-a)/\sqrt{M}$ after $M$ samples, up to constants depending on the variance. More elaborate procedures, such as amplitude estimation or Hadamard-test-based methods, may improve the scaling under additional circuit assumptions.
\end{remark}

\section{Grid functions, amplitude encoding, and the $L^2$--$\ell_2$ scaling}

Quantum algorithms for PDEs usually begin after spatial discretization. Suppose $\Omega\subset\R^d$ is discretized by $N$ grid points $x_j$, $j=0,\ldots,N-1$. For simplicity assume a uniform mesh with volume element $h^d$. A function $f$ is represented by the grid vector
\begin{equation}
  \bm f_h=(f(x_0),\ldots,f(x_{N-1}))^T\in\C^N.
\end{equation}
A natural quantum representation of this grid vector is the amplitude encoding
\begin{equation}
  \ket{\bm f_h} = \frac{1}{\norm{\bm f_h}_2}\sum_{j=0}^{N-1} f(x_j)\ket{j},
  \qquad
  \norm{\bm f_h}_2=\left(\sum_{j=0}^{N-1}|f(x_j)|^2\right)^{1/2}.
  \label{eq:amplitude-encoding-grid-vector}
\end{equation}
This encoding is homogeneous: $\bm f_h$ and $c\bm f_h$ define the same quantum state up to a global phase when $c\neq0$. Therefore the scale $\norm{\bm f_h}_2$ must be stored or estimated separately if physical units or absolute magnitudes are needed.

\subsection{The norm conversion}

The continuum norm is approximated by the quadrature formula
\begin{equation}
  \norm{f}_{L^2(\Omega)}^2
  \approx
  h^d\sum_{j=0}^{N-1}|f(x_j)|^2
  = h^d\norm{\bm f_h}_2^2.
\end{equation}
Hence
\begin{equation}
  \norm{f}_{L^2(\Omega)}\approx h^{d/2}\norm{\bm f_h}_2.
  \label{eq:L2-l2-scaling}
\end{equation}
Equivalently, if $\ket{\bm f_h}$ is the normalized amplitude encoding in \eqref{eq:amplitude-encoding-grid-vector}, then its $j$th amplitude is approximately
\begin{equation}
  \frac{f(x_j)}{\norm{\bm f_h}_2}
  \approx
  \frac{h^{d/2}f(x_j)}{\norm{f}_{L^2(\Omega)}}.
\end{equation}
Thus amplitude encoding stores the $L^2$-normalized function values multiplied by the square root of the cell volume. This is the same scaling that appears when finite element coefficient vectors are related to $L^2$ inner products through a mass matrix; see, for example, standard finite-difference and finite-element treatments in \cite{LarssonThomee2009,TveitoWinther2005}.

\begin{remark}[Mass matrices]
For a finite element function $u_h(x)=\sum_j u_j\varphi_j(x)$, let
\begin{equation}
 M_{ij}=\int_\Omega\overline{\varphi_i(x)}\varphi_j(x)\,dx,
 \qquad \|u_h\|_{L^2}^2=\bm u^\dag M\bm u.
\end{equation}
The coefficient state is $\bm u/\|\bm u\|_2$, whereas the coefficients of
the $L^2$-normalized function are $\bm u/(\bm u^\dag M\bm u)^{1/2}$.
The mass-normalized coordinate is $\bm x=M^{1/2}\bm u$, for which
$\|\bm x\|_2=\|u_h\|_{L^2}$; its quantum state is
$M^{1/2}\bm u/(\bm u^\dag M\bm u)^{1/2}$. These representations should not
be identified. Even when $M=cI$, the normalized coefficient vector represents
a function of $L^2$ norm $\sqrt c$. Mass lumping simplifies the coordinate
transformation, but its diagonal weights must still be included.
Alternatively, retain coefficient coordinates and include $M$ in the observable.
\end{remark}

\subsection{Inner products of functions}

If $f$ and $g$ are represented by amplitude states $\ket{\bm f_h}$ and $\ket{\bm g_h}$, then
\begin{equation}
  \braket{\bm f_h}{\bm g_h}
  =\frac{\bm f_h^\dag\bm g_h}{\norm{\bm f_h}_2\norm{\bm g_h}_2}.
\end{equation}
The continuum inner product is approximated by
\begin{equation}
  \ip{f}{g}_{L^2}
  \approx
  h^d \bm f_h^\dag\bm g_h
  =
  \left(h^{d/2}\norm{\bm f_h}_2\right)
  \left(h^{d/2}\norm{\bm g_h}_2\right)
  \braket{\bm f_h}{\bm g_h}.
  \label{eq:inner-product-scaling}
\end{equation}
Equation \eqref{eq:inner-product-scaling} is often the first piece of bookkeeping needed in a PDE application. The quantum circuit estimates the normalized overlap; the numerical analyst must multiply back the discrete norms and quadrature weights.  This is the same bookkeeping that underlies finite-difference and finite-element error estimates in standard numerical PDE texts such as Larsson--Thom\'ee and Tveito--Winther \cite{LarssonThomee2009,TveitoWinther2005}.

\section{State preparation}
\label{sec:state-preparation}

A state preparation circuit for a vector $\bm v\in\C^N$ is a unitary $U_v$ such that
\begin{equation}
  U_{\bm v}\ket{0^n}=\ket{\bm v}:=\frac{1}{\norm{\bm v}_2}\sum_{j=0}^{N-1}v_j\ket{j}.
\end{equation}
State preparation can be regarded as a data-input problem. Quantum speedups can disappear if preparing $\ket{\bm v}$ costs $\Theta(N)$. For PDEs, we therefore distinguish between at least three cases:
\begin{enumerate}[label=(\roman*)]
  \item $\bm v$ has a compact formula, such as samples of a smooth function;
  \item $\bm v$ is generated by a previous quantum subroutine;
  \item $\bm v$ is arbitrary classical data and must be loaded from memory.
\end{enumerate}
Only the first two cases usually support efficient state preparation without a strong memory-access assumption.

\subsection{Grover--Rudolph preparation for smooth densities}

Grover and Rudolph introduced an efficient recursive procedure for preparing states associated with efficiently integrable probability distributions \cite{GroverRudolph2002}. We describe it in a form useful for grid functions.

Let $\rho:[0,1]\to\R_+$ be a probability density and assume that integrals of $\rho$ over dyadic intervals can be computed efficiently. Specifically,  let $N=2^n$ and define cells
\begin{equation}
  I_j=[j/N,(j+1)/N),\qquad
  p_j=\int_{I_j}\rho(x)\,dx.
\end{equation}
The target state is
\begin{equation}
  \ket{p}=\sum_{j=0}^{N-1}\sqrt{p_j}\ket{j}.
\end{equation}
which is automatically normalized. 
The construction is a binary tree (Figure~\ref{fig:grover-rudolph-tree}).  At each node, one rotates the next qubit according to the relative mass in the two children of the current interval.

\begin{figure}[htbp]
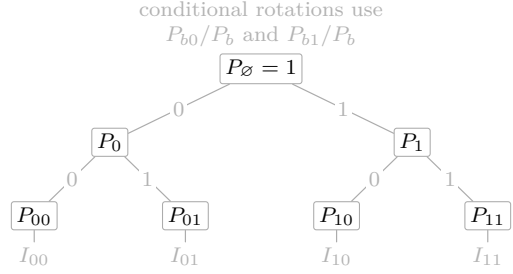

\centering
\GroverRudolphTree
\caption{The binary tree behind Grover--Rudolph state preparation.  The node weight $P_b$ is the integral of the density over the dyadic interval indexed by the bit string $b$.  The rotation at node $b$ places the next qubit in a superposition with probabilities $P_{b0}/P_b$ and $P_{b1}/P_b$.}
\label{fig:grover-rudolph-tree}
\end{figure}

Due to the dyadic partitioning, the interval $I_b$ indexed by a bit string $b=b_1\cdots b_m$, the strings $b0$ and $b1$ are obtained by appending one additional binary digit.  The corresponding intervals $I_{b0}$ and $I_{b1}$ are respectively the left and right halves of $I_b$.  Define
\begin{equation}
  P_b=\int_{I_b}\rho(x)\,dx,
  \qquad
  P_{b0}=\int_{I_{b0}}\rho(x)\,dx,
  \qquad
  P_{b1}=\int_{I_{b1}}\rho(x)\,dx.
  \label{eq:grover-rudolph-node-weights}
\end{equation}
The additivity identity $P_b=P_{b0}+P_{b1}$ is the local mass conservation property of the binary tree. The controlled rotation at node $b$ uses the angle $\theta_b$ satisfying
\begin{equation}
  \cos^2\theta_b=\frac{P_{b0}}{P_b},
  \qquad
  \sin^2\theta_b=\frac{P_{b1}}{P_b}.
  \label{eq:grover-rudolph-angle}
\end{equation}
After all $n$ levels, the amplitude of the leaf indexed by $j=j_1\cdots j_n$ is\footnote{Here bits are generated from the most significant downward, so $j_1$ is the most significant bit; in the register convention $\ket{j}=\ket{j_{n-1}}\ot\cdots\ot\ket{j_0}$ introduced earlier this is the bit $j_{n-1}$.}
\begin{equation}
  \prod_{m=0}^{n-1}\sqrt{\frac{P_{j_1\cdots j_{m+1}}}{P_{j_1\cdots j_m}}}
  =\sqrt{P_j}=\sqrt{p_j},
\end{equation}
where the empty string is used at $m=0$.  Thus the telescoping product of conditional probabilities gives the target amplitudes.

\begin{remark}[Preparing a function, not a density]
For a nonzero complex-valued function $f$, choosing
$\rho=|f|^2/\|f\|_{L^2}^2$ prepares
\[
 \sqrt{p_j}=\frac{\bigl(\int_{I_j}|f(x)|^2\,dx\bigr)^{1/2}}{\|f\|_{L^2}}.
\]
Applying $e^{i\arg f(x_j)}$ gives a cell-mass approximation to the sampled
function state. It is not exactly the amplitude encoding of $f(x_j)$;
the difference is a discretization error to be controlled. To prepare the
nodal state exactly, use $p_j=|f(x_j)|^2/\sum_k|f(x_k)|^2$ and assume
efficient coherent computation of its hierarchical partial sums. An
efficient value formula alone does not guarantee efficient partial sums.
At zero amplitudes the phase is immaterial. If $P_b=0$, the node's
controlled rotation is arbitrary because that branch has zero amplitude.
\end{remark}

\begin{remark}[Higher dimensions and bit ordering]
Use separate coordinate registers $\ket{j_1}\cdots\ket{j_d}$, with bits
within each coordinate generated from most to least significant.
For example, in two dimensions define
\[
 p_{ij}=\int_{I_i}\int_{J_j}\rho(x,y)\,dy\,dx,\qquad p_i=\sum_jp_{ij}.
\]
First prepare $\sum_i\sqrt{p_i}\ket{i}$, then prepare
$\sum_j\sqrt{p_{ij}/p_i}\ket{j}$ controlled on $i$. The joint state is
$\sum_{i,j}\sqrt{p_{ij}}\ket{i}\ket{j}$; zero-mass rows require no particular
conditional state. If $\rho(x,y)=\rho_x(x)\rho_y(y)$, the registers can be
prepared independently. Nonseparable densities require conditional masses,
including partial boxes at each bit step. Interleaving coordinate bits
is also valid, but changes the boxes and controls. Efficiency requires
coherent computation of these masses to the required accuracy.
\end{remark}

An alternative, more general state-preparation network was given by Kaye and Mosca \cite{KayeMosca2004}.  Their construction also uses conditional norms and phases associated with a binary decomposition. It is efficient for structured families for which these quantities can be computed efficiently, but it does not make arbitrary classical data inexpensive to load. For multivariate functions represented by Fourier or Chebyshev expansions,
state-preparation circuits can also be built by linearly combining
block encodings of the basis functions
\cite{RosenkranzEtAl2025MultivariateStatePreparation}.

\section{The quantum Fourier transform}
\label{sec:qft-iqft}

The quantum Fourier transform (QFT) is the unitary version of the discrete Fourier transform acting on amplitudes.  Let $N=2^n$ and let $\omega_N=e^{2\pi i/N}$.  The QFT is defined by the following unitary; see Nielsen and Chuang \cite[Sec.~5.1]{NielsenChuang2010} for detailed derivations:
\begin{equation}
  F_N\ket{j}
  =\frac{1}{\sqrt N}\sum_{k=0}^{N-1}\omega_N^{jk}\ket{k},
  \qquad j=0,\ldots,N-1.
  \label{eq:qft-definition}
\end{equation}
The inverse QFT is $F_N^\dag$, obtained by replacing $\omega_N$ by $\omega_N^{-1}$:
\begin{equation}
  F_N^\dag\ket{k}
  =\frac{1}{\sqrt N}\sum_{j=0}^{N-1}\omega_N^{-jk}\ket{j}.
\end{equation}
If a grid vector is amplitude encoded as $\ket{\bm u}=\sum_j u_j\ket{j}/\norm{u}$, then applying $F_N$ produces the normalized vector of discrete Fourier coefficients.  This is why QFT diagonalizes periodic finite-difference operators in later chapters.

The circuit follows from the binary factorization of a Fourier basis state.  With the binary fraction
\begin{equation}
  0.j_\ell j_{\ell-1}\cdots j_0
  :=\frac{j_\ell}{2}+\frac{j_{\ell-1}}{2^2}+\cdots+\frac{j_0}{2^{\ell+1}},
\end{equation}
one obtains, up to reversal of the output-qubit order,
\begin{equation}
  F_N\ket{j_{n-1}\cdots j_0}
  =\frac{1}{2^{n/2}}
  \bigotimes_{\ell=0}^{n-1}
  \left(\ket{0}+e^{2\pi i(0.j_\ell\cdots j_0)}\ket{1}\right).
  \label{eq:qft-binary-factorization}
\end{equation}
Each factor is generated by one Hadamard gate followed by controlled phase rotations; the final SWAP gates restore the conventional bit order; Figure~\ref{fig:qft-circuit} shows the four-qubit circuit.  

\begin{figure}[htbp]
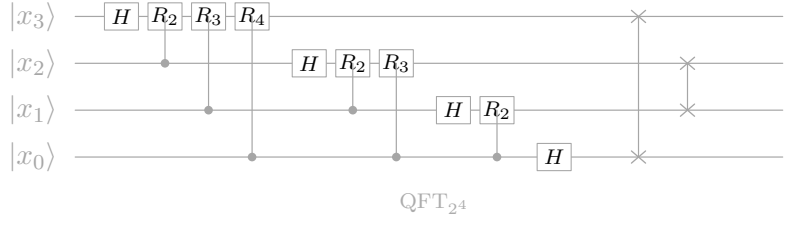

\centering
\QFTFourCircuit
\caption{A QFT circuit on four qubits.  Here $R_m=\operatorname{diag}(1,e^{2\pi i/2^m})$.  The final swaps implement bit reversal.  The inverse QFT reverses the gate order and replaces each $R_m$ by $R_m^\dag$.}
\label{fig:qft-circuit}
\end{figure}

The classical fast Fourier transform maps an explicitly stored vector in $\C^N$ to another explicitly stored vector in $O(N\log N)$ arithmetic operations.  The QFT instead maps an amplitude-encoded vector on $n=\log_2 N$ qubits to another amplitude-encoded vector using $O(n^2)$ one- and two-qubit gates for the exact circuit, and fewer gates for standard approximate versions \cite{Coppersmith2002,NielsenChuang2010}.  The apparent exponential gap should be interpreted carefully: the QFT does not print all Fourier coefficients.  It prepares a state whose amplitudes are those coefficients, and subsequent measurements or controlled operations can access only selected information.

The circuit structure is nevertheless important.  Rather than updating $N$ Fourier coefficients one at a time, the gates act coherently on all amplitudes of the quantum state.  The controlled phase gates are diagonal and many commute, while gates on disjoint qubit pairs can be scheduled in parallel.  Hence the exact circuit uses $O(n^2)$ gates, and its depth can be reduced further subject to connectivity and approximation tolerance.  The saving is therefore a combination of coherent amplitude processing and circuit parallelism, not the classical production of $N$ output numbers.  This is one of the reasons QFT-based methods are attractive for structured grids, where the PDE operator is diagonal or nearly diagonal in the Fourier basis.

\begin{remark}[QFT versus FFT]
For numerical analysts, the safest analogy is the following.  The FFT is an efficient explicit change of basis for a stored vector.  The QFT is an efficient coherent change of basis for an amplitude-encoded vector.  It is extremely powerful when a quantum algorithm needs to apply phases, filters, or measurements in the frequency basis; it is not a replacement for the FFT when the task is to output the full list of Fourier coefficients as classical data.
\end{remark}

\section{Quantum phase estimation}
\label{sec:quantum-phase-estimation}

Quantum phase estimation (QPE) is the standard primitive for converting an
eigenphase into a binary register.  It is the main reason the inverse QFT was
introduced in the previous section.  Suppose that a unitary $U$ satisfies
\begin{equation}
  U\ket{q_j}=e^{2\pi i\theta_j}\ket{q_j},
  \qquad 0\leq \theta_j<1.
\end{equation}
QPE uses an $r$-qubit clock register, controlled powers
$U,U^2,U^4,\ldots,U^{2^{r-1}}$, and an inverse QFT on the clock register.  On
an exact eigenstate, the idealized map is
\begin{equation}
  \ket{0^r}\ket{q_j}
  \longmapsto
  \ket{\widetilde \theta_j}\ket{q_j},
  \label{eq:qpe-eigenstate-map}
\end{equation}
where $\widetilde \theta_j$ is an $r$-bit approximation of $\theta_j$.  More
precisely, $r$ clock qubits resolve phases to accuracy $O(2^{-r})$, with the
usual constant-success-probability qualification; additional qubits or median
amplification improve the success probability \cite[Sec.~5.2]{NielsenChuang2010}.

The real power of QPE is that it acts coherently on superpositions.  If
\begin{equation}
  \ket{\psi}=\sum_j c_j\ket{q_j},
\end{equation}
then QPE produces
\begin{equation}
  \ket{0^r}\ket{\psi}
  \longmapsto
  \sum_j c_j\ket{\widetilde \theta_j}\ket{q_j},
  \label{eq:qpe-superposition-map0}
\end{equation}
up to the phase-estimation error.  Thus QPE does not merely return one
eigenvalue; it attaches approximate spectral information to every eigenmode
present in the input state.  This viewpoint is central to eigenvalue problems,
linear-system algorithms, and amplitude estimation.  The latter applies QPE
to the two-dimensional Grover rotation introduced in
Subsection~\ref{subsec:amplitude-estimation}.  In Chapter~\ref{chap:elliptic-quantum},
we will apply the same idea to elliptic eigenvalue problems.

For Hamiltonian simulation, one solves the Schr\"odinger equation,
\begin{equation}
  U=e^{-iHt_0}.
\end{equation}
If $H\ket{q_j}=\lambda_j\ket{q_j}$, then the measured phase is
$\theta_j=-t_0\lambda_j/(2\pi)$ modulo one.  The modulo-one ambiguity is
important: either $t_0$ must be small enough to avoid aliasing on the spectral
window of interest, or additional information must be used to unwrap the
phase.  This mesh-dependent aliasing issue will appear again for discretized
PDE operators, whose largest eigenvalues often grow like a power of $h^{-1}$.

More advanced variants, including iterative, adaptive, and fixed-point phase
estimation, reduce ancilla counts or improve robustness.  We will not need
those refinements in this introductory chapter, but they are important in
resource-sensitive implementations.

\section{Estimating inner products}
\label{sec:inner-product-tests}

Many quantities of interest in numerical PDEs are inner products: masses, energies, weak-form residuals, correlations, and linear functionals. Quantum circuits provide several ways to estimate such quantities. We record the two most common primitives.

\subsection{The swap test}

Given two states $\ket{\bm u}$ and $\ket{\bm v}$, the swap test estimates $|\braket{\bm u}{\bm v}|^2$. It uses one ancilla qubit and a controlled-SWAP operation (Figure~\ref{fig:swap-test-circuit}).  It can also be viewed as a Hadamard test whose controlled unitary is the register-swap operator.  The test appears naturally in quantum fingerprinting and is a standard primitive for estimating fidelities and overlaps \cite{BuhrmanCleveWatrousDeWolf2001,NielsenChuang2010,LinWiebe2026QASC}.

\begin{figure}[htbp]
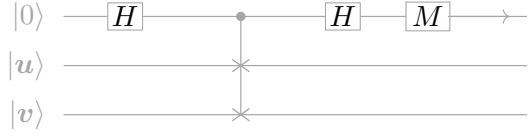

\centering
\SwapTestCircuit
\caption{Swap test for estimating $|\braket{\bm u}{\bm v}|^2$.  The controlled operation swaps the two input registers only when the ancilla is $\ket{1}$, and the box $M$ denotes computational-basis measurement of the ancilla.}
\label{fig:swap-test-circuit}
\end{figure}

We will show a sequence of states to make the construction transparent.  We attach an ancilla register and starting from 
\begin{equation}
  \ket{\Psi_0}=\ket{0}\ket{\bm u}\ket{\bm v},
\end{equation}
the first Hadamard gate gives
\begin{equation}
  \ket{\Psi_1}
  =\frac{1}{\sqrt2}
  \left(\ket{0}\ket{\bm u}\ket{\bm v}
  +\ket{1}\ket{\bm u}\ket{\bm v}\right).
\end{equation}
The controlled-SWAP acts only on the second branch,
\begin{equation}
  \ket{\Psi_2}
  =\frac{1}{\sqrt2}
  \left(\ket{0}\ket{\bm u}\ket{\bm v}
  +\ket{1}\ket{\bm v}\ket{\bm u}\right).
\end{equation}
After the second Hadamard gate on the ancilla,
\begin{align}
  \ket{\Psi_3}
  ={}&\frac12\ket{0}
  \left(\ket{\bm u}\ket{\bm v}+\ket{\bm v}\ket{\bm u}\right) \notag\\
  &+\frac12\ket{1}
  \left(\ket{\bm u}\ket{\bm v}-\ket{\bm v}\ket{\bm u}\right).
  \label{eq:swap-test-final-state}
\end{align}
Therefore
\begin{align}
  \Pr(\text{ancilla}=0)
  &=\frac14
  \left\|\ket{\bm u}\ket{\bm v}+\ket{\bm v}\ket{\bm u}\right\|^2 \notag\\
  &=\frac14\left(
  2+\braket{\bm u}{\bm v}\braket{\bm v}{\bm u}
   +\braket{\bm v}{\bm u}\braket{\bm u}{\bm v}\right) \notag\\
  &=\frac{1+|\braket{\bm u}{\bm v}|^2}{2}.
  \label{eq:swap-test-probability}
\end{align}
This calculation follows the step-by-step presentation in \cite{Lin2022QASC} and is also consistent with the standard circuit treatment in Nielsen and Chuang \cite{NielsenChuang2010,LinWiebe2026QASC}.  Repeating the test estimates $|\braket{\bm u}{\bm v}|^2$ by sampling.  The swap test is most useful when the magnitude of an overlap is enough, for example in fidelity estimation.

\begin{remark}[What the swap test does and does not give]
The swap test gives the squared magnitude of the inner product. It does not give the sign or phase of $\braket{\bm u}{\bm v}$. If the real or imaginary part is needed and one has access to state-preparation unitaries $U_u$ and $U_v$, one typically uses a Hadamard-test variant.
\end{remark}

\subsection{The Hadamard test}

Suppose $W$ is a unitary and $\ket{\psi}$ is a state. The Hadamard test estimates $\Real\bra{\psi}W\ket{\psi}$ or $\Imag\bra{\psi}W\ket{\psi}$ by controlling $W$ with one ancilla \cite{Lin2022QASC}.  It is a concrete instance of the enlarged-register construction in Figure~\ref{fig:joint-unitary-registers}: the ancilla and system registers undergo the joint controlled unitary
\begin{equation}
  \ket{0}\bra{0}\otimes I+\ket{1}\bra{1}\otimes W.
\end{equation}

\begin{figure}[htbp]
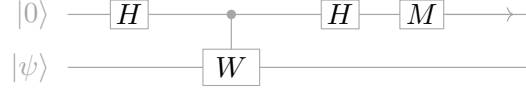

\centering
\HadamardTestCircuit
\caption{Hadamard test for estimating $\Real\bra{\psi}W\ket{\psi}$.  The box $M$ denotes computational-basis measurement of the ancilla.  To estimate the imaginary part, insert $S^\dag$ before the final Hadamard, where $S=\operatorname{diag}(1,i)$ is the phase gate.}
\label{fig:hadamard-test-circuit}
\end{figure}

For the real part, the state evolves as
\begin{align}
  \ket{0}\ket{\psi}
  &\xmapsto{H\otimes I}
  \frac{1}{\sqrt2}(\ket{0}+\ket{1})\ket{\psi} \notag\\
  &\xmapsto{\text{controlled-}W}
  \frac{1}{\sqrt2}
  \left(\ket{0}\ket{\psi}+\ket{1}W\ket{\psi}\right) \notag\\
  &\xmapsto{H\otimes I}
  \frac12\ket{0}(\ket{\psi}+W\ket{\psi})
  +\frac12\ket{1}(\ket{\psi}-W\ket{\psi}).
  \label{eq:hadamard-test-state-sequence}
\end{align}
Consequently, after measuring the ancilla, the outcome follows 
\begin{align}
  \Pr(0)
  &=\frac14\left\|\ket{\psi}+W\ket{\psi}\right\|^2
   =\frac12\left(1+\Real\bra{\psi}W\ket{\psi}\right),\\
  \Pr(1)
  &=\frac12\left(1-\Real\bra{\psi}W\ket{\psi}\right),
\end{align}
and hence
\begin{equation}
  \Pr(0)-\Pr(1)=\Real\bra{\psi}W\ket{\psi}.
  \label{eq:hadamard-test-real}
\end{equation}

For the imaginary part, apply $S^\dag$ to the ancilla before the final Hadamard, where
\begin{equation}
  S=\begin{bmatrix}1&0\\0&i\end{bmatrix},
  \qquad
  S^\dag=\begin{bmatrix}1&0\\0&-i\end{bmatrix}
  \label{eq:S-phase-gate}
\end{equation}
is the phase gate.  With this convention, the final state becomes
\begin{equation}
  \frac12\ket{0}(\ket{\psi}-iW\ket{\psi})
  +\frac12\ket{1}(\ket{\psi}+iW\ket{\psi}),
\end{equation}
so that
\begin{equation}
  \Pr(0)=\frac12\left(1+\Imag\bra{\psi}W\ket{\psi}\right).
  \label{eq:hadamard-test-imaginary}
\end{equation}

If $U_u\ket{0^n}=\ket{\bm u}$ and $U_v\ket{0^n}=\ket{\bm v}$, then taking $W=U_u^\dag U_v$ gives
\begin{equation}
  \bra{0^n}U_u^\dag U_v\ket{0^n}=\braket{\bm u}{\bm v}.
\end{equation}
Thus the real and imaginary Hadamard tests recover the signed complex inner product, provided controlled versions of the state-preparation circuits are available.

\section{Postselection and amplitude amplification}
\label{sec:postselection-amplification}

Many quantum algorithms produce the desired state only after a favorable measurement outcome. This mechanism is called \emph{postselection}. It is simple but important enough to isolate.

\subsection{Postselection}

Suppose a circuit prepares
\begin{equation}
  \mathcal A\ket{0}
  =
  \sqrt{p}\ket{0}\ket{\phi}
  +
  \sqrt{1-p}\ket{1}\ket{\phi_\perp},
  \label{eq:postselection-decomposition}
\end{equation}
where the first register is an ancilla flag. If we measure the flag and keep only the runs in which the outcome is $0$, then the remaining register is in the state $\ket{\phi}$. The price is that this happens with probability $p$. Repeating the experiment until success requires $O(1/p)$ trials on average.

A subtlety that has no exact classical analogue is that postselection can be costly when it is repeated many times.  Suppose a computation applies $m$ nonunitary steps and postselects after each one, with success probabilities $p_1,\ldots,p_m$.  The probability that the whole run survives is
\begin{equation}
  p_{\rm tot}=\prod_{r=1}^m p_r.
  \label{eq:repeated-postselection-product}
\end{equation}
If $p_r\leq p<1$, then $p_{\rm tot}\leq p^m$, which is exponentially small in the number of postselections.  Even when the measurements are weak, say $p_r=1-\delta_r$, one has the approximation
\begin{equation}
  \prod_{r=1}^m(1-\delta_r)\approx \exp\left(-\sum_{r=1}^m\delta_r\right).
  \label{eq:weak-measurement-product}
\end{equation}
Thus repeatedly applying a small nonunitary map and postselecting at every step, as pictured in Figure~\ref{fig:repeated-postselection-circuit}, can destroy the success probability.  This is why many quantum algorithms try to keep the full process coherent.

\begin{figure}[t]
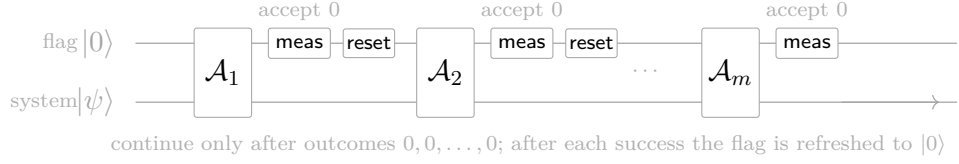

\centering
\RepeatedPostselectionCircuit
\caption{Repeated postselection with a reusable flag ancilla.  Each box labeled $\mathsf{meas}$ is a mid-circuit computational-basis measurement; ``accept $0$'' means that the run continues only when the outcome is $0$.  After a successful measurement, the flag is refreshed to $\ket{0}$ before the next stage.  The total success probability is the product in \eqref{eq:repeated-postselection-product}.}
\label{fig:repeated-postselection-circuit}
\end{figure}

This operation is analogous to conditioning a random variable on an event, but it is not merely a matter of data analysis after the fact. The unsuccessful quantum runs are discarded, and the state available for subsequent computation is the normalized conditional state. Therefore a postselected subroutine with small success probability can dominate the cost of an algorithm.

\begin{remark}[Postselection changes normalization]
Before measuring the flag, the desired component in \eqref{eq:postselection-decomposition} has norm $\sqrt p$. After successful postselection, it is renormalized to unit norm. This renormalization is useful because it produces a valid quantum state, but it also hides the magnitude $p$, which must be tracked or estimated separately when it carries physical information.
\end{remark}

\subsection{Amplitude amplification}

A typical situation can be written more abstractly as
\begin{equation}
  \mathcal A\ket{0}
  =
  \sqrt{p}\ket{\mathrm{good}}\ket{\phi}
  +
  \sqrt{1-p}\ket{\mathrm{bad}}\ket{\phi_\perp},
  \label{eq:good-bad-decomposition}
\end{equation}
where $p$ is the success probability. Classical repetition needs $O(1/p)$ attempts. Amplitude amplification boosts the success probability using $O(1/\sqrt p)$ applications of $\mathcal A$ and $\mathcal A^\dag$ \cite{BrassardHoyerMoscaTapp2000}.

Let $\sin^2\theta=p$, $0\leq\theta\leq\pi/2$. Define two reflections: one about the initial state and one that changes the sign of the good subspace. Their product acts as a rotation by angle $2\theta$ in the two-dimensional subspace generated by the good and bad components in \eqref{eq:good-bad-decomposition}. After $k$ amplification steps, the success probability becomes
\begin{equation}
  \sin^2((2k+1)\theta).
\end{equation}
Choosing $k\approx \pi/(4\theta)$ makes the success probability close to one when $p$ is known. When $p$ is unknown, variants based on randomized iteration counts or amplitude estimation are used.

\subsection{Quantities of interest and recovering the missing norm}
\label{subsec:qoi-and-output-norm}

Let $P_{\rm qoi}$ be a projector representing a quantity of interest, such as membership in a subdomain or an event in a sampling problem. If
\begin{equation}
  p_{\rm qoi}=\bra{\psi}P_{\rm qoi}\ket{\psi},
\end{equation}
then estimating $p_{\rm qoi}$ is equivalent to estimating a measurement probability. Standard sampling gives error $O(M^{-1/2})$ after $M$ shots. Amplitude estimation can reduce the error to $O(R^{-1})$ after $R$ coherent uses of the underlying circuit, under stronger circuit assumptions. In PDE applications this matters when $p_{\rm qoi}$ is a physically meaningful scalar output, such as probability mass in a region, a normalized energy fraction, or a success probability associated with a linear-system solver.

The original amplitude-amplification and amplitude-estimation framework of Brassard, H\o yer, Mosca, and Tapp gives the basic quadratic improvement from sampling error $O(M^{-1/2})$ after $M$ independent shots to coherent error scaling $O(R^{-1})$ after $R$ circuit uses when a scalar quantity has been encoded as an amplitude \cite{BrassardHoyerMoscaTapp2000}.  Rall's block-encoding framework for estimating physical quantities gives a useful modern interpretation of the same idea: if an observable, correlation function, density-of-states quantity, or matrix element can be written as a block-encoded matrix element, then a Hadamard test or related overlap circuit turns it into an amplitude, and amplitude estimation gives the quadratic improvement in the target scalar accuracy \cite{Rall2020PhysicalQuantities}.  This is the viewpoint used later in the PDE chapters when we estimate observables directly instead of first reconstructing a full solution vector.

There is a second issue that is especially important for linear-system and ODE algorithms.  These algorithms often prepare a normalized vector
\begin{equation}
  \ket{\bm v}=\frac{\bm v}{\norm{\bm v}_2},
  \label{eq:normalized-output-state}
\end{equation}
rather than the unnormalized vector $\bm v$ itself.  If the desired scalar is a physical amplitude, norm, or energy depending on $\bm v$, then the missing factor $\norm{\bm v}_2$ must be recovered separately.

A common situation is that a circuit prepares a flagged state
\begin{equation}
  U\ket{0^a}\ket{\bm b}
  =
  \ket{0^a}\frac{\bm v}{\alpha}
  +\ket{\Phi_\perp},
  \qquad
  (\bra{0^a}\otimes I)\ket{\Phi_\perp}=0.
  \label{eq:flagged-unnormalized-vector}
\end{equation}
Measuring the ancilla and observing $0^a$ succeeds with probability
\begin{equation}
  p_{\rm succ}=\frac{\norm{\bm v}_2^2}{\alpha^2}.
  \label{eq:succ-prob-output-norm}
\end{equation}
Therefore
\begin{equation}
  \norm{\bm v}_2=\alpha\sqrt{p_{\rm succ}}.
  \label{eq:recover-output-norm}
\end{equation}
The success probability can be estimated by sampling or amplitude estimation.  Once $\norm{\bm v}_2$ is known, normalized measurements can be converted back to unnormalized physical quantities.  For example, if $O$ is an observable on the system register, then
\begin{equation}
  \bm v^\dag O \bm v
  =
  \norm{\bm v}_2^2\,
  \bra{\bm v}O\ket{\bm v}.
  \label{eq:unnormalized-observable-recovery}
\end{equation}
This separation between preparing the direction $\bm v/\norm{\bm v}$ and estimating the amplitude $\norm{\bm v}$ is often hidden in high-level descriptions of quantum linear solvers.  For PDE applications, it should be made explicit because physical quantities such as $L^2$ norms, energies, fluxes, and residuals are not invariant under normalization.

\subsection{Amplitude estimation}
\label{subsec:amplitude-estimation}

Postselection and amplitude amplification treat the success probability $p$ in \eqref{eq:good-bad-decomposition} as an obstacle. In many PDE applications, however, $p$ \emph{is} the answer: by \eqref{eq:recover-output-norm}, success probabilities encode output norms, and projector expectations $p_\Pi=\bra{\psi}\Pi\ket{\psi}$ encode quantities of interest such as probability mass in a subdomain. Amplitude estimation is the primitive that estimates such a $p$ with quadratically fewer circuit uses than direct sampling \cite{BrassardHoyerMoscaTapp2000}.

Write $p=\sin^2\theta$ as in amplitude amplification. The product of the two reflections used there is a rotation by $2\theta$ in the two-dimensional good/bad subspace, hence a unitary with eigenvalues $e^{\pm 2i\theta}$. Applying quantum phase estimation, as described in Section~\ref{sec:quantum-phase-estimation}, to this rotation, with the state $\mathcal A\ket{0}$ as input, returns an estimate of $\theta$. Using $R$ controlled applications of $\mathcal A$ and $\mathcal A^\dag$, the estimate $\tilde p=\sin^2\tilde\theta$ satisfies
\begin{equation}
  |\tilde p-p|
  \;=\;
  O\!\left(\frac{\sqrt{p(1-p)}}{R}+\frac{1}{R^2}\right)
  \label{eq:amplitude-estimation-error}
\end{equation}
with constant success probability. Direct sampling with $M$ shots gives error $O(\sqrt{p(1-p)/M})$; matching a target additive error $\epsilon$ therefore costs $M=O(\epsilon^{-2})$ samples classically but only $R=O(\epsilon^{-1})$ coherent uses of $\mathcal A$. This is the quadratic saving invoked when we say that a scalar output can be estimated to accuracy $\epsilon$ with $\widetilde O(\epsilon^{-1})$ queries.

\begin{remark}[What amplitude estimation costs]
The improvement is not free. The $R$ applications of $\mathcal A$ occur \emph{coherently} inside one long circuit, so the circuit depth grows by the factor $R$, and controlled versions of $\mathcal A$ are required. When $\mathcal A$ itself contains a full PDE solve, this multiplies the solver depth by $\epsilon^{-1}$. Whether amplitude estimation or simple sampling is preferable is therefore a balance between total gate count and circuit depth, the same kind of trade-off as between a single long time integration and an ensemble of short ones.
\end{remark}

\section{Block encodings}

Classical matrix--vector multiplication takes an explicitly stored vector
$\bm x$ and returns the entries of $A\bm x$. A quantum block encoding has a
different input and output model: it acts on an amplitude-encoded state
$\ket{\bm x}=\bm x/\|\bm x\|_2$ and embeds the scaled action of $A$ in a
specified ancilla block of a larger unitary. One use of this unitary does not
return the entries of $A\bm x$, nor does it generally prepare
$A\bm x/\|A\bm x\|_2$ deterministically. Extracting that state requires
postselection or amplitude amplification, with a cost depending on the
normalization and on $\|A\bm x\|_2/\|\bm x\|_2$. Alternatively, later
subroutines can use the block encoding coherently without measuring this
ancilla. Thus a block encoding is a composable matrix-access primitive;
comparisons with classical multiplication must also account for input
preparation and the requested output. See
\cite{GilyenSuLowWiebe2019,LinWiebe2026QASC} for detailed constructions.

\begin{definition}[Block encoding]
Let $A\in\C^{N\times N}$. A unitary $U$ acting on $a+n$ qubits is an $(\alpha,a,\eps)$ block encoding of $A$ if
\begin{equation}
  \left\|A-
  \alpha(\bra{0^a}\otimes I_N)U(\ket{0^a}\otimes I_N)
  \right\|\leq \eps.
  \label{eq:block-encoding-definition}
\end{equation}
When $\eps=0$, we say the block encoding is exact. In block matrix form, exactness means
\begin{equation}
  U=\begin{bmatrix}
  A/\alpha & *\\
  * & *
  \end{bmatrix},
\end{equation}
with respect to the decomposition in which the ancilla register is $\ket{0^a}$ or orthogonal to $\ket{0^a}$.
\end{definition}

The positive scalar $\alpha$ is called the \emph{subnormalization factor}.
For an exact block encoding, $\alpha\geq\|A\|_2$. A smaller factor improves
the success amplitude, although its implementation cost must also be counted.
For $\bm x\ne0$, the exact successful component of
$U\ket{0^a}\ket{\bm x}$ is
\begin{equation}
  \ket{0^a}\frac{A\bm x}{\alpha\|\bm x\|_2}.
\end{equation}
For an approximate encoding this component has vector error at most
$\eps/\alpha$. In the exact case, postselection on $0^a$ succeeds with probability
\begin{equation}
  p=\frac{\norm{A\bm x}_2^2}{\alpha^2\norm{\bm x}_2^2}.
  \label{eq:block-postselection-probability}
\end{equation}
Thus the normalization $\alpha$ directly affects the success probability and the complexity of amplitude amplification.
In this sense, a block encoding should be viewed as a unitary implementation of a scaled matrix action together with a postselection flag. Without postselection, the full output is still a valid quantum state in the enlarged space; after successful postselection, the system register contains the normalized vector $A\bm x/\norm{A\bm x}_2$.

\begin{remark}[Block encoding as a matrix access model]
A block encoding is not merely a storage format. It is a promise that $A/\alpha$ appears as a submatrix of a unitary that can be implemented efficiently. Thus the real input model is the circuit for $U$, not the abstract existence of a unitary dilation.
\end{remark}

\subsection{A small calculus of block encodings}
\label{subsec:block-encoding-calculus}

Block encodings can be combined in much the same way that matrices are combined in numerical linear algebra.  We record two rules that will be used repeatedly below; for more details, 
see \cite{GilyenSuLowWiebe2019,LinWiebe2026QASC}.

First, if $U_A$ is an $(\alpha,a,\epsilon_A)$ block encoding of $A$ and $U_B$ is a $(\beta,b,\epsilon_B)$ block encoding of $B$, then, after placing the ancillas in separate registers, one obtains a block encoding of $AB$ with normalization $\alpha\beta$.  More explicitly, on the register order
\[
  \mathcal H_{a}\otimes\mathcal H_{b}\otimes\mathcal H_{\rm sys},
\]
let $\widetilde U_A$ denote $U_A$ acting on the $a$-ancilla register and the system register, with identity on the $b$-ancilla register, and let $\widetilde U_B$ denote $U_B$ acting on the $b$-ancilla register and the same system register, with identity on the $a$-ancilla register.  In the exact case,
\begin{equation}
  (\bra{0^{a+b}}\otimes I)\widetilde U_A\widetilde U_B(\ket{0^{a+b}}\otimes I)
  =\frac{AB}{\alpha\beta}.
  \label{eq:block-encoding-product-rule}
\end{equation}
With errors, the block-encoding error is at most
\begin{equation}
  \epsilon_{AB}\leq \alpha\epsilon_B+\beta\epsilon_A+\epsilon_A\epsilon_B.
\end{equation}
This product rule is the quantum analogue of composing two matrix-vector operations, but the normalization factors multiply.

Second, if one has block encodings of $A_0,\ldots,A_{s-1}$, then the (linear combination of unitaries) LCU construction below block encodes a weighted sum
\begin{equation}
  A=\sum_{\ell=0}^{s-1}a_\ell A_\ell
\end{equation}
with normalization proportional to $\sum_\ell |a_\ell|\alpha_\ell$, where $\alpha_\ell$ is the normalization of the $\ell$th block encoding.  Thus the block-encoding calculus looks familiar algebraically, but the normalizations play the role of hidden condition and success-probability parameters.

\subsection{Linear combination of unitaries}

A  way to build block encodings is the linear-combination-of-unitaries
(LCU) method. We first specialize to unitary summands, whose individual
block-encoding normalizations are one. Suppose
\begin{equation}
  A=\sum_{\ell=0}^{s-1} a_\ell U_\ell,
  \qquad
  a_\ell\geq0,
\end{equation}
where each $U_\ell$ is unitary. Let
\begin{equation}
  \alpha=\sum_{\ell=0}^{s-1}a_\ell.
\end{equation}
Let $m=\lceil\log_2 s\rceil$, and pad the selector space to dimension $2^m$ if necessary. Define a preparation unitary $P$ by
\begin{equation}
  P\ket{0^m}=\sum_{\ell=0}^{s-1}\sqrt{a_\ell/\alpha}\ket{\ell},
\end{equation}
with zero amplitude on the padded selector labels. Define the unitary
\begin{equation}
  \operatorname{SELECT}(U)=
  \sum_{\ell=0}^{s-1}\ket{\ell}\bra{\ell}\otimes U_\ell
  +
  \sum_{\ell=s}^{2^m-1}\ket{\ell}\bra{\ell}\otimes I.
  \label{eq:lcu-select-padded}
\end{equation}
The second sum acts on the padded selector labels and is needed for unitarity; those labels have zero amplitude after PREPARE. In the selector basis, this is simply a block diagonal matrix,
\begin{equation}
  \operatorname{SELECT}(U)
  =
  \begin{bmatrix}
  U_0 &        &        &        &        \\
      & U_1    &        &        &        \\
      &        & \ddots &        &        \\
      &        &        & U_{s-1} &        \\
      &        &        &        & I_{\rm pad}
  \end{bmatrix}.
  \label{eq:lcu-select-block-diagonal}
\end{equation}
If $w_\ell=\sqrt{a_\ell/\alpha}$, then the top-left block after preparing and unpreparing the selector register is the weighted compression
\begin{equation}
  \begin{bmatrix}w_0I & w_1I & \cdots & w_{s-1}I\end{bmatrix}
  \begin{bmatrix}
  U_0 &        &        &   \\
      & U_1    &        &   \\
      &        & \ddots &   \\
      &        &        & U_{s-1}
  \end{bmatrix}
  \begin{bmatrix}w_0I\\w_1I\\\vdots\\w_{s-1}I\end{bmatrix}
  =\sum_{\ell=0}^{s-1}\frac{a_\ell}{\alpha}U_\ell.
  \label{eq:lcu-matrix-compression}
\end{equation}
Thus LCU is not conceptually far from a classical block matrix construction: the selector register chooses one diagonal block, and the preparation vector supplies the coefficients.
Then
\begin{equation}
  \left(\bra{0^m}P^\dag\otimes I\right)
  \operatorname{SELECT}(U)
  \left(P\ket{0^m}\otimes I\right)
  =\frac{A}{\alpha}.
\end{equation}
Therefore
\begin{equation}
  U_A=(P^\dag\otimes I)\operatorname{SELECT}(U)(P\otimes I)
\end{equation}
block encodes $A$ with subnormalization $\alpha$.

Negative or complex coefficients are absorbed into the unitaries.
More generally, suppose $V_\ell$ block encodes $A_\ell/\alpha_\ell$
and $A=\sum_\ell a_\ell A_\ell$. One can pad the work registers to a common size,
prepare selector amplitudes $\sqrt{|a_\ell|\alpha_\ell/\alpha}$, and
SELECT $e^{i\arg a_\ell}V_\ell$, where
\[
 \alpha=\sum_\ell|a_\ell|\alpha_\ell.
\]
Projecting the selector and work registers onto zero gives $A/\alpha$.
The formula $\alpha=\sum_\ell|a_\ell|$ below refers to unitary summands. Structured nonunitary summands can also reduce decomposition size relative
to a Pauli expansion. Gnanasekaran and Surana develop such decompositions
and unitary completions for structured sparse matrices arising from PDEs
\cite{GnanasekaranSurana2026LCThings}.

The coefficient sum
\begin{equation}
  \alpha=\sum_{\ell=0}^{s-1}|a_\ell|=\norm{a}_1
  \label{eq:lcu-coefficient-one-norm}
\end{equation}
is therefore not merely a normalization convention.  Applied to a normalized state $\ket{\psi}$, the LCU block succeeds under postselection with probability
\begin{equation}
  p_{\rm LCU}=\frac{\norm{A\ket{\psi}}^2}{\alpha^2},
  \label{eq:lcu-success-probability}
\end{equation}
and amplitude amplification introduces a factor of order $\alpha/\norm{A\ket{\psi}}$.  Thus a decomposition with a large coefficient $1$-norm may be expensive even when cancellations make $\norm{A}$ small.  The quantity $\alpha$ plays a role analogous to a stability or conditioning constant for the chosen decomposition.

Why can LCU nevertheless be useful?  Classically, explicitly forming or applying a sum of $s$ unrelated matrices generally requires visiting all $s$ terms.  Quantumly, a single SELECT circuit applies the appropriate $U_\ell$ coherently to every selector branch, and PREPARE supplies all weights in superposition.  If PREPARE and SELECT have compact structured circuits, one block-encoding query can therefore represent a sum with many terms while using only $O(\log s)$ selector qubits.  This is not an automatic speedup: for unstructured data, preparing the coefficients or implementing the multi-controlled SELECT operation may itself cost $O(s)$.  The advantage comes from structure in the coefficient state and in the family $\{U_\ell\}$, as emphasized in \cite{LinWiebe2026QASC}.

\subsection{Sparse-matrix block encodings}
\label{sec:sparse-block-encoding-access}

Sparse access means coherent access to locations and values. For a general
nonsymmetric matrix, the construction below uses both row and column lists,
analogous to supplying both compressed sparse row and column data.

Assume $A\in\C^{N\times N}$ has at most $s$ nonzeros per row and column.
Pad each list to length $s$ using distinct marked dummy entries.
Let $f_r(i,\ell)$ be the column in row slot $\ell$, and $f_c(j,\ell)$
the row in column slot $\ell$. Assume coherent oracles
\begin{equation}
 \begin{aligned}
 O_{\mathrm{loc}}:&\ \ket{i}\ket{\ell}\ket{0}
 \mapsto\ket{i}\ket{\ell}\ket{f_r(i,\ell)},\\
 O_{\mathrm{loc}}^c:&\ \ket{j}\ket{\ell}\ket{0}
 \mapsto\ket{j}\ket{\ell}\ket{f_c(j,\ell)},
 \end{aligned}
 \label{eq:sparse-location-oracle}
\end{equation}
and
\begin{equation}
 O_{\mathrm{val}}:\ket{i}\ket{j}\ket{0}
 \mapsto\ket{i}\ket{j}\ket{A_{ij}}.
 \label{eq:sparse-value-oracle}
\end{equation}
Also assume efficient inverse list indexing on listed entries, or equivalent
in-place location oracles, so that the slot can be uncomputed after producing
a location. The three-register location oracle alone does not supply this
erasure for free. Dummy entries have value zero. Structured stencils often
supply these routines by reversible arithmetic. For Hermitian matrices the
row lists also specify column lists; otherwise column access is additional.

Write $A_{ij}=|A_{ij}|e^{i\theta_{ij}}$ and choose
\[
 \alpha\geq\max\left\{\max_i\sum_j|A_{ij}|,\max_j\sum_i|A_{ij}|\right\}.
\]
Use a common enlarged space with a two-qubit flag and isometries
\begin{align}
 P\ket{i}
 &=\frac{1}{\sqrt\alpha}\sum_j\sqrt{|A_{ij}|}\ket{00}\ket{i,j}
       +\ket{01}\ket{g_i^r},
 \label{eq:sparse-left-isometry}\\
 Q\ket{j}
 &=\frac{1}{\sqrt\alpha}\sum_i e^{i\theta_{ij}}\sqrt{|A_{ij}|}
                  \ket{00}\ket{i,j}
       +\ket{10}\ket{g_j^c}.
 \label{eq:sparse-right-isometry}
\end{align}
The unnormalized garbage vectors have squared norms
$1-\sum_j|A_{ij}|/\alpha$ and $1-\sum_i|A_{ij}|/\alpha$, respectively.
Each family retains its input label and is orthogonal for distinct inputs.
The flags make the row and column garbage sectors mutually orthogonal
and orthogonal to the successful sector. Thus
\begin{equation}
 \bra{i}P^\dag Q\ket{j}=\frac{A_{ij}}{\alpha},\qquad
 P^\dag Q=\frac{A}{\alpha}.
 \label{eq:sparse-isometry-overlap}
\end{equation}
If $U_P,U_Q$ implement these isometries from a common zero work register,
then $U_P^\dag U_Q$ block encodes $A$ with normalization $\alpha$.

For a direct implementation, take $\alpha=sM$ with
$M\geq\max_{ij}|A_{ij}|$. Prepare the uniform state over $s$ slots,
compute location and value, and rotate a success flag with amplitude
\[
 \sqrt{|A_{ij}|/M}=\sqrt{s|A_{ij}|/\alpha}.
\]
Multiplication by the slot amplitude $1/\sqrt s$ gives
$\sqrt{|A_{ij}|/\alpha}$. Add the phase on the column preparation's
successful branch. Uncompute value and slot registers there, and route
row and column failures to flags $01$ and $10$. This realizes the
isometries up to arithmetic and rotation errors. A smaller normalization
based on row and column sums requires an additional state-preparation
or amplification argument.

\begin{remark}[Relation to standard sparse formats]
Classically, a matrix-vector product with an $s$-sparse matrix costs $O(sN)$ arithmetic operations because all rows must be visited.  Quantumly, the goal is different: construct a unitary whose action on amplitudes has the matrix embedded in one block.  The oracle calls in \eqref{eq:sparse-location-oracle}--\eqref{eq:sparse-value-oracle} are made coherently on superpositions of indices.  This is why the cost is counted in oracle queries and reversible arithmetic gates, not in floating-point operations over all rows.
\end{remark}

\begin{remark}[Structured sparse matrices]
For structured sparse matrices these oracles should not remain abstract.  The row locations may be computed by shifts and modular arithmetic, and the values may be constants or simple functions of the grid point.  Camps et al.  give explicit circuits for block encodings of several structured sparse matrices, including banded circulant and tridiagonal matrices \cite{CampsLinVanBeeumenYang2024}.  The finite-difference examples in Section~\ref{sec:finite-difference-block-encodings} are the simplest instances of this principle.
\end{remark}

\section{Quantum singular value transformation}
\label{sec:qsvt-basic}

Once a normalized block $A/\alpha$ is block encoded, QSVT can transform its singular values by a polynomial.  In this section, to lighten notation, $A$ denotes the normalized matrix appearing in the block encoding: its singular values lie in $[0,1]$, and in the Hermitian eigenvalue-transformation setting its eigenvalues lie in $[-1,1]$. This is the closest quantum analogue of applying a polynomial filter to a matrix in numerical linear algebra; see \cite{GilyenSuLowWiebe2019,LinWiebe2026QASC} for the formal construction and a computational-science-oriented exposition.

Let $A\in\C^{N\times N}$ have singular value decomposition
\begin{equation}
  A=\sum_k \sigma_k \ket{u_k}\bra{v_k},
  \qquad
  0\leq \sigma_k\leq 1.
\end{equation}
Suppose $U$ is a block encoding of $A$. QSVT constructs, from alternating uses of $U$, $U^\dag$, and simple phase rotations on an ancilla, a new unitary whose top block is approximately
\begin{equation}
  p^{(\mathrm{SV})}(A)
  :=\sum_k p(\sigma_k)\ket{u_k}\bra{v_k},
\end{equation}
for a polynomial $p$ satisfying certain parity and boundedness constraints. The foundational theory is due to Gily\'en, Su, Low, and Wiebe \cite{GilyenSuLowWiebe2019}. For odd polynomials, the transformed block maps right singular vectors to left singular vectors, as displayed. For even polynomials, the natural output is instead $\sum_k p(\sigma_k)\ket{v_k}\bra{v_k}$, acting within the right singular space. For Hermitian $A$ this distinction disappears in the eigenvalue form \eqref{eq:quantum-eigenvalue-transformation}, which is the case used most often in this lecture note.

When $A$ is Hermitian, singular-value transformation can be arranged as an eigenvalue transformation.  If
\begin{equation}
  A=\sum_k \lambda_k\ket{q_k}\bra{q_k},
  \qquad -1\leq \lambda_k\leq 1,
\end{equation}
then the transformed block acts as
\begin{equation}
  p(A)=\sum_k p(\lambda_k)\ket{q_k}\bra{q_k}.
  \label{eq:quantum-eigenvalue-transformation}
\end{equation}
This Hermitian specialization is often called \emph{quantum eigenvalue transformation}.  It is the form most familiar to numerical analysts because it is literally a polynomial functional calculus for a Hermitian matrix.

\begin{remark}[Analogy with Krylov and Chebyshev methods]
In Krylov methods, one approximates $f(A)\bm b$ by $p_m(A)\bm b$ for a polynomial $p_m$ of degree $m$ \cite{Saad2003Iterative}. QSVT also implements polynomial approximations, but it does so inside a unitary circuit using a block encoding of $A$. The degree of $p$ is reflected in the number of calls to the block encoding.
\end{remark}

\subsection{Example: inverse filters}

Let $A$ be Hermitian positive definite and suppose its scaled spectrum lies in
\begin{equation}
  \operatorname{spec}(A)\subseteq [1/\kappa,1].
\end{equation}
A QSVT-based linear-system algorithm chooses a polynomial $p_m$ such that
\begin{equation}
  \max_{x\in[1/\kappa,1]}\abs{p_m(x)-\frac{c}{x}}\leq \eps,
  \label{eq:qsvt-inverse-approximation}
\end{equation}
where one typically takes $c=\Theta(1/\kappa)$.  QSVT requires the implemented polynomial to obey the boundedness condition $|p_m(x)|\leq1$ on all of $[-1,1]$, not only on the spectral interval $[1/\kappa,1]$.  Applying the polynomial to an input state $\ket{\bm b}=\sum_k b_k\ket{q_k}$ produces a successful branch proportional to
\begin{equation}
  p_m(A)\ket{\bm b}\approx c\sum_k\frac{b_k}{\lambda_k}\ket{q_k}
  =cA^{-1}\ket{\bm b}.
\end{equation}
This is the QSVT version of an inverse polynomial filter.  The condition number enters because the polynomial must approximate $1/x$ near the endpoint $x=1/\kappa$ while remaining bounded on $[-1,1]$.  The scaling constant $c=\Theta(1/\kappa)$ also affects the postselection probability: in a worst-case instance the successful branch can have norm smaller by a factor comparable to $1/\kappa$, so output normalization or amplitude amplification can reintroduce condition-number dependence.

\begin{lemma}[Block encoding of an inverse]
\label{lem:block-encoding-inverse}
Let $A=A^\dag$ be positive definite and suppose that $U_A$ is an exact
$(\alpha,a,0)$ block encoding of $A$, with
\begin{equation}
  \operatorname{spec}(A/\alpha)\subseteq[1/\kappa,1].
\end{equation}
For every $0<\epsilon<1/2$, QSVT constructs, using
\begin{equation}
  O\!\left(\kappa\log\frac{\kappa}{\epsilon}\right)
\end{equation}
queries to $U_A$ and $U_A^\dag$, a block encoding whose top-left block $B$
satisfies
\begin{equation}
  \left\|B-\frac{\alpha}{2\kappa}A^{-1}\right\|\leq \epsilon .
  \label{eq:block-encoded-inverse-error}
\end{equation}
Equivalently, this gives a block encoding of $A^{-1}$ with subnormalization $2\kappa/\alpha$ and scaled error $(2\kappa/\alpha)\epsilon$, up to the additional QSVT ancillas.  The factor $\alpha/(2\kappa)$ appears in the successful block because the inverse must be rescaled to remain a contraction.  Thus the lemma gives a bounded block encoding of a scaled inverse, not a direct deterministic application of $A^{-1}$.
\end{lemma}

\begin{proof}[Idea]
Apply quantum eigenvalue transformation to a polynomial $p$ that approximates
$1/(2\kappa x)$ on $[1/\kappa,1]$ while obeying $|p(x)|\leq1$ for all
$x\in[-1,1]$.  The required degree is linear in $\kappa$ up to logarithmic
factors \cite{GilyenSuLowWiebe2019,LinWiebe2026QASC}.  The resulting block
acts as $p(A/\alpha)\approx (\alpha/(2\kappa))A^{-1}$.
\end{proof}

\subsection{Diagonal matrices and the cost of inversion}
\label{subsec:diagonal-block-encoding}

A diagonal matrix is the easiest matrix to apply classically.  If
\begin{equation}
  D=\operatorname{diag}(d_0,d_1,\ldots,d_{N-1}),
  \label{eq:diagonal-matrix}
\end{equation}
then a classical matrix-vector product merely multiplies each component by $d_j$.  In a quantum block encoding, one must still account for normalization.  Suppose that the diagonal entries can be computed reversibly in fixed-point form and that
\begin{equation}
  \alpha\geq \max_j |d_j|=\norm{D}_2.
  \label{eq:diagonal-alpha}
\end{equation}
A block encoding can be constructed by first computing the magnitude and phase $d_j=|d_j|e^{i\phi_j}$, then applying
\begin{equation}
  \ket{j}\ket{0}_{\rm val}\ket{0}_{\rm anc}
  \longmapsto
  \ket{j}\ket{d_j}_{\rm val}
  \left(
  \frac{|d_j|}{\alpha}\ket{0}_{\rm anc}
  +\sqrt{1-\frac{|d_j|^2}{\alpha^2}}\ket{1}_{\rm anc}
  \right),
  \label{eq:diagonal-controlled-rotation}
\end{equation}
and finally applying a controlled phase $e^{i\phi_j}$ on the successful branch before uncomputing the value register.  The top-left ancilla block is $D/\alpha$.  Thus diagonal matrices are friendly to block encoding when their entries are computable, but the factor $\alpha$ is still part of the algorithmic cost. This magnitude-rotation-plus-phase structure is exactly the diagonal analogue of the sparse construction in Section~\ref{sec:sparse-block-encoding-access}.

The inverse is more subtle.  Assume, after scaling and restricting to the positive-definite real case, that
\begin{equation}
  d_j\in[1/\kappa,1].
  \label{eq:diagonal-spectrum-kappa}
\end{equation}
A generic QSVT inverse filter applied to a block encoding of $D$ must approximate $x\mapsto 1/x$ on $[1/\kappa,1]$, and therefore has query complexity proportional to the condition number, up to logarithmic factors.  In other words, the fact that $D$ is diagonal does not by itself make a block encoding of $D^{-1}$ cheap if the only primitive is a block encoding of $D$.  This point is important for quantum preconditioning.  A classical Jacobi preconditioner is almost trivial to apply, but a quantum implementation must still either implement the reciprocal entries coherently or pay the usual inverse-filter cost.

Tong, An, Wiebe, and Lin introduced a different primitive called fast inversion, which directly block encodes inverse eigenvalues by reversible arithmetic when sufficient structure is available \cite{TongAnWiebeLin2021FastInversion}.  This can be powerful, but it is an additional access assumption.  The lesson for numerical analysts is that ``diagonal'' and ``cheap to invert'' are not automatically synonymous in the block-encoding model.

\subsection{Example: heat semigroups and exponential filters}

For a parabolic equation after spatial discretization,
\begin{equation}
  \frac{d u}{dt}=-Lu,
  \qquad L\succeq0,
\end{equation}
we need
\begin{equation}
  u(t)=e^{-tL}u(0).
\end{equation}
If $L/\alpha$ is block encoded and $\operatorname{spec}(L/\alpha)\subseteq[0,1]$, then one approximates
\begin{equation}
  x\mapsto e^{-\alpha t x},
  \qquad x\in[0,1].
  \label{eq:qsvt-heat-filter}
\end{equation}
The cost depends on which operator is supplied. A generic block encoding
of $L/\alpha$ alone does not imply a square-root cost in $\alpha t$:
QSVT requires boundedness on $[-1,1]$, where zero is an interior point.
With the stronger access assumption $L=G^\dag G$ and a block encoding
of $G/\alpha_G$, an even polynomial approximating
$e^{-t\alpha_G^2x^2}$ transforms the right singular space into $e^{-tL}$.
Its degree is
$O(\alpha_G\sqrt{t\log(1/\eps)}+\log(1/\eps))$.
Alternatively, one can assume a normalization-one block encoding of
$I-L/\alpha$; this does not follow for free from an encoding of $L/\alpha$.
Chapter~\ref{chap:parabolic-quantum} makes these access distinctions explicit
\cite{GilyenSuLowWiebe2019}.

\subsection{Example: oscillatory propagators}

For a Hermitian Hamiltonian $H$, Schr\"odinger evolution is
\begin{equation}
  \ket{\psi(t)}=e^{-itH}\ket{\psi(0)}.
\end{equation}
Hamiltonian simulation can be regarded as implementing polynomial approximations to the real and imaginary parts,
\begin{equation}
  \cos(tx),\qquad \sin(tx),
\end{equation}
through the same block-encoding and polynomial-transformation toolkit.  More precisely, suppose that $H/\alpha$ has a block encoding.  Qubitization, equivalently the Hamiltonian-simulation specialization of QSVT, implements $e^{-itH}$ using
\begin{equation}
  O\!\left(\alpha t+\frac{\log(1/\eps)}{\log\log(1/\eps)}\right)
  \label{eq:optimal-hamiltonian-simulation-basic}
\end{equation}
queries to the block encoding, up to constant-factor conventions in the access model \cite{LowChuang2019,GilyenSuLowWiebe2019}.  We usually write this as $\widetilde O(\alpha t+\log(1/\eps))$.  The leading linear dependence on the scaled time $\alpha t$ is not merely an artifact of the method: in the black-box sparse-Hamiltonian model, no-fast-forwarding lower bounds rule out sublinear dependence for general Hamiltonians \cite{BerryAhokasCleveSanders2007}.

With access to $G/\alpha_G$, the heat filter is the bounded even function
$e^{-t\alpha_G^2x^2}$ and has degree essentially $\alpha_G\sqrt t$.
Hamiltonian simulation resolves oscillations of $e^{-i\tau x}$ throughout
the signal interval and has generic query complexity linear in $\tau$.
Fourier diagonalization may permit fast-forwarding, but uses additional structure.

For wave equations, this is the route by which a second-order real PDE is converted into a first-order Hermitian dynamics in the hyperbolic chapter.

\begin{remark}[What QSVT gives]
QSVT gives a new block encoding.  If the desired output is a state proportional to $p(A)\ket{\bm b}$, one still has to account for postselection, output normalization, and measurement.  Thus QSVT solves the spectral transformation problem; it does not by itself solve the data-loading or readout problems.
\end{remark}

\section{Finite-difference operators as block encodings}
\label{sec:finite-difference-block-encodings}

We now specialize the preceding constructions to finite-difference matrices. These examples are deliberately elementary. They show how the matrices that appear in standard numerical PDE discretizations can be expressed through a small number of shift unitaries. The stencils below are the standard model finite-difference operators used in introductory numerical PDE texts \cite{LarssonThomee2009,TveitoWinther2005}; Figure~\ref{fig:three-point-laplacian-stencil} shows the basic one-dimensional example.

Throughout this section let $N=2^n$ and define the cyclic shift $S$ by
\begin{equation}
  S\ket{j}=\ket{j+1\pmod N},
  \qquad
  S^\dag\ket{j}=\ket{j-1\pmod N}.
\end{equation}
The quantum circuit for this shift matrix $S$ is modular addition by one; the circuit for $S^\dag$ is modular subtraction by one. It should not be confused with the one-qubit phase gate $S$ in \eqref{eq:S-phase-gate}. These are precisely the left- and right-shift components used in explicit block encodings of banded circulant matrices by Camps--Lin--Van Beeumen--Yang \cite{CampsLinVanBeeumenYang2024}.

\begin{figure}[htbp]
\centering
\OneDLaplacianStencil
\caption{The three-point finite-difference stencil for the positive one-dimensional Laplacian $-\partial_{xx}$.  The same local stencil is represented quantumly by a short LCU of the identity, the cyclic shift, and the inverse cyclic shift.}
\label{fig:three-point-laplacian-stencil}
\end{figure}

\subsection{Forward difference operator}

Consider the periodic forward difference
\begin{equation}
  (D_+u)_j=\frac{u_{j+1}-u_j}{h}.
\end{equation}
With our shift convention,
\begin{equation}
  D_+=\frac{S^\dag-I}{h}.
\end{equation}
Equivalently, $-D_+=(I-S^\dag)/h$. Since both $I$ and $-S^\dag$ are unitary, the LCU construction gives an exact block encoding.

\begin{proposition}[Block encoding of the forward difference]
Let
\begin{equation}
  \operatorname{SELECT}_{D}
  =\ket{0}\bra{0}\otimes I
  +\ket{1}\bra{1}\otimes (-S^\dag),
\end{equation}
and let $P_D=H$ on the one-qubit selector register. Then
\begin{equation}
  U_D=(H\otimes I)\operatorname{SELECT}_{D}(H\otimes I)
\end{equation}
is an exact $(2/h,1,0)$ block encoding of $-D_+$:
\begin{equation}
  (\bra{0}\otimes I)U_D(\ket{0}\otimes I)
  =\frac{I-S^\dag}{2}
  =\frac{h(-D_+)}{2}.
\end{equation}
\end{proposition}

\begin{remark}[Boundary conditions]
For nonperiodic forward differences, the wrap-around entry must be removed or modified. In the structured sparse approach, this is done by adding boundary-controlled rotations or corrections, exactly as in the tridiagonal nonperiodic modification of the Camps--Lin--Van Beeumen--Yang circulant circuit. Conceptually, one starts from the periodic shift circuit and then zeros out the entries that cross the boundary.
\end{remark}

\subsection{Three-point Laplacian in one dimension}

Let $L_h$ denote the positive periodic one-dimensional discrete Laplacian,
\begin{equation}
  (L_hu)_j=\frac{2u_j-u_{j-1}-u_{j+1}}{h^2}.
\end{equation}
Then
\begin{equation}
  L_h=\frac{2I-S-S^\dag}{h^2}.
  \label{eq:1d-positive-laplacian-shift-form}
\end{equation}

\begin{figure}[htbp]
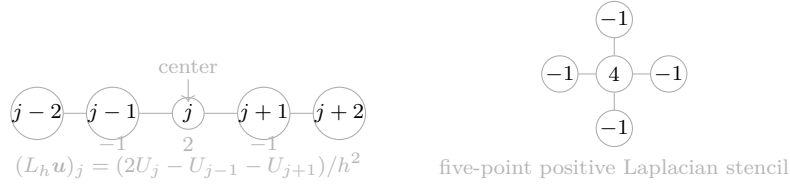

\centering
\StencilDiagram\qquad\FivePointStencilDiagram
\caption{The one-dimensional three-point and two-dimensional five-point positive Laplacian stencils.  These are the finite-difference matrices later written as LCUs of shift unitaries.}
\label{fig:laplacian-stencils}
\end{figure}

This is the standard three-point stencil for $-\partial_{xx}$ with periodic boundary conditions, shown with its two-dimensional five-point counterpart in Figure~\ref{fig:laplacian-stencils}; it is the finite-difference model operator used throughout classical numerical PDE texts \cite{LarssonThomee2009,TveitoWinther2005}.

\begin{proposition}[LCU block encoding of the one-dimensional Laplacian]
Let the selector register have two qubits and define
\begin{equation}
  \operatorname{SELECT}_{L}
  =
  \ket{0}\bra{0}\otimes I
  +\ket{1}\bra{1}\otimes (-S)
  +\ket{2}\bra{2}\otimes (-S^\dag)
  +\ket{3}\bra{3}\otimes I.
\end{equation}
Let $P_L$ satisfy
\begin{equation}
  P_L\ket{0}
  =\frac{\sqrt2}{2}\ket{0}
  +\frac{1}{2}\ket{1}
  +\frac{1}{2}\ket{2}.
\end{equation}
The amplitude of $\ket{3}$ is zero. Then
\begin{equation}
  U_L=(P_L^\dag\otimes I)\operatorname{SELECT}_{L}(P_L\otimes I)
\end{equation}
obeys
\begin{equation}
  (\bra{0}\otimes I)U_L(\ket{0}\otimes I)
  =\frac{2I-S-S^\dag}{4}
  =\frac{h^2L_h}{4}.
\end{equation}
Thus $U_L$ is an exact $(4/h^2,2,0)$ block encoding of $L_h$.
\end{proposition}

\begin{remark}[Connection with the explicit tridiagonal circuit]
The matrix in \eqref{eq:1d-positive-laplacian-shift-form} is a banded circulant matrix with stencil coefficients $a_0=2/h^2$ and $a_-=a_+=-1/h^2$ (the symbols $\alpha$, $\gamma$ are reserved here for the block-encoding normalization and, later, the Schr\"odinger coefficient). It is therefore a direct specialization of the banded-circulant construction. In the sparse-oracle language, $f(j,\ell)$ selects $j-1$, $j$, or $j+1$ modulo $N$, and the value circuit supplies the three stencil weights. The LCU form above is the same structure written in the shorter language of shift unitaries.
\end{remark}

\subsection{Two-dimensional five-point Laplacian}

Let the grid have $N_x=2^{n_x}$ points in the $x$ direction and $N_y=2^{n_y}$ points in the $y$ direction. The full state space is
\begin{equation}
  \C^{N_x}\otimes\C^{N_y},
\end{equation}
with basis states $\ket{i}\ket{j}$. Define
\begin{equation}
  S_x=S\otimes I,
  \qquad
  S_y=I\otimes S,
\end{equation}
where each $S$ is the corresponding cyclic shift in one coordinate. For equal mesh spacing $h$, the positive periodic five-point Laplacian is
\begin{equation}
  L_{h}^{(2D)}
  =\frac{4I-S_x-S_x^\dagger-S_y-S_y^\dagger}{h^2}.
  \label{eq:2d-laplacian-shift-form}
\end{equation}

\begin{proposition}[LCU block encoding of the two-dimensional Laplacian]
Let the selector register have three qubits and define SELECT by
\begin{align}
  \operatorname{SELECT}_{2D}
  ={}&\ket{0}\bra{0}\otimes I
  +\ket{1}\bra{1}\otimes (-S_x)
  +\ket{2}\bra{2}\otimes (-S_x^\dagger)\notag\\
  &+\ket{3}\bra{3}\otimes (-S_y)
  +\ket{4}\bra{4}\otimes (-S_y^\dagger)
  +\sum_{\ell=5}^{7}\ket{\ell}\bra{\ell}\otimes I.
\end{align}
Let $P_{2D}$ satisfy
\begin{equation}
  P_{2D}\ket{0}
  =\frac{1}{\sqrt2}\ket{0}
  +\frac{1}{\sqrt8}\left(\ket{1}+\ket{2}+\ket{3}+\ket{4}\right).
\end{equation}
Then
\begin{equation}
  U_{2D}=(P_{2D}^\dag\otimes I)\operatorname{SELECT}_{2D}(P_{2D}\otimes I)
\end{equation}
satisfies
\begin{equation}
  (\bra{0}\otimes I)U_{2D}(\ket{0}\otimes I)
  =\frac{4I-S_x-S_x^\dagger-S_y-S_y^\dagger}{8}
  =\frac{h^2L_h^{(2D)}}{8}.
\end{equation}
Thus $U_{2D}$ is an exact $(8/h^2,3,0)$ block encoding of $L_h^{(2D)}$.
\end{proposition}

\begin{remark}[Anisotropic grid]
For spacings $h_x$ and $h_y$,
\begin{equation}
  L_{h_x,h_y}^{(2D)}
  =
  \frac{2I-S_x-S_x^\dagger}{h_x^2}
  +
  \frac{2I-S_y-S_y^\dagger}{h_y^2}.
\end{equation}
An LCU block encoding has subnormalization
\begin{equation}
  \alpha=\frac{4}{h_x^2}+\frac{4}{h_y^2}.
\end{equation}
The selector preparation amplitudes are proportional to the square roots of the absolute stencil coefficients.
\end{remark}

\subsection{General \texorpdfstring{$d$}{d}-dimensional periodic Laplacian}

The same construction extends to $d$ dimensions. Let $S_r$ shift the $r$th coordinate. For a uniform mesh,
\begin{equation}
  L_h^{(d)}=\frac{1}{h^2}\sum_{r=1}^d(2I-S_r-S_r^\dagger).
\end{equation}
The LCU coefficient sum is
\begin{equation}
  \alpha=\frac{4d}{h^2}.
\end{equation}
The block encoding uses a selector register large enough to choose the diagonal term and the $2d$ nearest-neighbor shifts. This is the finite-difference analogue of assembling a stencil matrix from shift matrices.

\begin{remark}[Cost model]
The cost is dominated by controlled modular additions and subtractions. These operations have circuits of size polynomial in the number of grid qubits. Thus for tensor-product periodic grids, the block encoding has complexity polynomial in $\log N$, not polynomial in $N$. Boundary conditions and variable coefficients introduce additional arithmetic and control logic.
\end{remark}

\section{From finite differences to PDE algorithms}

The previous section gave block encodings of discrete differential operators.  A PDE algorithm typically adds four more ingredients.

\subsection{Initial and forcing data}

One must prepare amplitude encodings of initial data, forcing terms, or boundary data. Smooth data may be handled by Grover--Rudolph-type preparation. Data produced by a preceding quantum routine may already be in amplitude form. Arbitrary classical data remain expensive unless one assumes a suitable memory model.

\subsection{Operator functions}

Time-dependent PDEs often require functions of matrices. Examples include
\begin{align}
  \text{heat equation:}\qquad &u(t)=e^{-tL}u_0,\\
  \text{wave equation:}\qquad &u(t)=\cos(t\sqrt L)u_0 + L^{-1/2}\sin(t\sqrt L)v_0,\\
  \text{elliptic equation:}\qquad &u=L^{-1}f.
\end{align}
Once $L$ is block encoded and spectrally scaled, QSVT gives a route to approximating these matrix functions by polynomials.

\subsection{Normalization and success probabilities}

If a block encoding of $L$ has subnormalization $\alpha$, then applying a polynomial $p(L/\alpha)$ produces a normalized quantum output only after postselection or embedded unitary evolution. The norm of the desired vector controls the success probability. This is not a technicality: in PDE problems, smoothing, dissipation, stiffness, and ill-conditioning can strongly affect output norms.

\subsection{Reading out useful information}

A quantum algorithm rarely returns all entries of $u$. Instead it returns a state proportional to $u$ or allows estimation of selected quantities:
\begin{equation}
  \ip{g}{u},\qquad
  \norm{u}_{L^2(\omega)}^2,
  \qquad
  \bm u^\dag B\bm u,
  \qquad
  \text{or}\qquad
  \Pr(x\in\omega).
\end{equation}
This is natural for many PDE tasks but inappropriate if the goal is to print the full solution field.

These four ingredients form the recurring checklist of the lecture note, and the PDE chapters can be read as case studies in how each equation type distributes its cost among them.  For elliptic problems, the operator function is an ill-conditioned inverse, and the condition number is not an input parameter but a consequence of the discretization (Chapter~\ref{chap:elliptic-quantum}).  For hyperbolic problems, the evolution becomes exactly unitary in the right variables, and the burden shifts to state preparation, sources, and readout (Chapter~\ref{chap:hyperbolic-quantum}).  For parabolic problems, the same smoothing that makes the operator cheap makes the normalized output state expensive (Chapter~\ref{chap:parabolic-quantum}).  For nonlinear problems, even the choice of linear representation is a choice of output model (Chapter~\ref{chap:nonlinear-quantum}).

\section{A worked example: Poisson equation on a periodic grid}

\begin{figure}[htbp]
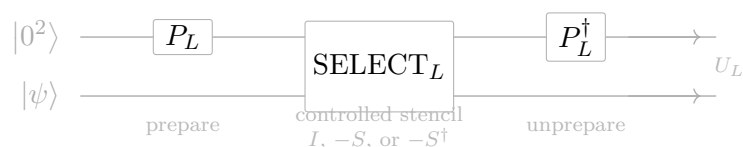

\centering
\LaplacianBlockEncodingCircuit
\caption{Three-stage LCU workflow for the periodic three-point Laplacian.  The first box prepares the selector amplitudes, $\operatorname{SELECT}_L$ applies $I$, $-S$, or $-S^\dag$ conditioned on the selector, and $P_L^\dag$ unprepares the selector.  Projecting the selector onto $\ket{0^2}$ extracts the block $h^2L_h/4$.}
\label{fig:laplacian-block-encoding-circuit}
\end{figure}

Consider the one-dimensional periodic Poisson problem
\begin{equation}
  -u''(x)=f(x),
  \qquad x\in[0,1],
\end{equation}
with mean-zero right-hand side. The finite-difference discretization is
\begin{equation}
  L_h \bm u_h=\bm f_h,
  \qquad
  L_h=\frac{2I-S-S^\dag}{h^2}.
\end{equation}
The constant vector lies in the nullspace, so the inverse is defined on the orthogonal complement of constants.

A quantum linear-system route consists of the following steps.
\begin{enumerate}[label=\textbf{Step \arabic*.},leftmargin=*]
\item Prepare $\ket{\bm f_h}$, for example using smooth-function state preparation if $f$ is given analytically.
\item Use the block encoding of $L_h$ from Section~\ref{sec:finite-difference-block-encodings}; the three-stage circuit is shown in Figure~\ref{fig:laplacian-block-encoding-circuit}. Since the block encoding represents $L_h/\alpha$ with $\alpha=4/h^2$, the nonzero eigenvalues of $L_h/\alpha$ lie in $[\Theta(N^{-2}),1]$.
\item Use QSVT to approximate $x\mapsto 1/x$ on the nonzero part of the spectrum. The degree depends on the condition number and desired accuracy.
\item Apply the approximate inverse to obtain a state proportional to $\bm u_h=L_h^+\bm f_h$, where $L_h^+$ denotes the Moore--Penrose pseudo-inverse on the mean-zero subspace.
\item Estimate quantities of interest using overlap estimation, amplitude estimation, or observable measurements. Convert back to continuum units using the $h^{d/2}$ scaling.
\end{enumerate}
The last step should be read in the sense of Sections~\ref{sec:measurement-observables} and~\ref{subsec:amplitude-estimation}.  If a normalized solution state is prepared, direct sampling estimates an observable expectation with the usual $O(\epsilon^{-2})$ shot dependence, while amplitude estimation can reduce this to $O(\epsilon^{-1})$ coherent uses of the solution-preparation circuit.  The normalization factor hidden by the linear-system postselection must still be recovered if the desired PDE output is an unnormalized quantity such as $\bm g_h^\dag\bm u_h$ or an $L^2$ norm.

\begin{remark}[Conditioning is still present]
Quantum representation does not remove the condition number of the discrete elliptic operator. For the one-dimensional periodic Laplacian, the smallest nonzero eigenvalue is $\Theta(1)$ while the largest eigenvalue is $\Theta(h^{-2})$ for $L_h$; after scaling by $\alpha=\Theta(h^{-2})$, the smallest nonzero scaled eigenvalue is $\Theta(h^2)$. Thus inverse approximation still sees a condition number of order $h^{-2}$.
\end{remark}

\section{A worked example: Schr\"odinger dynamics with the three-point Laplacian}
\label{sec:schrodinger-three-point-laplacian}

The Poisson example used the block encoding of $L_h$ to apply an approximate inverse.  We now use the same block encoding for a time-dependent problem.  Consider the one-dimensional periodic Schr\"odinger equation with kinetic energy only,
\begin{equation}
  i\partial_t u(t,x)=-\gamma \partial_{xx}u(t,x),
  \qquad x\in[0,1],
  \label{eq:continuous-schrodinger-kinetic}
\end{equation}
where $\gamma>0$ is a physical or nondimensional constant.  For example, in nondimensional units one may take $\gamma=1/2$ or $\gamma=1$.  Using the positive periodic discrete Laplacian from \eqref{eq:1d-positive-laplacian-shift-form}, the semidiscrete equation is
\begin{equation}
  i\frac{d}{dt}\bm u_h(t)=H_h\bm u_h(t),
  \qquad
  H_h=\gamma L_h
  =\frac{\gamma}{h^2}\left(2I-S-S^\dag\right).
  \label{eq:discrete-schrodinger-kinetic}
\end{equation}
Thus
\begin{equation}
  \bm u_h(t)=e^{-itH_h}\bm u_h(0).
  \label{eq:discrete-schrodinger-propagator}
\end{equation}
Unlike the Poisson problem, no inverse is involved.  The map in \eqref{eq:discrete-schrodinger-propagator} is unitary because $H_h$ is Hermitian.  Therefore the Euclidean norm of the grid vector is preserved:
\begin{equation}
  \norm{\bm u_h(t)}_2=\norm{\bm u_h(0)}_2.
\end{equation}
Consequently, if the initial condition is amplitude encoded as
\begin{equation}
  \ket{\bm u_h(0)}
  =
  \frac{1}{\norm{\bm u_h(0)}_2}
  \sum_{j=0}^{N-1}(\bm u_h(0))_j\ket{j},
\end{equation}
then the desired quantum output is simply
\begin{equation}
  \ket{\bm u_h(t)}
  =
  e^{-itH_h}\ket{\bm u_h(0)}.
\end{equation}

The block encoding of $L_h$ from Section~\ref{sec:finite-difference-block-encodings} satisfies
\begin{equation}
  (\bra{0^2}\otimes I)U_L(\ket{0^2}\otimes I)
  =
  \frac{h^2L_h}{4}.
\end{equation}
Hence the same unitary $U_L$ is also a block encoding of the Hamiltonian $H_h=\gamma L_h$, with subnormalization
\begin{equation}
  \alpha_H=\frac{4\gamma}{h^2},
  \qquad
  (\bra{0^2}\otimes I)U_L(\ket{0^2}\otimes I)
  =
  \frac{H_h}{\alpha_H}.
  \label{eq:Hh-block-encoding-from-Lh}
\end{equation}
It is useful to define the scaled Hermitian matrix
\begin{equation}
  A_h:=\frac{H_h}{\alpha_H}
  =
  \frac{h^2L_h}{4}.
\end{equation}
The eigenvalues of $L_h$ are
\begin{equation}
  \lambda_k(L_h)
  =
  \frac{4}{h^2}\sin^2\left(\frac{\pi k}{N}\right),
  \qquad k=0,\ldots,N-1,
  \label{eq:periodic-laplacian-eigenvalues}
\end{equation}
and therefore
\begin{equation}
  \lambda_k(A_h)
  =
  \sin^2\left(\frac{\pi k}{N}\right)
  \in[0,1].
\end{equation}
Thus the Schr\"odinger propagator can be written as
\begin{equation}
  e^{-itH_h}
  =
  e^{-i\tau A_h},
  \qquad
  \tau=\alpha_H t.
  \label{eq:scaled-schrodinger-propagator}
\end{equation}

The QSVT viewpoint is now exactly the same as the polynomial functional calculus used in numerical analysis.  We want to approximate the scalar function
\begin{equation}
  f_\tau(x)=e^{-i\tau x},
  \qquad x\in[0,1],
  \label{eq:schrodinger-qsvt-target-function}
\end{equation}
and then apply the resulting polynomial to $A_h=H_h/\alpha_H$.  If $p_m$ is a degree-$m$ polynomial satisfying
\begin{equation}
  \max_{x\in[0,1]}
  \abs{p_m(x)-e^{-i\tau x}}
  \leq \eps,
  \label{eq:schrodinger-polynomial-approximation}
\end{equation}
then the corresponding matrix approximation obeys
\begin{equation}
  \norm{p_m(A_h)-e^{-i\tau A_h}}
  \leq \eps.
\end{equation}
QSVT, or more specifically its Hamiltonian-simulation specialization often called quantum signal processing or qubitization, compiles such a bounded polynomial transformation into a sequence of phase rotations and alternating calls to the block encoding $U_L$ and its adjoint.  Schematically, it constructs a unitary $V_\Phi$ such that
\begin{equation}
  \left\|
  (\bra{0^{a'}}\otimes I)V_\Phi(\ket{0^{a'}}\otimes I)
  -
  p_m(A_h)
  \right\|
  \leq
  \eps,
  \label{eq:qsvt-schrodinger-block}
\end{equation}
where $a'$ denotes the selector qubits together with any additional signal or work qubits used by the QSVT construction.  Combining \eqref{eq:schrodinger-polynomial-approximation} and \eqref{eq:qsvt-schrodinger-block} gives
\begin{equation}
  \left\|
  (\bra{0^{a'}}\otimes I)V_\Phi(\ket{0^{a'}}\otimes I)
  -
  e^{-itH_h}
  \right\|
  \lesssim
  \eps.
\end{equation}

For numerical analysts, a useful way to view the polynomial approximation
underlying Hamiltonian simulation is through the Jacobi--Anger expansion.
Although
\[
    \operatorname{spec}(A_h)\subset[0,1],
\]
the signal variable in a QSVT transformation ranges over the full interval
$[-1,1]$.  It is therefore convenient to approximate the propagator on this
full signal interval rather than to approximate only on the physical spectral
interval.

For $x\in[-1,1]$, the Jacobi--Anger identities give
\begin{align}
  \cos(\tau x)
  &=
  J_0(\tau)
  +
  2\sum_{k=1}^{\infty}
  (-1)^k J_{2k}(\tau)T_{2k}(x),
  \label{eq:jacobi-anger-cosine}
  \\
  \sin(\tau x)
  &=
  2\sum_{k=0}^{\infty}
  (-1)^k J_{2k+1}(\tau)T_{2k+1}(x),
  \label{eq:jacobi-anger-sine}
\end{align}
where $T_k$ is the Chebyshev polynomial of the first kind and $J_k$ is
the Bessel function of the first kind.  Define the truncated polynomials
\begin{align}
  C_R(x)
  &:=
  J_0(\tau)
  +
  2\sum_{k=1}^{R}
  (-1)^k J_{2k}(\tau)T_{2k}(x),
  \label{eq:truncated-jacobi-anger-cosine}
  \\
  S_R(x)
  &:=
  2\sum_{k=0}^{R}
  (-1)^k J_{2k+1}(\tau)T_{2k+1}(x).
  \label{eq:truncated-jacobi-anger-sine}
\end{align}
The polynomial $C_R$ is even and has degree $2R$, whereas $S_R$ is odd
and has degree $2R+1$.  These are precisely the parity conditions required
by the corresponding QSVT transformations.

For every $\delta>0$, the truncation order can be chosen so that
\begin{align}
  \max_{x\in[-1,1]}
  \left|C_R(x)-\cos(\tau x)\right|
  &\leq \delta,
  \\
  \max_{x\in[-1,1]}
  \left|S_R(x)-\sin(\tau x)\right|
  &\leq \delta.
\end{align}
Consequently,
\begin{equation}
  P_R(x):=C_R(x)-iS_R(x)
\end{equation}
approximates $e^{-i\tau x}$ uniformly on $[-1,1]$ with error at most
$2\delta$.

The raw truncations satisfy
\[
  \|C_R\|_{\infty,[-1,1]},
  \|S_R\|_{\infty,[-1,1]}
  \leq 1+\delta.
\]
A harmless rescaling, for example
\[
  \widetilde C_R=\frac{C_R}{1+\delta},
  \qquad
  \widetilde S_R=\frac{S_R}{1+\delta},
\]
therefore produces QSVT-admissible real polynomials bounded by one on the
full signal interval.  QSVT implements the even polynomial
$\widetilde C_R(A_h)$ and the odd polynomial $\widetilde S_R(A_h)$ in
separate phase sequences.  A one-qubit linear combination of these blocks,
followed by robust oblivious amplitude amplification, gives a block encoding
of
\[
  e^{-i\tau A_h}=e^{-itH_h}
\]
to the desired precision.  This is the block-Hamiltonian simulation
construction of Gilyén, Su, Low, and Wiebe~\cite{GilyenSuLowWiebe2019}.

The Chebyshev series above specify the target approximation polynomials;
they do not directly give the QSVT phase angles.  The latter are obtained
from a polynomial-to-phase synthesis procedure.

\begin{remark}[Parity, boundedness, and complex propagators]
A real QSVT polynomial must have definite parity and must be bounded by one
on the full signal interval $[-1,1]$.  The Jacobi--Anger truncation
$C_R$ is even, while $S_R$ is odd, so they can be implemented by separate
QSVT phase sequences after a small rescaling that enforces the boundedness
condition.  The complex propagator
\[
  e^{-i\tau x}=\cos(\tau x)-i\sin(\tau x)
\]
is then obtained by combining the cosine and sine blocks with one additional
ancilla and applying robust oblivious amplitude amplification.  Thus the
full complex exponential is not being treated as one real polynomial of
definite parity.
\end{remark}

A quantum algorithm for the kinetic Schr\"odinger equation can therefore be summarized as follows.
\begin{enumerate}[label=\textbf{Step \arabic*.},leftmargin=*]
\item Prepare the amplitude-encoded initial condition $\ket{\bm u_h(0)}$.
\item Use the LCU block encoding $U_L$ of the three-point Laplacian.  For $H_h=\gamma L_h$, the subnormalization is $\alpha_H=4\gamma/h^2$.
\item Choose a polynomial $p_m$ approximating $x\mapsto e^{-i\alpha_Htx}$ on $[0,1]$.
\item Use QSVT/QSP to implement the corresponding transformation of the block-encoded matrix $H_h/\alpha_H$.
\item The resulting system state approximates $\ket{\bm u_h(t)}=e^{-itH_h}\ket{\bm u_h(0)}$.
\item Estimate observables such as probability mass in a subregion, Fourier-mode occupation, overlaps with test functions, or expected position/momentum-type quantities.
\end{enumerate}
Because Schr\"odinger evolution is unitary, the output norm does not decay or grow.  The remaining readout cost is therefore exactly the measurement cost discussed earlier in the chapter: direct sampling repeats the whole simulation for each shot, whereas amplitude estimation uses controlled simulations and their inverses to obtain quadratic improvement in the target scalar accuracy.  This distinction will reappear in the wave and heat chapters, where the simulation cost and the measurement cost must be multiplied to obtain the end-to-end complexity.

\begin{remark}[No fast-forwarding in the general black-box model]
For a generic sparse Hamiltonian specified through black-box access, simulation cannot in general have sublinear dependence on the scaled evolution time: there are sparse Hamiltonians for which constant-accuracy simulation requires $\Omega(\norm{H}t)$ queries, up to the normalization conventions of the access model \cite{BerryAhokasCleveSanders2007}.  This no-fast-forwarding theorem explains why the generic QSVT/qubitization cost is essentially linear in
\begin{equation}
  \alpha_Ht=\frac{4\gamma t}{h^2}.
\end{equation}
Fast-forwarding is therefore a structural exception, not a property of arbitrary sparse matrices.
\end{remark}

\begin{remark}[Why not simply diagonalize by the QFT?]
For this constant-coefficient periodic example, there is also a direct spectral implementation.  Since $L_h$ is diagonalized by the discrete Fourier transform,
\begin{equation}
  e^{-itH_h}
  =
  F_N^\dag
  \operatorname{diag}\left(
  \exp\left[
  -it\gamma\frac{4}{h^2}\sin^2\left(\frac{\pi k}{N}\right)
  \right]
  \right)_{k=0}^{N-1}
  F_N .
\end{equation}
A quantum circuit can apply the QFT, compute the eigenvalue and its phase by reversible arithmetic, apply the diagonal phase, and uncompute.  Because the matrix is explicitly diagonalizable, this is a structured fast-forwarding mechanism and is not ruled out by the generic lower bound above.

\end{remark}

\section{How to read query and gate counts}
\label{sec:query-gate-counts}

Quantum algorithmic complexity is usually reported in several layers.  For PDE algorithms, it is helpful to keep the following quantities separate:
\begin{itemize}
  \item $N$: number of spatial degrees of freedom, usually with $n=\lceil\log_2N\rceil$ qubits for the solution register;
  \item $C_{\mathrm{prep}}$: gate cost of preparing the input state;
  \item $C_U$: gate cost of one block-encoding circuit $U_A$ or its inverse;
  \item $Q$: number of calls, or queries, to the block encoding;
  \item $\alpha$: block-encoding subnormalization factor;
  \item $p_{\mathrm{succ}}$: postselection success probability;
  \item $M$: number of measurement shots used to estimate the desired output quantity.
\end{itemize}

A typical QSVT-based subroutine has the form
\begin{equation}
  \ket{\bm b}
  \xrightarrow{\;Q\ \text{queries to }U_A\;}
  \text{block encoding of }p(A/\alpha)
  \xrightarrow{\;\text{postselect}\;}
  \frac{p(A/\alpha)\ket{\bm b}}{\norm{p(A/\alpha)\ket{\bm b}}}.
\end{equation}
The leading gate count has the schematic form
\begin{equation}
  C_{\mathrm{total}}
  \approx \Big(
  C_{\mathrm{prep}}
  + Q\,C_U
  + C_{\mathrm{poly}}(Q,n,\log(1/\eps))
  \Big)R_{\mathrm{post}}\,C_{\mathrm{meas}}.
  \label{eq:total-cost-schematic}
\end{equation}
HHere $C_{\mathrm{poly}}$ includes phase rotations, arithmetic, and
controls. For a state-output procedure with success probability $p$,
$R_{\mathrm{post}}$ is $O(1/p)$ for repetition or $O(1/\sqrt p)$
for amplitude amplification, including the required inverse circuits
and reflections. $C_{\mathrm{meas}}$ counts repetitions or coherent
queries for the desired measurement accuracy. When directly estimating
an unnormalized matrix element of a block encoding,
$R_{\mathrm{post}}=1$; its block normalization instead enters the
required estimation precision.

The query count $Q$ is usually comparable to the degree of the polynomial transformation.  For example, inverse filters have degree depending on the condition number, while Hamiltonian simulation depends essentially on the scaled time $\alpha t$.  The gate count then multiplies this query count by the cost of implementing the block encoding itself.  This is why explicit finite-difference and finite-element block encodings matter: they determine $C_U$ and often also the normalization $\alpha$.

\begin{remark}[Comparison with classical complexity]
A classical sparse matrix-vector product is often counted as $O(sN)$ operations.  A quantum block encoding of the same sparse matrix may have gate cost polynomial in $\log N$ for structured matrices, but it acts on amplitude-encoded vectors and returns limited measurement information.  Therefore one should not compare $O(sN)$ directly with $\operatorname{polylog}(N)$ without also comparing input preparation, output readout, conditioning, and the quantity of interest.
\end{remark}

\begin{remark}[Query complexity versus wall-clock runtime]
Query complexity isolates the number of calls to an idealized primitive such as a block encoding or sparse oracle.  Gate complexity expands those primitives into elementary gates.  Physical runtime further depends on architecture, connectivity, error correction, and parallelism.  In this lecture note we mostly discuss query and gate complexity, because these are the architecture-independent quantities most closely connected to numerical linear algebra.
\end{remark}

\section{Summary}

\begin{center}
\begin{tabular}{p{0.25\textwidth}p{0.34\textwidth}p{0.31\textwidth}}
\toprule
Concept & Numerical linear algebra analogue & Main caution \\
\midrule
Quantum state & Unit vector in $\C^N$ & Entries are not directly readable \\
Amplitude encoding & Normalized coefficient vector & Norm is stored separately \\
Measurement & Sampling an observable's spectral measure & One shot gives an eigenvalue, not an expectation \\
Postselection & Conditioning on a successful flag & Small success probability is costly \\
Swap test & Overlap estimator & Gives $|\ip{u}{v}|^2$, not phase \\
Hadamard test & Quadratic-form estimator & Needs controlled unitaries \\
Block encoding & Matrix access model & Subnormalization matters \\
Diagonal matrix & Entrywise multiplication & Inverse still costs unless reciprocals are encoded \\
LCU/SELECT & Block diagonal compression & Coefficients enter through normalization \\
Repeated postselection & Sequential conditioning & Products of probabilities can be exponentially small \\
Pauli strings & Orthogonal matrix basis & Efficient only when expansion is compact \\
Matrix norms & Stability and scaling estimates & Max norm and spectral norm differ by dimension \\
QSVT/QET & Polynomial/eigenvalue matrix function & Degree controls queries \\
Amplitude amplification & Boosting/postselection & Depends on success probability \\
Amplitude estimation & Quadratically faster scalar estimation & Long coherent circuits with controlled access \\
Grover--Rudolph/Kaye--Mosca & Tree-structured data loading & Requires efficient conditional norms/phases \\
QFT & Coherent Fourier transform & Does not output all Fourier coefficients \\
Phase estimation & Coherent eigenvalue readout & Aliasing; resolution costs controlled powers \\
Sparse access & CSR/CSC-like matrix access & Must be implemented reversibly \\
Query/gate counts & Cost model for primitives & Must include input and measurement \\
\bottomrule
\end{tabular}
\end{center}

The guiding principle for PDE applications is that discretization, normalization, block encoding, polynomial transformation, postselection, and readout must be analyzed together. A quantum algorithm is not just a faster matrix-vector multiply. It is a full pipeline for preparing data, transforming amplitudes, and extracting limited but useful scalar information.

\section{Exercises}

\begin{exercise}[Grover--Rudolph tree]
For a four-cell distribution with probabilities $p_{00}$, $p_{01}$, $p_{10}$, $p_{11}$, write the two levels of conditional rotations explicitly and verify that the final amplitudes are $\sqrt{p_{00}},\sqrt{p_{01}},\sqrt{p_{10}},\sqrt{p_{11}}$.
\end{exercise}

\begin{exercise}[QFT versus FFT]
Let $N=2^n$.  Compare the arithmetic cost of a classical FFT on an explicitly stored vector of length $N$ with the gate count of an exact QFT circuit on $n$ qubits.  Explain why this comparison does not mean that the QFT outputs all Fourier coefficients in polylogarithmic time.
\end{exercise}

\begin{exercise}[Phase estimation without aliasing]
Let $H=H^\dag$ with $\operatorname{spec}(H)\subseteq[0,\lambda_{\max}]$, and set $U=e^{-iHt_0}$.  Give a condition on $t_0$ under which distinct eigenvalues in this window correspond to distinct phases $\theta\in[0,1)$, and determine how many clock qubits are needed to resolve two eigenvalues separated by $\Delta\lambda$.  For the discrete Laplacian with $\lambda_{\max}=\Theta(h^{-2})$, how must $t_0$ scale with the mesh size?
\end{exercise}

\begin{exercise}[Sparse isometry construction]
Verify \eqref{eq:sparse-isometry-overlap} for a real nonnegative tridiagonal matrix.  What changes if the off-diagonal entries are negative?
\end{exercise}

\begin{exercise}[Normalization]
Let $f(x)=\sin(2\pi x)$ on $[0,1]$ and let $x_j=j/N$. Compute $\norm{\bm f_h}_2$ asymptotically and verify the scaling $\norm{f}_{L^2}\approx h^{1/2}\norm{\bm f_h}_2$.
\end{exercise}

\begin{exercise}[Swap test]
Derive \eqref{eq:swap-test-probability} directly by expanding the state after the final Hadamard gate.
\end{exercise}

\begin{exercise}[Hadamard-test sign convention]
Repeat the derivation of the imaginary-part Hadamard test.  Show that inserting $S^\dag$ before the final Hadamard gives
\[
  \Pr(0)=\frac12\left(1+\Imag\bra{\psi}W\ket{\psi}\right),
\]
whereas inserting $S$ gives the opposite sign.  This is a useful reminder that phase conventions in circuit identities should be checked algebraically.
\end{exercise}

\begin{exercise}[Observable measurement]
Let $O=\sum_r\lambda_r\Pi_r$ be a Hermitian observable. Starting from the Born probabilities $p_r=\bra{\psi}\Pi_r\ket{\psi}$, verify that the expected measurement outcome is $\bra{\psi}O\ket{\psi}$.
\end{exercise}

\begin{exercise}[Postselection cost]
For the state in \eqref{eq:postselection-decomposition}, compute the expected number of repetitions required to obtain one successful flag measurement. Then compare this with the $O(1/\sqrt p)$ query scaling from amplitude amplification.
\end{exercise}

\begin{exercise}[Forward difference]
Using $S\ket{j}=\ket{j+1\pmod N}$, verify which of $S-I$ or $S^\dag-I$ corresponds to $(D_+u)_j=(u_{j+1}-u_j)/h$ when vectors are represented by their coefficients in the computational basis.
\end{exercise}

\begin{exercise}[Dirichlet correction]
Write the nonperiodic tridiagonal matrix for the one-dimensional Dirichlet Laplacian. Express it as the periodic circulant Laplacian plus a low-rank boundary correction. Discuss how the correction could be incorporated by LCU.
\end{exercise}

\begin{exercise}[Two-dimensional stencil]
For $N_x=N_y=4$, explicitly write the five-point periodic Laplacian as a Kronecker sum. Identify the shift matrices $S_x$ and $S_y$.
\end{exercise}

\begin{exercise}[Condition number]
Compute the eigenvalues of $L_h=(2I-S-S^\dag)/h^2$ on a periodic grid and estimate the condition number on the mean-zero subspace.
\end{exercise}

\begin{exercise}[Diagonal inverse versus diagonal application]
Let $D=\operatorname{diag}(d_0,\ldots,d_{N-1})$ with $d_j\in[1/\kappa,1]$.  Explain why applying $D$ and applying $D^{-1}$ are both $O(N)$ operations classically if all diagonal entries are stored.  Then explain why a block encoding of $D$ does not automatically give a constant-query block encoding of $D^{-1}$ in the QSVT model.  Relate your answer to the polynomial approximation of $1/x$ on $[1/\kappa,1]$.
\end{exercise}

\begin{exercise}[Repeated postselection]
Suppose a quantum subroutine consists of $m$ stages and each stage succeeds with probability $1-\delta$.  Show that the total success probability is $(1-\delta)^m\approx e^{-m\delta}$ for small $\delta$.  If $\delta=c/m$, what constant success probability is obtained in the limit $m\to\infty$?  What happens if $\delta$ is independent of $m$?
\end{exercise}

\begin{exercise}[Pauli expansion]
For a one-qubit Hermitian matrix
\begin{equation}
  H=\begin{bmatrix}a&c-id\\c+id&b\end{bmatrix},
\end{equation}
with $a,b,c,d\in\R$, write $H$ as a linear combination of $I,X,Y,Z$.  Then verify the coefficient formula \eqref{eq:pauli-expansion} in this case.
\end{exercise}

\begin{exercise}[Recovering an unnormalized norm]
A block-encoding circuit with normalization $\alpha$ prepares the desired vector $\bm v/\alpha$ in the successful ancilla branch.  If the success probability is estimated to be $p$, show that $\norm{\bm v}=\alpha\sqrt p$.  If a normalized-state measurement estimates $\bra{\bm v}O\ket{\bm v}$ to accuracy $\eta$, how does the uncertainty in $p$ enter the estimate of the unnormalized quadratic form $\bm v^\dag O \bm v$?
\end{exercise}

\begin{exercise}[Exponentiating a Pauli string]
Let
\[
P=P_{n-1}\otimes\cdots\otimes P_0,
\qquad
P_j\in\{I,X,Y,Z\}.
\]
Show that \(P^2=I\), and hence
\[
e^{-i\theta P}
=
\cos(\theta)I-i\sin(\theta)P.
\]
Describe a circuit for this unitary using single-qubit basis changes, a CNOT
parity ladder, and one \(R_z(2\theta)\) rotation. If \(P\) has Pauli weight
\(w\), estimate the number of one- and two-qubit gates required.
\end{exercise}

\begin{exercise}[Concentration of measurement outcomes]
Let \(O=O^\dagger\) have eigenvalues in \([a,b]\), and let
\(X_1,\ldots,X_m\) be independent outcomes obtained by measuring \(O\) in
the state \(\ket{\psi}\). Define
\[
\widehat\mu_m=\frac1m\sum_{r=1}^m X_r,
\qquad
\mu=\langle\psi|O|\psi\rangle.
\]
Use the bounded-difference inequality to show that
\[
\Pr\!\left(|\widehat\mu_m-\mu|\geq\epsilon\right)
\leq
2\exp\left(-\frac{2m\epsilon^2}{(b-a)^2}\right).
\]
Specialize this result to a Pauli observable, whose outcomes are
\(\pm1\), and determine how many measurements suffice to achieve additive
error \(\epsilon\) with failure probability at most \(\delta\).
\end{exercise}

\begin{exercise}[Pauli expansion of the one-dimensional Laplacian]
Let \(N=2^n\) and define the noncyclic shift
\[
S_N=\sum_{j=0}^{N-2}\ket{j+1}\bra{j}.
\]
Using
\[
\sigma_+=\ket{1}\bra{0}=\frac{X-iY}{2},
\qquad
\sigma_-=\ket{0}\bra{1}=\frac{X+iY}{2},
\]
show that
\[
S_N
=
\sum_{k=0}^{n-1}
I^{\otimes(n-k-1)}
\otimes\sigma_+
\otimes\sigma_-^{\otimes k}.
\]
Expand \(S_N+S_N^\dag\) into Pauli strings and show that it contains at
most
\[
\sum_{k=0}^{n-1}2^k=N-1
\]
terms. Conclude that the tridiagonal finite-difference Laplacian
\[
L_h=\frac{1}{h^2}\left(2I-S_N-S_N^\dag\right)
\]
has a Pauli expansion with \(O(N)\) terms.
\end{exercise}
% Body-only LaTeX file. It is intended to be incorporated by \include{...}
% into a larger book manuscript. No documentclass or package preamble is included.

\chapter{Quantum Algorithms for Elliptic PDEs}
\label{chap:elliptic-quantum}

\section{Classical discretization and the direct QLSA route}
\label{sec:elliptic-classical-discretization}

Elliptic boundary value problems appear whenever a system has reached a static equilibrium.  In electrostatics, the potential satisfies a Poisson equation with charge density as the source.  In steady heat conduction, the temperature satisfies a diffusion equation with heat sources.  In linear elasticity, displacement fields are determined by the balance between elastic stresses and applied loads.  These examples have different physical meanings, but they share a common mathematical structure: a coercive second-order operator, a boundary condition, and a variational energy principle.  Classical background on elliptic equations may be found in \cite{Evans2010PDE} and finite-difference and finite-element discretizations are treated extensively in \cite{LarssonThomee2009,MortonMayers2005,BrennerScott2008}.

The simplest and most familiar model is the Poisson equation.  We state the model problem in the variable-coefficient Dirichlet form
\begin{equation}
  -\nabla\cdot\bigl(a(\bm{x})\nabla u(\bm{x})\bigr)=f(\bm{x}),
  \qquad \bm{x}\in\Omega,
  \qquad u|_{\partial\Omega}=0,
  \label{eq:elliptic-model}
\end{equation}
where $\Omega\subset\mathbb{R}^d$ is bounded and $a(\bm{x})\in\mathbb{R}^{d\times d}$ is symmetric and uniformly elliptic: there are constants $0<a_{\min}\le a_{\max}<\infty$ such that
\begin{equation}
  a_{\min}|\bm{\xi}|^2
  \leq
  \bm{\xi}^{T}a(\bm{x})\bm{\xi}
  \leq
  a_{\max}|\bm{\xi}|^2,
  \qquad \bm{\xi}\in\mathbb{R}^d,
  \quad  \bm{x}\in\Omega.
  \label{eq:uniform-ellipticity}
\end{equation}
We may express \eqref{eq:elliptic-model} simply as $Lu=f$, where $L=-\nabla\cdot(a\nabla)$ is the elliptic operator.  This notation will be used again in the hyperbolic and parabolic chapters: the wave equation evolves according to $u_{tt}+Lu=0$, while the heat equation evolves according to $u_t+Lu=0$.  Thus the same spatial operator gives rise to three different classes of PDEs.

The Poisson equation is the special case $a(\bm{x})=I$:
\begin{equation}
  -\Delta u=f,
  \qquad u|_{\partial\Omega}=0,
  \qquad
  \Delta u=\sum_{j=1}^{d}\frac{\partial^2u}{\partial x_j^2}.
\end{equation}
The sign convention makes $L$ a positive definite operator under appropriate boundary conditions, e.g., the homogeneous Dirichlet condition.  This positivity is the source of stability in the continuous problem. After discretization, stability is inherited by methods that preserve a discrete analogue of it, in the form of a discrete maximum principle or a coercivity estimate.

Whatever discretization is chosen, the computational problem usually takes the algebraic form
\begin{equation}
  A_h\bm{x}_h=\bm{b}_h.
  \label{eq:discrete-linear-system}
\end{equation}
Quantum algorithms targeting \eqref{eq:discrete-linear-system} are known as quantum linear-system algorithms (QLSA). 
A quantum linear-system algorithm does not normally return every component of $\bm{x}_h$.  Its natural output is instead a normalized state
\begin{equation}
  \ket{\bm{x}_h}
  :=\frac{\bm{x}_h}{\|\bm{x}_h\|_2},
  \label{eq:elliptic-solution-state}
\end{equation}
from which selected overlaps or observables are estimated.

For our discussion of QLSA here, the notation and access model are those of Chapter~\ref{chap:basic-elements}: $\bm{b}_h$ is amplitude encoded as in \eqref{eq:amplitude-encoding-grid-vector}, continuum and Euclidean norms are compared using \eqref{eq:L2-l2-scaling}, and $A_h$ is supplied through a block encoding in the sense of \eqref{eq:block-encoding-definition}.  Throughout this chapter, $\bm{x}_h$ denotes the algebraic vector in the coordinates used for amplitude encoding.  In finite differences this is simply the vector of grid values, often denoted $\bm{u}_h$ locally.  In finite elements, we reserve $\bm{y}$ for the coefficient vector in a nodal basis and use $\bm{x}=M_h^{1/2}\bm{y}$ for the mass-normalized coordinate vector.  Likewise, $\bm{b}_h$ denotes the right-hand side after the corresponding coordinate normalization.  We use $N_h$ generically for the total number of spatial degrees of freedom when the precise finite-difference or finite-element convention is not important.

Thus one must distinguish the discretization error, the quantum algorithmic error, and the measurement error:
\begin{equation}
  u
  \longrightarrow
  u_h
  \longrightarrow
  \ket{\bm{x}_h}
  \longrightarrow
  \text{quantity of interest}.
  \label{eq:elliptic-end-to-end-pipeline}
\end{equation}

Figure~\ref{fig:elliptic-pipeline} previews the most direct route through this pipeline.  Its third box announces the theme of the chapter: for elliptic PDEs, the condition number is not an independent input to the linear-algebra problem but a consequence of the discretization, growing as $\Theta(h^{-2})$ exactly when the mesh is refined to meet an accuracy target.   This section follows the direct route and pays the full condition number.  Section~\ref{sec:factorized-dilation} exploits the gradient factorization $A_h=G_h^\dag G_h$ to work at first order, where only the square root of the condition number appears.  Section~\ref{sec:qft-spectral-filtering} bypasses the generic inverse filter altogether when geometry, coefficients, and boundary conditions permit explicit diagonalization; Section~\ref{sec:elliptic-eigenvalue-qpe} then treats the associated eigenvalue problem.  Across all three routes, one pattern is worth watching for from the start: a saving in the spectral condition number never simply disappears.  It moves elsewhere in the pipeline---into an output-block weight, a postselection probability, or the preparation of a transformed input---and an honest end-to-end accounting must follow it there.

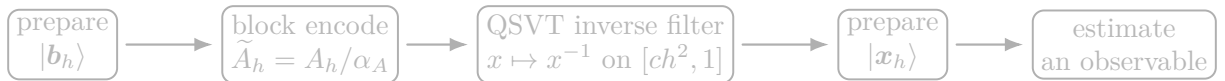
\begin{figure}[htbp]
\centering
\resizebox{\textwidth}{!}{%
\begin{tikzpicture}[x=1cm,y=1cm,line cap=round,line join=round,>=Latex,
    every node/.style={font=\small},
    flow/.style={gray!62,line width=0.8pt,->,shorten >=3pt,shorten <=3pt},
    box/.style={draw=gray!62,fill=white,rounded corners=4pt,
                line width=0.7pt,align=center,inner sep=3pt,text=gray!56}]
  \node[box] (prep) at (0,0)    {prepare\\$\ket{\bm{b}_h}$};
  \node[box] (enc)  at (3.1,0)  {block encode\\$\widetilde A_h=A_h/\alpha_A$};
  \node[box] (inv)  at (6.8,0)  {QSVT inverse filter\\$x\mapsto x^{-1}$ on $[c h^2,1]$};
  \node[box] (out)  at (10.4,0) {prepare\\$\ket{\bm{x}_h}$};
  \node[box] (meas) at (13.3,0) {estimate\\an observable};
  \draw[flow] (prep)--(enc);
  \draw[flow] (enc)--(inv);
  \draw[flow] (inv)--(out);
  \draw[flow] (out)--(meas);
\end{tikzpicture}%
}
\caption{The direct elliptic QLSA pipeline.  Even if one application of the block encoding costs only $\operatorname{polylog}(N)$ gates, the inverse filter must resolve the normalized spectral interval down to $\Theta(h^2)$, producing a query count $\widetilde O(h^{-2})$ before state-preparation and readout costs are included.}
\label{fig:elliptic-pipeline}
\end{figure}

\paragraph{Notation for sizes.}
The following table fixes the size notation used in this chapter.

\begin{center}
\begin{tabular}{ll}
\toprule
symbol & meaning in this chapter \\
\midrule
$h$ & mesh width or grid size \\
$L_{\rm box}$ & side length of a model tensor-product box \\
$N_{\rm cell}=L_{\rm box}/h$ & number of grid cells per coordinate direction \\
$N_h=N_{\rm cell}-1$ & number of one-dimensional interior grid unknowns \\
 $N_{\rm dof}$ & total number of spatial degrees of freedom \\
$N=2^n$ & dimension of a quantum register after power-of-two padding \\
\bottomrule
\end{tabular}
\end{center}

This small dictionary is included because classical PDE texts often use $N$ for the number of grid points, whereas quantum algorithms often reserve $N=2^n$ for the dimension of an $n$-qubit register.  Padding to the next power of two is a bookkeeping issue, not a change in the PDE discretization.

\subsection{Finite differences in one space dimension}
\label{subsec:fd-1d}

Let $\Omega=(0,L_{\rm box})$ and consider
\begin{equation}
  -u''(x)=f(x),
  \qquad u(0)=u(L_{\rm box})=0.
\end{equation}
Choose a mesh width $h$ with $N_{\rm cell}=L_{\rm box}/h\in\mathbb{N}$ and grid points 
\begin{equation}
    \label{1dgrids}
    x_j=jh, \quad j=0,1,\ldots,N_{\rm cell}.
\end{equation}
Accordingly, the unknowns are the $N_h=N_{\rm cell}-1$ interior values
\[
  U_j\approx u(x_j),\qquad j=1,\ldots,N_h.
\]
This convention is the one used later in the time-dependent chapters: $h$ is the grid/cell size and $L_{\rm box}/h$ is the number of cells per coordinate direction.  For $L_{\rm box}=1$, this is equivalent to the familiar notation $h=1/(N_h+1)$ if $N_h$ denotes the number of interior unknowns.

The centered second difference gives
\begin{equation}
  -\frac{U_{j+1}-2U_j+U_{j-1}}{h^2}=f(x_j),
  \qquad j=1,\ldots,N_h.
  \label{eq:fd-1d-equation}
\end{equation}
With
\begin{equation}
  \bm{u}_h=(U_1,\ldots,U_{N_h})^{\dag},
  \qquad
  (\bm{b}_h)_j=f(x_j),
\end{equation}
we obtain
\begin{equation}
  A_h\bm{u}_h=\bm{b}_h,
  \qquad
  A_h=\frac1{h^2}
  \begin{bmatrix}
    2 & -1 \\
    -1 & 2 & -1 \\
       & \ddots & \ddots & \ddots \\
       &        & -1 & 2 & -1 \\
       &        &    & -1 & 2
  \end{bmatrix}_{N_h\times N_h}.
  \label{eq:one-d-laplacian}
\end{equation}
This tridiagonal system is the simplest nontrivial elliptic discretization, and it serves as the running example for every quantum construction in this chapter.  Three properties of the matrix $A_h$ organize everything that follows: a stability estimate, an explicit spectrum, and a gradient factorization.

First, Taylor expansion gives a local truncation error $O(h^2)$ for smooth solutions.  The discrete maximum principle then converts this consistency estimate into a global maximum-norm error estimate; see, for example, the finite-difference treatments in \cite[Chapter~4]{LarssonThomee2009} and \cite{MortonMayers2005}:
\begin{equation}
  \max_{1\leq j\leq N_h}|U_j-u(x_j)|
  \leq
  C h^2\|u^{(4)}\|_{L^\infty(0,L_{\rm box})}.
  \label{eq:fd-second-order-error}
\end{equation}
The maximum principle is important conceptually: it is one finite-dimensional analogue of elliptic stability.  A quantum algorithm can change how the linear system is solved, but it does not remove the underlying discretization and stability analysis.

Second, the matrix is diagonalized by a discrete sine transform (DST), so its eigenvalues are explicit:
\begin{equation}
  \lambda_k(A_h)
  =\frac4{h^2}\sin^2\!\left(\frac{k\pi}{2(N_h+1)}\right),
  \qquad k=1,\ldots,N_h.
  \label{eq:fd-spectrum}
\end{equation}
Consequently, the spectral range is explicit:
\begin{equation}
  \lambda_{\min}(A_h)=\Theta(L_{\rm box}^{-2}),
  \qquad
  \lambda_{\max}(A_h)=\Theta(h^{-2}),
  \qquad
  \kappa(A_h)=\Theta\!\left(\frac{L_{\rm box}^2}{h^2}\right).
  \label{eq:elliptic-kappa}
\end{equation}
For fixed physical domain size $L_{\rm box}=O(1)$ this reduces to the usual $\Theta(h^{-2})$ condition-number growth.  Written in terms of the number of cells per coordinate, $N_{\rm cell}=L_{\rm box}/h$, the condition number is $\Theta(N_{\rm cell}^2)$.

Third, the same matrix is the product of a discrete gradient and a discrete divergence.  Let $G_h\in\mathbb{R}^{(N_h+1)\times N_h}$ denote the forward difference with homogeneous boundary values:
\begin{equation}\label{G1-grad}
  (G_h\bm{u}_h)_j=\frac{U_j-U_{j-1}}{h},
  \qquad j=1,\ldots,N_h+1,
\end{equation}
where $U_0=U_{N_h+1}=0$.  Then
\begin{equation}
  A_h=G_h^{\dag}G_h,
  \label{eq:fd-gradient-factorization}
\end{equation}
and
\begin{equation}
  \|G_h\|=\Theta(h^{-1}).
  \label{eq:factor-scale}
\end{equation}
This first-order scale will be exploited in Section~\ref{sec:factorized-dilation}.

\subsection{The five-point Laplacian on a square}
\label{subsec:fd-2d}

For a square box with the same mesh width in both directions, let $N_h$ denote the number of interior grid points per coordinate.  With the unknowns $U_{i,j}$ approximating $u(x_i,x_j)$, the five-point scheme is
\begin{equation}
  -\Delta_h U_{i,j}
  :=\frac{4U_{i,j}-U_{i+1,j}-U_{i-1,j}-U_{i,j+1}-U_{i,j-1}}{h^2}
  =f(x_{i,j}).
  \label{eq:five-point}
\end{equation}
Clearly this is a generalization of the central difference scheme \eqref{eq:fd-1d-equation} to two dimensions. 

Equivalently, with appropriate ordering, we can express the matrix via tensor products,
\begin{equation}
  A_h=I_{N_h}\otimes T_h+T_h\otimes I_{N_h},
  \qquad
  T_h=\frac1{h^2}\operatorname{tridiag}(-1,2,-1).
  \label{eq:kronecker-laplacian}
\end{equation}
The method is again second order for smooth solutions, and its eigenvectors are tensor products of the sine modes from \eqref{eq:fd-spectrum}.  The eigenvalues are
\begin{equation}
  \lambda_{k,\ell}(A_h)
  =\frac4{h^2}\left[
  \sin^2\!\left(\frac{k\pi}{2(N_h+1)}\right)
  +
  \sin^2\!\left(\frac{\ell\pi}{2(N_h+1)}\right)
  \right].
  \label{eq:fd-spectrum-2d}
\end{equation}
Thus $\kappa(A_h)=\Theta((L_{\rm box}/h)^2)$ in any fixed dimension.  In $d$ dimensions,
\begin{equation}
  A_h=\sum_{m=1}^{d}
  I_{N_h}^{\otimes(m-1)}\otimes T_h\otimes I_{N_h}^{\otimes(d-m)}.
  \label{eq:d-dimensional-kronecker-sum}
\end{equation}

Finally, the corresponding discrete gradient can be written explicitly.  Let $G_1$ be the one-dimensional forward-difference matrix from \eqref{G1-grad}; then in two dimensions
\begin{equation}
  G_h=
  \begin{bmatrix}
    I_{N_h}\otimes G_1\\
    G_1\otimes I_{N_h}
  \end{bmatrix},
  \label{eq:2d-discrete-gradient}
\end{equation}
up to the convention used for lexicographic ordering.  Again one has the scaling,
\begin{equation}
  A_h=G_h^{\dag}G_h,
  \qquad
  \|G_h\|=\Theta(h^{-1}).
  \label{eq:2d-factorization}
\end{equation}

\subsection{Finite elements, mass normalization, and the stiffness matrix}
\label{subsec:fem}

Finite elements make the variational structure explicit and extend naturally to nonrectangular domains.  The weak problem is: find $u\in H_0^1(\Omega)$ such that
\begin{equation}
  a(u,v)=\ell(v),
  \qquad v\in H_0^1(\Omega),
  \label{eq:weak-form}
\end{equation}
where
\begin{equation}
  a(u,v)=\int_\Omega \nabla v(x)^\dag a(x)\nabla u(x)\,dx,
  \qquad
\ell(v)=\int_\Omega f(x)\overline{v(x)}\,dx.
  \label{eq:elliptic-bilinear-form}
\end{equation}
Let $V_h\subset H_0^1(\Omega)$ have basis $\{\phi_j\}_{j=1}^{N_h}$.  Writing
\begin{equation}
  u_h(x)=\sum_{j=1}^{N_h}(\bm{y})_j\phi_j(x),
\end{equation}
the weak form \eqref{eq:weak-form} implies that
\begin{equation}
  K_h\bm{y}=\bm{f}_h,
  \label{eq:fem-linear-system}
\end{equation}
where
\begin{equation}
  (K_h)_{ij}=a(\phi_j,\phi_i),
  \qquad
  (\bm{f}_h)_i=\ell(\phi_i).
  \label{eq:stiffness-load}
\end{equation}
For conforming $P_1$ elements on a shape-regular, quasi-uniform mesh,
assume $u\in H^2(\Omega)$ and the dual elliptic regularity needed for the
$L^2$ estimate. Then the following estimates are standard in finite element theory; see \cite[Chapter~5]{LarssonThomee2009} and \cite{BrennerScott2008}:
\begin{equation}
  \|u-u_h\|_{H^1(\Omega)}\leq Ch\|u\|_{H^2(\Omega)},
  \qquad
  \|u-u_h\|_{L^2(\Omega)}\leq Ch^2\|u\|_{H^2(\Omega)}.
  \label{eq:p1-fem-error}
\end{equation}

The coefficient vector $\bm{y}$ is not an $L^2$-orthonormal coordinate representation.  With the mass matrix
\begin{equation}
  (M_h)_{ij}=\int_\Omega \overline{\phi_i(x)}\,\phi_j(x)\,dx,
  \label{eq:fem-mass-matrix}
\end{equation}
we have
\begin{equation}
  \|u_h\|_{L^2(\Omega)}^2=\bm{y}^\dag M_h\bm{y}.
\end{equation}
The mass-normalized variables are therefore
\begin{equation}
  \bm{x}=M_h^{1/2}\bm{y},
  \qquad
  A_h=M_h^{-1/2}K_hM_h^{-1/2},
  \qquad
  \bm{b}_h=M_h^{-1/2}\bm{f}_h.
  \label{eq:mass-normalized-system}
\end{equation}
Then
\begin{equation}
  A_h\bm{x}=\bm{b}_h,
  \qquad
  \|\bm{x}\|_2=\|u_h\|_{L^2(\Omega)}.
\end{equation}
For a fixed-degree nodal basis on a shape-regular, quasi-uniform mesh,
the suitably scaled mass matrix is spectrally equivalent to the identity,
so $\kappa(M_h)=O(1)$. If its block encoding also satisfies
$\alpha_M/\lambda_{\min}(M_h)=O(1)$, a suitably scaled inverse-square-root
filter has degree $O(\log(1/\epsilon))$. Thus this spectral transformation
need not introduce an additional algebraic power of $h^{-1}$.
The cost of constructing the encoding and converting input and output
coordinates must still be included.

Similar to the finite difference approach in the previous section, the stiffness matrix also has a Gram factorization.  If $B_h$ collects elementwise discrete gradients and $W_h$ is the positive quadrature-weight matrix, then
\begin{equation}
  K_h=B_h^\dag W_h B_h.
\end{equation}
Consequently,
\begin{equation}
  A_h=G_h^\dag G_h,
  \qquad
  G_h=W_h^{1/2}B_hM_h^{-1/2}.
  \label{eq:fem-gradient-factorization}
\end{equation}
For fixed polynomial degree on quasi-uniform meshes,
\begin{equation}
  \|G_h\|=\Theta(h^{-1}),
  \qquad
  \|A_h\|=\Theta(h^{-2}),
  \qquad
  \kappa(A_h)=\Theta(h^{-2})\quad \text{on a fixed domain}.
  \label{eq:fem-scales}
\end{equation}
Thus $A_h$ is not merely a sparse positive matrix: it is the discrete energy operator associated with the weighted gradient $G_h$.

\subsection{Direct block encoding of $A_h$ and its mesh-dependent cost}
\label{subsec:elliptic-block-encoding}

The original quantum linear-system literature focused on optimizing the dependence on the matrix condition number, beginning with HHL and continuing through LCU- and QSVT-based solvers \cite{HarrowHassidimLloyd2009,GilyenSuLowWiebe2019}.  Elliptic PDEs add an important feature: the condition number is not an independent input parameter.  It grows with the number of grid points required by the PDE discretization.

For fixed-order finite differences, and for mass-lumped finite elements in fixed dimension, the normalized operator has $O(1)$ nonzeros per row.  The sparse-access construction of Section~\ref{sec:sparse-block-encoding-access}, or the explicit structured constructions of Section~\ref{sec:finite-difference-block-encodings}, then gives a block encoding of
\begin{equation}
  \widetilde A_h:=\frac{A_h}{\alpha_A},
  \qquad
  \alpha_A=\Theta(h^{-2}).
  \label{eq:elliptic-direct-normalization}
\end{equation}
 Because $\lambda_{\min}(A_h)=\Theta(L_{\rm box}^{-2})$, the normalized spectrum satisfies
\begin{equation}
  \operatorname{spec}(\widetilde A_h)
  \subseteq[c h^2/L_{\rm box}^2,1]
  \label{eq:elliptic-normalized-spectrum}
\end{equation}
for a mesh-independent constant $c>0$.

The QSVT inverse construction in \eqref{eq:qsvt-inverse-approximation} must approximate $1/x$ down to $x=\Theta(h^2/L_{\rm box}^2)$.  Its polynomial degree, and hence the number of calls to the block encoding of $A_h$, is
\begin{equation}
  Q_A
  =O\!\left(\kappa(A_h)\log\frac{\kappa(A_h)}{\epsilon}\right)
  =\widetilde O\!\left(\frac{L_{\rm box}^2}{h^2}\log\frac1{\epsilon}\right).
  \label{eq:elliptic-direct-qsvt-cost}
\end{equation}
This is a query count.  The gate count multiplies it by the cost of one block-encoding circuit, as explained in Section~\ref{sec:query-gate-counts}. Meanwhile, state preparation, postselection, and observable estimation remain additional costs.

It is useful to express the same scaling in terms of both the mesh width and the number of unknowns.  If $N_{\rm cell}=L_{\rm box}/h$ is the number of cells per coordinate and $N_{\rm dof}=\Theta(N_{\rm cell}^d)$, then
\begin{equation}
  Q_A=\widetilde O(N_{\rm cell}^{2})
  =\widetilde O(N_{\rm dof}^{2/d})
  \qquad (L_{\rm box}\ \text{fixed}).
  \label{eq:elliptic-direct-size-cost}
\end{equation}
Thus an efficient, even polylogarithmic, circuit for one sparse-matrix query does not by itself give a polylogarithmic elliptic solver.  The inverse filter must still resolve a mesh-dependent spectral interval.  If the second-order discretization error is balanced with a target accuracy by taking $h^2=\Theta(\epsilon)$, then the stated direct inverse-filter query bound becomes $\widetilde O(\epsilon^{-1})$ on a fixed box.  This mesh dependence is a central difficulty for quantum PDE algorithms and motivates the alternatives in the next two sections.

Finally, quantum registers have dimensions that are powers of two, while a Dirichlet grid usually produces $N_h=N_{\rm cell}-1$ unknowns.  In practice one pads the matrix and vector to the next power of two, for example by adding an identity block or by adding dummy eigenvalues outside the spectral window of interest.  Such padding changes constants but not the $h$-dependent estimates above.

\section{First-order factorization and Hermitian dilation}
\label{sec:factorized-dilation}

The direct QLSA route treats the second-order stiffness matrix \(A_h\) as
the primitive operator.  Since
\[
  \kappa(A_h)=\Theta\!\left((L_{\rm box}/h)^2\right),
\]
a direct inverse filter for \(A_h\) has the mesh-dependent scale discussed in
the previous section.  Elliptic structure suggests a more first-order
viewpoint.  With an appropriate finite difference or finite element
discretization, the stiffness matrix has a gradient--divergence
factorization,
\begin{equation}
  A_h=G_h^\dag G_h.
  \label{eq:elliptic-general-factorization}
\end{equation}
This is the discrete analogue of
\[
  -\nabla\cdot(a\nabla u)
  =
  (\sqrt a\,\nabla)^\dag(\sqrt a\,\nabla)u .
\]
For instance, the finite-difference factorization in
\eqref{eq:fd-gradient-factorization} and the finite-element factorization in
\eqref{eq:fem-gradient-factorization} are both of this form.  The matrix
\(G_h\) should be thought of as a weighted discrete gradient, while
\(G_h^\dag\) is the corresponding negative discrete divergence.

\subsection{The first-order spectral scale}
\label{subsec:elliptic-first-order-scale}

Recall that \(L_{\rm box}\) denotes the characteristic length of the
computational domain.  For conforming Dirichlet discretizations, the nonzero singular
values of \(G_h\) satisfy
\begin{equation}
  \sigma_{\min}(G_h)=\Theta(L_{\rm box}^{-1}),
  \qquad
  \sigma_{\max}(G_h)=\Theta(h^{-1}),
  \qquad
  \kappa_{\mathrm{eff}}(G_h)=\Theta(L_{\rm box}/h).
  \label{eq:elliptic-factor-spectrum}
\end{equation}
Here \(\kappa_{\mathrm{eff}}\) means the ratio of the largest singular value
to the smallest nonzero singular value.

It is useful to see this in one dimension.  On \((0,L_{\rm box})\), let \(G_h\)
be the forward-difference gradient with homogeneous Dirichlet boundary
conditions. The nodal operator $G_h^\dag G_h$ is diagonalized by
sine modes; the edge operator $G_hG_h^\dag$ has the corresponding
cosine modes, including a constant null mode. The nonzero singular values are
\begin{equation}
  \sigma_k(G_h)
  =
  \frac{2}{h}
  \sin\!\left(\frac{k\pi}{2(N_h+1)}\right),
  \qquad
  k=1,\ldots,N_h .
  \label{eq:elliptic-gradient-singular-values-1d}
\end{equation}
Thus
\[
  \sigma_1(G_h)\sim \frac{\pi}{L_{\rm box}},
  \qquad
  \sigma_{N_h}(G_h)\sim \frac{2}{h}.
\]
This gives \eqref{eq:elliptic-factor-spectrum}.  In several dimensions and
on quasi-uniform meshes, the lower bound follows from a discrete Poincar\'e
inequality,
\[
  \|\bm x\|_2
  \le
  C L_{\rm box}\|G_h\bm x\|_2,
\]
while the upper bound follows from the inverse estimate
\[
  \|G_h\bm x\|_2
  \le
  C h^{-1}\|\bm x\|_2 .
\]

Since \(A_h=G_h^\dag G_h\), the first-order relation
\begin{equation}
  \kappa_{\mathrm{eff}}(G_h)
  =
  \sqrt{\kappa(A_h)}
  \label{eq:elliptic-square-root-condition}
\end{equation}
is immediate.  This square-root relation is the source of the possible
improvement, but it is not by itself a complete quantum algorithm.  The
algorithm must also specify which right-hand side is prepared, which output
block is measured, and in which norm the output state is normalized.

\subsection{Hermitian dilation of the discrete gradient}
\label{subsec:hermitian-dilation-factor}

For a rectangular matrix \(G\in\mathbb C^{m\times n}\), define its Hermitian
dilation by
\begin{equation}
  \mathcal D(G)
  =
  \begin{bmatrix}
    0&G^\dag\\
    G&0
  \end{bmatrix}.
  \label{eq:hermitian-dilation}
\end{equation}
If
\[
  G=\sum_j\sigma_j\ket{\bm u_j}\bra{\bm v_j}
\]
is a singular-value decomposition, then \(\mathcal D(G)\) has nonzero
eigenvalues \(\pm\sigma_j\), with eigenvectors
\[
  \frac{1}{\sqrt2}
  \begin{bmatrix}
    \bm v_j\\
    \pm\bm u_j
  \end{bmatrix}.
\]
Moreover,
\begin{equation}
  \mathcal D(G)^+
  =
  \begin{bmatrix}
    0&G^+\\
    (G^\dag)^+&0
  \end{bmatrix},
  \label{eq:hermitian-dilation-pseudoinverse}
\end{equation}
where \(^{+}\) denotes the Moore--Penrose pseudo-inverse.

Suppose that \(G_h/\alpha_G\) has a block encoding with
\[
  \alpha_G=\Theta(\|G_h\|)=\Theta(h^{-1}).
\]
Then the same access gives a block encoding of
\(\mathcal D(G_h)/\alpha_G\).  Its nonzero spectrum lies in
\begin{equation}
  \operatorname{spec}
  \left(
    \frac{\mathcal D(G_h)}{\alpha_G}
  \right)
  \setminus\{0\}
  \subseteq
  [-1,-c\,h/L_{\rm box}]
  \cup
  [c\,h/L_{\rm box},1].
  \label{eq:elliptic-dilation-spectrum}
\end{equation}
Thus a QSVT pseudo-inverse filter for the dilation has polynomial degree
\[
  \widetilde O(L_{\rm box}/h),
\]
rather than the \(\widetilde O((L_{\rm box}/h)^2)\) scale of the direct
second-order inverse filter.

The identity
\[
  \mathcal D(G_h)^2
  =
  \begin{bmatrix}
    G_h^\dag G_h&0\\
    0&G_hG_h^\dag
  \end{bmatrix}
\]
also shows how the elliptic operator is embedded. Two access models
must be distinguished. Forming a block encoding of $\mathcal D(G_h)^2$
and treating it as a generic input recovers the second-order
condition-number scale $\kappa(A_h)$. Alternatively, composing two
scaled pseudo-inverse filters for $\mathcal D(G_h)$ with separate
ancillas gives a scaled $A_h^{-1}$ in the upper-left block, still using
$\widetilde O(L_{\rm box}/h)$ queries to the first-order encoding.
The inverse normalizations multiply, however, and the successful output
amplitude must be tracked. The square-root saving can persist at the
operator-filter level. In the end, whether it survives state preparation and
readout depends on the data and output. The mixed and transformed
formulations below expose these costs in complementary ways.

\subsection{Mixed first-order formulation}
\label{subsec:elliptic-mixed-first-order}

The most direct PDE formulation is to introduce the discrete flux
\[
  \bm p=G_h\bm x .
\]
Then
\[
  G_h^\dag G_h\bm x=\bm b
\]
is equivalent to
\begin{equation}
  \bm p-G_h\bm x=0,
  \qquad
  G_h^\dag\bm p=\bm b .
  \label{eq:elliptic-first-order-flux-system}
\end{equation}
In physical terms, \(\bm p\) is the flux (in Darcy flow, the velocity): the
first equation is the constitutive law and the second is conservation.  This
pair leads to the Hermitian indefinite system
\begin{equation}
  \mathsf A_h^{\rm mix}
  \begin{bmatrix}
    \bm p\\
    \bm x
  \end{bmatrix}
  =
  \begin{bmatrix}
    0\\
    -\bm b
  \end{bmatrix},
  \qquad
  \mathsf A_h^{\rm mix}
  =
  \begin{bmatrix}
    I&-G_h\\
    -G_h^\dag&0
  \end{bmatrix}.
  \label{eq:elliptic-mixed-saddle-system}
\end{equation}
This system is indefinite, but modern QLSA and QSVT methods can invert
Hermitian matrices whose spectra are separated from zero.

There is a small but important scaling issue in
\eqref{eq:elliptic-mixed-saddle-system}.  The identity block is dimensionless,
while \(G_h\) is a discrete derivative.  If \(L_{\rm box}\) is varied and the raw
matrix is used without scaling, the resulting Euclidean condition number is
not a dimensionally invariant quantity.  The better formulation is to
nondimensionalize the gradient.  This observation suggests the following choice of scaling parameter $s$, 
\[
  s=L_{\rm box},
  \qquad
  \widetilde{\bm p}=sG_h\bm x ,
\]
and solve
\begin{equation}
  \mathsf A_{h,s}^{\rm mix}
  \begin{bmatrix}
    \widetilde{\bm p}\\
    \bm x
  \end{bmatrix}
  =
  \begin{bmatrix}
    0\\
    -s^2\bm b
  \end{bmatrix},
  \qquad
  \mathsf A_{h,s}^{\rm mix}
  =
  \begin{bmatrix}
    I&-sG_h\\
    -sG_h^\dag&0
  \end{bmatrix}.
  \label{eq:elliptic-scaled-mixed-system}
\end{equation}
The singular values of \(sG_h\) lie in
\[
  [\,\Theta(1),\Theta(L_{\rm box}/h)\,].
\]
On the two-dimensional subspace associated with a singular value
\(\rho\) of \(sG_h\), the matrix reduces to
\[
  \begin{bmatrix}
    1&-\rho\\
    -\rho&0
  \end{bmatrix},
\]
whose eigenvalues are
\begin{equation}
  \lambda_\pm(\rho)
  =
  \frac{1\pm\sqrt{1+4\rho^2}}{2}.
  \label{eq:elliptic-mixed-eigenvalues}
\end{equation}
Note that \(\lambda_+\lambda_-=-\rho^2<0\): every singular value contributes
one positive and one negative eigenvalue, so the system is genuinely
indefinite, with spectrum separated from zero on both sides---exactly the
setting that QSVT inverse filters handle.  Since \(\rho_{\min}=\Theta(1)\) and
\(\rho_{\max}=\Theta(L_{\rm box}/h)\), we obtain the meaningful estimate
\begin{equation}
  \kappa(\mathsf A_{h,L_{\rm box}}^{\rm mix})
  =
  \Theta(L_{\rm box}/h).
  \label{eq:elliptic-mixed-condition-number}
\end{equation}

If one ignores this scaling and uses the raw dimensional matrix in
\eqref{eq:elliptic-mixed-saddle-system}, then the smallest eigenvalue may
behave like \(\sigma_{\min}(G_h)^2=\Theta(L_{\rm box}^{-2})\), while the largest
behaves like \(\sigma_{\max}(G_h)=\Theta(h^{-1})\).  This produces the
artifact
\begin{equation}
  \kappa(\mathsf A_h^{\rm mix})
  =
  \Theta(L_{\rm box}^2/h).
  \label{eq:elliptic-mixed-condition-number-large-box}
\end{equation}
We will regard \eqref{eq:elliptic-mixed-condition-number-large-box} as a
warning about units, not as the fundamental PDE scale.

The price of the mixed formulation is that the quantum output lives in an
energy-like norm.  With the scaled flux variable,
\[
  \left\|
  \begin{bmatrix}
    \widetilde{\bm p}\\
    \bm x
  \end{bmatrix}
  \right\|_2^2
  =
  \|\bm x\|_2^2
  +
  L_{\rm box}^2\|G_h\bm x\|_2^2 .
\]
This is the discrete analogue of
\[
  \|u\|_{L^2(\Omega)}^2
  +
  L_{\rm box}^2\|\nabla u\|_{L^2(\Omega)}^2 .
\]
Thus the mixed formulation is natural when fluxes, gradients, stresses, or
energy quantities are part of the desired output.  If only the potential
\(\bm x\) is wanted, the probability of projecting onto the \(\bm x\)-block is
\begin{equation}
  p_x
  =
  \frac{\|\bm x\|_2^2}
       {\|\bm x\|_2^2+L_{\rm box}^2\|G_h\bm x\|_2^2}.
  \label{eq:elliptic-mixed-x-block-prob-comparison}
\end{equation}
For smooth, low-frequency solutions, \(p_x\) may be \(O(1)\).  For mesh-scale
components, it can be as small as \(O((h/L_{\rm box})^2)\), so extracting only
\(\bm x\) can reintroduce the same first-order scale that the mixed system
removes from the spectral condition number.

\subsection{Transformed right-hand sides, Sum-QLS, and normalization}
\label{subsec:orsucci-dunjko-viewpoint}

The mixed formulation solves a coupled first-order PDE system.  The
factorized QLS framework of Orsucci and Dunjko gives a different access-model
route: instead of solving the mixed saddle-point system, it applies the
pseudo-inverse of the rectangular factor \(G_h\) to a suitably transformed
right-hand side \cite{OrsucciDunjko2021}.  Their work also explains why
positive definiteness alone is not enough to guarantee a
\(\sqrt{\kappa}\)-type quantum speedup in the standard black-box QLSA model.

In the notation of this chapter, choose a generalized right inverse
\[
  R_h\in\mathbb C^{m_h\times N_h}
\]
of \(G_h^\dag\), satisfying
\begin{equation}
  G_h^\dag R_h=I,
  \qquad
  \bm b'_h:=R_h\bm b_h .
  \label{eq:elliptic-transformed-rhs}
\end{equation}
Then the original solution is recovered from the least-squares problem
\begin{equation}
  \bm x_*
  =
  \underset{\bm x}{\operatorname{argmin}}\,
  \|G_h\bm x-\bm b'_h\|_2 .
  \label{eq:elliptic-factorized-regression}
\end{equation}
Indeed,
\begin{align}
  \bm x_*
  &=
  G_h^+\bm b'_h
  =
  G_h^+R_h\bm b_h
  \nonumber\\
  &=
  (G_h^\dag G_h)^{-1}G_h^\dag R_h\bm b_h
  =
  A_h^{-1}\bm b_h .
  \label{eq:elliptic-factorized-solution}
\end{align}
Equivalently, using the Hermitian dilation,
\begin{equation}
  \mathcal D(G_h)^+
  \begin{bmatrix}
    0\\
    \bm b'_h
  \end{bmatrix}
  =
  \begin{bmatrix}
    G_h^+\bm b'_h\\
    0
  \end{bmatrix}
  =
  \begin{bmatrix}
    \bm x_*\\
    0
  \end{bmatrix}.
  \label{eq:elliptic-dilation-solve}
\end{equation}

Because \(G_h\) is rectangular, only the component of \(\bm b'_h\) in
\(\operatorname{Range}(G_h)\) contributes.  Let
\[
  \Pi_{G_h}:=G_hG_h^+
\]
and, for the normalized state \(\ket{\bm b'_h}\), define
\begin{equation}
  \gamma_h
  :=
  \|\Pi_{G_h}\ket{\bm b'_h}\|^2 .
  \label{eq:elliptic-factor-overlap}
\end{equation}
The component orthogonal to \(\operatorname{Range}(G_h)\) is annihilated by
\(G_h^+\).  Hence a pseudo-inversion algorithm must project onto or amplify
the useful component, which introduces an overhead proportional to
\(1/\sqrt{\gamma_h}\).

The crucial additional cost is the preparation of
\(\bm b'_h=R_h\bm b_h\).  Suppose, for example, that \(R_h/\alpha_R\) has a
block encoding and that \(\ket{\bm b_h}\) can be prepared.  Applying the
block encoding of \(R_h/\alpha_R\) to \(\ket{\bm b_h}\) produces the
successful branch
\[
  \frac{R_h\ket{\bm b_h}}{\alpha_R},
\]
with probability
\begin{equation}
  p_R
  =
  \frac{\|R_h\bm b_h\|_2^2}
       {\alpha_R^2\|\bm b_h\|_2^2}.
  \label{eq:elliptic-R-preparation-prob}
\end{equation}
Thus preparing the normalized transformed state costs, schematically,
\begin{equation}
  C_{R b}
  =
  \widetilde O\left(
  (C_b+C_R)
  \frac{\alpha_R\|\bm b_h\|_2}{\|R_h\bm b_h\|_2}
  \right),
  \label{eq:elliptic-R-preparation-cost}
\end{equation}
if amplitude amplification is used.  Here \(C_b\) is the cost of preparing
\(\ket{\bm b_h}\) and \(C_R\) is the cost of applying the block encoding of
\(R_h/\alpha_R\).  If \(R_h\) itself is implemented by a QSVT inverse filter
for a global elliptic operator, then this cost may already contain the
condition number of the original problem and the factorized advantage is lost.
Thus \(R_h\) must be cheap because of additional structure, not merely because
it exists algebraically.

This point is visible in the Moore--Penrose choice.  The Moore--Penrose right
inverse of \(G_h^\dag\) is
\begin{equation}
  R_h=(G_h^\dag)^+
  =
  G_h(G_h^\dag G_h)^{-1}
  =
  G_hA_h^{-1}.
  \label{eq:elliptic-canonical-right-inverse}
\end{equation}
With this choice,
\begin{equation}
  R_h\bm b_h
  =
  G_h\bm x_*,
  \label{eq:elliptic-flux-right-hand-side}
\end{equation}
which is precisely the unknown flux.  It has perfect range overlap
\(\gamma_h=1\), but preparing it directly would require solving the elliptic
problem.  The two extremes bracket the design space: the Moore--Penrose choice
is perfectly aligned with \(\operatorname{Range}(G_h)\) but as hard as the
original problem, while a generic cheap \(R_h\) may waste most of its
amplitude outside the range.  A useful \(R_h\) must sit in between,
structured enough to be applied cheaply and aligned enough that
\(\gamma_h\) is not small.

The Sum-QLS construction of Orsucci and Dunjko identifies settings where a
nontrivial \(R_h\) is cheap.  In their model, the matrix is supplied as a sum
of local strictly positive-definite terms,
\begin{equation}
  A=\sum_{j=1}^{J}H^{(j)},
  \qquad
  H^{(j)}=L^{(j)}L^{(j)\dag}.
  \label{eq:orsucci-sum-qls}
\end{equation}
They form
\begin{equation}
  L=
  \begin{bmatrix}
    L^{(1)}&\cdots&L^{(J)}
  \end{bmatrix},
  \qquad
  R=
  \frac1J
  \begin{bmatrix}
    (L^{(1)})^{-1}\\
    \vdots\\
    (L^{(J)})^{-1}
  \end{bmatrix},
  \label{eq:orsucci-local-right-inverse}
\end{equation}
so that \(LL^\dag=A\) and \(LR=I\).  Under their access assumptions, the
local inverses \((L^{(j)})^{-1}\) are small or structured enough that
\(R\bm b\) can be prepared efficiently.  In that situation, the pseudo-inverse
step has the first-order scale
\begin{equation}
  \widetilde O\left(
  \frac{\kappa_{\mathrm{eff}}(G_h)}{\sqrt{\gamma_h}}
  \right)
  =
  \widetilde O\left(
  \frac{L_{\rm box}/h}{\sqrt{\gamma_h}}
  \right),
  \label{eq:elliptic-factor-query-cost}
\end{equation}
in addition to the transformed-state preparation cost.

Finite element assembly has a superficially similar local decomposition into
element contributions.  The analogy is useful, but the hypotheses are
different: standard element stiffness matrices are typically positive
semidefinite and rank deficient, and the number of elements grows with the
number of degrees of freedom.  Therefore the specific Sum-QLS theorem does
not apply automatically to a general finite element matrix.  The mixed
formulation is often the more direct PDE construction; the Sum-QLS viewpoint
becomes attractive when the generalized right inverse can be applied cheaply.

\subsection{Comparison with classical elliptic solvers}
\label{subsec:elliptic-classical-quantum-comparison}

Recall from the size table in \cref{sec:elliptic-classical-discretization} that
\begin{equation}
  N_{\rm cell}=\frac{L_{\rm box}}{h},
  \qquad
  N_{\rm dof}\asymp N_{\rm cell}^d .
  \label{eq:elliptic-nside-ndof-comparison}
\end{equation}
Then
\begin{equation}
  \kappa(A_h)=\Theta(N_{\rm cell}^2),
  \qquad
  \kappa_{\mathrm{eff}}(G_h)=\Theta(N_{\rm cell}).
  \label{eq:elliptic-comparison-kappa-ndof}
\end{equation}
The following table gives a schematic, fixed-access comparison.  For quantum
algorithms, the middle columns show query scales to normalized block
encodings, before gate costs, input preparation, postselection, and
measurement.  For classical algorithms, they show arithmetic work.  The output factors
\(\chi_A\) and \(\chi_{\rm fac}\) appearing in the last column are defined at
the end of the subsection.

For a dimension-aware comparison, let
\begin{equation}
  N_h:=\frac{L_{\rm box}}{h}-1,
  \qquad
  N_{\rm dof}\asymp N_h^d ,
  \label{eq:elliptic-Nh-Ndof-comparison}
\end{equation}
where $N_h$ is the number of interior grid points per coordinate
direction (the number of intervals differs by one), while \(N_{\rm dof}\) is the total number of spatial degrees of
freedom.  Thus
\[
  \kappa(A_h)=\Theta(N_h^2),
  \qquad
  \kappa_{\rm eff}(G_h)=\Theta(N_h).
\]
The table below gives a schematic comparison.  For the quantum rows, the
scale is the number of queries to a normalized block encoding, before gate
costs, state preparation, postselection, and measurement are included.  For
the classical rows, the scale is arithmetic work for producing the full
discrete vector.

\begin{center}
\small
\setlength{\tabcolsep}{4pt}
\begin{tabular}{llll}
\toprule
method & scale in \(N_h\) & scale in \(N_{\rm dof}\) & additional factor to track \\
\midrule
direct QSVT on \(A_h\)
&
\(\widetilde O(N_h^{2})\)
&
\(\widetilde O(N_{\rm dof}^{2/d})\)
&
normalization \(\chi_A\)
\\[0.3em]

mixed first-order system
&
\(\widetilde O(N_h)\)
&
\(\widetilde O(N_{\rm dof}^{1/d})\)
&
energy norm and \(x\)-block probability
\\[0.3em]

classical CG, no preconditioner
&
\(\widetilde O(N_h^{d+1})\)
&
\(\widetilde O(N_{\rm dof}^{1+1/d})\)
&
sparse matrix-vector products
\\[0.3em]

classical multigrid/BPX/PCG
&
\(\widetilde O(N_h^{d})\)
&
\(\widetilde O(N_{\rm dof})\)
&
problem-dependent hierarchy
\\
\bottomrule
\end{tabular}
\end{center}

The direct QSVT row has the usual output-normalization factor
\begin{equation}
  \chi_A
  :=
  \frac{\|\bm b_h\|_2}
       {\|A_h^{-1}\bm b_h\|_2}.
  \label{eq:elliptic-direct-output-factor}
\end{equation}
For sufficiently accurate \(L^2\)-normalized discretizations, this ratio
approaches the continuum quantity
\[
  \frac{\|f\|_{L^2(\Omega)}}{\|u\|_{L^2(\Omega)}} .
\]

For the mixed formulation, the natural output is not just \(\bm x_*\), but
the scaled energy variable
\[
  \begin{bmatrix}
    \widetilde{\bm p}_*\\
    \bm x_*
  \end{bmatrix},
  \qquad
  \widetilde{\bm p}_*=L_{\rm box}G_h\bm x_* .
\]
Its norm is
\begin{equation}
  \left\|
  \begin{bmatrix}
    \widetilde{\bm p}_*\\
    \bm x_*
  \end{bmatrix}
  \right\|_2^2
  =
  \|\bm x_*\|_2^2
  +
  L_{\rm box}^2\|G_h\bm x_*\|_2^2.
  \label{eq:elliptic-mixed-energy-norm-table}
\end{equation}
This is the discrete analogue of
\[
  \|u\|_{L^2(\Omega)}^2+
  L_{\rm box}^2\|\nabla u\|_{L^2(\Omega)}^2.
\]
Thus the mixed system is naturally normalized in an \(H^1\)-type energy norm.
If the desired output is only the potential block \(\bm x_*\), the probability
of obtaining that block is given by \cref{eq:elliptic-mixed-x-block-prob-comparison}.

The table should not be read as a claim of quantum advantage.  The quantum
entries are idealized block-encoding query counts, while the classical entries
are full-vector arithmetic costs.  Classical multigrid, domain decomposition,
and BPX-type preconditioners already achieve nearly linear complexity for many
standard elliptic problems
\cite{Saad2003Iterative,BramblePasciakXu1990}.  The meaningful quantum
question is more specific: after input preparation, block-encoding
construction, output normalization, and measurement are included, can a
quantum method improve the scaling for a particular quantity of interest?

\section{QFT-based spectral filtering for special elliptic problems}
\label{sec:qft-spectral-filtering}

Fourier methods are one of the oldest tools in the analysis of Poisson's equation: separation of variables turns the PDE into scalar algebraic equations for Fourier coefficients.  The quantum version preserves this same idea, but performs the change of basis coherently on amplitudes.  This gives a third route, different in kind from the previous two.  It applies only when the geometry, coefficients, and boundary conditions make the elliptic operator explicitly diagonalizable.  In that setting, one need not query a generic block encoding of $A_h$ and approximate its inverse by a QSVT polynomial.  Instead, one transforms to a known eigenbasis and evaluates the reciprocal eigenvalues by reversible arithmetic.

\subsection{Fourier diagonalization}
\label{subsec:classical-fourier-picture}

Consider first the constant-coefficient periodic Poisson equation on the flat torus $\mathbb T^d$ that came from periodic boundary conditions,
\begin{equation}
  -\Delta u=f,
  \qquad
  \int_{\mathbb T^d}f(x)\,dx=0.
  \label{eq:periodic-poisson}
\end{equation}
The zero-mean condition removes the nullspace with constant functions.  Using the Fourier convention $e^{2\pi i\bm k\cdot x}$, one obtains
\begin{equation}
  \widehat u_{\bm{k}}
  =g(\bm{k})\widehat f_{\bm{k}},
  \qquad
  g(\bm{k})=\frac1{4\pi^2 |\bm{k}|^2},
  \qquad \bm{k}\neq\bm{0}.
  \label{eq:poisson-spectral-filter}
\end{equation}
If the convention $e^{i\bm k\cdot x}$ is used instead, the factor $4\pi^2$ is absent.  We write it explicitly here to avoid ambiguity.

For a periodic finite-difference grid with $N=2^n$ points per coordinate, the QFT diagonalizes the circulant Laplacian:
\begin{equation}
  A_h=F_h^{\dag}\Lambda_hF_h,
  \label{eq:qft-diagonalization-elliptic}
\end{equation}
where, in one dimension, we have the familiar eigenvalues,
\begin{equation}
  (\Lambda_h)_{kk}
  =\frac{4 }{h^2}\sin^2\!\left(\frac{\pi k}{N}\right).
  \label{eq:periodic-discrete-symbol}
\end{equation}
The multidimensional symbol is the sum over coordinate directions.  Homogeneous Dirichlet and Neumann problems on rectangular boxes lead similarly to sine and cosine transforms.  Thus the QFT route is best understood as a structured-grid spectral method, not as a general elliptic solver.

Because the periodic Laplacian has a zero mode, the inverse below is always the pseudo-inverse on the mean-zero subspace.  For data satisfying the compatibility condition in \eqref{eq:periodic-poisson}, the constant Fourier mode is absent and this distinction is harmless.  A QFT-based solver consists of three operations:
\begin{equation}
  \ket{\bm{b}_h}
  \xrightarrow{\ F_h\ }
  \ket{\widehat{\bm{b}}_h}
  \xrightarrow{\ \Lambda_h^{+}\ }
  \Lambda_h^{+}\ket{\widehat{\bm{b}}_h}
  \xrightarrow{\ F_h^{\dag}\ }
  A_h^{-1}\ket{\bm{b}_h}.
  \label{eq:qft-elliptic-pipeline}
\end{equation}
Here $\Lambda_h^+$ denotes the pseudo-inverse: the zero Fourier mode is mapped to zero.  Because the right-hand side in \eqref{eq:periodic-poisson} has zero mean, that mode has zero amplitude and the pseudo-inverse agrees with the physical solution operator on the relevant subspace.  The QFT on each $n$-qubit coordinate register uses $O(n^2)$ elementary gates, or $O(dn^2)$ gates in $d$ dimensions.  The nonunitary reciprocal filter can be implemented as a diagonal block encoding, using the diagonal-matrix construction discussed in Subsection~\ref{subsec:diagonal-block-encoding}.

\subsection{Direct arithmetic block encoding of the reciprocal filter}
\label{subsec:diagonal-filter-block-encoding}

The middle arrow in \eqref{eq:qft-elliptic-pipeline} is a diagonal-matrix block encoding of the kind introduced in Subsection~\ref{subsec:diagonal-block-encoding}.  We spell it out here because this is the only nontrivial step and where the condition number reemerges.

Let \(\lambda_{\min}\) and \(\lambda_{\max}\) denote the smallest and largest
nonzero diagonal entries of \(\Lambda_h\), so that
\[
  \kappa=\frac{\lambda_{\max}}{\lambda_{\min}}
  =\Theta(h^{-2})
\]
on a fixed domain.  The zero Fourier mode is assigned reciprocal value zero,
which is the pseudo-inverse convention for the periodic Poisson problem.
Define
\begin{equation}
  r_{\bm k}
  =
  \begin{cases}
  1/(\Lambda_h)_{\bm k\bm k}, & \bm k\neq\bm 0,\\[2mm]
  0, & \bm k=\bm 0,
  \end{cases}
  \qquad
  \rho_{\bm k}:=\frac{r_{\bm k}}{\alpha_{\rm inv}},
  \qquad
  \alpha_{\rm inv}\ge \|\Lambda_h^+\|=\lambda_{\min}^{-1}.
  \label{eq:qft-reciprocal-filter-values}
\end{equation}
Thus \(0\le \rho_{\bm k}\le 1\).  The scale
\(\alpha_{\rm inv}=\Theta(\lambda_{\min}^{-1})\) is the diagonal analogue of
the normalization factor in a block encoding of \(A_h^{-1}\).  It is also the
place where the condition number enters the postselection probability.

The arithmetic implementation uses a reversible eigenvalue-inversion oracle.
More precisely, let \(U_\rho\) be a unitary acting on the Fourier-mode
register and a work register such that, up to fixed-point arithmetic error,
\begin{equation*}
  U_\rho:
  \ket{\bm k}\ket{0}_{\rm val}
  \longmapsto
  \ket{\bm k}\ket{\widetilde\rho_{\bm k}}_{\rm val}.
\end{equation*}
The first step is therefore
\begin{equation}
  \ket{\bm k}\ket{0}_{\rm val}\ket{0}_{a}
  \longmapsto
  \ket{\bm k}\ket{\widetilde\rho_{\bm k}}_{\rm val}\ket{0}_{a}.
  \label{eq:qft-reciprocal-arithmetic}
\end{equation}
This is the structured-grid version of the ``fast inversion'' primitive:
the eigenvalue is not learned by phase estimation, but computed directly from
the explicit Fourier formula for \((\Lambda_h)_{\bm k\bm k}\).  This is why
the construction is special to diagonalizable tensor-product problems.

Next, a controlled rotation acts on the ancilla qubit.  If
\[
  R_y(\theta)\ket{0}
  =
  \cos(\theta/2)\ket{0}
  +
  \sin(\theta/2)\ket{1},
\]
then choosing
\[
  \theta_{\bm k}=2\arccos(\widetilde\rho_{\bm k})
\]
gives
\begin{equation}
  \ket{\widetilde\rho_{\bm k}}_{\rm val}\ket{0}_{a}
  \longmapsto
  \ket{\widetilde\rho_{\bm k}}_{\rm val}
  \left(
     \widetilde\rho_{\bm k}\ket{0}_{a}
     +\sqrt{1-\widetilde\rho_{\bm k}^{\,2}}\ket{1}_{a}
  \right).
  \label{eq:qft-reciprocal-controlled-rotation}
\end{equation}
This is the same reciprocal-controlled-rotation idea that appears in HHL and
in the post-QPE view of linear-system algorithms: the reciprocal is encoded
as the amplitude of a successful ancilla branch.  The difference here is that
the reciprocal is obtained by reversible arithmetic from the Fourier index
\(\bm k\), rather than from a phase-estimation register.

Now apply \(U_\rho^\dag\) to uncompute the value register.  On a Fourier-space
input
\[
  \ket{\widehat{\bm b}_h}
  =
  \sum_{\bm k}\widehat b_{\bm k}\ket{\bm k},
\]
the state before the final ancilla measurement has the form
\begin{equation*}
  \sum_{\bm k}
  \widehat b_{\bm k}
  \ket{\bm k}\ket{0}_{\rm val}
  \left(
    \rho_{\bm k}\ket{0}_{a}
    +
    \sqrt{1-\rho_{\bm k}^{2}}\ket{1}_{a}
  \right),
\end{equation*}
where arithmetic errors have been suppressed.  Therefore, conditioned on
measuring the ancilla in \(\ket{0}_a\), the successful branch is
proportional to
\begin{equation}
  \sum_{\bm k}\rho_{\bm k}\widehat b_{\bm k}\ket{\bm k}
  =
  \frac{\Lambda_h^+}{\alpha_{\rm inv}}
  \ket{\widehat{\bm b}_h}.
  \label{eq:qft-reciprocal-success-branch}
\end{equation}
After applying the inverse QFT, this gives a block-encoding implementation of
\(A_h^+/\alpha_{\rm inv}\) on the mean-zero subspace.

The success probability is
\begin{equation*}
  p_{\rm spec}
  =
  \sum_{\bm k}
  |\rho_{\bm k}|^2|\widehat b_{\bm k}|^2
  =
  \frac{
    \|A_h^+\bm b_h\|_2^2
  }{
    \alpha_{\rm inv}^{\,2}\|\bm b_h\|_2^2
  } .
\end{equation*}
This formula is the useful warning.  Even though the diagonal reciprocal is
computed directly by arithmetic, the normalization
\(\alpha_{\rm inv}\simeq\lambda_{\min}^{-1}\) can make the successful branch
small.  For a right-hand side concentrated in high-frequency modes, the
amplitude can be as small as \(O(1/\kappa)\), so amplitude amplification can
reintroduce a condition-number cost.  For smooth or low-frequency data, the
Fourier coefficients \(\widehat b_{\bm k}\) are concentrated where
\(r_{\bm k}\) is large, and the effective success probability can be much
better.  This is the spectral-filtering advantage exploited by the structured
QFT approach.

This arithmetic construction is closely related to the diagonal inverse
block encodings discussed in Chapter~\ref{chap:basic-elements} and to the
fast-inversion primitive of Tong, An, Wiebe, and Lin
\cite{TongAnWiebeLin2021FastInversion}.  It bypasses a generic QSVT
polynomial approximation to \(1/x\) by using an explicit formula for the
eigenvalues.  The tradeoff is that it depends on separability, efficient
reversible evaluation of the diagonal entries, and the same normalization and
postselection issues that appear in other block-encoding constructions.

Finally, the factorization viewpoint is also visible in Fourier space.  For
the periodic Laplacian,
\[
  (\Lambda_h)_{\bm k\bm k}
  =
  \sum_{r=1}^d |g_r(\bm k)|^2,
\]
where \(g_r(\bm k)\) is the Fourier symbol of the discrete gradient in the
\(r\)-th coordinate.  Thus the eigenvalues are squares of first-order symbols.
One could therefore implement first-order filters in terms of
\(\sqrt{(\Lambda_h)_{\bm k\bm k}}\), paralleling the \(G_h\)-based
factorization in Section~\ref{sec:factorized-dilation}.  For solving Poisson
directly, however, the reciprocal \(1/\lambda_{\bm k}\) must still be applied
or encoded through a suitable mixed or transformed formulation.  The QFT
arithmetic method makes this distinction transparent: it can compute either
\(\lambda_{\bm k}^{-1}\) or \(\lambda_{\bm k}^{-1/2}\), but the final
normalization and the desired output determine which construction is useful.

The framework of Huang, Antonioli, and Barbaresco develops this construction systematically for constant-coefficient elliptic, Helmholtz, and diffusion problems \cite{HuangAntonioliBarbaresco2026QuantumSpectralPDE}.  Its principal quantum contribution is not the Fourier diagonalization---that part is classical---but a modular reversible-arithmetic circuit for block encoding frequency-dependent filters and reciprocals, building on fast inversion for explicitly computable diagonal entries \cite{TongAnWiebeLin2021FastInversion}.  Under efficient arithmetic oracles, the schematic coherent-filter cost is of the form
\begin{equation}
  C_{\rm spectral}
  =O\!\left(dn^2+\operatorname{polylog}(dN)+\log\!\left(\frac{\kappa}{\epsilon}\right)\right),
  \qquad N=2^n,
  \label{eq:huang-spectral-depth}
\end{equation}
where the QFTs contribute the $dn^2$ term and the reversible diagonal block encoding contributes the remaining terms.  This bound concerns the coherent filter circuit; it assumes that $\ket{\bm b_h}$ is already available and does not include the cost of postselection, amplitude amplification, or final readout.

The postselection probability is
\begin{equation}
  p_{\rm spec}
  =
  \sum_{\bm k}|\rho_{\bm k}|^2|\widehat b_{\bm k}|^2
  =
  \frac{\|A_h^{-1}\bm b_h\|_2^2}{\alpha_{\rm inv}^2\|\bm b_h\|_2^2}
  \label{eq:qft-spectral-success}
\end{equation}
for normalized input.  This formula is the simplest way to see how the condition number comes back.  Since $\alpha_{\rm inv}=\lambda_{\min}^{-1}$, one has $\rho_{\bm k}=\lambda_{\min}/\lambda_{\bm k}$.  If $\bm b_h$ is concentrated in the highest-frequency eigenspace, then $p_{\rm spec}=\Theta(\kappa^{-2})$, and amplitude amplification costs $\Theta(\kappa)=\Theta(h^{-2})$.  If $\bm b_h$ is concentrated in the low modes, then $p_{\rm spec}$ can be $O(1)$.

This explains when filtering is useful.  Smooth right-hand sides tend to have rapidly decaying high-frequency Fourier coefficients, and elliptic smoothing further suppresses high-frequency components in the solution.  A spectral filter can exploit this favorable spectral measure, especially when the desired output is a low-frequency observable rather than the full normalized solution vector.  But smoothness is not a substitute for an access model: one must still prepare the right-hand side state, compute the reciprocal accurately, and account for the postselection probability in \eqref{eq:qft-spectral-success}.

The QFT route therefore does not generically produce the $\Theta(h^{-1})$ first-order scale of Section~\ref{sec:factorized-dilation}.  It exploits a known eigenbasis rather than a mixed or factorized PDE system.  Its great advantage is that, on separable constant-coefficient problems, the entire spectral action can be described by QFTs and reversible arithmetic.  Its limitation is equally clear: variable coefficients, irregular domains, unstructured meshes, and general mixed boundary conditions destroy exact Fourier diagonalization.

\section{Elliptic eigenvalue problems, mass lumping, and phase estimation}
\label{sec:elliptic-eigenvalue-qpe}

The same elliptic operator defines an eigenvalue problem.  Because the weak form, stiffness matrix, mass matrix, and mass-normalized coordinates have already been introduced in Section~\ref{subsec:fem}, we only record the additional ingredients needed for eigenvalue approximation and quantum phase estimation.

\subsection{Galerkin eigenvalues and their approximation}
\label{subsec:fem-eigenvalue-problem}

The continuous weak eigenproblem is: find $\lambda\in\mathbb{R}$ and $0\neq\varphi\in H_0^1(\Omega)$ such that
\begin{equation}
  a(\varphi,v)=\lambda(\varphi,v)_{L^2(\Omega)},
  \qquad v\in H_0^1(\Omega).
  \label{eq:weak-eigenproblem}
\end{equation}
The conforming finite element approximation uses the same space $V_h$, stiffness matrix $K_h$, and mass matrix $M_h$ as in \eqref{eq:fem-linear-system} and \eqref{eq:fem-mass-matrix}:
\begin{equation}
  K_h\bm{y}_{j,h}=\lambda_{j,h}M_h\bm{y}_{j,h}.
  \label{eq:generalized-fem-eigenproblem}
\end{equation}
The min--max principle immediately gives the one-sided Rayleigh--Ritz bound
\begin{equation}
  \lambda_j\leq\lambda_{j,h},
  \qquad 1\leq j\leq\dim(V_h).
  \label{eq:galerkin-eigenvalue-upper-bound}
\end{equation}
Thus conforming Galerkin projection overestimates the exact eigenvalues.  For a fixed simple eigenpair with $\varphi_j\in H^2(\Omega)$, normalized in $L^2$, conforming $P_1$ elements satisfy \cite{Boffi2010Eigenvalue}
\begin{equation}
  0\leq\lambda_{j,h}-\lambda_j\leq Ch^2,
  \qquad
  \|\varphi_j-\varphi_{j,h}\|_{H^1(\Omega)}\leq Ch,
  \qquad
  \|\varphi_j-\varphi_{j,h}\|_{L^2(\Omega)}\leq Ch^2.
  \label{eq:fem-eigenvalue-errors}
\end{equation}
More generally, eigenvalue errors are quadratic in the corresponding energy-norm approximation error.

For the one-dimensional finite-difference matrix, the discrete sine eigenpairs are already given by \eqref{eq:fd-spectrum}.  For fixed mode index $j$ and $L_{\rm box}=1$,
\begin{equation}
  \lambda_{j,h}
  =j^2\pi^2-\frac{j^4\pi^4}{12}h^2+O(j^6h^4).
  \label{eq:fd-eigenvalue-error}
\end{equation}
Thus the low finite-difference eigenvalues are second-order accurate but, in contrast to conforming Galerkin eigenvalues, they underestimate the continuum eigenvalues.  This sign difference is a useful reminder that finite differences and finite elements may share the same convergence order while preserving different variational inequalities.

\subsection{Mass normalization and mass lumping}
\label{subsec:mass-lumping-eigenvalues}

With the exact mass matrix, the generalized problem \eqref{eq:generalized-fem-eigenproblem} becomes the standard Hermitian problem
\begin{equation}
  A_h\bm{z}_{j,h}=\lambda_{j,h}\bm{z}_{j,h},
  \qquad
  \bm{z}_{j,h}=M_h^{1/2}\bm{y}_{j,h},
  \label{eq:standard-fem-eigenproblem}
\end{equation}
where $A_h$ is the mass-normalized stiffness matrix from \eqref{eq:mass-normalized-system}.  The mathematical reduction is simple; implementing $M_h^{-1/2}$ can be nontrivial.

For low-order nodal elements with positive lumping weights, first form
the row sums of the full mass matrix, before eliminating prescribed
Dirichlet degrees of freedom:
\begin{equation}
 (M_L^{\rm full})_{ii}=\sum_j(M^{\rm full})_{ij},
 \qquad (M_L^{\rm full})_{ij}=0\quad(i\ne j).
 \label{eq:mass-lumped-definition}
\end{equation}
Let $M_{L,h}$ be its restriction to the free degrees of freedom. The
lumped eigenproblem is $K_h\bm y=\lambda M_{L,h}\bm y$, with Hamiltonian
\begin{equation}
 A_{L,h}=M_{L,h}^{-1/2}K_hM_{L,h}^{-1/2}.
 \label{eq:lumped-hamiltonian}
\end{equation}
Its diagonal scalings are local. Under the usual mesh, quadrature, and
regularity assumptions, low eigenvalues retain the expected convergence
rate, but they need not be upper bounds for the continuum eigenvalues.
In the following phase-estimation discussion, $A_h$ denotes whichever
Hermitian discretization has been chosen, with its corresponding
discretization error and access model.

\subsection{Phase estimation as spectral sampling}
\label{subsec:qpe-elliptic}

Quantum phase estimation is the circuit-level counterpart of expanding a vector in an eigenbasis.  It is more than a subroutine for returning one eigenvalue.  Its elegant feature is that it acts coherently on a superposition of eigenmodes.  Suppose a block encoding or Hamiltonian simulation of $A_h $ is available and prepare
\begin{equation}
  \ket{\bm{z}_0}=\sum_j c_j\ket{\bm{z}_{j,h} },
  \qquad
  A_h \ket{\bm{z}_{j,h} }
  =\lambda_{j,h} \ket{\bm{z}_{j,h} }.
  \label{eq:qpe-input-expansion}
\end{equation}
Choose a simulation time $t_0$ and write
\[
  e^{-it_0A_h }\ket{\bm{z}_{j,h} }
 =e^{2\pi i\theta_j}\ket{\bm{z}_{j,h} }, 
  \quad
  \theta_j=-\frac{t_0\lambda_{j,h} }{2\pi}\pmod 1.
\]
To recover the eigenvalue without aliasing, one chooses $t_0$ so that the relevant spectral interval fits inside one phase period; for example, $t_0\lambda_{\max}(A_h )<2\pi$ for the full spectrum.  Since $\lambda_{\max}(A_h )=\Theta(h^{-2})$, this choice makes a single simulation step have length $t_0=O(h^2)$ unless additional spectral information or rescaling is used. The no-aliasing condition is not the only cost.  If adjacent eigenvalues are
separated by \(\Delta\lambda\), their QPE phases are separated by
\(t_0\Delta\lambda/(2\pi)\).  Since \(t_0=O(h^2)\) and low-frequency elliptic
eigenvalue spacings are often \(O(L_{\rm box}^{-2})\), resolving individual
low modes may require controlled powers of size
\(O(L_{\rm box}^2/h^2)\).  Thus high-resolution spectral sampling can
reintroduce the elliptic gap scale, even though a single unaliased simulation
step has length \(O(h^2)\).

With $r$ ancilla qubits, standard QPE applies controlled powers of this unitary and an inverse QFT on the phase register.  Ideally, it performs the map
\begin{equation}
  \sum_j c_j\ket{\bm{z}_{j,h} }\ket{0^r}
  \longmapsto
  \sum_j c_j\ket{\bm{z}_{j,h} }\ket{\widetilde\theta_j},
  \label{eq:qpe-superposition-map}
\end{equation}
where $\widetilde\theta_j$ is an $r$-bit approximation to $\theta_j$.  Measuring the phase register returns an approximate eigenvalue with probability close to $|c_j|^2$; leaving the register unmeasured produces an entangled state between eigenvectors and eigenvalue labels.  This is the quantum analogue of spectral sampling.  For the standard circuit, the number of phase qubits determines the resolution, and the controlled powers determine the total simulation time; see Nielsen and Chuang for the usual accuracy and success-probability estimates \cite[Chapter~5]{NielsenChuang2010}.

Postselecting an eigenvalue window projects the spatial register onto the corresponding discrete eigenspace.  This makes QPE useful for spectral projectors, low-mode filtering, and eigenstate preparation when the initial state has good overlap.  It also reveals a limitation: QPE does not create the ground state from nothing.  In simple Dirichlet problems, smooth positive trial functions may have useful overlap with the first eigenfunction.

There are many variants of QPE.  Iterative and adaptive schemes reduce the number of coherent ancillas or change how measurements are processed; robust and low-depth variants are designed for early fault-tolerant regimes; and quantum signal processing gives another way to build spectral filters without explicitly storing an eigenvalue register.  We do not need those refinements here, but they are important in modern resource estimates; see the survey-style discussion in \cite{LinWiebe2026QASC} and the low-depth energy-estimation method of Ding and Lin \cite{DingLin2023QCELS}.

The mesh scale remains visible.  Since $\|A_h \|=\Theta(h^{-2})$, generic Hamiltonian simulation or phase estimation based directly on this operator inherits the normalization of the elliptic matrix unless QFT diagonalization, first-order factorization, or preconditioning is used.  Generalized-eigenvalue phase-estimation methods offer another route when direct access to the matrix pair $(K_h,M_h)$ is available \cite{parker2020quantum}.

\section{What to remember}
\label{subsec:elliptic-remember}

\begin{enumerate}
  \item Standard second-order discretizations have $L^2$ error $O(h^2)$ under sufficient regularity, but their mass-normalized stiffness matrices have $\kappa(A_h)=\Theta(h^{-2})$.
  \item A generic QSVT inverse filter therefore needs $\widetilde O(h^{-2})$ queries to a block encoding of $A_h$, even when each query has a compact sparse or structured circuit.
  \item The factorization $A_h=G_h^{\dag}G_h$ exposes the first-order scale $\kappa(G_h)=\Theta(h^{-1})$ on a fixed domain and can yield a square-root improvement under suitable access and range assumptions.
  \item QFT-based arithmetic filters avoid the generic inverse-polynomial construction for separable constant-coefficient problems, but worst-case postselection can still retain $\Theta(h^{-2})$ dependence.
  \item Conforming finite elements overestimate elliptic eigenvalues and give $O(h^2)$ eigenvalue errors for regular $P_1$ eigenfunctions.  Mass lumping simplifies phase estimation but changes the variational problem.
  \item Preconditioning, dynamics, stopping rules, and quantum-compatible finite element access are central to extending these introductory algorithms to realistic elliptic systems.
\end{enumerate}

\section{Outlook}
\label{sec:elliptic-outlook}

The preceding sections illustrate three basic routes: direct inversion of $A_h$, first-order factorization, and explicit spectral diagonalization.  Several broader issues remain important, but we discuss them only briefly.

\paragraph{Preconditioning.}
Classical elliptic solvers do not accept $\kappa(A_h)=\Theta(h^{-2})$ as unavoidable.  Multigrid and BPX preconditioners use nested spaces to produce operators whose condition numbers are independent of the finest mesh size.  A quantum preconditioner must itself admit an efficient block encoding and must not make state preparation or readout more expensive.  The BPX-based construction of Deiml and Peterseim is an important example of treating multilevel structure as part of the quantum algorithm rather than as hidden classical preprocessing \cite{BramblePasciakXu1990,DeimlPeterseim2025}.

\paragraph{Dynamics-based solvers.}
An elliptic solution can be realized as the steady state of a stable evolution, for example
\begin{equation}
  \dot{\bm{x}}(t)=-B_hA_h\bm{x}(t)+B_h\bm{b}_h,
  \label{eq:elliptic-relaxation-outlook}
\end{equation}
where $B_h$ may be a preconditioner.  Quantum ODE solvers, dilation methods, and Schr\"odingerization can then replace a direct inverse filter by simulation of the relaxation \cite{JinLiuYu2024Schrodingerization}.  The useful complexity parameter is the stability and decay of the preconditioned evolution, not merely $\kappa(A_h)$.

\paragraph{Residual-based stopping.}
Classical iterative solvers stop when a residual is sufficiently small.  A quantum dynamical solver can carry a residual register and estimate its weight without reconstructing the full solution \cite{LiDynamicStopping2026}.  This makes the runtime instance dependent and can exploit smooth right-hand sides whose high-frequency components decay rapidly.

\paragraph{Ground-state preparation.}
For elliptic eigenproblems, quantum imaginary-time evolution applies the filter
\begin{equation}
  e^{-\tau A_h}\ket{\bm{z}_0}
  =\sum_j c_j e^{-\tau\lambda_{j,h}}\ket{\bm{z}_{j,h}},
\end{equation}
which suppresses excited modes after normalization.  QITE and related dissipative methods provide alternative ground-state
filters, but still require nonzero initial ground-state overlap, with
costs depending on that overlap and on the spectral gap ~\cite{MottaSunTanORourkeYeMinnichBrandaoChan2020}.

\paragraph{Systems of elliptic equations.}
Linear elasticity, coupled diffusion, and mixed formulations lead to vector-valued unknowns and block matrices.  Strongly elliptic systems retain coercive energy estimates, whereas saddle-point formulations require inf--sup stability and suitable block preconditioners.  Quantum algorithms must preserve these structural distinctions rather than treating every block matrix as a generic sparse system.

\paragraph{General finite element workflows.}
On unstructured meshes, efficient access should be formulated in terms of local elements, basis-function supports, quadrature data, coefficients, and boundary constraints.  The remaining challenges include mass normalization, local-to-global indexing, variable coefficients, adaptive meshes, and observable extraction.  These issues determine whether a mathematically good finite element discretization also gives an efficient quantum access model.

\paragraph{Nonsymmetric elliptic problems.}
The discussion in this chapter has focused mainly on self-adjoint elliptic
operators, for which the finite difference or finite element matrix is
Hermitian positive definite.  This structure is lost when the differential
equation contains an advective term,
\begin{equation}
  -\nabla\cdot\bigl(a(\bm{x})\nabla u\bigr)
  +\bm{\beta}(\bm{x})\cdot\nabla u
  +c(\bm{x})u
  =f,
  \label{eq:elliptic-advection-diffusion-outlook}
\end{equation}
If the inverse block encoding is $A_h^{-1}/\alpha_{\rm inv}$, the
actual dimensionless amplification factor is
$\alpha_{\rm inv}\chi_A$. On a fixed domain this normalization can be
mesh independent, but it must still be restored when comparing physical
scales or changing the domain.
 Standard
Dirichlet, Neumann, and real Robin conditions preserve symmetry for the
diffusion operator, whereas oblique, impedance-type, or nonsymmetrically
enforced boundary conditions may produce a nonsymmetric discretization.
Upwind or streamline-stabilized treatments of
\eqref{eq:elliptic-advection-diffusion-outlook} generally lead to
\begin{equation}
  A_h=K_h+C_h,
  \qquad
  C_h\neq C_h^\dag,
  \label{eq:nonsymmetric-elliptic-matrix}
\end{equation}
where \(K_h\) is the diffusion stiffness matrix and \(C_h\) represents
advection, boundary, or stabilization terms.

For such matrices, the eigenvalue-transformation viewpoint used for
Hermitian positive-definite systems no longer applies directly.  One
possible route is to introduce the Hermitian dilation
\begin{equation}
  \mathcal D(A_h)
  =
  \begin{bmatrix}
    0&A_h^\dag\\
    A_h&0
  \end{bmatrix},
  \label{eq:nonsymmetric-elliptic-dilation}
\end{equation}
whose nonzero eigenvalues are the signed singular values of \(A_h\).
QSVT may then implement a singular-value inverse filter, but it does not alleviate the $O(1/h^2)$ scaling of the condition number.

\section{Exercises and further directions}
\label{sec:elliptic-exercises}

\begin{exercise}[One-dimensional factorization]
Verify directly that $G_h^{\dag}G_h$ is the matrix in \eqref{eq:one-d-laplacian}.  Compare the singular values of $G_h$ with the square roots of \eqref{eq:fd-spectrum}.
\end{exercise}

\begin{exercise}[Two-dimensional Kronecker structure]
Starting from \eqref{eq:five-point}, derive \eqref{eq:kronecker-laplacian} and show that its eigenvectors are tensor products of one-dimensional sine modes.  Then verify the stacked-gradient representation in \eqref{eq:2d-discrete-gradient}.
\end{exercise}

\begin{exercise}[Mass normalization]
For $P_1$ finite elements on a uniform one-dimensional mesh, write down $K_h$ and $M_h$.  Compare the spectra of $K_h$, $M_h^{-1}K_h$, and $M_h^{-1/2}K_hM_h^{-1/2}$.
\end{exercise}

\begin{exercise}[Direct QSVT scale]
Use \eqref{eq:elliptic-kappa} to derive the normalized spectral interval \eqref{eq:elliptic-normalized-spectrum} and the query estimate \eqref{eq:elliptic-direct-qsvt-cost}.
\end{exercise}

\begin{exercise}[Mixed first-order conditioning]
Derive the two eigenvalues \eqref{eq:elliptic-mixed-eigenvalues} of the mixed matrix \(\mathsf A_h^{\rm mix}\) associated with a singular value \(\sigma\) of \(G_h\).  Show that if \(\sigma_{\min}=\Theta(1)\), then \(\kappa(\mathsf A_h^{\rm mix})=\Theta(h^{-1})\), whereas if \(\sigma_{\min}=\Theta(L_{\rm box}^{-1})\ll1\), the unscaled system has the larger condition number in \eqref{eq:elliptic-mixed-condition-number-large-box}.
\end{exercise}

\begin{exercise}[Galerkin eigenvalue bound]
Use the min--max principle and $V_h\subset H_0^1(\Omega)$ to prove \eqref{eq:galerkin-eigenvalue-upper-bound}.
\end{exercise}

\begin{exercise}[Mass lumping]
For $P_1$ elements in one dimension, compute the consistent and row-sum lumped mass matrices.  Compare the resulting generalized eigenvalues with \eqref{eq:fd-spectrum}.
\end{exercise}

\begin{exercise}[Spectral-filter success probability]
Evaluate \eqref{eq:qft-spectral-success} when $\bm{b}_h$ is a single low-frequency mode and when it is the highest-frequency mode.  Explain the difference between arithmetic circuit depth and end-to-end state-preparation cost.
\end{exercise}

\begin{exercise}[Residual stopping]
Assume $\|\bm{x}_*-\bm{x}(t)\|\leq C\|\bm{r}(t)\|$ from \cref{eq:elliptic-relaxation-outlook}.  If a normalized joint state contains solution and residual blocks, derive a bound on the relative error from the probability of measuring the residual block.
\end{exercise}
% Body-only LaTeX file.  It is intended to be incorporated by \include{...}
% into a larger book manuscript.  No documentclass, packages, theorem
% declarations, or macro preamble are included here.

\chapter{Quantum Algorithms for Hyperbolic PDEs}
\label{chap:hyperbolic-quantum}

\section[Classical wave discretizations]{Classical wave discretizations: finite differences and finite elements}
\label{sec:hyperbolic-classical}

Hyperbolic equations model wave propagation rather than relaxation to static equilibrium.  This is the main conceptual difference from the elliptic chapter.  In an elliptic problem, the operator determines an equilibrium configuration. In contrast, in a hyperbolic problem, the same spatial operator generates a time evolution via Newton's second law.

Mathematically, the relation with Chapter~\ref{chap:elliptic-quantum} is direct.  The stiffness operator in the elliptic model \eqref{eq:elliptic-model} becomes, after mass normalization, the spatial operator whose square root determines wave frequencies.  The divergence--gradient factorization \eqref{eq:fem-gradient-factorization} therefore reappears below as the wave factorization \eqref{eq:wave-factorization}.  This is the first place where the distinction between an operator inverse problem and a unitary propagation problem becomes visible at the level of quantum circuits.

The time-dependent Schr\"odinger equation is itself commonly regarded as a wave equation because it propagates a complex wavefunction.  Mathematically, however, it is a first-order dispersive evolution rather than a second-order hyperbolic equation.  A central point of this chapter is that the classical second-order wave equation can nevertheless be rewritten, after an appropriate spatial discretization and a change to energy variables, as an ordinary Schr\"odinger equation with a Hermitian Hamiltonian. This paves the way to native quantum simulation of wave dynamics.

This observation also fixes the point of view used in the chapter.  We do not try to reproduce a classical time-marching code gate by gate.  Instead, we first build a semidiscrete Hamiltonian system from a standard finite difference or finite element approximation, and then use quantum Hamiltonian simulation to evolve that system.  Thus the role of numerical analysis is to identify a stable, consistent, and physically meaningful finite-dimensional model, while the role of quantum algorithms is to implement the resulting unitary or nearly unitary evolution and to extract a limited number of quantities of interest.  The distinction is important: a quantum computer is not an efficient printer of the full grid function, but it can be an efficient device for preparing and measuring structured states associated with the discrete wave energy.

The organizing principle of the chapter is energy conservation, and it is worth stating as a dial rather than a dichotomy.  When the discrete dynamics conserves a Euclidean energy exactly, the quantum computer operates in its native mode: the evolution is unitary, there is no postselection penalty, and the only costs are simulation depth and measurement.  Every departure from exact conservation---body forces, nonhomogeneous boundary data, absorbing layers, upwind viscosity---shows up as a specific, quantifiable quantum overhead: an LCU normalization, a success amplitude, or a dilation.  The chapter is arranged along this dial.  This section records the classical discretizations and their energy identities.  Section~\ref{sec:wave-to-schrodinger} performs the main translation from the wave equation to a Schr\"odinger equation.  Section~\ref{sec:wave-forcing-boundary-lcu} treats sources and boundary data by Duhamel's principle and LCU.  Section~\ref{sec:qft-wave-fast-forward} shows how explicit Fourier structure can fast-forward special cases, and Section~\ref{sec:wave-qsvt-general-mesh} assembles the end-to-end pipeline for general meshes.  Finally, Section~\ref{sec:first-order-hyperbolic-dilation} turns to first-order systems, where classical stabilization deliberately breaks conservation, and Schr\"odingerization, LCHS, and moment-matching dilation restore a unitary description on a larger space.

We begin with the scalar second-order wave equation
\begin{equation}
  u_{tt}(\bm{x},t)
  -\nabla\cdot\bigl(c(\bm{x})^2\nabla u(\bm{x},t)\bigr)
  =f(\bm{x},t),
  \qquad \bm{x}\in\Omega,
  \label{eq:wave-model}
\end{equation}
with initial conditions
\begin{equation}
  u(\bm{x},0)=u_0(\bm{x}),
  \qquad
  u_t(\bm{x},0)=v_0(\bm{x}),
  \label{eq:wave-initial-data}
\end{equation}
and, for the moment, homogeneous Dirichlet boundary condition
$u|_{\partial\Omega}=0$.  For $f=0$, the natural continuous energy is
\begin{equation}
  E(t)=\frac12\|u_t(t)\|_{L^2(\Omega)}^2
  +\frac12\int_\Omega c(\bm{x})^2
      |\nabla u(\bm{x},t)|^2\,d\bm{x}.
  \label{eq:wave-continuous-energy}
\end{equation}
Under the homogeneous boundary condition, integration by parts gives
$dE/dt=0$.  This identity is the main guide for both the classical discretization and the quantum construction.  A good discretization should not merely approximate $u_{tt}$ and $\nabla u$. Instead, it should preserve, or at least respect, the underlying energy structure.

There are two useful ways to rewrite the second-order equation as a first-order system.  The first uses the displacement--velocity pair.  With
\begin{equation}
  v=u_t,
  \qquad
  \bm{z}=\begin{bmatrix}u\\v\end{bmatrix},
\end{equation}
the homogeneous equation has the formal representation
\begin{equation}
  \frac{d}{dt}\bm{z}
  =
  \begin{bmatrix}0&I\\-L&0\end{bmatrix}\bm{z},
  \qquad
  L=-\nabla\cdot(c^2\nabla).
  \label{eq:wave-displacement-velocity-form}
\end{equation}
This form keeps the displacement explicit and is convenient when the desired output is a displacement observable.  The second formulation uses a pressure, strain, or deformation-gradient variable together with velocity.  We write
\begin{equation}
  \bm{p}=c\nabla u,
  \qquad
  v=u_t.
  \label{eq:wave-pressure-velocity-continuous}
\end{equation}
For the homogeneous equation,
\begin{equation}
  \bm{p}_t=c\nabla v,
  \qquad
  v_t=\nabla\cdot(c\bm{p}).
  \label{eq:wave-pressure-velocity-form}
\end{equation}
This pressure--velocity form is closer to the conserved energy \eqref{eq:wave-continuous-energy}.  It is especially natural when the output is elastic or acoustic energy, strain energy, flux, or modal energy.  Throughout the chapter we use \(\bm{p}_h=G_h\bm{y}_h\) for the discrete pressure/strain block and \(\bm{v}_h=\dot{\bm{y}}_h\) for the discrete velocity block.  The main Hamiltonian formulation below is based on these variables, while Subsection~\ref{subsec:augmented-displacement-formulation} explains how to carry the displacement as an additional block when needed.

For numerical analysts, this is just the familiar passage from a second-order equation to a first-order system in energy variables.  For quantum computing, it has an additional meaning.  The pressure--velocity variables are exactly the variables whose Euclidean norm is conserved by the semidiscrete dynamics.  They are therefore the natural amplitudes of a quantum state.  If one used the raw displacement vector instead, the conserved norm would involve the stiffness matrix and mass matrix, and the corresponding generator would not be Hermitian in the standard amplitude inner product.

\subsection{Centered finite differences}
\label{subsec:wave-fd}
As examples of classical algorithms, we again start with a finite difference scheme.
On the interval \(\Omega=(0,L_{\rm box})\), with \(c\) constant and \(f=0\), consider
\begin{equation}
  u_{tt}=c^2u_{xx},
  \qquad
  u(0,t)=u(L_{\rm box},t)=0.
  \label{eq:wave-1d}
\end{equation}

For a simple finite difference approximation for this model, we use the same grid points \eqref{1dgrids} with $N_h$ interior points and spacing $h$. 
  The standard centered fully discrete scheme in both space and time is
\begin{equation}
  \frac{U_j^{n+1}-2U_j^n+U_j^{n-1}}{k^2}
  =c^2\frac{U_{j+1}^n-2U_j^n+U_{j-1}^n}{h^2},
  \qquad j=1,\ldots,N_h.
  \label{eq:leapfrog-wave}
\end{equation}
Here \(U_j^n\) approximates \(u(x_j,t_n)\), and
\begin{equation}
  \bm{u}_h^n=(U_1^n,\ldots,U_{N_h}^n)^T
\end{equation}
denotes the vector of interior nodal values.  

Because \eqref{eq:leapfrog-wave}
is second order in time, it requires two starting vectors.  The first is
obtained directly from the displacement data,
\[
  \bm u_h^0\approx (u_0(x_1),\ldots,u_0(x_{N_h}))^T .
\]
The second initial vector is chosen by a Taylor expansion using the initial velocity.  For
\(\ddot{\bm u}_h=-A_h\bm u_h+\bm f_h(t)\), a second-order initialization is
\begin{equation}
  \bm u_h^1
=
\bm u_h^0+k\bm v_h^0
-\frac{k^2}{2}A_h\bm u_h^0
+\frac{k^2}{2}\bm f_h(0),
  \label{eq:leapfrog-initial-step}
\end{equation}
with \(\bm v_h^0\) obtained from \(v_0= u_t(x,0) \).  In the homogeneous constant-speed
case, this is the familiar formula obtained by replacing \(u_{tt}(0)\) with
\(c^2u_{xx}(0)\).  A lower-order initialization would contaminate the global
second-order accuracy, which is why this first step is part of the classical
scheme rather than a harmless implementation detail.

We use the weighted discrete norm
\begin{equation}
  \|\bm w\|_{\ell_h^2}
  :=\left(h\sum_{j=1}^{N_h} |w_j|^2\right)^{1/2},
  \label{eq:wave-discrete-l2-norm}
\end{equation}
which is the one-dimensional version of the \(h^{d/2}\)-scaling between \(L^2\) functions and Euclidean amplitude vectors discussed in Chapter~\ref{chap:basic-elements}.

The error analysis of a finite difference approximation like \eqref{eq:leapfrog-wave} often begins with the local error.  A Taylor expansion of the fully discrete formula, after substitution of the exact solution, gives the local residual
\begin{equation}
\begin{aligned}
  \tau_j^n
  :={}&
  \left[
  u_{tt}(x_j,t_n)
  -\frac{u(x_j,t_{n+1})-2u(x_j,t_n)+u(x_j,t_{n-1})}{k^2}
  \right]
  \\
  &-c^2\left[
  u_{xx}(x_j,t_n)
  -\frac{u(x_{j+1},t_n)-2u(x_j,t_n)+u(x_{j-1},t_n)}{h^2}
  \right]
  =O(k^2+h^2).
\end{aligned}
  \label{eq:wave-local-error}
\end{equation}
This displayed form is longer than the original difference equation, but it makes the two sources of second-order consistency visible: the centered time difference and the centered spatial difference.

Consistency alone does not imply convergence.  For the wave equation, stability requires the Courant--Friedrichs--Lewy restriction
\begin{equation}
  \nu:=\frac{ck}{h}\leq1.
  \label{eq:wave-cfl}
\end{equation}
For a uniform stability margin \(\nu\leq\nu_0<1\), and sufficiently smooth exact solutions, a representative estimate is
\begin{equation}
  \max_{0\leq n\leq T/k}
  \|\bm{u}_h^n-\bm{u}_{\mathrm{ex}}^n\|_{\ell_h^2}
  \leq C(1+T)(h^2+k^2),
  \qquad
  \bm{u}_{\mathrm{ex}}^n=
  \bigl(u(x_1,t_n),\ldots,u(x_{N_h},t_n)\bigr)^T,
  \label{eq:wave-fd-error}
\end{equation}
where the constant depends on a finite number of derivatives of \(u\), but not on \(h\), \(k\), or \(T\) \cite{LarssonThomee2009}.  The factor \(1+T\) is important: wave equations transport consistency errors rather than smoothing them.

This should be contrasted with the heat equation in Chapter~\ref{chap:parabolic-quantum}.  Parabolic smoothing damps high-frequency local errors, whereas the wave equation carries local errors along characteristics and reflects them from boundaries.  This is why the classical CFL time step and the quantum Hamiltonian-simulation normalization both contain the scale \(h^{-1}\): the fastest numerical waves cross one grid cell in time \(O(h)\).

The centered scheme in \eqref{eq:leapfrog-wave} is useful for explaining classical accuracy and the CFL restriction.  The quantum algorithm developed below does \emph{not} have to apply this classical time step coherently.  Instead, we typically discretize only in space and retain continuous time.  Define
\begin{equation}
  \bm{u}_h(t)=(U_1(t),\ldots,U_{N_h}(t))^T.
\end{equation}
For constant \(c\), the semidiscrete equation is
\begin{equation}
  \ddot{\bm{u}}_h(t)+A_h\bm{u}_h(t)=\bm{0},
  \qquad
  A_h=\frac{c^2}{h^2}
  \begin{bmatrix}
  2&-1&&&\\
  -1&2&-1&&\\
  &\ddots&\ddots&\ddots&\\
  &&-1&2&-1\\
  &&&-1&2
  \end{bmatrix}.
  \label{eq:fd-wave-matrix}
\end{equation}
Thus \(A_h\) is \(c^2\) times the one-dimensional elliptic matrix in \eqref{eq:one-d-laplacian}.  It is symmetric positive definite, and
\begin{equation}
  \|A_h^{1/2}\|=\Theta(h^{-1}).
  \label{eq:wave-square-root-scale}
\end{equation}
This is the hyperbolic analogue of the first-order scale \(\|G_h\|=\Theta(h^{-1})\) encountered in Chapter~\ref{chap:elliptic-quantum}.

The same construction extends naturally to variable wave speed.  Toward this end, we write
\begin{equation}
  a(x)=c(x)^2,
  \qquad
  -\partial_x(a(x)\partial_x u)
\end{equation}
for the positive spatial operator. Then the natural discretization is obtained by first approximating the edge flux
\begin{equation}
  F_{j+1/2}(t)
  \approx
  a_{j+1/2}\frac{U_{j+1}(t)-U_j(t)}{h},
  \qquad
  a_{j+1/2}=a(x_{j+1/2}),
\end{equation}
and then differencing the fluxes.  This gives
\begin{equation}
  (A_h\bm{u}_h)_j
  =\frac{1}{h^2}
  \left[
    a_{j+1/2}(U_j-U_{j+1})
    +a_{j-1/2}(U_j-U_{j-1})
  \right].
  \label{eq:variable-c-wave-stencil}
\end{equation}
Thus \(A_h\) remains symmetric.  Its diagonal entries are
\((a_{j-1/2}+a_{j+1/2})/h^2\), and its off-diagonal entries are
\(-a_{j+1/2}/h^2\).  This is the wave-equation counterpart of the variable-coefficient elliptic stiffness matrix from Chapter~\ref{chap:elliptic-quantum}.  The conservative flux form is also the intuitive reason for the symmetry: the same edge coefficient \(a_{j+1/2}\) couples the two neighboring nodes \(j\) and \(j+1\).

This edge-centered viewpoint is also the most useful one for quantum access.  A sparse-matrix oracle does not need to store the full matrix; it only needs to compute the neighboring indices and the local coefficients \(a_{j\pm1/2}\).  This is analogous to the sparse access model in Section~\ref{sec:sparse-block-encoding-access}: the finite difference stencil supplies the locations, while the coefficient field supplies the values.  For smooth or piecewise smooth \(c(x)\), these values can often be computed reversibly from the grid index.

In several dimensions, we again keep the time variable continuous and write the spatially semidiscrete equation in the same form as
\begin{equation}
  \ddot{\bm{u}}_h(t)+A_h\bm{u}_h(t)=\bm{f}_h(t),
  \label{eq:wave-fd-multidim-system}
\end{equation}
where $A_h$ is the multi-dimensional centered difference matrix. 
On Cartesian grids, \(A_h\) is assembled from the same directional second differences used in the five-point elliptic operator \eqref{eq:five-point}; for variable coefficients it is assembled from directional edge fluxes.  It has \(O(d)\) nonzero entries per row and \(\|A_h^{1/2}\|=\Theta(h^{-1})\) under uniform ellipticity of \(c(\bm x)^2\).  Higher-order and extended stencils lead to the same semidiscrete form, but their factorizations can be more delicate; see Subsection~\ref{subsec:fd-factorization-wave}.

\subsection{Finite elements and the mass matrix}
\label{subsec:wave-fem}

Finite elements begin from the same bilinear form and stiffness matrix used for the elliptic problem in Section~\ref{subsec:fem}.  For homogeneous Dirichlet boundary conditions, define
\begin{equation}
  a(u,v)=\int_\Omega c(\bm{x})^2
  \nabla u(\bm{x})\cdot\nabla v(\bm{x})\,d\bm{x}.
  \label{eq:wave-bilinear-form}
\end{equation}
The variational wave equation is: find $u(\cdot, t)\in H_0^1(\Omega)$ such that
\begin{equation}
  (u_{tt}(\cdot, t),v)+a(u(\cdot, t),v)=(f(t),v),
  \qquad v\in H_0^1(\Omega).
  \label{eq:wave-weak-form}
\end{equation}
Let $V_h\subset H_0^1(\Omega)$ be a conforming finite element space with basis
$\{\phi_1,\ldots,\phi_{N_h}\}$, and write
\begin{equation}
  u_h(\bm{x},t)=\sum_{j=1}^{N_h}(\bm{q}_h(t))_j\phi_j(\bm{x}).
\end{equation}

The semidiscrete method, after enforcing the weak form, is given by, 
\begin{equation}
  M_h\ddot{\bm{q}}_h(t)+K_h\bm{q}_h(t)=\bm{b}_h(t),
  \label{eq:fem-wave-mk}
\end{equation}
where
\begin{equation}
  (M_h)_{ij}=(\phi_j,\phi_i),
  \qquad
  (K_h)_{ij}=a(\phi_j,\phi_i),
  \qquad
  (\bm{b}_h(t))_i=(f(t),\phi_i).
  \label{eq:mass-stiffness-wave}
\end{equation}
Here $(\cdot,\cdot)$ refers to the $L^2$ inner product. 

The stiffness and mass matrices are precisely the matrices introduced in
\eqref{eq:stiffness-load} and \eqref{eq:fem-mass-matrix}, now appearing in a time-dependent equation.  Under homogeneous Dirichlet conditions, $M_h$ and $K_h$ are symmetric positive definite.  For $f=0$, multiplication by $\dot{\bm{q}}_h(t)^\dag$ gives
\begin{equation}
  \frac{d}{dt}
  \left[
  \frac12\dot{\bm{q}}_h(t)^\dag M_h\dot{\bm{q}}_h(t)
  +\frac12\bm{q}_h(t)^\dag K_h\bm{q}_h(t)
  \right]=0.
  \label{eq:fem-wave-energy}
\end{equation}
Thus the finite element semidiscretization preserves the continuous energy structure automatically.  From the quantum point of view, this identity also tells us which vector should be normalized.  The Euclidean norm of the coefficient vector \(\bm{q}_h\) is not the physical \(L^2\) norm.  The physical kinetic energy uses \(M_h\), while the strain energy uses \(K_h\).  This is why the mass-normalized variable \(\bm{y}_h=M_h^{1/2}\bm{q}_h\) will appear repeatedly below.

For piecewise linear elements on a quasi-uniform mesh, compatible projections of the initial data, and a sufficiently smooth solution, representative estimates are
\begin{align}
  \|u_h(t)-u(t)\|_{L^2(\Omega)}
  &\leq C h^2
  \left(
    \|u(t)\|_{H^2(\Omega)}
    +\int_0^t\|u_{tt}(s)\|_{H^2(\Omega)}\,ds
  \right),
  \label{eq:fem-wave-l2-error}\\
  \|\nabla(u_h(t)-u(t))\|_{L^2(\Omega)}
  &\leq C h
  \left(
    \|u(t)\|_{H^2(\Omega)}
    +\int_0^t\|u_{tt}(s)\|_{H^1(\Omega)}\,ds
  \right),
  \label{eq:fem-wave-h1-error}
\end{align}
up to the usual initial projection terms.  The integrals over $[0,t]$ again express the accumulation of consistency error in a nonsmoothing evolution.  Classical optimal-order analyses include the work of Dupont, Baker, Baker and Bramble, Dougalis, and Baker--Dougalis--Serbin
\cite{Dupont1973L2GalerkinHyperbolic,Baker1976HyperbolicFEM,
BakerBramble1979Hyperbolic,Dougalis1979MultistepGalerkin,
BakerDougalisSerbin1979TwoStepHyperbolic}.

The finite element formulation is useful pedagogically because it separates three issues that are often merged in a finite difference formula.  The matrix \(K_h\) represents the elastic or acoustic restoring force, the matrix \(M_h\) represents the discrete \(L^2\) inner product, and the energy identity tells us which combination of position and velocity is conserved.  Quantum amplitude encoding uses the ordinary Euclidean norm on coefficient vectors, so the mass matrix cannot be ignored.  The mass normalization in \eqref{eq:mass-normalized-y} is precisely the finite element analogue of the \(L^2\)-to-\(\ell_2\) scaling discussed in Chapter~\ref{chap:basic-elements}.

\subsection{Mass lumping}
\label{subsec:mass-lumping-wave}

The finite element mass matrix $M_h$ is sparse, but $M_h^{-1/2}$ is generally dense.  In principle, if a normalized block encoding of $M_h$ is available, QSVT can approximate the scalar map $x\mapsto x^{-1/2}$ on the spectral interval of $M_h$ and thereby produce a block encoding of $M_h^{-1/2}$.  This is a legitimate route, especially because the condition number of a properly scaled mass matrix is mesh independent on a shape-regular quasi-uniform mesh.  It nevertheless requires an additional spectral transformation and produces a nonlocal operator at the matrix level.

For an introductory construction, mass lumping is much simpler.  Replace $M_h$ by a positive diagonal approximation
\begin{equation}
  M_{h,L}=\operatorname{diag}(m_1,\ldots,m_{N_h}),
  \qquad m_i>0,
  \label{eq:mass-lumped}
\end{equation}
obtained, for example, by nodal quadrature or row summation.  The lumped system is
\begin{equation}
  M_{h,L}\ddot{\bm{q}}_h(t)+K_h\bm{q}_h(t)=\bm{b}_h(t),
  \label{eq:lumped-wave-system}
\end{equation}
and the normalized stiffness
\begin{equation}
  A_{h,L}=M_{h,L}^{-1/2}K_hM_{h,L}^{-1/2}
  \label{eq:lumped-normalized-stiffness}
\end{equation}
retains the locality of $K_h$.  This is the same mass-lumping mechanism used for the elliptic eigenproblem in Subsection~\ref{subsec:mass-lumping-eigenvalues}.

From the quantum perspective, the main advantage is not that lumping is a better finite element approximation; it is that it keeps the transformation to Euclidean coordinates local.  The diagonal entries \(m_i^{-1/2}\) can be applied by reversible arithmetic or by a simple diagonal block encoding, while a consistent mass matrix would require an additional matrix-function primitive.  Thus mass lumping trades a small classical quadrature modification for a much simpler quantum access model.

\medskip 
    The formulas above used homogeneous boundary conditions and, for most of the discussion, $f=0$.  This is the cleanest starting point because the homogeneous wave equation becomes unitary after the energy-variable transformation.  Body forces and nonhomogeneous boundary data can be incorporated by combining the homogeneous propagator with the variation-of-constants formula.  We give the quantum formulation in Section~\ref{sec:wave-forcing-boundary-lcu}.

\section{From the wave equation to a Schr\"odinger equation}
\label{sec:wave-to-schrodinger}

This section performs the main translation step of the chapter.  The continuous PDE is first replaced by a finite-dimensional conservative mechanical system.  We then choose coordinates in which the conserved energy is the Euclidean norm.  Only after these two numerical-analysis steps does the system become a genuine Schr\"odinger equation, which is then directly suitable for Hamiltonian simulation.  Both steps are familiar operations---semidiscretization and a change to energy variables---but the quantum setting sharpens the second one: it is not enough for the conserved quantity to be \emph{equivalent} to the Euclidean norm, as in a standard stability argument; it must \emph{be} the Euclidean norm, exactly, because that is the norm a quantum state carries.

The homogeneous semidiscrete wave equation has the form
\begin{equation}
  M_h\ddot{\bm{q}}_h(t)+K_h\bm{q}_h(t)=\bm{0},
  \qquad M_h=M_h^\dag\succ0,
  \qquad K_h=K_h^\dag\succeq0.
  \label{eq:wave-mk-homogeneous}
\end{equation}
Assume first that the boundary conditions remove the nullspace, so that \(K_h\succ0\).  Introduce the mass-normalized displacement
\begin{equation}
  \bm{y}_h(t)=M_h^{1/2}\bm{q}_h(t),
  \qquad
  A_h=M_h^{-1/2}K_hM_h^{-1/2}.
  \label{eq:mass-normalized-y}
\end{equation}
Then
\begin{equation}
  \ddot{\bm{y}}_h(t)+A_h\bm{y}_h(t)=\bm{0},
  \qquad A_h=A_h^\dag\succ0.
  \label{eq:y-second-order-wave}
\end{equation}
This is the dynamical counterpart of the mass-normalized elliptic system \eqref{eq:mass-normalized-system}.

The direct displacement--velocity system is
\begin{equation}
  \frac{d}{dt}
  \begin{bmatrix}
  \bm{y}_h\\ \dot{\bm{y}}_h
  \end{bmatrix}
  =
  \begin{bmatrix}
  0&I\\-A_h&0
  \end{bmatrix}
  \begin{bmatrix}
  \bm{y}_h\\ \dot{\bm{y}}_h
  \end{bmatrix}.
  \label{eq:canonical-first-order-wave}
\end{equation}
Although this system preserves the physical energy, its coefficient matrix is not skew-Hermitian in the ordinary Euclidean inner product.  It is Hamiltonian in the classical symplectic sense, but a quantum computer simulates unitary evolution in a Hilbert-space inner product.  We therefore change variables so that the conserved energy becomes an ordinary Euclidean norm.  A direct way to do this is to use a square-root factorization
\begin{equation}
  A_h=G_h^\dag G_h.
  \label{eq:wave-factorization}
\end{equation}
For finite differences, \(G_h\) is a scaled discrete gradient.  For finite elements, it is the mass-normalized gradient factor associated with \eqref{eq:fem-gradient-factorization}.  Define the pressure/strain block and the velocity block by
\begin{equation}
  \bm{p}_h(t)=G_h\bm{y}_h(t),
  \qquad
  \bm{v}_h(t)=\dot{\bm{y}}_h(t).
  \label{eq:wave-energy-variables}
\end{equation}
Then
\begin{equation}
  \dot{\bm{p}}_h(t)=G_h\bm{v}_h(t),
  \qquad
  \dot{\bm{v}}_h(t)=-G_h^\dag\bm{p}_h(t).
  \label{eq:wave-pv-system}
\end{equation}
Equivalently,
\begin{equation}
  \frac{d}{dt}\bm{\psi}_h(t)
  =S_h\bm{\psi}_h(t),
  \qquad
  \bm{\psi}_h(t)=
  \begin{bmatrix}\bm{p}_h(t)\\\bm{v}_h(t)\end{bmatrix},
  \qquad
  S_h=
  \begin{bmatrix}0&G_h\\-G_h^\dag&0\end{bmatrix}.
  \label{eq:wave-skew-system}
\end{equation}
Since \(S_h^\dag=-S_h\), multiplication by \(i\) gives the Schr\"odinger equation
\begin{equation}
  i\frac{d}{dt}\bm{\psi}_h(t)=H_h\bm{\psi}_h(t),
  \qquad
  H_h=iS_h
  =i\begin{bmatrix}0&G_h\\-G_h^\dag&0\end{bmatrix},
  \qquad H_h=H_h^\dag.
  \label{eq:wave-schrodinger-hamiltonian}
\end{equation}
This factorized Hamiltonian is the structural basis of the wave-equation algorithm of Costa, Jordan, and Ostrander \cite{CostaJordanOstrander2019Wave}, where finite difference forms were considered.  Its norm is tied to the energy:
\begin{equation}
  \|\bm{\psi}_h(t)\|_2^2
  =\|G_h\bm{y}_h(t)\|_2^2
   +\|\dot{\bm{y}}_h(t)\|_2^2
  =\bm{y}_h(t)^\dag A_h\bm{y}_h(t)
   +\|\dot{\bm{y}}_h(t)\|_2^2.
  \label{eq:wave-energy-state-norm}
\end{equation}
Thus the quantum state naturally represents the pressure--velocity, or energy, variables rather than the displacement alone.  This is often an advantage: modal energy, local flux, kinetic energy, and strain energy are all functions of \(G_h\bm{y}_h\) and \(\dot{\bm{y}}_h\).

Notice the contrast with the elliptic and parabolic chapters.  In an elliptic problem, the main primitive is an inverse or a resolvent; in a parabolic problem, it is a contraction semigroup.  Here the homogeneous evolution is unitary after the correct change of variables.  Consequently, there is no final postselection penalty associated with the norm of the propagated state.  The measurement bottleneck instead comes from the usual need to estimate observables from repeated state preparations.

\subsection{Finite differences as a factored system}
\label{subsec:fd-factorization-wave}

For the one-dimensional Dirichlet discretization, define the first-difference matrix
\begin{equation}
  D_h=\frac1h
  \begin{bmatrix}
  1&0&\cdots&0\\
  -1&1&\ddots&\vdots\\
  &\ddots&\ddots&0\\
  0&\cdots&-1&1\\
  0&\cdots&0&-1
  \end{bmatrix}.
  \label{eq:dirichlet-gradient-matrix}
\end{equation}
It maps \(N_h\) interior nodal values to \(N_h+1\) edge differences, including the two boundary edges.  For constant \(c\),
\begin{equation}
  A_h=c^2D_h^\dag D_h,
  \qquad
  G_h=cD_h.
  \label{eq:fd-gradient-factorization-wave}
\end{equation}
For variable coefficient \(a(x)=c(x)^2\), let \(W_h\) be the positive diagonal matrix of edge coefficients.  Then
\begin{equation}
  A_h=D_h^\dag W_hD_h,
  \qquad
  G_h=W_h^{1/2}D_h,
  \label{eq:variable-coeff-fd-factorization-wave}
\end{equation}
which is the factorized form of the conservative stencil \eqref{eq:variable-c-wave-stencil}.  Notice that the stiffness matrix itself has entries of size \(O(h^{-2})\), while the factor \(G_h\) has entries of size \(O(h^{-1})\).  The wave Hamiltonian uses the first-order factor, not the stiffness matrix directly.  This is why the generic hyperbolic scale is \(T/h\), in contrast to the direct elliptic inverse scale \(h^{-2}\) discussed in Chapter~\ref{chap:elliptic-quantum}.  The corresponding Hamiltonian is sparse, its entries are \(O(h^{-1})\), and
\begin{equation}
  \|H_h\|=\|G_h\|=\Theta(h^{-1}).
  \label{eq:wave-hamiltonian-scale}
\end{equation}
In higher dimensions, \(G_h\) is obtained by stacking directional difference matrices.  In two dimensions, for example,
\begin{equation}
  G_h=c
  \begin{bmatrix}D_x\\D_y\end{bmatrix},
  \qquad
  G_h^\dag G_h=c^2(D_x^\dag D_x+D_y^\dag D_y)=A_h.
  \label{eq:two-d-gradient-factorization-wave}
\end{equation}

A simple example of a wider stencil is a one-dimensional lattice with both nearest-neighbor and next-nearest-neighbor couplings.  On a periodic grid, let
\begin{equation}
  (D_1\bm{u})_j=\frac{U_{j+1}-U_j}{h},
  \qquad
  (D_2\bm{u})_j=\frac{U_{j+2}-U_j}{h}.
  \label{eq:wide-stencil-differences}
\end{equation}
This is also the ideas in \cite{CostaJordanOstrander2019Wave}. If the two spring constants \(\eta_1,\eta_2\) are both nonnegative, then
\begin{equation}
  A_h=\eta_1D_1^\dag D_1+\eta_2D_2^\dag D_2
  =G_h^\dag G_h,
  \qquad
  G_h=
  \begin{bmatrix}
    \sqrt{\eta_1}D_1\\
    \sqrt{\eta_2}D_2
  \end{bmatrix}.
  \label{eq:extended-stencil-factorization-positive}
\end{equation}
This gives a five-point stencil whose Fourier symbol is
\begin{equation}
  \lambda(\theta)
  =\frac{4\eta_1}{h^2}\sin^2\frac{\theta}{2}
  +\frac{4\eta_2}{h^2}\sin^2\theta .
  \label{eq:extended-stencil-symbol}
\end{equation}

The more interesting lattice-dynamics point in \cite{Li2025PRL} is that the
next-nearest-neighbor coefficient need not be positive term by term.  What is
needed for a stable wave equation is nonnegativity of the total symbol.  For
example, take \(\eta_2<0\).  Since
\[
  \sin^2\theta=4\sin^2(\theta/2)\cos^2(\theta/2),
\]
the symbol in \eqref{eq:extended-stencil-symbol} is nonnegative provided
\begin{equation}
  \eta_1+4\eta_2\ge0 .
  \label{eq:wide-stencil-stability-condition}
\end{equation}
The stacked construction \eqref{eq:extended-stencil-factorization-positive} is
now unavailable, because \(\sqrt{\eta_2}\) is imaginary.  The question becomes
whether the nonnegative symbol can still be written as \(|g(\theta)|^2\) for a
\emph{local} difference operator---that is, whether the factorization can be
carried out without leaving the space of short stencils.  It can.  Readers may
recognize what follows as a small instance of the Fej\'er--Riesz factorization
of nonnegative trigonometric polynomials.  Choose real numbers
\(\alpha,\beta\) such that
\begin{equation}
  \alpha\beta=\eta_2,
  \qquad
  \alpha^2+\beta^2=\eta_1+2\eta_2 .
  \label{eq:wide-stencil-alpha-beta}
\end{equation}
These two conditions simply match the constant and the \(\cos\theta\)
coefficients of the two trigonometric polynomials, and such \(\alpha,\beta\)
exist under \eqref{eq:wide-stencil-stability-condition}.  Define
\begin{equation}
  (G_h\bm u)_j
  =\frac{1}{h}\left[
    \alpha(U_{j+1}-U_j)
    +\beta(U_{j+2}-U_{j+1})
  \right].
  \label{eq:wide-stencil-local-factor}
\end{equation}
The Fourier symbol of \(G_h\) is
\[
  g(\theta)=\frac{1}{h}(e^{i\theta}-1)(\alpha+\beta e^{i\theta}),
\]
and therefore
\begin{equation}
  |g(\theta)|^2
  =\frac{4\eta_1}{h^2}\sin^2\frac{\theta}{2}
   +\frac{4\eta_2}{h^2}\sin^2\theta.
  \label{eq:wide-stencil-factor-symbol}
\end{equation}
As a concrete stable choice with a negative next-nearest-neighbor coefficient,
take \(\eta_1=5\), \(\eta_2=-1\).  Then
\(\eta_1+4\eta_2=1>0\), and one may choose
\(\alpha=(1+\sqrt5)/2\), \(\beta=-(\sqrt5-1)/2\).  This example illustrates
why the factor \(G_h\), rather than the individual stencil coefficients, is
the quantum object of interest.  Even a wider stencil with some negative
couplings can lead to a Hermitian wave Hamiltonian once its nonnegative symbol
has been factored as \(A_h=G_h^\dag G_h\).

\subsection{Generic Hamiltonian-simulation complexity}
\label{subsec:generic-wave-hsim-complexity}

Suppose that $G_h/\alpha_G$ has an efficient block encoding with
\begin{equation}
  \alpha_G=\Theta(\|G_h\|)=\Theta(h^{-1}).
  \label{eq:G-block-encoding-scale-wave}
\end{equation}
Then $H_h/\alpha_G$ has a block encoding with comparable cost.  Indeed, the Hermitian dilation from Chapter~\ref{chap:elliptic-quantum},
\[
  \mathcal D(G_h)=\begin{bmatrix}0&G_h^\dag\\G_h&0\end{bmatrix},
\]
is converted to $H_h=i\begin{bmatrix}0&G_h\\-G_h^\dag&0\end{bmatrix}$
by reordering the two blocks and applying a phase. Specifically, for
$G_h:\mathbb C^n\to\mathbb C^m$, let
$R:\mathbb C^n\oplus\mathbb C^m\to\mathbb C^m\oplus\mathbb C^n$
exchange the summands, and set $V=\operatorname{diag}(iI_m,I_n)$.
Then $H_h=VR\mathcal D(G_h)R^\dag V^\dag$.
With padded registers this changes the query cost by only a constant factor.  Qubitization or QSVT Hamiltonian simulation implements
$e^{-iH_hT}$ with query complexity
\begin{equation}
  \widetilde O\!\left(
    \alpha_GT+\log\frac1{\epsilon}
  \right)
  =\widetilde O\!\left(
    \frac{T}{h}+\log\frac1{\epsilon}
  \right).
  \label{eq:qsvt-wave-complexity}
\end{equation}
Thus the best generic mesh-dependent scaling exposed by the first-order formulation is linear in $T/h$, rather than $T/h^2$.  Up to logarithmic factors, this linear dependence on the scaled evolution time is optimal in the black-box setting: the no-fast-forwarding theorem rules out sublinear-time simulation for general sparse Hamiltonians \cite{BerryAhokasCleveSanders2007,LowChuang2019}; see also the discussion in Childs' lecture notes \cite{ChildsQuantumAlgorithmsNotes}.  Any improvement beyond $T/h$ must therefore exploit additional structure, such as the explicitly computable dispersion relation used in Section~\ref{sec:qft-wave-fast-forward}.

It is useful to compare this scale with a classical explicit method while
keeping the physical parameters visible.  Let \(L_{\rm box}\) be a
characteristic side length of the computational domain, let \(h\) be the mesh
size, and let \(c_{\max}\) be the maximum wave speed.  On a uniform grid,
\[
  N_h\asymp \left(\frac{L_{\rm box}}{h}\right)^d
\]
spatial degrees of freedom are used.  A stable leapfrog computation from \cref{eq:leapfrog-wave} satisfies
a CFL restriction of the form
\[
  k \le C_{\rm CFL}\frac{h}{c_{\max}},
\]
where the constant depends on the spatial dimension and the stencil but not
on \(h\), \(T\), or \(L_{\rm box}\).  Hence the number of time steps is
\[
  N_t
  =
  O\left(\frac{c_{\max}T}{h}\right).
\]
Since each sparse time step costs \(O(N_h)\) arithmetic operations, the
classical explicit work is
\[
  O(N_tN_h)
  =
  O\left(
    \frac{c_{\max}T}{h}
    \left(\frac{L_{\rm box}}{h}\right)^d
  \right).
\]

The quantum Hamiltonian-simulation cost for the energy-state formulation is
instead governed by the largest wave frequency,
\[
  \|H_h\|T
  \asymp
  \frac{c_{\max}T}{h}.
\]
Thus the same CFL scale \(c_{\max}/h\) appears, but without the extra factor
counting all grid degrees of freedom.  The domain size \(L_{\rm box}\) enters
the quantum algorithm through the number of qubits, state preparation, and
readout model, rather than through the Hamiltonian-simulation time scale
itself.

This is different from the elliptic chapter.  For elliptic equations, solving
\(Lu=f\) requires applying an inverse, and the smallest eigenvalue scales like
\(L_{\rm box}^{-2}\).  Therefore the elliptic condition number contains
\((L_{\rm box}/h)^2\), or \(L_{\rm box}/h\) after first-order factorization.
For the wave equation, forward simulation is controlled by the largest
frequency, which scales like \(c_{\max}/h\), not by the low-frequency gap.
The low-frequency scale would reappear only if one tries to invert the wave
operator, recover displacement from strain by a pseudo-inverse, or perform
high-resolution spectral estimation near zero frequency.

The quantum cost above reflects only coherent evolution of an
amplitude-encoded energy state.  It does not include input preparation or
readout of all \(N_h\) grid values.  This is why the most meaningful quantum
outputs are low-dimensional observables, overlaps, modal weights, or sampled
energy distributions rather than the full wave field.

\section{Body forces and nonhomogeneous boundary conditions}
\label{sec:wave-forcing-boundary-lcu}

A common question from numerical analysts is what happens when the wave equation has a body force or nonzero boundary data.  For the wave equation this can be formulated cleanly, because the homogeneous part is already a Hamiltonian-simulation problem.  The nonhomogeneous term enters through variation of constants.  After a quadrature approximation for the time integral, the result can be implemented by a controlled superposition of homogeneous propagators and source-state preparations, closely related to the inhomogeneous constructions used in LCU and linear-combination-of-Hamiltonian-simulations (LCHS) algorithms \cite{AnLiuLin2023}.  In the language of the conservation dial from the chapter introduction, this section takes the first step away from exact conservation: the dynamics is still built entirely from unitary propagators, but a coefficient bookkeeping---an LCU normalization---now enters the cost.

Start from
\begin{equation}
  M_h\ddot{\bm{q}}_h(t)+K_h\bm{q}_h(t)=\bm{b}_h(t).
  \label{eq:forced-wave-mk}
\end{equation}
After the same mass normalization,
\begin{equation}
  \bm{y}_h=M_h^{1/2}\bm{q}_h,
  \qquad
  A_h=M_h^{-1/2}K_hM_h^{-1/2},
  \qquad
  \bm{f}_h(t)=M_h^{-1/2}\bm{b}_h(t),
  \label{eq:forced-wave-mass-normalized}
\end{equation}
we obtain
\begin{equation}
  \ddot{\bm{y}}_h(t)+A_h\bm{y}_h(t)=\bm{f}_h(t).
  \label{eq:forced-wave-second-order-normalized}
\end{equation}
Assume \(A_h=G_h^\dag G_h\) and define
\begin{equation}
  \bm{\psi}_h(t)=
  \begin{bmatrix}\bm{p}_h(t)\\\bm{v}_h(t)\end{bmatrix},
  \qquad
  \bm{p}_h=G_h\bm{y}_h,
  \qquad
  \bm{v}_h=\dot{\bm{y}}_h.
  \label{eq:forced-wave-energy-variables}
\end{equation}
Then
\begin{equation}
  \dot{\bm{\psi}}_h(t)
  =S_h\bm{\psi}_h(t)+\bm{g}_h(t),
  \qquad
  S_h=\begin{bmatrix}0&G_h\\-G_h^\dag&0\end{bmatrix},
  \qquad
  \bm{g}_h(t)=\begin{bmatrix}\bm{0}\\\bm{f}_h(t)\end{bmatrix}.
  \label{eq:forced-wave-first-order-skew}
\end{equation}
With \(H_h=iS_h\), this is
\begin{equation}
  i\dot{\bm{\psi}}_h(t)
  =H_h\bm{\psi}_h(t)+i\bm{g}_h(t),
  \label{eq:forced-wave-schrodinger-inhomogeneous}
\end{equation}
and the homogeneous propagator is
\begin{equation}
  U_h(\tau)=e^{-iH_h\tau}=e^{S_h\tau}.
  \label{eq:forced-wave-homogeneous-propagator}
\end{equation}
Variation of constants gives
\begin{equation}
  \bm{\psi}_h(T)
  =U_h(T)\bm{\psi}_h(0)
  +\int_0^T U_h(T-s)\bm{g}_h(s)\,ds.
  \label{eq:forced-wave-duhamel}
\end{equation}
The same Hamiltonian-simulation circuit used for the homogeneous equation is therefore reused at every quadrature node.

At the PDE level, \eqref{eq:forced-wave-duhamel} says that the final wave is a superposition of waves emitted by the source at earlier times.  A numerical quadrature rule turns this continuum superposition into a finite sum.  A quantum LCU implementation mirrors this interpretation: the clock or quadrature register chooses the emission time, the source oracle prepares the emitted wave packet, and the homogeneous propagator transports it from that time to \(T\).

A quadrature rule with nodes \(s_m\) and weights \(w_m\) gives
\begin{equation}
  \bm{\psi}_{h,M}(T)
  :=U_h(T)\bm{\psi}_h(0)
  +\sum_{m=1}^{M_q}w_mU_h(T-s_m)\bm{g}_h(s_m)
  \approx\bm{\psi}_h(T).
  \label{eq:forced-wave-duhamel-quadrature}
\end{equation}
To see explicitly how this becomes an LCU, define the normalized source
states
\begin{equation}
  \ket{\widehat{\bm{g}}_m}
  =\frac{\bm{g}_h(s_m)}{\|\bm{g}_h(s_m)\|},
  \qquad
  \ket{\widehat{\bm\psi}_0}
  =\frac{\bm{\psi}_h(0)}{\|\bm{\psi}_h(0)\|}.
  \label{eq:forced-wave-normalized-source-states}
\end{equation}
The magnitudes \(\|\bm{g}_h(s_m)\|\) are not part of the normalized quantum
state; they must be placed in the LCU coefficients.  Introduce
\begin{equation}
  \Lambda
  =\|\bm{\psi}_h(0)\|
  +\sum_{m=1}^{M_q}|w_m|\,\|\bm{g}_h(s_m)\|.
  \label{eq:forced-wave-lcu-norm}
\end{equation}
This is the same normalization parameter that appears in the LCU construction
of Chapter~\ref{chap:basic-elements}: it is the sum of the absolute
coefficients of the terms being coherently added.

A convenient PREPARE oracle is
\begin{equation}
  \ket{0}
  \longmapsto
  \frac{1}{\sqrt\Lambda}
  \left(
  \sqrt{\|\bm{\psi}_h(0)\|}\ket{0}
  +\sum_{m=1}^{M_q}
    \sqrt{|w_m|\,\|\bm{g}_h(s_m)\|}\ket{m}
  \right),
  \label{eq:forced-wave-lcu-prepare}
\end{equation}
together with a state oracle
\begin{equation}
  \ket{m}\ket{0}
  \longmapsto
  \begin{cases}
  \ket{0}\ket{\widehat{\bm\psi}_0}, & m=0,\\[0.3em]
  \ket{m}\ket{\widehat{\bm{g}}_m}, & 1\le m\le M_q,
  \end{cases}
  \label{eq:source-history-loading}
\end{equation}
and a SELECT operation that applies \(U_h(T)\) for \(m=0\) and
\(U_h(T-s_m)\) for \(m\ge1\).  The phase of \(w_m\) is included in the
controlled SELECT operation.  After unpreparing the quadrature register, the
successful branch is proportional to
\begin{equation}
  \frac{1}{\Lambda}
  \left(
  U_h(T)\bm{\psi}_h(0)
  +\sum_{m=1}^{M_q}w_mU_h(T-s_m)\bm{g}_h(s_m)
  \right).
  \label{eq:forced-wave-lcu-success-branch}
\end{equation}
This spells out why \(\Lambda\) controls the success probability: the raw
success amplitude is \(\|\bm\psi_{h,M}(T)\|/\Lambda\).  Since the propagators
are unitary, the triangle inequality gives
\(\|\bm\psi_{h,M}(T)\|\le\Lambda\) always, with equality exactly when all the
terms align in phase.  The LCU therefore pays for cancellation: if the
propagated initial wave and the emitted source waves interfere destructively
at time \(T\), the success amplitude shrinks in proportion.  The source state in
\eqref{eq:source-history-loading} is normalized, while its magnitude is
carried by the coefficient register.  Without such coherent source-history
loading, \eqref{eq:forced-wave-duhamel-quadrature} is only a classical
quadrature formula, not yet a quantum LCU.  This is the same issue that
appears in LCHS and other inhomogeneous linear-ODE algorithms: the homogeneous
propagator may be efficient, but the time-dependent forcing must also be
coherently accessible.

Several common PDE sources have this structure.  If \(f(\bm{x},t)\) is given by an analytic formula, the value oracle can compute it from the space-time grid indices.  If the forcing has a separated form \(f(\bm{x},t)=\sum_r a_r(t)f_r(\bm{x})\), then the source history can be loaded from a small superposition over the separation index and the quadrature index.  Boundary sources, such as the Neumann example below, are even more localized: the spatial state may be supported only on boundary degrees of freedom, while the time dependence is carried by a scalar function.

Consequently, after amplitude amplification, a schematic coherent cost is
\begin{equation}
  \widetilde O\!\left(
  \frac{\Lambda}{\|\bm{\psi}_{h,M}(T)\|}
  \left[
    C_{\rm source}+\frac{T}{h}+\log\frac1{\epsilon}
  \right]
  \right),
  \label{eq:forced-wave-lcu-cost}
\end{equation}
plus the cost needed to make the quadrature error sufficiently small.  The formula keeps visible the two distinct issues: Hamiltonian simulation of the homogeneous wave and coherent preparation of the source superposition.

\subsection{A mixed Dirichlet--Neumann boundary example}
\label{subsec:wave-neumann-example}

Boundary conditions deserve some care.  In elliptic and parabolic problems,
a Neumann boundary condition represents a prescribed flux.  In mechanical wave
problems, the analogous condition is often a prescribed traction.  For the
one-dimensional wave equation
\begin{equation}
  u_{tt}-(a u_x)_x=f,
  \qquad
  a=c^2,
  \label{eq:wave-neumann-model}
\end{equation}
we consider the mixed boundary conditions
\begin{equation}
  u(0,t)=0,
  \qquad
  a u_x(L_{\rm box},t)=g_R(t).
  \label{eq:wave-dirichlet-neumann-bc}
\end{equation}
The case \(g_R=0\) is the homogeneous Neumann condition.  If \(u\) is a
transverse displacement of a string or a longitudinal displacement of a bar,
then \(a u_x\) is the boundary force, or traction, after nondimensionalization.

For simplicity take \(a=c^2\) constant and use a vertex-centered grid
\[
  x_j=jh,\qquad j=0,\ldots,N_{\rm cell},
  \qquad h=\frac{L_{\rm box}}{N_{\rm cell}}.
\]
The left boundary value is fixed, \(U_0(t)=0\), while the right boundary value
\(U_{N_{\rm cell}}(t)\) is retained as an unknown.  Thus in this local example
there are \(N_{\rm cell}\) unknowns,
\[
  \bm u_h(t)=
  (U_1(t),\ldots,U_{N_{\rm cell}}(t))^\top .
\]

For the interior nodes \(1\le j\le N_{\rm cell}-1\), the standard centered
second difference gives
\begin{equation}
  \ddot U_j
  =
  \frac{a}{h^2}
  \left(U_{j+1}-2U_j+U_{j-1}\right)
  +f_j(t).
  \label{eq:wave-neumann-interior-fd}
\end{equation}
At the right boundary,  a common strategy to incorporate a Neumann condition is to introduce a ghost point \(U_{N_{\rm cell}+1}\) and use
the centered approximation
\begin{equation}
  a\frac{U_{N_{\rm cell}+1}-U_{N_{\rm cell}-1}}{2h}
  =
  g_R(t).
  \label{eq:wave-neumann-ghost-bc}
\end{equation}
Hence
\[
  U_{N_{\rm cell}+1}
  =
  U_{N_{\rm cell}-1}
  +
  \frac{2h}{a}g_R(t).
\]
Substituting this into the centered second difference at
\(j=N_{\rm cell}\) gives
\begin{align}
  \ddot U_{N_{\rm cell}}
  &=
  \frac{a}{h^2}
  \left(
    U_{N_{\rm cell}+1}
    -2U_{N_{\rm cell}}
    +U_{N_{\rm cell}-1}
  \right)
  +f_{N_{\rm cell}}(t)
  \nonumber\\
  &=
  \frac{2a}{h^2}
  \left(
    U_{N_{\rm cell}-1}-U_{N_{\rm cell}}
  \right)
  +
  \frac{2}{h}g_R(t)
  +
  f_{N_{\rm cell}}(t).
  \label{eq:wave-neumann-last-row}
\end{align}
Therefore the semidiscrete equation can be written as
\begin{equation}
  \ddot{\bm u}_h(t)
  +
  \widetilde A_h\bm u_h(t)
  =
  \bm f_h(t)
  +
  \frac{2g_R(t)}{h}\bm e_{N_{\rm cell}},
  \label{eq:wave-neumann-fd-system}
\end{equation}
where the matrix from assembling the finite difference formulas is given by,
\begin{equation}
  \widetilde A_h
  =
  \frac{a}{h^2}
  \begin{bmatrix}
  2&-1&&&\\
  -1&2&-1&&\\
  &\ddots&\ddots&\ddots&\\
  &&-1&2&-1\\
  &&&-2&2
  \end{bmatrix}.
  \label{eq:wave-neumann-fd-matrix}
\end{equation}
The last row is different because the right boundary node represents a half
cell.  Consequently \(\widetilde A_h\) is not symmetric in the standard
Euclidean inner product.

This nonsymmetry is not a physical loss of self-adjointness.  It can be attributed to a
mass-weighting issue.  Let
\begin{equation}
  M_{DN}
  =
  \operatorname{diag}
  \left(1,1,\ldots,1,\frac12\right),
  \label{eq:wave-neumann-mass}
\end{equation}
where the harmless common factor \(h\) has been omitted, and define the
symmetric stiffness matrix
\begin{equation}
  K_{DN}
  =
  \frac{a}{h^2}
  \begin{bmatrix}
  2&-1&&&\\
  -1&2&-1&&\\
  &\ddots&\ddots&\ddots&\\
  &&-1&2&-1\\
  &&&-1&1
  \end{bmatrix}.
  \label{eq:wave-neumann-stiffness}
\end{equation}
Then
\begin{equation}
  \widetilde A_h=M_{DN}^{-1}K_{DN}.
  \label{eq:wave-neumann-mass-stiffness}
\end{equation}
This is exactly the mass-lumped finite element or finite-volume form of the
mixed Dirichlet--Neumann discretization.

To obtain a symmetric matrix, introduce the mass-normalized variable
\[
  \bm y_h=M_{DN}^{1/2}\bm u_h .
\]
Then \eqref{eq:wave-neumann-fd-system} becomes
\begin{equation}
  \ddot{\bm y}_h(t)
  +
  A_h^{DN}\bm y_h(t)
  =
  M_{DN}^{1/2}
  \left(
    \bm f_h(t)
    +
    \frac{2g_R(t)}{h}\bm e_{N_{\rm cell}}
  \right),
  \label{eq:wave-neumann-normalized-system}
\end{equation}
where
\begin{equation}
  A_h^{DN}
  =
  M_{DN}^{1/2}\widetilde A_h M_{DN}^{-1/2}
  =
  M_{DN}^{-1/2}K_{DN}M_{DN}^{-1/2}
  =
  (A_h^{DN})^\dag .
  \label{eq:wave-neumann-symmetric-operator}
\end{equation}
Thus the apparently nonsymmetric matrix \(\widetilde A_h\) is similar to a
symmetric positive semidefinite matrix.  For the homogeneous Neumann condition
\(g_R=0\), the boundary contributes no forcing term.  For nonzero traction,
the boundary data appears as a time-dependent forcing term supported at the
right boundary node.

This example illustrates a useful principle.  Boundary conditions should be
incorporated at the level of the semidiscrete variational, finite-volume, or
finite-difference system before choosing the quantum encoding.  A matrix that
looks non-Hermitian in unweighted coordinates may become Hermitian after the
proper mass normalization.  Genuine non-Hermitian dilation is needed only when
the discretization itself contains dissipation, absorbing layers, upwinding,
or other non-self-adjoint mechanisms.

The extension of this simple one-dimensional ghost-point calculation to
higher-dimensional domains is more delicate.  On tensor-product grids, one can
derive analogous boundary rows face by face.  On curved or unstructured
domains, Neumann and traction conditions are most naturally handled through
the weak form and the boundary load vector.  Developing efficient
block-encodings that preserve this mass-weighted self-adjoint structure for
general high-dimensional geometries remains an interesting open problem.

\section{QFT diagonalization, fast forwarding, and splitting}
\label{sec:qft-wave-fast-forward}

Constant-coefficient finite differences on uniform tensor-product grids have more structure than a generic sparse Hamiltonian: the spatial operator is diagonalized by Fourier, sine, or cosine transforms.  Recent QFT-based wave-equation circuits exploit exactly this structure \cite{WrightMcKeeverFirstJohnstonTillayChaneyRosenkranzLubasch2024Wave,LubaschKikuchiWrightMcKeever2025FourierPDE}.

The analogy with the classical FFT solver is close but not identical.  Classically, one transforms all grid values to Fourier space, evolves each mode, and transforms back.  Quantumly, the Fourier transform is applied coherently to the amplitude-encoded state.  The mode index is then available in a register, so the mode frequency can be computed by reversible arithmetic and used to control a rotation.  The circuit acts on all Fourier modes in superposition.  This is why the cost is polylogarithmic in the number of grid points in the ideal structured setting, rather than proportional to the number of modes.

\subsection{Uniform grids with periodic boundary conditions}
\label{subsec:periodic-grid-qft}

Consider the periodic one-dimensional wave equation with \(N=2^n\) grid points.  This \(N\) denotes the periodic register size, not the number of interior Dirichlet unknowns used earlier.  If a nonperiodic discretization has a non-power-of-two number of degrees of freedom, one pads the register as in Chapter~\ref{chap:elliptic-quantum}.  The discrete Fourier transform diagonalizes the positive Laplacian:
\begin{equation}
  F_NA_hF_N^\dag
  =\operatorname{diag}(\omega_k^2),
  \qquad
  \omega_k=\frac{2c}{h}
  \left|\sin\left(\frac{\pi k}{N}\right)\right|,
  \qquad k=0,\ldots,N-1.
  \label{eq:periodic-wave-dispersion}
\end{equation}
After the transformation, each Fourier mode satisfies a decoupled oscillator equation
\begin{equation}
  \ddot{\widehat u}_k(t)+\omega_k^2\widehat u_k(t)=0.
  \label{eq:mode-oscillator}
\end{equation}
For \(\omega_k>0\), define the two-component mode energy vector
\begin{equation}
  \bm{\chi}_k(t)
  =\begin{bmatrix}
     \omega_k\widehat u_k(t)\\
     \widehat v_k(t)
   \end{bmatrix},
  \qquad \widehat v_k=\dot{\widehat u}_k.
\end{equation}
This uses the positive spectral factor $\omega_k$. If the original
strain coordinate is $\widehat w_k=g_k\widehat u_k$ with
$|g_k|=\omega_k$, first multiply it by $\overline{g_k}/|g_k|$ for
$g_k\ne0$, and undo this phase when returning to strain coordinates.
For $\omega_k=0$, the velocity is constant, while the mean displacement
$\widehat u_k(t)=\widehat u_k(0)+t\widehat v_k(0)$ is absent from the
energy coordinate and must be tracked separately if requested.
The vector $\bm\chi_k$ obeys a 2 by 2 Schr\"odinger equation, 
\begin{equation}
  \frac{d}{dt}\bm{\chi}_k(t)
  =\omega_k
  \begin{bmatrix}0&1\\-1&0\end{bmatrix}
  \bm{\chi}_k(t)
  =i\omega_k Y\bm{\chi}_k(t),
  \label{eq:qft-mode-rotation}
\end{equation}
where \(Y\) is the Pauli matrix from Chapter~\ref{chap:basic-elements}.  Hence
\begin{equation}
  \bm{\chi}_k(T)=e^{i\omega_k T Y}\bm{\chi}_k(0),
  \label{eq:qft-mode-propagator}
\end{equation}
or, equivalently, the two eigenstates of \(Y\) acquire phases \(e^{\pm i\omega_k T}\).

Although there are \(N\) Fourier modes, the quantum circuit does not apply \(N\) rotations sequentially.  The mode index \(k\) is stored in a binary register.  Reversible arithmetic computes \(\omega_k T\) coherently in superposition, and a single controlled rotation acts on an energy qubit with angle determined by that register.  The workflow is
\begin{enumerate}[label=(\roman*)]
  \item apply the QFT to the spatial register;
  \item reversibly compute \(\omega_k T\) modulo \(2\pi\) from the binary index \(k\);
  \item apply the controlled \(Y\)-rotation \(e^{i\omega_k T Y}\);
  \item uncompute the angle register and apply the inverse QFT.
\end{enumerate}
This is one coherent circuit acting on the Fourier register, not a loop over modes.  Classically, applying all \(N\) mode rotations would require at least \(O(N)\) arithmetic operations after an FFT.  Quantumly, the register contains a superposition of all \(k\), and the same reversible arithmetic circuit applies the rotation associated with the binary value of \(k\) to every amplitude in superposition.  The price is that the result is still a quantum state: one cannot read all rotated Fourier coefficients without \(O(N)\) measurements.

The exact \(n\)-qubit QFT costs \(O(n^2)=O(\log^2N)\) elementary gates, as reviewed in Section~\ref{sec:qft-iqft}.  The additional phase arithmetic has cost polynomial in \(n\) and in the number of precision bits.  Thus the explicit spectral circuit does not incur the generic multiplicative \(T/h\) Hamiltonian-simulation count.  The factor \(1/h\) appears in the magnitude of the phase \(\omega_k T\), and therefore in the arithmetic precision needed to compute that phase, rather than as the number of simulated time steps.

This is the wave-equation analogue of the Fourier-space reciprocal-filter construction used for Poisson-type equations in \cite{HuangAntonioliBarbaresco2026QuantumSpectralPDE}.  The elliptic filter \(\lambda_{\bm k}^{-1}\) is replaced by the unitary phase or rotation \(e^{\pm i\omega_{\bm k}T}\).  In several dimensions,
\begin{equation}
  \lambda_{\bm k}
  =\frac{4}{h^2}
  \sum_{\ell=1}^d
  \sin^2\left(\frac{\pi k_\ell}{N_\ell}\right),
  \qquad
  \omega_{\bm k}=c\sqrt{\lambda_{\bm k}},
  \label{eq:wave-qft-multidim-frequency}
\end{equation}
so the arithmetic includes the summation and square-root step.  The advantage is that the structured spectral formula can bypass generic sparse Hamiltonian simulation.  The limitation is equally important: this fast-forwarding relies on separability and an explicit diagonal symbol.  It does not extend directly to general finite element meshes or variable-coefficient wave speeds.

For Dirichlet and Neumann conditions on rectangular grids, the QFT is replaced by quantum sine and cosine transforms.  These transforms are QFT relatives obtained by even or odd extensions and have comparable polylogarithmic circuit cost; see, for example, the QFT discussion in \cite{NielsenChuang2010,Coppersmith2002}.  In \(d\) dimensions, tensor-product transforms give formulas analogous to \eqref{eq:wave-qft-multidim-frequency}, with the appropriate boundary-dependent modification.  The QFT-based circuits of Wright et al. and Lubasch et al. also exploit smooth initial data to reduce circuit demands, but state preparation and observable readout remain nontrivial parts of the end-to-end cost \cite{WrightMcKeeverFirstJohnstonTillayChaneyRosenkranzLubasch2024Wave,LubaschKikuchiWrightMcKeever2025FourierPDE}.

\subsection{Klein--Gordon type equations}
\label{subsec:klein-gordon-splitting}

A useful extension is the Klein--Gordon equation with a spatially varying mass term,
\begin{equation}
  u_{tt}-c^2 \Delta u+m(\bm{x})^2u=0.
  \label{eq:kg-equation}
\end{equation}
The term \(m(\bm{x})^2u\) is a local restoring force.  When \(m\) is constant, it creates a spectral gap and makes low-frequency waves oscillate even at zero spatial frequency.  When \(m(\bm{x})\) varies, it acts like a spatially varying local oscillator strength, producing scattering, localization, or trapping effects in much the same way that a potential modifies Schr\"odinger dynamics.

After discretization and mass normalization, write
\begin{equation}
  \ddot{\bm{y}}_h+(A_{K,h}+A_{M,h})\bm{y}_h=\bm{0},
  \label{eq:kg-discrete}
\end{equation}
where \(A_{K,h}\) is the discretized \(-c^2 \Delta\) operator and \(A_{M,h}\) is the positive diagonal matrix associated with \(m(\bm{x})^2\).  Choose
\begin{equation}
  A_{K,h}=G_{K,h}^\dag G_{K,h},
  \qquad
  A_{M,h}=G_{M,h}^\dag G_{M,h},
  \qquad
  G_{M,h}=A_{M,h}^{1/2},
  \label{eq:kg-factorizations}
\end{equation}
where \(G_{M,h}\) is diagonal in physical space.  Stack the factors,
\begin{equation}
  G_h=\begin{bmatrix}G_{K,h}\\G_{M,h}\end{bmatrix},
  \qquad
  G_h^\dag G_h=A_{K,h}+A_{M,h}.
  \label{eq:kg-stacked-factor}
\end{equation}
The stacked form is more than notation.  It separates the part of the energy stored in spatial deformation from the part stored in the local oscillator term.  When \(m(\bm x)\) is diagonal in physical space, \(G_{M,h}\) is also diagonal after a simple finite difference discretization and is therefore easy to block encode or implement by controlled rotations.  The kinetic factor \(G_{K,h}\), on the other hand, is diagonal in Fourier space for constant coefficients.  This is the reason a split-step method is natural.

The energy state has two pressure/strain blocks sharing one velocity block,
\begin{equation}
  \bm{\psi}_h=
  \begin{bmatrix}\bm{p}_{K,h}\\\bm{p}_{M,h}\\\bm{v}_h\end{bmatrix},
  \qquad
  \bm{p}_{K,h}=G_{K,h}\bm{y}_h,
  \quad
  \bm{p}_{M,h}=G_{M,h}\bm{y}_h,
  \quad
  \bm{v}_h=\dot{\bm{y}}_h,
  \label{eq:kg-energy-variables}
\end{equation}
and again $\bm \psi_h$ solves a Schr\"odinger equation. We notice that its Hamiltonian splits naturally as \(H_h=H_{K,h}+H_{M,h}\), with
\begin{equation}
  H_{K,h}=i
  \begin{bmatrix}
  0&0&G_{K,h}\\
  0&0&0\\
  -G_{K,h}^\dag&0&0
  \end{bmatrix},
  \qquad
  H_{M,h}=i
  \begin{bmatrix}
  0&0&0\\
  0&0&G_{M,h}\\
  0&-G_{M,h}^\dag&0
  \end{bmatrix}.
  \label{eq:kg-hamiltonian-split}
\end{equation}
Both terms are Hermitian.  The \(H_{M,h}\) evolution consists of independent physical-space rotations, while \(H_{K,h}\) becomes a collection of Fourier-space rotations after the QFT.

A direct approach for simulating $\bm \psi_h(t)$ is Hamiltonian simulations by qubitization of QSVT    applied to $H_h$. On the other hand, this splitting in \cref{eq:kg-hamiltonian-split} is the wave-equation analogue of split-operator methods for the Schr\"odinger equation.  The mass term is diagonal in physical space, so it is cheap when \(m(\bm{x})\) can be evaluated from the grid index.  The Laplacian is diagonal in Fourier space, so it is cheap on a uniform periodic or separable grid.  The noncommutativity of these two diagonal representations is exactly what the splitting error measures.

The most widely used splitting scheme in classical algorithms is Strang splitting. For time step of size $\Delta t$, it approximates the evolution by,
\begin{equation}
  e^{-i \Delta t H_h}
   \approx e^{-i\Delta tH_{M,h}/2}
   e^{-i\Delta tH_{K,h}}
   e^{-i\Delta tH_{M,h}/2}.
  \label{eq:strang-splitting}
\end{equation}
If the relevant nested commutators are bounded, the one-step error is
\begin{equation}
  \|e^{-i\Delta t(H_{K,h}+H_{M,h})}-S_{\Delta t}\|
  \leq C_{\rm comm}\Delta t^3,
\end{equation}
and the global error over time \(T\) is \(O(C_{\rm comm}T\Delta t^2)\).  This constant is generally mesh dependent.  If \(G_{K,h}\) has size \(O(h^{-1})\) and \(G_{M,h}\) is bounded independently of \(h\), then a rough estimate for the nested commutators gives the scaling
\begin{equation}
  C_{\rm comm}=O(h^{-2}).
  \label{eq:kg-commutator-scale}
\end{equation}
Thus a Strang calculation with error \(O(\epsilon)\) over time \(T\) may require roughly
\[
  T/\Delta t=O\!\left(T^{3/2}h^{-1}\epsilon^{-1/2}\right)
\]
steps in a worst-case norm estimate.  Compared with the generic count \(\widetilde O(T/h)\) of qubitization, the number of steps is larger by a factor \(T^{1/2}\epsilon^{-1/2}\); what splitting buys in exchange is the structure of each step.  The method can still be attractive because each step uses very structured QFT and diagonal-rotation circuits.  It should not, however, be read as a mesh-independent splitting method.  Each step follows the workflow
\begin{equation}
\begin{aligned}
  &\text{physical-space rotation}
  \longrightarrow \text{QFT}
  \longrightarrow \text{Fourier-space rotation}
  \\
  &\longrightarrow \text{QFT}^{\dagger}
  \longrightarrow \text{physical-space rotation}.
\end{aligned}
  \label{eq:qft-strang-workflow}
\end{equation}

Higher-order Trotter splitting schemes can also be applied to reduce the dependence on $T$ and $\epsilon$. The commutator error scaling has been extensively studied \cite{childs2021theory}.

\section{General meshes: access, state preparation, and readout}
\label{sec:wave-qsvt-general-mesh}

For unstructured finite element meshes, variable coefficients, complicated domains, and nonseparable boundary conditions, Fourier diagonalization is unavailable.  The robust approach is to block encode the sparse factor \(G_h\) and simulate the Hermitian Hamiltonian \eqref{eq:wave-schrodinger-hamiltonian}.  Its generic cost has already been stated in \eqref{eq:qsvt-wave-complexity}.  Figure~\ref{fig:hyperbolic-pipeline} summarizes the resulting pipeline; the rest of this section walks through its three stages---input, evolution, and readout---and records what each stage costs.

This is the point at which the access model becomes part of the numerical method.  In a finite difference code, one might think of \(G_h\) as a stencil.  In a finite element code, one may instead access it through local element matrices, quadrature rules, and incidence relations between elements and degrees of freedom.  Both descriptions are sparse, but they lead to different reversible oracles.  The book-level abstraction is a block encoding of \(G_h/\alpha_G\); the implementation-level question is how the mesh data and coefficient data are made coherent.

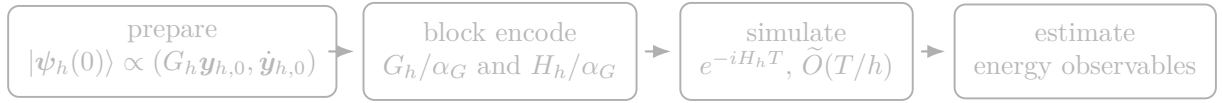
\begin{figure}[htbp]
\centering
\resizebox{\textwidth}{!}{%
\begin{tikzpicture}[x=1cm,y=1cm,line cap=round,line join=round,>=Latex,
    every node/.style={font=\small},
    flow/.style={gray!62,line width=0.8pt,->,shorten >=3pt,shorten <=3pt},
    box/.style={draw=gray!62,fill=white,rounded corners=4pt,
                line width=0.7pt,align=center,inner sep=7pt,text=gray!56}]
  \node[box] (prep) at (0,0)
    {prepare\\$\ket{\bm{\psi}_h(0)}\propto
    (G_h\bm{y}_{h,0},\dot{\bm{y}}_{h,0})$};
  \node[box] (enc) at (4.2,0)
    {block encode\\$G_h/\alpha_G$ and $H_h/\alpha_G$};
  \node[box] (sim) at (8.0,0)
    {simulate\\$e^{-iH_hT}$, $\widetilde O(T/h)$};
  \node[box] (meas) at (11.7,0)
    {estimate\\energy observables};
  \draw[flow] (prep)--(enc);
  \draw[flow] (enc)--(sim);
  \draw[flow] (sim)--(meas);
\end{tikzpicture}%
}
\caption{The generic wave-equation pipeline.  The homogeneous energy-variable evolution is unitary and requires no final postselection.  The Duhamel--LCU construction for sources introduces a separate source-loading and success-amplitude factor.}
\label{fig:hyperbolic-pipeline}
\end{figure}

The input state is proportional to
\begin{equation}
  \bm{\psi}_h(0)
  =\begin{bmatrix}
    G_h\bm{y}_h(0)\\
    \dot{\bm{y}}_h(0)
  \end{bmatrix}.
  \label{eq:wave-initial-energy-state}
\end{equation}
Thus state preparation requires access not only to the initial displacement and velocity, but also to the discrete gradient of the displacement.  On structured grids this can be generated by arithmetic circuits or by applying a block encoding of \(G_h\).  For mass-lumped finite elements, one first prepares the coefficient vectors, applies the diagonal mass normalization, and then applies the local gradient factor.

For smooth initial data on a tensor-product grid, Grover--Rudolph type state preparation or Fourier preparation may be appropriate.  For finite element data, a more realistic input may be a mesh-based data structure or a classical preprocessing routine that prepares nodal values.  In either case, the cost of preparing \(G_h\bm{y}_{h,0}\) should not be silently hidden: it is part of the end-to-end algorithm, just as assembling a stiffness matrix or applying a gradient operator is part of a classical workflow.

The output quantum state is the normalized pressure--velocity state
\begin{equation}
  \ket{\bm{\psi}_h(T)}
  :=
  \frac{\bm{\psi}_h(T)}{\|\bm{\psi}_h(T)\|},
  \qquad
  \bm{\psi}_h(T)
  =
  \begin{bmatrix}
    G_h\bm{y}_h(T)\\
    \dot{\bm{y}}_h(T)
  \end{bmatrix}.
  \label{eq:wave-output-energy-state}
\end{equation}
For an observable \(O=O^\dag\) acting on the energy variables, the normalized quantity of interest is
\begin{equation}
  \mu_O(T)
  =
  \bra{\bm{\psi}_h(T)}O\ket{\bm{\psi}_h(T)}
  =
  \frac{\bm{\psi}_h(T)^\dag O\bm{\psi}_h(T)}
       {\|\bm{\psi}_h(T)\|^2}.
  \label{eq:wave-energy-observable}
\end{equation}
This is the observable-estimation problem introduced in Section~\ref{sec:measurement-observables}.  After rescaling \(O\), we may assume \(\|O\|\leq1\).  One preparation of \(\ket{\bm{\psi}_h(T)}\) costs
\begin{equation}
  Q_{\rm sim}(T,\epsilon)
  =\widetilde O\left(\frac{T}{h}+\log\frac1\epsilon\right)
  \label{eq:wave-one-simulation-cost}
\end{equation}
queries to the block encoding of the wave Hamiltonian.

With direct sampling, \(O(\epsilon^{-2}\log(1/\delta))\) independent preparations and measurements estimate \(\mu_O(T)\) to additive error \(\epsilon\) with failure probability at most \(\delta\).  Hence
\begin{equation}
  Q_{\rm tot}^{\rm samp}
  =\widetilde O\left(
    \left(\frac{T}{h}+\log\frac1\epsilon\right)
    \frac{1}{\epsilon^2}
    \log\frac1\delta
  \right).
  \label{eq:wave-sampling-total-cost}
\end{equation}
With amplitude estimation, the measurement dependence improves to \(O(\epsilon^{-1})\), provided one can coherently apply the state-preparation circuit and its inverse:
\begin{equation}
  Q_{\rm tot}^{\rm AE}
  =\widetilde O\left(
    \left(\frac{T}{h}+\log\frac1\epsilon\right)
    \frac{1}{\epsilon}
    \log\frac1\delta
  \right).
  \label{eq:wave-ae-total-cost}
\end{equation}
This is the same sampling-versus-amplitude-estimation distinction discussed in Chapter~\ref{chap:basic-elements}.  The improvement in measurement complexity comes at the price of greater coherent depth.  For near-term or early fault-tolerant devices, direct sampling may be more realistic; for asymptotic resource estimates, amplitude estimation is the natural benchmark.  In either case, the measurement cost multiplies the cost of preparing the wave state, so it should not be omitted from an end-to-end PDE complexity estimate.

Finally, the normalized expectation can be converted into the corresponding unnormalized physical quadratic form:
\begin{equation}
  \bm{\psi}_h(T)^\dag O\bm{\psi}_h(T)
  =\|\bm{\psi}_h(T)\|^2\mu_O(T).
  \label{eq:wave-unnormalized-observable}
\end{equation}
For the homogeneous wave equation,
\begin{equation}
  \|\bm{\psi}_h(T)\|=\|\bm{\psi}_h(0)\|,
  \label{eq:wave-energy-norm-preserved}
\end{equation}
so this normalization factor is determined by the initial discrete energy.  This is simpler than the nonunitary linear-system and heat-equation settings discussed in Subsection~\ref{subsec:qoi-and-output-norm}, where the output norm may itself be an unknown quantity.

Recovering every displacement coefficient is not generally an efficient quantum output task.  The favorable setting is instead one in which the desired result is a small collection of energy fractions, local energies, modal weights, fluxes, or other observables.  If displacement observables are essential, one can carry the displacement as an additional dynamical block, as described next.

\subsection{An augmented displacement formulation}
\label{subsec:augmented-displacement-formulation}

So far, we have relied on the first order form, using the pressure and velocity as the primary variables, ot leverage efficient Hamiltonian simulations. To include the displacement in the simulation, we can define
\begin{equation}
  \bm{z}_h(t)
  =\begin{bmatrix}
    \bm{y}_h(t)\\
    \bm{p}_h(t)\\
    \bm{v}_h(t)
  \end{bmatrix},
  \qquad
  \bm{p}_h=G_h\bm{y}_h,
  \quad
  \bm{v}_h=\dot{\bm{y}}_h.
  \label{eq:augmented-wave-state}
\end{equation}
Then the time evolution becomes
\begin{equation}
  \frac{d}{dt}\bm{z}_h(t)
  =L_{\rm aug}\bm{z}_h(t),
  \qquad
  L_{\rm aug}
  =\begin{bmatrix}
  0&0&I\\
  0&0&G_h\\
  0&-G_h^\dag&0
  \end{bmatrix}.
  \label{eq:augmented-wave-ode}
\end{equation}
The accumulator equation \(\dot{\bm{y}}_h=\bm{v}_h\) makes \(L_{\rm aug}\) non-skew-Hermitian in the Euclidean inner product.  It therefore cannot be simulated directly by ordinary Hamiltonian simulation.  The dilation introduced in the next section gives a systematic way to map this augmented dynamics, as well as dissipative upwind dynamics, to a Schr\"odinger equation on a larger space. In particular, since the non-Hermitian part of $L_{\rm aug}$ only scales $O(1)$, the overall time complexity is still $O(T/h).$

The need for this augmented system is purely an output issue.  The homogeneous pressure--velocity dynamics already gives a unitary quantum evolution, but it does not store \(\bm{y}_h\) as a state component.  Carrying \(\bm{y}_h\) turns the evolution into an accumulator system, analogous to integrating velocity to recover displacement in a classical code.  Accumulators are useful, but they are not energy-preserving by themselves, which is why dilation rather than direct Hamiltonian simulation is needed.

\section[First-order systems and dilation]{First-order hyperbolic systems, Schr\"odingerization, and dilation}
\label{sec:first-order-hyperbolic-dilation}

Many models in transport, acoustics, elasticity, and electromagnetism are
written directly as first-order systems.  These systems are closer to the
pressure--velocity formulation than to the second-order displacement
formulation.  The main new issue is that stable finite difference methods for
first-order hyperbolic equations often create numerical dissipation.  That
dissipation is essential classically, but it destroys unitarity and therefore
forces us to use some form of embedding or dilation.  This is the far end of
the conservation dial: the departure from unitarity is no longer a bookkeeping
matter, as it was for sources, but is built into the discretization itself.
The three frameworks presented below---Schr\"odingerization, LCHS, and
moment-matching dilation---are three routes to the same destination, a unitary
evolution on a larger space whose projection reproduces the stable nonunitary
scheme.

Consider the constant-coefficient problem
\begin{equation}
  \bm{u}_t+A\bm{u}_x=\bm{0},
  \qquad
  \bm{u}(x,t)\in\mathbb R^m.
  \label{eq:first-order-system}
\end{equation}
The PDE system is called hyperbolic if $A$ is diagonalizable with real eigenvalues, 
\begin{equation}
  A=S\Lambda S^{-1},
  \qquad
  \Lambda=\operatorname{diag}(\lambda_1,\ldots,\lambda_m).
  \label{eq:A-eigendecomposition}
\end{equation}
The characteristic variables \(\bm{w}=S^{-1}\bm{u}\) satisfy decoupled advection
equations
\begin{equation}
  (w_j)_t+\lambda_j(w_j)_x=0.
  \label{eq:characteristic-advections}
\end{equation}
This is the classical reason that upwinding is built from the eigendecomposition
of the flux matrix.  Symmetrizable systems can be treated similarly after
changing the inner product.

Define
\begin{equation}
  \Lambda^+=\operatorname{diag}(\max(\lambda_j,0)),
  \qquad
  \Lambda^-=\operatorname{diag}(\min(\lambda_j,0)),
\end{equation}
and
\begin{equation}
  A^+=S\Lambda^+S^{-1},
  \qquad
  A^-=S\Lambda^-S^{-1},
  \qquad
  |A|=A^+-A^-=S|\Lambda|S^{-1}.
  \label{eq:A-flux-splitting}
\end{equation}
Let \(\bm u_j(t)\) denote the grid value at \(x_j\).  The semidiscrete upwind
method can be expressed as the matrix-vector level,
\begin{equation}
  \dot{\bm u}_j
  +A^+\frac{\bm u_j-\bm u_{j-1}}{h}
  +A^-\frac{\bm u_{j+1}-\bm u_j}{h}
  =\bm{f}_j(t),
  \label{eq:upwind-system}
\end{equation}
where \(\bm f_j(t)\) contains body forcing or boundary contributions.
Equivalently,
\begin{equation}
  \dot{\bm u}_j
  =-A\frac{\bm u_{j+1}-\bm u_{j-1}}{2h}
  +\frac{h}{2}|A|
   \frac{\bm u_{j+1}-2\bm u_j+\bm u_{j-1}}{h^2}
  +\bm f_j(t).
  \label{eq:upwind-skew-dissipative}
\end{equation}
The first term is the centered transport part.  The second term is numerical
viscosity.  This is the semidiscrete form of the familiar finite-difference
principle: centered differencing is often nondissipative, while upwind differencing
adds just enough damping to control one-sided propagation.

For the global vector
\begin{equation}
  \bm u_h(t)=(\bm u_0(t)^T,\ldots,\bm u_{N-1}(t)^T)^T,
\end{equation}
the semidiscrete equation has the form
\begin{equation}
  \dot{\bm u}_h(t)=L_h^{\rm up}\bm u_h(t)+\bm f_h(t).
  \label{eq:upwind-nonunitary-structure}
\end{equation}

This is a useful example of the tension between numerical stability and
quantum unitarity.  A centered discretization of a symmetric hyperbolic system
is naturally skew-Hermitian, but it may be oscillatory near discontinuities or
sharp gradients.  Upwinding repairs the classical stability problem by adding
dissipation.  The goal is therefore not to discard upwinding in order to keep
a unitary equation, but to embed the stable nonunitary scheme into a larger
unitary dynamics.

\subsection{Schr\"odingerization: the PDE-to-Schr\"odinger transformation}
\label{subsec:schrodingerization-first-order}

A systematic way to embed nonunitary PDE dynamics into Schr\"odinger-type
dynamics was introduced by Jin and collaborators under the name
\emph{Schr\"odingerization}
\cite{JinLiuYu2024Schrodingerization,JinLiu2025SchrodingerizationCVPerspective}.
The basic idea is to add one auxiliary continuous variable so that the
dissipative part of the PDE becomes a derivative in that new variable.  The
resulting equation is a Schr\"odinger equation on an enlarged space.

For the semidiscrete flow
\begin{equation}
  \dot{\bm x}(t)
  =
  \bigl(-iH(t)+K(t)\bigr)\bm x(t),
  \qquad
  H(t)=H(t)^\dag,
  \qquad
  K(t)=K(t)^\dag,
  \label{eq:schrodingerization-nonunitary-flow}
\end{equation}
introduce an auxiliary coordinate \(p\ge0\) and set, formally,
\begin{equation}
  \bm\Psi(t,p)=e^{-p}\bm x(t).
\end{equation}
Since \(\partial_p\bm\Psi=-\bm\Psi\), the term \(K(t)\bm x(t)\) can be
represented as \(-K(t)\partial_p\bm\Psi(t,p)\).  Thus
\eqref{eq:schrodingerization-nonunitary-flow} is lifted to
\begin{equation}
  i\partial_t\bm\Psi(t,p)
  =
  \left[
    H(t)-iK(t)\partial_p
  \right]\bm\Psi(t,p).
  \label{eq:schrodingerization-extended-wave}
\end{equation}
The half-line ansatz is a formal identity. In particular a unitary realization needs
an ancillary domain and an initial extension. One choice is $p\in\mathbb R$
with $\Psi(0,p)=e^{-|p|}\bm x(0)$. For finite-dimensional Hermitian
$H(t)$ and $K(t)$ with suitable time regularity, Fourier transformation
in $p$ gives Hermitian fibers $H(t)+\eta K(t)$ and hence a unitary
evolution on $L^2(\mathbb R)\otimes\mathbb C^n$.
On the positive-$p$ side, finite propagation gives the sufficient
recovery region
\[
 p>\int_0^t\max\{0,\lambda_{\max}(K(s))\}\,ds,
 \qquad \bm x(t)=e^p\Psi(t,p).
\]
For $K(t)\le0$, every $p>0$ lies in this region; recovery at $p=0$ can
be interpreted as a one-sided trace when regularity permits. For
general $K$, evaluation at $p=0$ is not a valid universal recovery rule.
Truncating the ancillary interval, discretizing it, and extracting
a recovery region introduce additional errors and success probabilities
that must be included in an algorithmic cost
\cite{JinLiuYu2024Schrodingerization}.

\subsection{LCHS: an LCU of Hamiltonian simulations}
\label{subsec:lchs-first-order}

Another way to view the same embedding problem is through the LCHS method of
An, Liu, and Lin \cite{AnLiuLin2023}, already encountered for sources in
Section~\ref{sec:wave-forcing-boundary-lcu}.  LCHS seeks a representation of a nonunitary flow as a
linear combination, or integral, of unitary Hamiltonian simulations.  This
makes the connection to the LCU primitive of Chapter~\ref{chap:basic-elements}
particularly explicit.

For an autonomous generator
\[
  L=-iH+K,
  \qquad H=H^\dag,
  \qquad K=K^\dag,
\]
one aims for an identity or approximation of the schematic form
\begin{equation}
  e^{t(-iH+K)}
  \approx
  \sum_{m=1}^{M} c_m\,
  e^{-it(H+\eta_m K)} .
  \label{eq:lchs-discrete-lcu}
\end{equation}
Each term is a Hamiltonian simulation with Hamiltonian \(H+\eta_mK\).  Once
the coefficients \(c_m\) and parameters \(\eta_m\) are chosen, the algorithm
has the standard LCU structure:
\[
  \operatorname{PREPARE}:
  \ket{0}\mapsto
  \frac{1}{\sqrt{\Lambda}}
  \sum_m \sqrt{|c_m|}\ket{m},
  \qquad
  \Lambda=\sum_m|c_m|,
\]
followed by
\[
  \operatorname{SELECT}:
  \ket{m}\ket{\psi}
  \mapsto
  \ket{m}\,e^{i\arg c_m}e^{-it(H+\eta_mK)}\ket{\psi},
\]
and postselection on the coefficient register.  Thus the LCU normalization
\(\Lambda\), the longest simulation time, and the cost of block encoding
\(H+\eta_mK\) determine the complexity.

In continuous form, \eqref{eq:lchs-discrete-lcu} becomes
\begin{equation}
  e^{t(-iH+K)}
  \approx
  \int_{\mathbb R} w(\eta)\,
  e^{-it(H+\eta K)}\,d\eta .
  \label{eq:lchs-continuous-integral}
\end{equation}
The contractive LCHS representation assumes $K\le0$, allowing zero
eigenvalues. For this finite-dimensional autonomous generator this is
equivalent to $\|e^{t(-iH+K)}\|\le1$ for every $t\ge0$.
If a shift is used to make $K$ nonpositive, its scalar growth factor
must be restored and counted in the normalization.
After quadrature, this is exactly a linear combination of Hamiltonian
simulations.  This representation is especially appealing in numerical PDEs:
one keeps the stable nonunitary discretization, but simulates a family of
Hermitian Hamiltonians instead of the non-Hermitian matrix directly.

The LCHS viewpoint is closely related to Schr\"odingerization.  In both cases,
a dissipative or nonunitary term is represented by a superposition of unitary
evolutions in an enlarged space.  The difference is mostly organizational:
Schr\"odingerization emphasizes the auxiliary PDE and continuous-variable
Hamiltonian, whereas LCHS emphasizes the quadrature and LCU implementation.

\subsection{Moment-matching dilation}
\label{subsec:moment-matching-dilation-wave}

The moment-matching dilation framework of \cite{Li2025MomentDilation}
abstracts both of these ideas.  Consider again
\begin{equation}
  \dot{\bm{x}}(t)=L(t)\bm{x}(t),
  \qquad
  L(t)=-iH(t)+K(t),
  \qquad
  H(t)=H(t)^\dag,
  \qquad
  K(t)=K(t)^\dag.
  \label{eq:general-nonhermitian-flow}
\end{equation}
The natural question is: what ancillary operator \(F\), right vector
\(\ket{r}\), and evaluation functional \(\bra{l}\) make a projected unitary
evolution reproduce this nonunitary flow exactly?

Let \(F^\dag=-F\) act on an ancillary space, and define the dilated Hamiltonian
\begin{equation}
  \widetilde H(t)
  =
  I_A\otimes H(t)+iF\otimes K(t).
  \label{eq:moment-dilated-hamiltonian}
\end{equation}
It is Hermitian because \(iF\) and \(K(t)\) are Hermitian.  Expanding its
propagator in a Taylor series in the autonomous case, or a Dyson series in the
time-dependent case, shows that every term containing \(q\) occurrences of
\(K\) carries the ancillary coefficient
\begin{equation}
  m_q:=\langle l|F^q|r\rangle .
\end{equation}
A detail is worth making explicit, because it is what makes the construction
work.  Within a given expansion term, the \(H\) factors and the \(K\) factors
are interleaved in some time order; but the \(H\) factors act as the identity
on the ancillary space, so the ancillary operator product is \(F^q\) for
\emph{every} interleaving.  The ancillary coefficient therefore depends only
on the count \(q\), not on the positions of the \(K\) factors, and a single
sequence of scalar moments controls the entire series.  Thus the formal Taylor or Dyson coefficients agree if
\begin{equation}
  \langle l|F^q|r\rangle=1,
  \qquad q=0,1,2,\ldots .
  \label{eq:moment-identities}
\end{equation}
When the resulting series converge and the encoding and recovery operations permit termwise application, these identities give
\begin{equation}
  (\langle l|\otimes I)
  \mathcal T\exp\left[
  -i\int_0^t\widetilde H(s)\,ds
  \right]
  (|r\rangle\otimes I)
  =
  \mathcal T\exp\left[
  \int_0^tL(s)\,ds
  \right].
  \label{eq:moment-dilation-formula}
\end{equation}
For bounded matrices, matching moments through order $m$ matches the
series through that order. Further, a usable finite-order approximation also needs
a remainder bound involving the operator norms and time. For unbounded
function-space realizations, domains, convergence, and the boundedness
or continuity of recovery must be checked separately. 

Schr\"odingerization is recovered by taking
\begin{equation}
  F=-\partial_p,
  \qquad
  r(p)=e^{-p},
  \qquad
  \langle l|\varphi\rangle=\varphi(0).
  \label{eq:schrodingerization-moment-triple}
\end{equation}
Formally on $p>0$, $F^qr=r$ for every $q\ge0$, yielding the stated
moments. This is not a whole-line normalizable eigenvector of a
skew-adjoint derivative. A unitary implementation uses an extension
and a valid recovery region, as described above.  Other
choices of the moment triple produce differential, integral, pseudodifferential,
finite-difference, and Bargmann--Fock dilations
\cite{Li2025MomentDilation}.

The analogy with reduced-order modeling is intentional.  In rational Krylov
and moment-matching methods, one chooses a small subspace so that selected
transfer-function moments are reproduced.  Here one chooses an ancillary
dynamics so that the moments of the dissipative part \(K(t)\) are reproduced
by projected unitary evolution.  The reward is that a nonunitary numerical
flow can be represented by a unitary flow on a larger space; the price is
ancillary dimension, projection success probability, and discretization of the
ancillary variable.

Finally, if \(iF\) has spectral resolution
\begin{equation}
  iF=\int_{\mathbb R}\eta\,dE_F(\eta),
\end{equation}
then each ancillary spectral sector evolves under \(H+\eta K\).  Projection
onto \(\bra l\) and \(\ket r\) gives an integral representation
\begin{equation}
  e^{t(-iH+K)}
  =
  \int_{\mathbb R}w_{l,r}(\eta)
    e^{-it(H+\eta K)}\,d\eta,
  \label{eq:dilation-lcu-integral}
\end{equation}
at least formally in the autonomous case.  After quadrature, this becomes the
LCHS/LCU form in \eqref{eq:lchs-discrete-lcu}.  Thus Schr\"odingerization,
LCHS, and moment-matching dilation are three views of the same underlying
principle: represent a stable nonunitary flow as the projection of a unitary
flow in a larger space.  The Gaussian-transmutation representation of the heat
semigroup in Chapter~\ref{chap:parabolic-quantum} is built in exactly this
spirit: a dissipative flow written as an integral of unitary evolutions, with
an explicit weight and an explicit quadrature.

\subsection{Implication for upwind hyperbolic discretizations}
\label{subsec:upwind-dilation-implication}

We now return to the upwind system
\[
  \dot{\bm u}_h=L_h^{\rm up}\bm u_h+\bm f_h(t),
  \qquad
  L_h^{\rm up}=-iH_h +K_h .
\]
When $A$ is Hermitian, the centered transport part gives \(H_h\), while the upwind
viscosity gives \(K_h \preceq0\).  The preceding discussion says that
we do not need to abandon this stable discretization in order to make a
quantum algorithm.  Instead, one can choose a dilation or LCHS representation
and simulate Hermitian Hamiltonians of the form
\[
  H_h +\eta K_h .
\]
For fixed spatial dimension and fixed system size \(m\), the operator norms
scale as
\[
  \|H_h \|+\|K_h \|
  =
  O\!\left(\frac{\max_j|\lambda_j|}{h}\right).
\]
Therefore the coherent simulation scale for the homogeneous upwind evolution
is expected to be
\begin{equation}
  \widetilde O\!\left(
    \frac{T\max_j|\lambda_j|}{h}
  \right),
  \label{eq:upwind-dilation-complexity}
\end{equation}
up to logarithmic factors, dilation-discretization costs, and postselection or
LCU normalization factors.  This is the same CFL-type scale that appears in
classical explicit hyperbolic solvers, but without the factor corresponding to
updating all grid values.  As throughout this chapter, the quantum advantage
can only be meaningful for state preparation, sampling, or low-dimensional
observables, not for printing the entire solution vector.

For nonzero forcing or boundary terms, one combines the dilation of the
homogeneous propagator with the Duhamel and LCU construction described in
Section~\ref{sec:wave-forcing-boundary-lcu}.  The source history must be loaded
coherently, and its norms enter the LCU normalization.  Thus the final
workflow is:

\begin{enumerate}
  \item construct a stable upwind or flux-splitting semidiscretization;
  \item write the global ODE as
        \(\dot{\bm u}_h=L_h^{\rm up}\bm u_h+\bm f_h(t)\);
  \item separate
        \(L_h^{\rm up}=-iH_h +K_h \);
  \item block encode \(H_h \) and \(K_h \);
  \item choose a Schr\"odingerization, LCHS, or moment-matching dilation;
  \item simulate the enlarged Hermitian dynamics and estimate the desired
        observable.
\end{enumerate}

The same dilation principle also applies to the augmented displacement
generator in \eqref{eq:augmented-wave-ode}.  In that case the nonunitarity is
introduced artificially to carry an additional displacement block; in the
upwind case, it is introduced by the stable numerical discretization itself.

\section{Summary and outlook}
\label{sec:hyperbolic-summary-outlook}

The main lesson of this chapter is that hyperbolic PDEs enter quantum algorithms through their energy structure.  Finite differences and finite elements give
\begin{equation}
  M_h\ddot{\bm{q}}_h+K_h\bm{q}_h=\bm{b}_h(t).
\end{equation}
After mass normalization and factorization
$A_h=M_h^{-1/2}K_hM_h^{-1/2}=G_h^\dag G_h$, the homogeneous equation becomes
\begin{equation}
  i\dot{\bm{\psi}}_h=H_h\bm{\psi}_h,
  \qquad
  H_h=i\begin{bmatrix}0&G_h\\-G_h^\dag&0\end{bmatrix}.
\end{equation}
For general meshes, optimal Hamiltonian simulation has the leading scale
$\widetilde O(T/h)$.  Explicit Fourier diagonalization can do better for special constant-coefficient tensor-product problems.  Upwind discretizations introduce controlled nonunitarity and motivate moment-matching dilation.

From a numerical-analysis perspective, the message is that the choice of variables matters as much as the choice of discretization.  The same second-order wave equation can lead to a displacement--velocity system, a pressure--velocity Hamiltonian, an augmented accumulator system, or a dissipative upwind system.  These formulations are classically equivalent or closely related, but they lead to different quantum primitives: Hamiltonian simulation, QFT phase arithmetic, LCU for sources, or dilation for nonunitary dynamics.

\paragraph{Boundary conditions.}
Nonhomogeneous Dirichlet, Neumann, and Robin data can be converted into source terms and treated by the Duhamel--LCU construction.  Absorbing layers, impedance conditions, and perfectly matched layers generally introduce dissipation and are more naturally handled by dilation.  From a numerical viewpoint, this is the familiar distinction between conservative boundary treatments and energy-absorbing ones.  From a quantum viewpoint, it is the distinction between unitary Hamiltonian simulation and nonunitary evolution with a success probability or dilation overhead.  Efficient preparation of liftings and boundary-load states remains a problem-specific issue. An alternative is to enforce boundary values and constraints through penalty
projections. Schleich et al. combine this approach with interaction-picture
LCHS, using block-encoded generators and efficiently simulable constraint
projections \cite{SchleichEtAl2025PenaltyProjections}.

\paragraph{Fast-forwardable wave dynamics.}
The QFT example shows that long-time wave propagation can be fast-forwarded when both the eigenbasis and dispersion relation are efficiently computable.  This is a very special situation.  The circuit computes the phase \(\omega_{\bm k}T\) directly instead of simulating the Hamiltonian for time \(T\).  Characterizing other PDE discretizations with this structure, while respecting the generic no-fast-forwarding obstruction, is an important direction.  Possible examples include separable media, periodic lattice models, and certain exactly diagonalizable discretizations.

\paragraph{Multiscale wave dynamics.}
High-contrast media, small geometric features, and locally refined meshes introduce several spatial scales.  The finest scale still controls $\|G_h\|$ and hence the generic Hamiltonian-simulation normalization.  This mirrors classical CFL restrictions, where the smallest cells and largest wave speeds dictate the time step.  Multilevel,  homogenization, or reduced-basis ideas may be needed to prevent the smallest mesh size from determining the full quantum cost.

\paragraph{Maxwell equations.}
Maxwell's equations are already first order and possess a natural energy inner product.  Compatible curl discretizations,  and edge elements can lead directly to Hermitian or skew-Hermitian block systems, while conductivity, absorbing boundaries, and material dispersion introduce nonunitary terms.  These systems connect naturally to both Hamiltonian simulation and dilation.
Sharma et al. have demonstrated a Schr\"odingerization-based simulation of
two-dimensional Maxwell dynamics on IonQ Forte Enterprise hardware, using
compressed circuits and a procedure to reconstruct signed field components
\cite{SharmaEtAl2026MaxwellHardware}.

\paragraph{More general spatial discretizations.}
Discontinuous Galerkin, spectral-element, mixed, mimetic, and high-order finite difference methods may offer better accuracy per degree of freedom than the lowest-order schemes used here.  Their quantum suitability depends not only on sparsity, but also on mass-matrix structure, factorization, block-encoding normalization, state preparation, and the observable to be measured.  The most promising discretization for a quantum algorithm may therefore not be the one with the fewest classical floating-point operations.  It may instead be the one whose algebraic structure leads to the cleanest block encoding and the most useful quantum output state.

\section*{What to remember}
\addcontentsline{toc}{section}{What to remember}

\begin{itemize}
  \item The homogeneous wave equation becomes a Schr\"odinger equation after
  spatial discretization, mass normalization, and passage to pressure/strain
  and velocity variables.
  \item The generic Hamiltonian-simulation normalization is
  \(\|G_h\|=\Theta(h^{-1})\), so the black-box wave cost scales like
  \(\widetilde O(T/h)\), not \(T/h^2\).
  \item QFT methods can fast-forward special constant-coefficient tensor-grid
  problems because both the eigenbasis and the dispersion relation are known
  explicitly.
  \item Forcing and nonhomogeneous boundary conditions enter through
  Duhamel's formula.  The propagator may be efficient, but coherent loading of
  the time-dependent source is an additional input assumption.
  \item Upwind schemes are stable because they add numerical dissipation.  That
  same dissipative term makes the semidiscrete dynamics nonunitary and leads
  naturally to dilation methods.
\end{itemize}

\section*{Exercises}
\addcontentsline{toc}{section}{Exercises}

\begin{exercise}[Continuous energy conservation]
For the homogeneous wave equation with homogeneous Dirichlet boundary
conditions, differentiate the energy in \eqref{eq:wave-continuous-energy}
and use integration by parts to prove
\[
  \frac{d}{dt}E(t)=0.
\]
\end{exercise}

\begin{exercise}[Finite-difference wave frequencies]
For the one-dimensional matrix \eqref{eq:fd-wave-matrix}, compute the
eigenvalues and verify that \(\|A_h^{1/2}\|=\Theta(h^{-1})\).
\end{exercise}

\begin{exercise}[Gradient factorization]
Show that the first-difference matrix \eqref{eq:dirichlet-gradient-matrix}
satisfies
\[
  D_h^\dag D_h=h^{-2}\operatorname{tridiag}(-1,2,-1).
\]
\end{exercise}

\begin{exercise}[Finite-element energy]
Starting from
\(M_h\ddot{\bm q}_h+K_h\bm q_h=\bm 0\), prove the finite element energy
identity \eqref{eq:fem-wave-energy}.
\end{exercise}

\begin{exercise}[Fourier-mode rotations]
Derive the periodic dispersion relation \eqref{eq:periodic-wave-dispersion}
and the mode rotation \eqref{eq:qft-mode-rotation}.
\end{exercise}

\begin{exercise}[Hermiticity of the wave Hamiltonian]
Verify that \(H_h\) in \eqref{eq:wave-schrodinger-hamiltonian} is Hermitian
and that its Schr\"odinger equation is equivalent to
\eqref{eq:wave-pv-system}.
\end{exercise}

\begin{exercise}[Upwind dissipation]
For scalar advection with positive velocity, show that upwinding equals a
centered difference plus numerical viscosity.
\end{exercise}

\begin{exercise}[Flux splitting]
For symmetric \(A=S\Lambda S^{-1}\), derive the flux splitting
\eqref{eq:A-flux-splitting} and the upwind method \eqref{eq:upwind-system}.
\end{exercise}

\begin{exercise}[Moment identities]
Verify the moment identities for the Schr\"odingerization triple in
\eqref{eq:schrodingerization-moment-triple}.
\end{exercise}

\begin{exercise}[Duhamel--LCU]
Use Duhamel's formula to write a quadrature-based LCU approximation for a
forced wave equation.  Identify which part of the construction is the
homogeneous propagator and which part is the source-state preparation.
\end{exercise}
% Body-only LaTeX file.  It is intended to be incorporated by \include{...}
% into a larger book manuscript.  No documentclass, packages, theorem
% declarations, or macro preamble are included here.

\chapter{Quantum Algorithms for Parabolic PDEs}
\label{chap:parabolic-quantum}

\section[Classical heat discretizations]{Classical heat discretizations}
\label{sec:parabolic-classical}

Parabolic equations are the dissipative counterpart of the elliptic and
hyperbolic problems considered in Chapters~\ref{chap:elliptic-quantum} and
\ref{chap:hyperbolic-quantum}.  We use the heat equation as the model
problem:
\begin{equation}
  u_t(\bm{x},t)-\nabla\cdot\bigl(a(\bm{x})\nabla u(\bm{x},t)\bigr)
  =f(\bm{x},t),
  \qquad \bm{x}\in\Omega,
  \label{eq:heat-model-variable}
\end{equation}
with initial data
\begin{equation}
  u(\bm{x},0)=u_0(\bm{x})
  \label{eq:heat-initial}
\end{equation}
 and, unless stated otherwise, homogeneous Dirichlet boundary conditions at the boundary $\partial \Omega$.  The
coefficient $a(\bm{x})$ is the heat conductivity, or diffusion coefficient.
In this chapter we take it to be scalar and uniformly positive,
\begin{equation}
  0<a_{\min}\le a(\bm{x})\le a_{\max}<\infty,
  \label{eq:heat-uniform-ellipticity}
\end{equation}
although many of the ideas extend to symmetric positive-definite matrix-valued
coefficients.  We write the continuous spatial operator as
\begin{equation}
  L=-\nabla\cdot(a\nabla),
  \label{eq:heat-continuous-operator}
\end{equation}
using the same notation as in Chapter~\ref{chap:elliptic-quantum}.  When
$a(\bm{x})=c(\bm{x})^2$, this is the same self-adjoint positive operator that
appears in the second-order wave equation
\begin{equation}
  u_{tt}+Lu=0 .
\end{equation}
The difference between the three problem classes is which matrix function of
$L$ is applied.  Elliptic problems use
$L^{-1}$, wave problems use oscillatory functions of $L^{1/2}$, and heat
problems use the contraction semigroup
\begin{equation}
  u(t)=e^{-tL}u_0
  \label{eq:heat-semigroup-continuous}
\end{equation}
when $f=0$.

For $f=0$, multiplying \eqref{eq:heat-model-variable} by $u$ and integrating by parts
gives the energy identity
\begin{equation}
  \frac12\frac{d}{dt}\|u(t)\|_{L^2(\Omega)}^2
  +\int_\Omega \nabla u(\bm{x},t)^T
  a(\bm{x})\nabla u(\bm{x},t)\,d\bm{x}=0.
  \label{eq:heat-energy}
\end{equation}
Thus the $L^2$ norm is nonincreasing.  Moreover, since the dissipation term
weights gradients, high-frequency components are
damped more rapidly than low-frequency components.  This smoothing is the
main numerical feature of the heat equation, and it will play both sides of
the ledger in this chapter.  It is helpful for approximation
and spectral filtering: a decaying, nonoscillatory filter admits short
polynomial and Gaussian representations, which is ultimately why the coherent
costs below scale like $\sqrt{T}/h$ rather than $T/h^2$.  But it is also the
source of a quantum output issue:
preparing the normalized state proportional to $e^{-tL}u_0$ can be expensive
if the norm of the physical solution has decayed significantly.  This
tension---the same decay that makes the \emph{operator} cheap makes the
normalized \emph{state} expensive---organizes the whole chapter, and it is
resolved only at the level of the output model, by asking for observables
rather than states.

The organization is as follows.  This section reviews the classical discretizations,
following the finite-difference and
finite-element analysis of parabolic problems in standard texts such as
\cite{LarssonThomee2009,TveitoWinther2005,Thomee2006ParabolicFEM}.
We use Crank--Nicolson as a fully discrete benchmark, and we also set up the
semidiscrete systems, because these are the objects most naturally
fed to QSVT, Gaussian dilation, and Hamiltonian-simulation based primitives.
The remaining sections then examine four quantum routes to the same
semigroup: a history-state QLSA built on Crank--Nicolson
(Section~\ref{sec:heat-cn-history-qlsa}), a direct QSVT implementation of the
exponential filter (Section~\ref{sec:heat-qsvt}), Gaussian dilation through
the first-order factor together with observable-driven output
(Section~\ref{subsec:heat-first-order-dilation}), and QFT-based spectral
implementations with positive-time smoothing
(Section~\ref{sec:heat-qft-spectral}).

\subsection{Finite-difference discretization}
\label{subsec:heat-fd-semidiscrete}

On $\Omega=(0,L_{\rm box})$, we again consider the grid points defined in \eqref{1dgrids}. 
For the constant-coefficient equation $u_t=u_{xx}+f$, the standard
second-order centered finite-difference method in space gives
\begin{equation}
  \dot{\bm{u}}_h(t)=-A_h\bm{u}_h(t)+\bm{f}_h(t),
  \qquad
  A_h=-\Delta_h.
  \label{eq:heat-fd-semidiscrete}
\end{equation}
Here $A_h$ is the same positive tridiagonal matrix introduced in
\eqref{eq:one-d-laplacian}; in several dimensions it is the Kronecker-sum
matrix in \eqref{eq:d-dimensional-kronecker-sum}.  Hence
\begin{equation}
  A_h=A_h^\dag\succeq0,
  \qquad
  \|A_h\|=\Theta(h^{-2}).
  \label{eq:heat-fd-operator-scale}
\end{equation}
For the homogeneous problem $\bm f_h=0$, the semidiscrete solution is,
\begin{equation}
  \bm{u}_h(t)=e^{-tA_h}\bm{u}_h(0).
  \label{eq:heat-fd-semigroup}
\end{equation}

For variable conductivity, the conservative finite-difference discretization is
obtained by approximating the flux $a u_x$ at cell edges.  Let
\begin{equation}
  a_{j+1/2}\approx a(x_{j+1/2}).
\end{equation}
Then
\begin{equation}
  (A_h\bm{u})_j
  =\frac{1}{h^2}
  \left[
    a_{j+1/2}(u_j-u_{j+1})
    +a_{j-1/2}(u_j-u_{j-1})
  \right],
  \qquad j=1,\ldots,N_h,
  \label{eq:heat-variable-fd-operator}
\end{equation}
with boundary values inserted according to the boundary condition.  This is
exactly the finite-difference analogue of the weak form: first compute a
finite-difference gradient, multiply by the edge conductivity, and then take a
negative discrete divergence.  The resulting matrix is symmetric positive
semidefinite for nonnegative $a_{j+1/2}$, and positive definite under
homogeneous Dirichlet boundary conditions.  In this notation,
\begin{equation}
  A_h=G_h^\dag W_hG_h,
  \label{eq:heat-variable-fd-factorization}
\end{equation}
where $G_h$ is the first-difference matrix and $W_h$ is the positive diagonal
matrix of edge conductivities, matching the notation used for wave equations.
Equivalently, after replacing $G_h$ by $W_h^{1/2}G_h$, this has the simple
form $A_h=G_h^\dag G_h$.  This is the same divergence-gradient structure used
for elliptic problems in Chapter~\ref{chap:elliptic-quantum}, and it is worth
noting now: the factorization will be leveraged in
Section~\ref{subsec:heat-first-order-dilation}.

Let $\bm{I}_h u(x, t)$ denote the vector of nodal samples of the exact solution,
and use the weighted discrete norm $\|\cdot\|_{\ell_h^2}$ introduced in
\eqref{eq:wave-discrete-l2-norm}.  For sufficiently smooth solutions, stability
of $e^{-tA_h}$ and the $O(h^2)$ consistency of the centered Laplacian yield the representative
semidiscrete estimate
\begin{equation}
  \max_{0\le t\le T}
\|\bm{u}_h(t)-\bm{I}_h u(\cdot ,t)\|_{\ell_h^2}
  \le
  C h^2
  \left(
    \|u_0\|_{H^4(\Omega)}
    +\int_0^T\|u_t(x,s)\|_{H^4(\Omega)}\,ds
  \right).
  \label{eq:heat-fd-semidiscrete-error}
\end{equation}
The exact regularity norm is not important here; the important point is that
spatial discretization produces a second-order approximation to the heat
semigroup.

To integrate \eqref{eq:heat-fd-semidiscrete} in time, the operator scale
\eqref{eq:heat-fd-operator-scale} matters: the semidiscrete system can be stiff,
and an explicit method such as forward Euler is stable only under the
restrictive step condition $k=O(h^2)$.  The standard remedy is an implicit,
unconditionally stable method, and Crank--Nicolson is the second-order
representative.  Now let $t_n=nk$ with stepsize $k$ and denote the fully discrete
Crank--Nicolson iterates by
$\bm{u}_h^n$.  Crank--Nicolson applied to
\eqref{eq:heat-fd-semidiscrete} is
\begin{equation}
  \left(I+\frac{k}{2}A_h\right)\bm{u}_h^{n+1}
  =
  \left(I-\frac{k}{2}A_h\right)\bm{u}_h^n
  +k\bm{f}_h^{n+1/2}.
  \label{eq:heat-cn-fd-vector}
\end{equation}
The local truncation error of Crank--Nicolson for a smooth solution of the
semidiscrete ODE is $O(k^3)$ per step, and the centered spatial discretization
has local error $O(h^2)$. Each step of Crank--Nicolson involves a rational function of $A_h$.  The eigenvalues, which determine the homogeneous amplification factor and given by
\begin{equation}
  r(\lambda)=\frac{1-k\lambda/2}{1+k\lambda/2},
  \qquad |r(\lambda)|\le1,
  \qquad \mathrm{since}\quad  \lambda\ge0,
  \label{eq:cn-amplification}
\end{equation}
is bounded by one, the method is unconditionally stable in the discrete
$L^2$ norm.  This also means that the discrete approximation inherits the contraction in time.

Stability converts the accumulated local truncation errors into
the global second-order estimate
\begin{equation}
  \|\bm{u}_h^n-\bm{I}_h u(t_n)\|_{\ell_h^2}
  \le C_T(h^2+k^2),
  \qquad 0\le t_n\le T.
  \label{eq:heat-cn-fd-error}
\end{equation}
Thus one may choose $k=O(h)$ to balance the spatial and temporal errors,
rather than the $k=O(h^2)$ restriction required by an explicit scheme.
The price of this freedom is that each step of
\eqref{eq:heat-cn-fd-vector} requires solving a linear system with the matrix
$I+\tfrac{k}{2}A_h$.  Classically, this is the familiar trade of stiff time
integration.  It is also worth remembering for what follows: when
Section~\ref{sec:heat-cn-history-qlsa} revisits Crank--Nicolson as one global
linear system, these per-step solves are exactly what gets stacked into the
quantum linear-system formulation.

\subsection{Finite-element semidiscretization}
\label{subsec:heat-fem-semidiscrete}

Let $V_h\subset H_0^1(\Omega)$ be the same conforming finite-element space used
in Section~\ref{subsec:fem}.  The weak heat equation is
\begin{equation}
  (u_t(\cdot,t),v)+a(u(\cdot, t),v)=(f(t),v),
  \qquad \forall v\in H_0^1(\Omega),
  \label{eq:heat-weak-form}
\end{equation}
where the bilinear form is the elliptic form from
\eqref{eq:elliptic-bilinear-form}.

Writing
\begin{equation}
  u_h(\bm{x},t)=\sum_{j=1}^{N_h}(\bm{q}_h(t))_j\phi_j(\bm{x}),
  \label{eq:heat-fem-expansion}
\end{equation}
and enforcing the weak form against each basis function $v=\phi_i$ yields the
semidiscrete finite-element equations
\begin{equation}
  M_h\dot{\bm{q}}_h(t)+K_h\bm{q}_h(t)=\bm{b}_h(t).
  \label{eq:heat-fem-semidiscrete}
\end{equation}
The mass matrix $M_h$, stiffness matrix $K_h$, and load vector $\bm{b}_h(t)$
are the same finite-element objects introduced in
\eqref{eq:stiffness-load} and \eqref{eq:fem-mass-matrix}.  For the homogeneous
problem,
\begin{equation}
  \frac12\frac{d}{dt}
  \bigl(\bm{q}_h(t)^\dag M_h\bm{q}_h(t)\bigr)
  +\bm{q}_h(t)^\dag K_h\bm{q}_h(t)=0.
  \label{eq:heat-fem-energy}
\end{equation}
This is the discrete counterpart of \eqref{eq:heat-energy}.

As in the elliptic and wave chapters, we can also introduce mass-normalized
coordinates, so that the generator becomes a single Hermitian matrix acting on
amplitudes:
\begin{equation}
  \bm{y}_h(t)=M_h^{1/2}\bm{q}_h(t),
  \qquad
  A_h=M_h^{-1/2}K_hM_h^{-1/2},
  \qquad
  \bm{f}_h(t)=M_h^{-1/2}\bm{b}_h(t).
  \label{eq:heat-mass-normalization}
\end{equation}
Then
\begin{equation}
  \dot{\bm{y}}_h(t)=-A_h\bm{y}_h(t)+\bm{f}_h(t),
  \qquad A_h=A_h^\dag\succeq0.
  \label{eq:heat-mass-normalized-ode}
\end{equation}
This is the parabolic analogue of the normalized elliptic system
\eqref{eq:mass-normalized-system} and the normalized wave system in
Chapter~\ref{chap:hyperbolic-quantum}.  For quasi-uniform meshes and fixed
polynomial degree, we still have
\begin{equation}
  \|A_h\|=\Theta(h^{-2}).
  \label{eq:heat-fem-normalized-scale}
\end{equation}
Thus finite differences and finite elements arrive at the same abstract
object: a semidiscrete contraction generated by a Hermitian positive matrix of
norm $\Theta(h^{-2})$.  Everything quantum in this chapter starts from this
object.

For continuous piecewise-linear elements and sufficiently smooth solutions,
standard Galerkin theory for parabolic equations gives, for compatible
initial data,
\begin{equation}
  \|u_h(\cdot,t)-u(\cdot,t)\|_{L^2(\Omega)}
  \le
  C h^2
  \left(
    \|u_0\|_{H^2(\Omega)}
    +\int_0^t\|u_t(\cdot,s)\|_{H^2(\Omega)}\,ds
  \right).
  \label{eq:heat-fem-semidiscrete-error}
\end{equation}
More refined estimates use the analytic smoothing of the parabolic semigroup;
we return to this point in Subsection~\ref{subsec:heat-positive-time-smoothing}.

\begin{remark}[Mass normalization and mass lumping]
The consistent mass matrix does not prevent a quantum treatment: a block
encoding of $M_h$ can be combined with QSVT approximations of $M_h^{-1/2}$ to
build a block encoding of $A_h=M_h^{-1/2}K_hM_h^{-1/2}$.  This is the most
literal version of the finite-element discretization.  For an introductory
presentation, however, mass lumping is often simpler.  Replacing $M_h$ by a
positive diagonal matrix $M_{L,h}$ gives
\begin{equation}
  A_{L,h}=M_{L,h}^{-1/2}K_hM_{L,h}^{-1/2},
  \label{eq:heat-lumped-normalized}
\end{equation}
which preserves the local sparsity pattern of $K_h$ and avoids an additional
matrix-function step.  This is the same simplification used for elliptic
eigenproblems and for wave equations in earlier chapters.
\end{remark}

\subsection{Forcing and boundary data through Duhamel's principle}
\label{subsec:heat-duhamel}

Both finite differences and finite elements reduce the inhomogeneous problem
to
\begin{equation}
  \dot{\bm{y}}_h(t)=-A_h\bm{y}_h(t)+\bm{f}_h(t).
  \label{eq:heat-inhomogeneous-ode}
\end{equation}
Variation of constants gives
\begin{equation}
  \bm{y}_h(T)
  =e^{-TA_h}\bm{y}_h(0)
  +\int_0^T e^{-(T-s)A_h}\bm{f}_h(s)\,ds.
  \label{eq:heat-duhamel}
\end{equation}
After quadrature, the source contribution becomes
\begin{equation}
  \int_0^T e^{-(T-s)A_h}\bm{f}_h(s)\,ds
  \approx
  \sum_{m=1}^{N_q}w_m e^{-(T-s_m)A_h}\bm{f}_h(s_m).
  \label{eq:heat-duhamel-quadrature}
\end{equation}
This expression has the form of an LCU of homogeneous heat propagators
applied to time-dependent source vectors.  The word ``LCU'' hides an access
assumption: one must coherently prepare the quadrature index, the weights, and
normalized source states $\ket{\bm{f}_h(s_m)}=\bm{f}_h(s_m)/\|\bm{f}_h(s_m)\|$,
with the norms loaded into the LCU amplitudes, or otherwise embed the source
into an autonomous enlarged system.  The same issue appeared for body forces in the
wave chapter.  Because the source term enters every route in this uniform way,
the rest of the chapter concentrates on the homogeneous propagator
$e^{-TA_h}$; the reader should keep in mind that each construction below can
be wrapped in a Duhamel LCU of exactly this form.

\section{Crank--Nicolson as a quantum linear system}
\label{sec:heat-cn-history-qlsa}

The previous section treated Crank--Nicolson as a classical time integrator.
There is also a quantum linear-system viewpoint: stack all time levels into a
single vector and solve one large sparse linear system.  The appeal is
conceptual economy---time evolution becomes a single linear-algebra object,
and the QLSA machinery of Chapter~\ref{chap:elliptic-quantum} applies
verbatim.  This history-state or
clock-state idea goes back to quantum algorithms for linear ODEs by Berry and
collaborators and appears in refined form in later work, including Krovi's
ODE algorithm \cite{Berry2014HighOrderODE,BerryChildsOstranderWang2017,Krovi2023ImprovedDE}. Stability-dependent quantum linear-ODE algorithms also provide explicit
resource bounds for both final-time and history-state preparation
\cite{JenningsEtAl2024LinearODECost}. 
It is also the most general route in this chapter: nothing in it uses
self-adjointness or positivity, which is why it remains available for the
nonnormal problems mentioned in the outlook.  The question we pursue here is
what this generality costs for the heat equation.

For the homogeneous recurrence, define
\begin{equation}
  D_+=I+\frac{k}{2}A_h,
  \qquad
  D_-=I-\frac{k}{2}A_h.
  \label{eq:heat-cn-Dpm}
\end{equation}
Let $\bm y_h^n$ denote the mass-normalized Crank--Nicolson iterate.  Stacking
\begin{equation}
  \bm{Y}
  =
  \begin{bmatrix}
    \bm{y}_h^0\\
    \bm{y}_h^1\\
    \vdots\\
    \bm{y}_h^{N_t}
  \end{bmatrix},
  \qquad N_tk=T,
  \label{eq:heat-cn-history-vector}
\end{equation}
gives the block lower-bidiagonal system
\begin{equation}
  \mathcal L_{\rm CN}\bm{Y}
  =
  \begin{bmatrix}
    \bm{y}_h^0\\0\\\vdots\\0
  \end{bmatrix},
  \qquad
  \mathcal L_{\rm CN}
  =
  \begin{bmatrix}
    I&&&\\
    -D_-&D_+&&\\
    &-D_-&D_+&\\
    &&\ddots&\ddots
  \end{bmatrix}.
  \label{eq:heat-cn-history-system}
\end{equation}
If $A_h$ is sparse, then $\mathcal L_{\rm CN}$ is sparse on the product of the
clock register and the spatial register.  A QLSA can be implemented by
block-encoding this matrix and applying the inverse filter described in
\eqref{eq:qsvt-inverse-approximation} of Chapter~\ref{chap:basic-elements}.

The conditioning of \eqref{eq:heat-cn-history-system} reflects both time and
space.  A quick way to see the two contributions is to look at the block
forward substitution.  The matrix itself has block norm
\[
  \|\mathcal L_{\rm CN}\|
  \le
  O\bigl(\|D_+\|+\|D_-\|+1\bigr)
  =
  O(1+k\|A_h\|).
\]
On the other hand, solving the homogeneous recurrence gives
\[
  \bm y_h^n=R_k^n\bm y_h^0,
  \qquad
  R_k:=D_+^{-1}D_- .
\]
Since \(A_h\succeq0\), the eigenvalues of \(R_k\) are
\[
  r_k(\lambda)=\frac{1-k\lambda/2}{1+k\lambda/2},
  \qquad |r_k(\lambda)|\le1,
\]
so the propagator part is contractive.  The inverse of the block
lower-triangular clock matrix nevertheless contains a causal sum over
\(N_t+1\) time levels.  This gives the simple estimate
\begin{equation}
  \kappa(\mathcal L_{\rm CN})
  =
  O\bigl((N_t+1)(1+k\|A_h\|)\bigr).
  \label{eq:heat-cn-history-condition}
\end{equation}
This is the same type of evolution-norm bookkeeping that appears in quantum
algorithms for linear ODEs, including the analyses of Berry et al. and Krovi:
the clock system sees both the accumulated time interval and the scale of the
linear generator.

One may formally left-precondition each time-step equation by \(D_+^{-1}\).
Then the recurrence becomes
\[
  \bm y_h^{n+1}=R_k\bm y_h^n,
  \qquad R_k=D_+^{-1}D_-,
\]
and the clock matrix has \(O(N_t)\) conditioning because \(R_k\) is a
contraction.  But this statement has moved the spatial difficulty into the
preconditioner: implementing \(D_+^{-1}\) is itself a linear-system problem
with
\[
  \kappa(D_+)
  =
  \frac{1+k\lambda_{\max}(A_h)/2}
       {1+k\lambda_{\min}(A_h)/2}
  =
  O(1+k\|A_h\|).
\]
Thus a preconditioned clock description is analytically cleaner, but the
parabolic scale \(\|A_h\|=\Theta(h^{-2})\) has not disappeared.  The lesson is
a kind of conservation of difficulty: the clock formulation lets one shift the
parabolic stiffness between the clock matrix and its preconditioner, but no
purely algebraic rearrangement removes it.

The QLSA output is the normalized history state
\begin{equation}
 \ket{Y}=
 \frac{\sum_{n=0}^{N_t}\|\bm y_h^n\|_2
               \ket{n}_{\rm clock}\ket{\bm y_h^n}}
      {\left(\sum_{n=0}^{N_t}\|\bm y_h^n\|_2^2\right)^{1/2}},
 \label{eq:heat-cn-history-state}
\end{equation}
where zero-norm slices are omitted.

For the homogeneous heat equation, Crank--Nicolson is unconditionally stable
and the amplification factor satisfies \(|r_k(\lambda)|\le1\).  Therefore
\[
  \|\bm y_h^n\|_2\le \|\bm y_h^0\|_2,
  \qquad n=0,\ldots,N_t.
\]
The largest time slice is the initial one, not an unknown intermediate time.
Define the discrete output ratio
\begin{equation}
  g_T^{\rm CN}
  :=
  \frac{\|\bm{y}_h^0\|_2}
       {\|\bm{y}_h^{N_t}\|_2}.
  \label{eq:heat-cn-output-ratio}
\end{equation}
Then the probability of measuring the final clock value satisfies
\begin{equation}
  p_{N_t}
  =
  \frac{\|\bm{y}_h^{N_t}\|_2^2}
       {\sum_{n=0}^{N_t}\|\bm{y}_h^n\|_2^2}
  \ge
  \frac{1}{(N_t+1)(g_T^{\rm CN})^2}.
  \label{eq:heat-cn-final-clock-prob}
\end{equation}
The bound exposes two obstructions, and they are of different kinds.  The
factor \(N_t+1\) is clock dilution: the answer is spread uniformly over the
time register.  The factor \((g_T^{\rm CN})^2\) is physical decay: the
final-time slice is genuinely small.
Padding the clock with additional copies of the final state can make the
clock probability constant \cite{Krovi2023ImprovedDE}, but the amount of padding grows with the same
norm ratio; padding cures the dilution, not the decay.  In the limit
\(k\to0\), \(g_T^{\rm CN}\) is the discrete analogue
of the heat output ratio \(g_T\) in \eqref{eq:heat-output-norm-ratio}.  Even
though the heat semigroup is stable,
\begin{equation}
  \sup_{0\le t\le T}\|e^{-tA_h}\|\le1,
  \qquad \|A_h\|=\Theta(h^{-2}),
  \label{eq:heat-ode-stability-parameters}
\end{equation}
the final norm may be much smaller than the initial norm.

Combining the clock conditioning, spatial scale, and output ratio gives the
schematic history-state cost
\begin{equation}
  \widetilde O\left(
    g_T^{\rm CN}\,T\|A_h\|\,\polylog(1/\epsilon)
  \right)
  =
  \widetilde O\left(
    g_T^{\rm CN}\,\frac{T}{h^2}\,\polylog(1/\epsilon)
  \right),
  \label{eq:heat-krovi-schematic-cost}
\end{equation}
up to sparsity, state-preparation, and clock-padding factors.  This approach
is conceptually important and applies to nonnormal or inhomogeneous ODEs.  For
a self-adjoint heat semigroup, however, QSVT and Gaussian dilation exploit
positivity and smoothing more directly, and the next two sections show that
they replace the factor \(T/h^2\) by its square root.

\section{QSVT implementation of the heat semigroup}
\label{sec:heat-qsvt}

The homogeneous semidiscrete heat problem asks us to apply
\begin{equation}
  E_T:=e^{-TA_h},
  \qquad A_h=A_h^\dag\succeq0.
  \label{eq:heat-semigroup-discrete}
\end{equation}
To obtain a square-root dependence on the diffusive time scale, assume
access to a first-order factor:
\[
 A_h=G_h^\dag G_h,\qquad \|G_h\|\leq\alpha_G,
\]
with an efficient block encoding of $G_h/\alpha_G$. A rectangular factor
is handled by input and output projectors, or by padding. Set
\begin{equation}
 X_h:=G_h/\alpha_G,\qquad \alpha_A:=\alpha_G^2,\qquad\tau:=T\alpha_A.
 \label{eq:heat-qsvt-normalized-generator}
\end{equation}
Then $X_h^\dag X_h=A_h/\alpha_A$ and
$e^{-TA_h}=e^{-\tau X_h^\dag X_h}$. Choose a real even polynomial with
\begin{equation}
 |p_{\tau,\epsilon}(x)|\leq1,\qquad
 |p_{\tau,\epsilon}(x)-e^{-\tau x^2}|\leq\epsilon,\qquad x\in[-1,1].
 \label{eq:heat-qsvt-bounded-subinterval-poly}
\end{equation}
A Chebyshev approximation, followed by a small rescaling to enforce
boundedness, gives degree
\begin{equation}
 m=O\left(\sqrt{\max\{\tau,\log(1/\epsilon)\}\log(1/\epsilon)}\right).
 \label{eq:heat-exp-degree}
\end{equation}
When $\tau\gtrsim\log(1/\epsilon)$ this becomes
\begin{equation}
 m=O\left(\sqrt{\tau\log(1/\epsilon)}\right).
 \label{eq:heat-exp-degree-diffusive}
\end{equation}
If $G_h=\sum_j\sigma_j\ket{u_j}\bra{v_j}$, even QSVT acts on the right
singular space as
$\sum_jp_{\tau,\epsilon}(\sigma_j/\alpha_G)\ket{v_j}\bra{v_j}$, including
the kernel. This approximates $E_T=e^{-TA_h}$ in operator norm. Thus
\begin{equation}
 Q_{\rm QSVT}
 =O\left(\alpha_G\sqrt{T\log(1/\epsilon)}+\log(1/\epsilon)\right)
 =O\left(\sqrt{T\alpha_A\log(1/\epsilon)}+\log(1/\epsilon)\right),
 \label{eq:qsvt-heat-cost}
\end{equation}
where queries are to the encoding of $G_h/\alpha_G$ and its adjoint.
For an efficiently encoded factor with $\alpha_G=\Theta(h^{-1})$,
\begin{equation}
 Q_{\rm QSVT}=\widetilde O\left(
 \frac{\sqrt T}{h}\sqrt{\log(1/\epsilon)}+\log(1/\epsilon)\right).
 \label{eq:qsvt-heat-cost-h}
\end{equation}
This normalization of $G_h$ is an access assumption; positivity of $A_h$
alone does not supply an efficient factor encoding.

There is also an endpoint version. If one independently has a
normalization-one block encoding of
\begin{equation}
 B_h:=I-A_h/\alpha_A,\qquad \alpha_A\geq\|A_h\|,
 \label{eq:shifted-B}
\end{equation}
then
\begin{equation}
 e^{-TA_h}=e^{-\tau(I-B_h)},\qquad \tau=T\alpha_A,
 \label{eq:qsvt-shifted-exp}
\end{equation}
uses
\begin{equation}
 f(x)=e^{-\tau(1-x)},\qquad x\in[-1,1].
 \label{eq:qsvt-shifted-scalar-function}
\end{equation}
Its endpoint layer has width $1/\tau$, giving the same square-root degree
estimate; even and odd parts can be combined with constant block
normalization. But elementary LCU starting from $A_h/\alpha_A$ gives
$B_h/2$, not $B_h$, and moves the spectral endpoint into the interior
of the signal interval. Removing that normalization is not a free affine
transformation. The square-root bound must therefore be tied to factor
access or independently supplied shifted-contraction access
\cite{GilyenSuLowWiebe2019}.

For factor access, the Gaussian $e^{-\tau x^2}$ has width
$1/\sqrt\tau$ at zero, explaining the degree scale $\sqrt\tau$.
Direct approximation of $e^{-\tau x}$ near zero has width $1/\tau$
at an interior point of $[-1,1]$. These are different approximation problems.

The estimate above is a coherent operator cost.  It is not automatically the
cost of preparing the normalized final heat state.  If
\begin{equation}
  \ket{\bm{y}_0}=\frac{\bm{y}_h(0)}{\|\bm{y}_h(0)\|_2},
  \qquad
  p_T=\|E_T\ket{\bm{y}_0}\|_2^2,
  \label{eq:heat-state-prep-success}
\end{equation}
then postselecting the block-encoding ancilla succeeds with probability
\(p_T\).  Preparing the normalized state
\(E_T\ket{\bm{y}_0}/\|E_T\ket{\bm{y}_0}\|\) therefore requires an additional
factor
\begin{equation}
  g_T:=p_T^{-1/2}
  =\frac{\|\bm{y}_h(0)\|_2}{\|\bm{y}_h(T)\|_2}.
  \label{eq:heat-output-norm-ratio}
\end{equation}
Under homogeneous Dirichlet boundary conditions, the smallest eigenvalue is
bounded below, \(\lambda_{1,h}=\Theta(1)\), so long-time heat flow can make
\(g_T\) exponentially large in \(T\).  This is the central distinction between
implementing the heat operator and producing a normalized heat-solution state.
An operator approximation with error $\delta_E$ gives normalized-state
error $O(\delta_Eg_T)$ when $\delta_Eg_T$ is small. Thus a target state
error $\epsilon_{\rm state}$ requires
$\delta_E=O(\epsilon_{\rm state}/g_T)$, which also enters the filter
degree through its precision dependence. Section~\ref{subsec:heat-direct-overlap}
avoids this state-normalization factor for an observable estimated to
fixed additive accuracy.

\section{First-order Hermitian dilation and Gaussian transmutation}
\label{subsec:heat-first-order-dilation}

Recall that our discretization in space has the factorization property
\begin{equation}
  A_h=G_h^\dag G_h,
  \label{eq:heat-factorization}
\end{equation}
as in the elliptic factorization \eqref{eq:elliptic-general-factorization},
and $\|G_h\|=\Theta(h^{-1})$ for standard second-order discretizations.

As in the elliptic and wave chapters, the Hermitian dilation is the right
device:
\begin{equation}
  \mathcal D(G_h)
  =
  \begin{bmatrix}
    0&G_h^\dag\\
    G_h&0
  \end{bmatrix}.
  \label{eq:heat-G-dilation}
\end{equation}
Then
\begin{equation}
  \mathcal D(G_h)^2
  =
  \begin{bmatrix}
    A_h&0\\
    0&G_hG_h^\dag
  \end{bmatrix},
  \label{eq:heat-G-square}
\end{equation}
so $A_h$ appears as the upper-left block: squaring the first-order dilation
reproduces the second-order operator, just as squaring the wave Hamiltonian
recovered the stiffness matrix in Chapter~\ref{chap:hyperbolic-quantum}.
The dilation is Hermitian, so $e^{-is\mathcal D(G_h)}$ is a unitary that
Hamiltonian simulation can implement.  What is still missing is a way to
assemble the \emph{decaying} function $e^{-T\mathcal D(G_h)^2}$ out of these
oscillatory unitaries.  Classical Fourier analysis gives the simple recipe: the Gaussian
is, up to scaling, its own Fourier transform.  The key observation is the
Gaussian Fourier identity
\begin{equation}
  e^{-Tx^2}
  =\frac{1}{2\sqrt{\pi T}}
  \int_{-\infty}^{\infty}
  e^{-s^2/(4T)}e^{-isx}\,ds
  \label{eq:gaussian-fourier-identity}
\end{equation}
therefore gives, by the spectral theorem,
\begin{equation}
  e^{-T\mathcal D(G_h)^2}
  =\frac{1}{2\sqrt{\pi T}}
  \int_{-\infty}^{\infty}
  e^{-s^2/(4T)}e^{-is\mathcal D(G_h)}\,ds.
  \label{eq:gaussian-matrix-identity}
\end{equation}
The upper-left block of the left-hand side is $e^{-TA_h}$.  Thus heat flow is
represented as a Gaussian average of unitary wave evolutions generated by the
first-order operator $\mathcal D(G_h)$.  Formulas of this type are classical
in PDE theory, where they express the heat kernel as a Gaussian average of
wave kernels; the quantum literature has rediscovered and operationalized
them.  This construction is closely related to the
Gaussian-LCHS construction of Kharazi et al. and to the transmutation approach
of Jin, Ma, and Zuazua
\cite{KharaziAlkadriMandadapuWhaley2026ReactionRates,JinMaZuazua2026Transmutation}.
It is also an instance of the broader dilation viewpoint introduced in
Chapter~\ref{chap:hyperbolic-quantum}: after diagonalizing or quadraturing the
ancillary Gaussian variable, the heat semigroup is represented as an LCU of
unitary evolutions, much like \eqref{eq:dilation-lcu-integral}.

Truncating and discretizing the integral gives an LCU
\begin{equation}
  e^{-T\mathcal D(G_h)^2}
  \approx
  \sum_{m=1}^{N_q}c_m(T)e^{-is_m\mathcal D(G_h)},
  \label{eq:gaussian-lcu}
\end{equation}
where
\begin{equation}
  \sum_m |c_m(T)|=O(1),
  \qquad
  \max_m |s_m|=O\left(\sqrt{T\log(1/\epsilon)}\right).
  \label{eq:gaussian-quadrature-scales}
\end{equation}
The $O(1)$ total weight gives only constant LCU normalization overhead.
The query bound also assumes efficient preparation of the quadrature
weights and coherent computation of the nodes, together with SELECT
implemented by controlled-time Hamiltonian simulation at the maximum
$|s_m|$. A sequential loop over all quadrature nodes would not establish
this gate bound. These preparation and arithmetic costs, and the
physical output-norm loss for state preparation, remain part of the
end-to-end algorithm. The longest Hamiltonian evolution in the integral runs
only to time $O(\sqrt{T\log(1/\epsilon)})$, because that is where the Gaussian
window has decayed below the tolerance.  The diffusive square root, which
appeared in Section~\ref{sec:heat-qsvt} as a polynomial degree, reappears here
as the width of a Gaussian.
If $\mathcal D(G_h)/\alpha_G$ has a block encoding with
$\alpha_G=\Theta(h^{-1})$, Hamiltonian simulation and LCU yield
\begin{equation}
  Q_{\rm Gauss}
  =
  \widetilde O\left(
    \alpha_G\sqrt{T\log(1/\epsilon)}
    +\log(1/\epsilon)
  \right)
  =
  \widetilde O\left(
    \frac{\sqrt T}{h}\sqrt{\log(1/\epsilon)}
  \right).
  \label{eq:gaussian-heat-cost}
\end{equation}
Even QSVT on $G_h/\alpha_G$ and Gaussian dilation through
$\mathcal D(G_h)/\alpha_G$ have the same leading coherent scaling under
the same factor-access assumption. They implement the Gaussian function
of a first-order operator. This comparison does not establish the same
query bound from a generic block encoding of $A_h$ alone.

\paragraph{A heuristic comparison with classical Crank--Nicolson.}
It is useful to compare \eqref{eq:gaussian-heat-cost} with a classical
Crank--Nicolson computation while keeping the physical parameters visible.
On a \(d\)-dimensional domain of characteristic side length \(L_{\rm box}\),
a mesh with spacing \(h\) has
\[
  \left(\frac{L_{\rm box}}{h}\right)^d
\]
spatial degrees of freedom, up to constants depending on the element type,
boundary treatment, and geometry.

A Crank--Nicolson step for the mass-normalized heat equation has the form
\[
  \left(I+\frac{k}{2}A_h\right)\bm y_h^{n+1}
  =
  \left(I-\frac{k}{2}A_h\right)\bm y_h^n .
\]
The method is unconditionally stable, but accuracy still constrains the time
step.  In nondimensional variables, the standard second-order estimate has
the form
\[
  O(h^2+k^2).
\]
Thus a balanced choice is \(k=\Theta(h)\), giving \(O(T/h)\) time steps.  If
one keeps physical units, the same statement should be read after
nondimensionalizing time by the diffusive scale \(L_{\rm box}^2/a_{\rm ref}\).

The cost of each Crank--Nicolson step is the cost of solving
\[
  \left(I+\frac{k}{2}A_h\right)\bm z=\bm r .
\]
Since
\[
  \|A_h\|=\Theta(a_{\max}h^{-2}),
\]
the step matrix has condition number bounded schematically by
\[
  \kappa\!\left(I+\frac{k}{2}A_h\right)
  \lesssim
  1+k\|A_h\|
  =
  O\!\left(1+\frac{k\,a_{\max}}{h^2}\right).
\]
Thus a linear solver still feels the parabolic spatial scale.
With multigrid, BPX, or another near-optimal elliptic preconditioner, each
implicit solve can be performed in essentially
\[
  O\!\left(\left(\frac{L_{\rm box}}{h}\right)^d\right)
\]
work, up to logarithmic factors and coefficient-dependent constants.  With
\(k=\Theta(h)\) in nondimensional units, the full classical
Crank--Nicolson computation therefore costs
\[
  \widetilde O\!\left(
  \frac{T}{h}
  \left(\frac{L_{\rm box}}{h}\right)^d
  \right)
  =
  \widetilde O\!\left(
  T\,L_{\rm box}^d\,h^{-d-1}
  \right)
  \]
in nondimensional variables.

The quantum cost in \eqref{eq:gaussian-heat-cost} is of a different type. It
is the coherent cost of applying the heat semigroup to an amplitude-encoded
state:
\[
  Q_{\rm heat}
  =
  \widetilde O\!\left(
  \sqrt{T\|A_h\|}
  \right)
  =
  \widetilde O\!\left(
  \frac{\sqrt{T}}{h}
  \right).
\]

\section{Direct estimation of a PDE quantity of interest}
\label{subsec:heat-direct-overlap}

As alluded to in previous sections, for dissipative dynamics, producing the normalized final state is often not
the actual computational goal.  A PDE calculation usually asks for a scalar output:
a heat content, a regional average, a reaction rate, or a weighted observable.
Let
\begin{equation}
  \mathcal Q(T)
  =\bm{\ell}_h^\dag\bm{y}_h(T)
  =\bm{\ell}_h^\dag e^{-TA_h}\bm{y}_h(0),
  \label{eq:heat-qoi}
\end{equation}
be the linear quantity expressed at the discrete level.

Introduce normalized states
\begin{equation}
  \ket{\bm{\psi}}
  =\frac{\bm{y}_h(0)}{\|\bm{y}_h(0)\|_2},
  \qquad
  \ket{\bm{\phi}}
  =\frac{\bm{\ell}_h}{\|\bm{\ell}_h\|_2}.
  \label{eq:heat-input-observable-states}
\end{equation}
Then
\begin{equation}
  \mathcal Q(T)
  =\|\bm{\ell}_h\|_2\,\|\bm{y}_h(0)\|_2
  \bra{\bm{\phi}}e^{-TA_h}\ket{\bm{\psi}}.
  \label{eq:heat-qoi-overlap}
\end{equation}
The factorization performs a clean division of labor: the vector norms appear as explicit classical prefactors that must
be known or estimated to sufficient accuracy,  while the quantum computer is asked only for a dimensionless matrix
element of modulus at most one.  This is the
same lesson emphasized by direct-overlap and observable-driven heat algorithms,
and by the end-to-end analysis of Linden, Montanaro, and Shao
\cite{LiHeatGaussianDilation,KharaziAlkadriMandadapuWhaley2026ReactionRates,LindenMontanaroShao2022Heat}.

Once either QSVT or Gaussian dilation has produced a block encoding of
$e^{-TA_h}$, the matrix element in \eqref{eq:heat-qoi-overlap} can be
estimated by the Hadamard test or by amplitude estimation, as reviewed in
Chapter~\ref{chap:basic-elements}.  If the target physical additive error is
$\epsilon$, then the normalized matrix element $\bra{\bm{\phi}}e^{-TA_h}\ket{\bm{\psi}}$ must be estimated to accuracy
\begin{equation}
  \frac{\epsilon}{\|\bm{\ell}_h\|_2\,\|\bm{y}_h(0)\|_2}.
  \label{eq:heat-qoi-normalized-error}
\end{equation}
Plain sampling squares this inverse accuracy, while amplitude estimation uses
it linearly.  If $C_E$ denotes the coherent cost of one block encoding of the
heat semigroup, the amplitude-estimation cost is schematically
\begin{equation}
  \widetilde O\left(
    (C_E+C_\psi+C_\phi)\,
    \max\left\{1,\frac{\|\bm{\ell}_h\|_2\,\|\bm{y}_h(0)\|_2}{\epsilon}\right\}
  \right),
  \label{eq:heat-direct-overlap-complexity}
\end{equation}
where $C_\psi$ and $C_\phi$ are the coherent preparation costs for
$\ket{\bm\psi}$ and $\ket{\bm\phi}$, including inverses when required.
An $O(1)$ block normalization is suppressed. Direct sampling replaces
the inverse-accuracy factor by its square. If the norm prefactor is
already below $\epsilon$, returning zero meets the additive tolerance. 
Direct sampling replaces the final factor by its square.  Crucially, neither
approach needs to prepare
\begin{equation}
  \frac{e^{-TA_h}\ket{\bm{\psi}}}{\|e^{-TA_h}\ket{\bm{\psi}}\|}.
\end{equation}
Therefore direct estimation can avoid the exponentially large $g_T$ factor in
\eqref{eq:heat-output-norm-ratio} when the desired output is an additive
physical observable.  This resolves, at the level of the output model, the
tension announced at the start of the chapter: the decay of the solution is
absorbed into a small matrix element to be estimated, rather than into a small
postselection probability to be amplified. The operator approximation
must itself have error
$O(\epsilon/(\|\bm\ell_h\|_2\|\bm y_h(0)\|_2))$ to fit the physical
additive-error budget. If relative accuracy is required for a decaying
observable, the additive tolerance shrinks with that observable and
the cost can again grow with time.

The work of Linden, Montanaro, and Shao illustrates why this distinction is
not cosmetic.  For their heat-content benchmark, a QLSA route based on a
history-state formulation is not asymptotically faster than the best classical
methods, while the strongest quantum improvement comes from applying amplitude
estimation to a carefully chosen stochastic representation
\cite{LindenMontanaroShao2022Heat}.  In other words, end-to-end heat-equation
complexity is controlled not only by the coherent propagation primitive, but
also by discretization, normalization, and measurement.

\section{QFT methods and positive-time spectral smoothing}
\label{sec:heat-qft-spectral}

For constant coefficients on periodic tensor-product grids, the Fourier basis
provides an explicit diagonalization.  Chapter~\ref{chap:elliptic-quantum}
already introduced this structure for Poisson-type problems.  With
\(\mathcal F_h=F_{N_{\rm cell}}^{\otimes d}\),
\begin{equation}
  A_h=\mathcal F_h^\dag\Lambda_h\mathcal F_h,
  \qquad
  \Lambda_h=\operatorname{diag}(\lambda_{\bm{k}}),
  \label{eq:heat-qft-diagonalization}
\end{equation}
and, by \eqref{eq:periodic-discrete-symbol},
\begin{equation}
  \lambda_{\bm{k}}
  =
  \frac{4}{h^2}
  \sum_{r=1}^d
  \sin^2\left(\frac{\pi k_r}{N_{\rm cell}}\right),
  \qquad
  \bm k=(k_1,\ldots,k_d).
  \label{eq:heat-qft-symbol}
\end{equation}

Here $N_{\rm cell}$ denotes the number of periodic grid points per coordinate direction. 

Thus each Fourier coefficient decays independently:
\begin{equation}
  \widehat y_{\bm k}(T)
  =
  e^{-T\lambda_{\bm k}}\widehat y_{\bm k}(0).
  \label{eq:heat-qft-mode-decay}
\end{equation}
Equivalently,
\begin{equation}
  e^{-TA_h}
  =
  \mathcal F_h^\dag
  \operatorname{diag}\left(e^{-T\lambda_{\bm{k}}}\right)
  \mathcal F_h.
  \label{eq:heat-qft-semigroup}
\end{equation}
A QFT implementation applies \(\mathcal F_h\), computes
\(\lambda_{\bm{k}}\) coherently by reversible arithmetic, performs a controlled
rotation with amplitude \(e^{-T\lambda_{\bm{k}}}\), uncomputes the arithmetic,
and applies the inverse QFT.  Explicit Fourier-space circuits for heat,
advection, wave, and Poisson equations are developed in
\cite{LubaschKikuchiWrightMcKeever2025FourierPDE}.

The formula \eqref{eq:heat-qft-mode-decay} makes the parabolic smoothing
visible.  Low frequencies, for which \(\lambda_{\bm k}T\ll1\), are transmitted
almost unchanged.  High frequencies, for which \(\lambda_{\bm k}T\gg1\), are
suppressed.  This is the same physical mechanism behind the square-root
complexity in QSVT and Gaussian dilation.  The largest discrete frequency is
\(\lambda_{\max}(A_h)=\Theta(h^{-2})\), so the dimensionless diffusion time is
\(\tau=T\|A_h\|=\Theta(T/h^2)\).  Bounded polynomial or Gaussian
representations of the heat filter have degree or propagation time
\[
  \widetilde O(\sqrt{\tau})
  =
  \widetilde O(\sqrt T/h),
\]
rather than \(\widetilde O(\tau)=\widetilde O(T/h^2)\).  In this sense, heat
smoothing gives a quadratic improvement over a naive time-marching or
Taylor-series scale.

It is useful to separate two favorable effects.  First, the QFT diagonalizes
the operator, so the matrix action becomes scalar arithmetic on the mode
register.  If the sine and exponential in \eqref{eq:heat-qft-symbol} are
computed explicitly and reversibly, the coherent gate count for the diagonal
filter can be polylogarithmic in the number of grid points, apart from the
arithmetic precision and postselection.  Second, even when the diagonal
filter is implemented through a generic polynomial, Chebyshev, or Gaussian
construction with efficient first-order-factor or shifted-contraction
access, heat smoothing gives the $\widetilde O(\sqrt T/h)$ scale.
The first effect uses an explicit eigenbasis while the second still requires
the appropriate operator access.

The price of normalized-state preparation remains.  If
\begin{equation}
  \ket{\bm{y}_0}=\sum_{\bm{k}}\widehat y_{\bm{k}}\ket{\bm{k}}
\end{equation}
in Fourier space, then the postselection probability of the diagonal heat
filter is
\begin{equation}
  p_T^{\rm QFT}
  =
  \sum_{\bm{k}}|\widehat y_{\bm{k}}|^2 e^{-2T\lambda_{\bm{k}}}
  =
  \|e^{-TA_h}\ket{\bm{y}_0}\|_2^2.
  \label{eq:heat-qft-success}
\end{equation}
Thus QFT diagonalization simplifies the operator implementation, but it does
not remove the physical decay of the solution.

\subsection{Positive-time smoothing and spectral truncation}
\label{subsec:heat-positive-time-smoothing}

So far, parabolic smoothing has mostly appeared on the cost side of the
ledger, through the decayed output norm.  Here it becomes an asset: smoothing
reduces the spatial resolution that the problem actually requires, an effect
that is visible to PDE analysis but invisible to a purely algebraic
condition-number discussion.  At positive
time, high modes have already been damped.  From
\eqref{eq:heat-qft-mode-decay}, modes satisfying
\begin{equation}
  \lambda_{\bm{k}}
  \ge
  \lambda_{\rm cut}
  :=
  \frac{1}{T}\log\frac{1}{\epsilon}
  \label{eq:heat-qft-lambda-cut}
\end{equation}
contribute at most $\epsilon$ in operator norm. For the Fourier window,
interpret each mode index as a signed frequency (an index near $N$ represents
a small negative frequency). On a periodic grid, the small-frequency expansion of \eqref{eq:heat-qft-symbol} gives
\[
  \lambda_{\bm k}\approx 4\pi^2|\bm k|^2/L_{\rm box}^2 .
\]
Thus the relevant Fourier window at time \(T\) is approximately
\begin{equation}
  |\bm k|
  \lesssim
  \frac{L_{\rm box}}{2\pi}
  \sqrt{\frac{1}{T}\log\frac{1}{\epsilon}}.
  \label{eq:heat-effective-fourier-window}
\end{equation}
This is the spectral form of parabolic smoothing: the longer the heat equation
runs, the fewer high-frequency modes survive above a fixed tolerance.

The same phenomenon appears in finite-element error estimates.  In addition
to uniform-in-time \(O(h^2)\) estimates for smooth solutions,
for shape-regular, quasi-uniform $P_1$ meshes and an elliptic operator
with the required $H^2$ regularity, analytic-semigroup smoothing gives
the positive-time estimate
\begin{equation}
  \|E(T)-E_h(T)P_h\|_{L^2\to L^2}
  \le
  C\min\left\{1,\frac{h^2}{T}\right\},
  \qquad T>0,
  \label{eq:heat-positive-time-error}
\end{equation}
where \(E(T)=e^{-TL}\), \(E_h(T)\) is the semidiscrete heat semigroup, and
\(P_h\) is typically the \(L^2\) projection onto the finite-element space.
Such positive-time estimates are standard in parabolic finite-element theory
\cite[Ch.~3]{Thomee2006ParabolicFEM}. This positive-time estimate is useful when \(T\) is fixed away from zero; it
does not justify using a mesh coarser than the physical feature scale or the
domain geometry permits.

If the physical output is measured at a fixed
positive time, \eqref{eq:heat-positive-time-error} may allow a coarser mesh
than a uniform-in-time worst-case estimate would suggest.  Balancing
\(h^2/T\approx\epsilon\) gives
\begin{equation}
  h\approx \sqrt{T\epsilon}.
  \label{eq:heat-positive-time-h-choice}
\end{equation}
Substituting this into the coherent scale \(\sqrt T/h\) gives an
\(O(\epsilon^{-1/2})\) dependence for the semigroup block-encoding part,
whereas choosing \(h=O(\sqrt\epsilon)\) from a time-uniform \(O(h^2)\) estimate
would give the usual \(O(\sqrt T/\sqrt\epsilon)\) scaling.  Note that the
\(T\)-dependence has cancelled entirely in the first case: running the
equation longer permits a coarser mesh at exactly the rate that offsets the
longer diffusion time.  The exact
conclusion for an end-to-end algorithm still depends on the observable norm,
state preparation, and measurement, but the positive-time smoothing estimate
is a useful PDE-level improvement that is invisible in a purely algebraic
condition-number discussion.

For rectangular domains with separable boundary conditions, the same
truncation idea can be implemented with sine or cosine transforms rather than
the periodic QFT.  This is the parabolic analogue of classical FFT, DST, and
DCT solvers for heat equations on boxes.

\section{Summary and outlook}
\label{sec:heat-summary-outlook}

The heat equation reuses the elliptic spatial operators of
Chapter~\ref{chap:elliptic-quantum}, but changes the matrix function from an
inverse to a decaying exponential.  The main conclusions are:

\begin{enumerate}
\item Finite differences and finite elements lead to semidiscrete systems
\begin{equation}
  \dot{\bm{y}}_h=-A_h\bm{y}_h+\bm{f}_h(t),
  \qquad A_h\succeq0,
  \qquad \|A_h\|=\Theta(h^{-2}).
\end{equation}
Crank--Nicolson is the standard second-order fully discrete method; its
unconditional stability permits $k=O(h)$ when the spatial error is $O(h^2)$.

\item A Crank--Nicolson history-state QLSA is possible, but its cost retains
both the clock length and the parabolic spatial scale.  In the self-adjoint
heat setting it is generally less favorable than QSVT or Gaussian dilation,
which exploit the bounded exponential filter more directly.  An open question
is whether sharper history-state analyses or revised preconditioned clock
systems can reduce this gap for parabolic PDEs with sources and boundary data.

\item Given $A_h=G_h^\dag G_h$ and an efficient encoding of
$G_h/\alpha_G$ with $\alpha_G=\Theta(h^{-1})$, even QSVT implements
the heat semigroup with coherent query complexity
\begin{equation}
 \widetilde O(1+\alpha_G\sqrt T)=\widetilde O(1+\sqrt T/h),
\end{equation}
with logarithmic precision factors suppressed. This uses a bounded
Gaussian polynomial in the singular values of $G_h$; positivity of
$A_h$ alone is insufficient for the same access-model conclusion.

\item Gaussian dilation achieves the same coherent scale by representing heat
flow as a Gaussian average of unitary wave evolutions generated by the
first-order dilation $\mathcal D(G_h)$.  It is most attractive when the
first-order factor is the natural access model.

\item QFT methods give an especially transparent spectral implementation for
constant-coefficient tensor-product problems.  They can exploit explicit
symbols, sine/cosine transforms, and positive-time spectral truncation, but
normalized-state preparation still pays for heat decay.

\item Direct estimation of physical quantities of interest can avoid producing
the decayed normalized heat state.  This is often the right end-to-end goal
for parabolic PDEs.
\end{enumerate}

Several issues remain open.

\paragraph{Source terms and boundary conditions.}
Source terms and nonhomogeneous boundary data require coherent Duhamel
constructions, source-history loading, or enlarged autonomous systems.  The
main difficulty is not the variation-of-constants formula itself, but the
quantum access model for the time-dependent forcing states and their
normalization factors.

\paragraph{Nonnormal parabolic problems.}
Convection--diffusion equations, absorbing boundary conditions, and stabilized
discretizations can lead to nonnormal generators.  These problems require
tools beyond self-adjoint spectral filters, such as dilations, singular-value
transformations, pseudospectral estimates, or problem-specific
preconditioning. For advection--diffusion, lattice-Boltzmann formulations also offer structured
collision--streaming operators. Tiwari et al. develop an LCU-based quantum
implementation and demonstrate two-dimensional transport on trapped-ion
hardware; Ray et al. extend this line of work to three-dimensional transport
with spatially varying advection
\cite{TiwariEtAl2025RealizableQLBM,RayEtAl2026QLBM3D}.
He, Yu, and Yu analyze both sequential quantum time marching and a global
quantum linear-system formulation for lattice-Boltzmann advection--diffusion
\cite{HeYuYu2026LBMTimeMarching}. For finite-difference transport discretizations,
Over et al. construct a direct block encoding of the explicit time-marching
operator \cite{OverEtAl2025DirectTimeMarching}.

\paragraph{History-state QLSA versus semigroup filters.}
The QLSA/history-state approach is broadly applicable because it starts from a
standard time discretization such as Crank--Nicolson.  However, in its basic
form it does not match the coherent complexity of direct QSVT or Gaussian
dilation for heat semigroups.  It remains open whether sharper conditioning
estimates, better preconditioning, or different clock formulations can close
this gap.

\paragraph{Positive-time smoothing and observables.}
The positive-time smoothing estimates used above suggest that, for many
quantities of interest, the mesh need not resolve modes that have already
been damped below the target tolerance.  Extending this observable-driven
viewpoint to general finite elements, systems of parabolic equations, and
stochastic parabolic models is an important direction for future work.

\section{Exercises}

\begin{exercise}[Semidiscrete heat operator]
Starting from the matrix $A_h$ in \eqref{eq:one-d-laplacian}, derive the
semidiscrete heat equation \eqref{eq:heat-fd-semidiscrete} and verify
\(\|A_h\|=\Theta(h^{-2})\).
\end{exercise}

\begin{exercise}[Crank--Nicolson stability]
Derive the Crank--Nicolson amplification factor \eqref{eq:cn-amplification}
and show that its modulus is at most one.  Explain why this unconditional
stability leads to a global $O(k^2)$ temporal error for smooth solutions.
\end{exercise}

\begin{exercise}[Finite-element energy]
Starting from $M_h\dot{\bm q}_h+K_h\bm q_h=0$, derive
\eqref{eq:heat-mass-normalized-ode} and the energy identity
\eqref{eq:heat-fem-energy}.
\end{exercise}

\begin{exercise}[Bounded polynomial for the heat semigroup]
Let $A_h=G_h^\dag G_h$ and suppose $G_h/\alpha_G$ is block encoded.
Explain why even QSVT with a bounded approximation to
$e^{-T\alpha_G^2x^2}$ acts as $e^{-TA_h}$ on the right singular space.
Derive the degree in \eqref{eq:qsvt-heat-cost}. Why does the same
argument not apply directly to a generic encoding of $A_h/\alpha_A$?
What normalization results from elementary LCU for $I-A_h/\alpha_A$?
\end{exercise}

\begin{exercise}[Gaussian identity]
Prove the Gaussian identity \eqref{eq:gaussian-fourier-identity} and use the
spectral theorem to derive \eqref{eq:gaussian-matrix-identity}.
\end{exercise}

\begin{exercise}[Clock probability]
For the Crank--Nicolson history state \eqref{eq:heat-cn-history-state}, prove
\eqref{eq:heat-cn-final-clock-prob}.  How much padding is needed to make the
probability of observing the final-time block constant?
\end{exercise}

\begin{exercise}[Positive-time mesh selection]
Use the smoothing estimate \eqref{eq:heat-positive-time-error} to choose $h$
so that the spatial error at a fixed time $T>0$ is at most $\epsilon$.  Compare
this choice with the time-uniform estimate $O(h^2)$.
\end{exercise}

\begin{exercise}[Direct quantity of interest]
Compare the cost of preparing the normalized final heat state with the cost of
directly estimating the quantity of interest in \eqref{eq:heat-qoi-overlap}.
\end{exercise}

% Body-only LaTeX file.  It is intended to be incorporated by \include{...}
% into a larger book manuscript.  No documentclass, packages, theorem
% declarations, or macro preamble are included here.

\chapter{A First Look at Nonlinear Problems}
\label{chap:nonlinear-quantum}

\section{Why nonlinear PDEs require a different viewpoint}
\label{sec:nonlinear-warning}

The algorithms in previous chapters relied heavily on the linear structure of the PDEs:  elliptic equations led to
linear systems, hyperbolic equations led to Hamiltonian dynamics, and
parabolic equations led to contraction semigroups.  Nonlinear partial
differential equations do not fit so directly into this picture.  The
superposition principle fails, stability may depend on the solution itself,
shocks or caustics may form, and the appropriate notion of solution may be
weak, entropy, viscosity, measure-valued, or probabilistic.  Thus one should
not expect a single quantum primitive to handle nonlinear PDEs in the same way
that Hamiltonian simulation handles a finite-dimensional Schr\"odinger equation.

There are special exceptions.  If a nonlinear evolution can be mapped by an
explicit transform to a stable linear problem, then the linear quantum methods
of the previous chapters may become directly relevant.  The classical
Cole--Hopf transformation for viscous Burgers and Hamilton--Jacobi equations is
the standard example: after an exponential change of variables, a nonlinear
viscous equation becomes a heat equation.  More recent work of Jin and Liu uses
entropy penalization to generalize this idea to viscosity solutions of convex
Hamilton--Jacobi equations, reducing the extraction of point values, gradients,
and minima to heat-like linear dynamics \cite{JinLiu2026HamiltonJacobi}.  These
examples are important because they preserve the physically relevant nonlinear
solution concept while giving a linear computational representation.

Such transformations, however, are problem dependent.  For a general nonlinear
PDE, after appropriate spatial discretization one obtains a nonlinear ODE
\begin{equation}
  \frac{d}{dt} {\bm u}(t)=\Phi(\bm u(t),t),
  \qquad \bm u(t)\in\mathbb C^N,
  \label{eq:nonlinear-method-of-lines-general}
\end{equation}
and there is no universal linear quantum evolution whose amplitudes equal
\(\bm u(t)\) for all possible nonlinear vector fields \(\Phi\).  This chapter
therefore presents two elementary linearization viewpoints.  They are included
mainly as templates, not as general-purpose algorithms.

The first is \emph{Carleman linearization}.  It maps polynomial nonlinear
dynamics to an infinite-dimensional linear flow on tensor powers, or monomials, of the
state.  After truncation, the problem becomes a large linear ODE and can be
treated by QLSA, quantum linear-ODE solvers, or dilation methods presented in previous chapters.  This
viewpoint is best interpreted as a route to \emph{trajectory simulation}: given
one initial condition, prepare or estimate the solution at a later time.

The second viewpoint is the \emph{Liouville}, \emph{Koopman}, or
\emph{Koopman--von Neumann} formulation.  Instead of following one trajectory,
one follows a density or an observable on phase space.  The nonlinear dynamics
becomes a linear, but also infinite dimensional, transport equation for the probability density, or equivalently
a Schr\"odinger-type equation for a wavefunction whose squared modulus is the
density.  This viewpoint is best interpreted as a route to \emph{probability
and ensemble simulation}: compute expectations, correlations, or
uncertainty-quantification observables.

Figure~\ref{fig:nonlinear-two-lifts} previews the two routes.  They start from
different objects---a single trajectory versus an ensemble---and they arrive at
different linear problems.  Much of this chapter consists of making the middle
arrows precise, and of asking, in the spirit of the earlier chapters, what each
route costs end to end: state preparation, linear evolution, and the extraction
of a physical answer.

For classical nonlinear hyperbolic problems, finite-volume methods, Riemann
solvers, limiters, and entropy conditions are central; a standard reference is
LeVeque's book \cite{LeVeque2002FVM}.  On the quantum side, Carleman
linearization has been analyzed for dissipative nonlinear differential equations
\cite{LiuKoldenKroviLoureiroTrivisaChilds2021}, including reaction--diffusion
models \cite{LiuAnFangWangLowJordan2023ReactionDiffusion}.  It has been
improved for more general linearized dynamics using evolution-norm bounds
\cite{Krovi2023ImprovedDE}, extended to regimes beyond the traditional
dissipative condition \cite{WuWangLi2025NonlinearDynamics}, and developed
further in recent embedding and nonlinear-solver frameworks
\cite{JenningsKorzekwaLostaglioSornborgerSubasiWang2025Carleman,CostaSchleichMoralesBerry2025Nonlinear}.
The Koopman--von Neumann formulation for quantum simulation of nonlinear
classical dynamics was developed in \cite{Joseph2020KvN}; related Liouville,
level-set, and observable-computation methods for nonlinear PDEs are discussed
in \cite{JinLiuYu2023LinearRepresentations,JinLiu2024NonlinearPDEObservables}. Nonlinear stochastic differential equations, which in the context of PDEs, can emerge from a semi-discrete approximation of stochastic PDE, has been recently considered in \cite{bravyi2025quantumsimulationnoisyclassical, bravyi2026quantumalgorithmsstochasticnonlinear, li2026quantumderivativepricingspdes}. Lattice-Boltzmann formulations provide another structured route to quantum
fluid simulation. Li et al. investigate a Carleman-linearized formulation
and its potential quantum advantage \cite{LiEtAl2025FluidDynamics}.
Sanavio et al. construct matrix-access block encodings for
Carleman-linearized lattice-Boltzmann evolution
\cite{SanavioEtAl2025CLBOracles}. Jennings et al. combine an incompressible
lattice-Boltzmann formulation with Carleman lifting and a quantum linear-system
solver, analyzing the full cost of extracting fluid observables
\cite{JenningsEtAl2025EndToEndCFD}. Their subsequent work incorporates walls,
inlets, outlets, and external forcing
\cite{JenningsEtAl2026IncompressibleFlows}. These analyses show why convergence,
conditioning, and readout must be included when assessing a possible quantum
advantage for fluid dynamics.
Related quantum lattice-gas work by Kocherla et al. treats diffusion and
Burgers dynamics using phase estimation and phase correction
\cite{KocherlaEtAl2024MesoscaleFluids}.

The purpose here is therefore not to claim a broad quantum advantage for
nonlinear PDEs.  Rather, it is to show the reader two linearization mechanisms
and to clarify what each mechanism computes.

\begin{figure}[t]
\centering
\begin{tikzpicture}[x=1cm,y=1cm,>=stealth, every node/.style={font=\small}]
  \node[draw,rounded corners,align=center,minimum width=3.0cm,minimum height=0.8cm] (traj) at (0,1.2)
    {trajectory\\$\bm u(t)$};
  \node[draw,rounded corners,align=center,minimum width=3.6cm,minimum height=0.8cm] (car) at (4.5,1.2)
    {tensor lift\\$1,\bm u,\bm u^{\otimes2},\ldots$};
  \node[draw,rounded corners,align=center,minimum width=3.2cm,minimum height=0.8cm] (lin1) at (9.2,1.2)
    {linear ODE\\Carleman};
  \draw[->] (traj) -- (car);
  \draw[->] (car) -- (lin1);

  \node[draw,rounded corners,align=center,minimum width=3.0cm,minimum height=0.8cm] (ens) at (0,-0.5)
    {ensemble\\$\rho(t,\bm x)$};
  \node[draw,rounded corners,align=center,minimum width=3.6cm,minimum height=0.8cm] (sqrt) at (4.5,-0.5)
    {square-root state\\$\psi=\sqrt{\rho}$};
  \node[draw,rounded corners,align=center,minimum width=3.2cm,minimum height=0.8cm] (lin2) at (9.2,-0.5)
    {linear PDE\\KvN/Liouville};
  \draw[->] (ens) -- (sqrt);
  \draw[->] (sqrt) -- (lin2);
\end{tikzpicture}
\caption{Two linear ways to look at nonlinear dynamics.  Carleman lifting keeps a single trajectory but enlarges the state by tensor powers.  Liouville/KvN lifting moves to probability or observable evolution on phase space.}
\label{fig:nonlinear-two-lifts}
\end{figure}
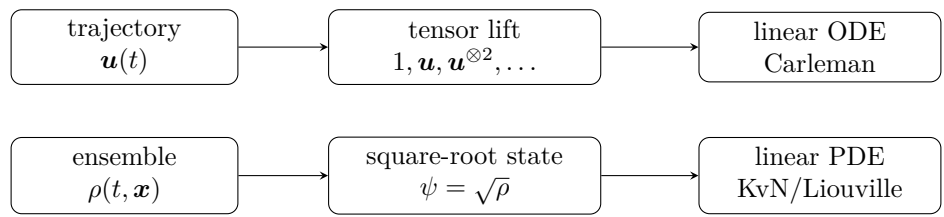

\section{From nonlinear PDEs to nonlinear ODEs}
\label{sec:nonlinear-method-of-lines}

A reasonable starting point is semidiscrete approximation, where the continuous PDE is first discretized in space, producing a
finite-dimensional nonlinear ODE of the form already displayed in
\eqref{eq:nonlinear-method-of-lines-general}.
For example, a scalar conservation law
\begin{equation}
  u_t+f(u)_x=0
  \label{eq:scalar-conservation-law}
\end{equation}
may be approximated by a finite-volume method
\begin{equation}
  \frac{d}{dt}u_j(t)
  =-\frac{1}{h}\left(
  \widehat f_{j+1/2}(\bm u(t))-\widehat f_{j-1/2}(\bm u(t))
  \right),
  \label{eq:finite-volume-nonlinear-ode}
\end{equation}
where \(\widehat f\) is a numerical flux.  For smooth solutions and simple
fluxes, this is a finite-dimensional nonlinear ODE.  For discontinuous
solutions, however, the numerical flux is part of the mathematical problem: it
selects the entropy solution and controls spurious oscillations \cite{LeVeque2002FVM}.

As a second, closely related example, the viscous Burgers equation
\begin{equation}
  u_t+u u_x=\nu u_{xx}
  \label{eq:burgers-pde}
\end{equation}
may lead to a semidiscrete system of the form
\begin{equation}
  \dot{\bm u}=-\nu A_h\bm u-\mathsf N_h(\bm u,\bm u),
  \label{eq:burgers-semidiscrete-quadratic}
\end{equation}
where \(A_h\) is the positive matrix discretizing \(-\partial_{xx}\), as in
Chapter~\ref{chap:elliptic-quantum}, and \(\mathsf N_h\) is a bilinear
convection operator.

The right-hand side of \eqref{eq:burgers-semidiscrete-quadratic} is the
structure worth remembering from this section: a linear part that is
dissipative and a quadratic part of moderate size.  The next section takes
exactly this form as its starting point, and the competition between the two
parts---linear decay against quadratic growth---will resurface as the small
parameter that governs the convergence of Carleman truncation.

In this chapter we assume that the spatial discretization has already been
chosen.  The additional classical questions---shock capturing, invariant
preservation, adaptivity, stiffness, and long-time stability---are essential in
real computations but beyond this introductory treatment.

\section{Carleman linearization for polynomial nonlinearities}
\label{sec:carleman-nonlinear}

Consider a nonautonomous forced quadratic system
\begin{equation}
  \dot{\bm u}
  =\bm f_0(t)+F_1(t)\bm u+F_2(t)(\bm u\otimes \bm u),
  \qquad \bm u(t)\in\mathbb C^N .
  \label{eq:quadratic-ode}
\end{equation}
Here \(\bm f_0(t)\in\mathbb C^N\),
\(F_1(t)\in\mathbb C^{N\times N}\), and
\(F_2(t)\in\mathbb C^{N\times N^2}\).  The semidiscrete Burgers equation
\eqref{eq:burgers-semidiscrete-quadratic} is a typical example after spatial
approximation: the linear part is dissipative and the convection term is
quadratic.  Higher-degree polynomial nonlinearities can be treated by adding
more off-diagonal blocks, but the quadratic case already contains the main
idea.

Why should tensor powers appear at all?  The basic difficulty with a quadratic
vector field is that monomials in the state do not close under differentiation:
the derivative of \(\bm u\) involves \(\bm u\otimes\bm u\), the derivative of
\(\bm u\otimes\bm u\) involves \(\bm u\otimes\bm u\otimes\bm u\), and so on.
Numerical analysts will recognize the pattern from moment hierarchies in
kinetic theory, where the evolution of each moment involves the next one.
Carleman's idea is not to fight the hierarchy but to embrace it: track every
tensor power simultaneously, and the dynamics becomes exactly linear, at the
price of infinitely many variables.  Truncation then plays the role of a
closure, and the analysis must quantify what the closure discards.

The Carleman construction thus replaces nonlinear functions of \(\bm u\) by new
linear coordinates.  Define
\begin{equation}
  \bm z_k(t)=\bm u(t)^{\otimes k},
  \qquad k=1,2,3,\ldots,
  \label{eq:carleman-tensor-powers}
\end{equation}
and set \(\bm z_0(t)=1\).  Differentiating a tensor power gives
\begin{equation}
  \frac{d}{dt}\bm z_k(t)
  =\sum_{j=1}^k
  \bm u^{\otimes(j-1)}\otimes \dot{\bm u}\otimes
  \bm u^{\otimes(k-j)} .
  \label{eq:tensor-power-derivative}
\end{equation}
Substituting \eqref{eq:quadratic-ode} yields, for \(k\ge1\),
\begin{equation}
  \dot{\bm z}_k
  =A_{k,k-1}(t)\bm z_{k-1}
  +A_{k,k}(t)\bm z_k
  +A_{k,k+1}(t)\bm z_{k+1}.
  \label{eq:carleman-block-row}
\end{equation}
The forcing term lowers tensor degree, the linear term preserves tensor degree,
and the quadratic term raises tensor degree by one.  More explicitly,
\begin{align}
  A_{k,k}(t)
  &=\sum_{j=1}^k
  I^{\otimes(j-1)}\otimes F_1(t)\otimes I^{\otimes(k-j)},
  \label{eq:carleman-Akk}\\
  A_{k,k+1}(t)
  &=\sum_{j=1}^k
  I^{\otimes(j-1)}\otimes F_2(t)\otimes I^{\otimes(k-j)}.
  \label{eq:carleman-Ak-kplus1}
\end{align}
The block \(A_{k,k-1}(t)\) is the analogous sum obtained by inserting
\(\bm f_0(t)\) into one tensor slot.  For example,
\(A_{1,0}\bm z_0=\bm f_0\).

Thus the nonlinear ODE has been embedded exactly into the infinite linear ODE
\begin{equation}
  \frac{d}{dt}
  \begin{bmatrix}
  \bm z_0\\ \bm z_1\\ \bm z_2\\ \vdots\\ \bm z_K\\ \bm z_{K+1}\\ \vdots
  \end{bmatrix}
  =
  \begin{bmatrix}
  0&&&&&&\\
  A_{1,0}&A_{1,1}&A_{1,2}&&&&\\
  &A_{2,1}&A_{2,2}&A_{2,3}&&&\\
  &&A_{3,2}&A_{3,3}&A_{3,4}&&\\
  &&&\ddots&\ddots&\ddots&\\
  &&&&A_{K,K-1}&A_{K,K}&A_{K,K+1}\\
  &&&&&A_{K+1,K}&A_{K+1,K+1}&\ddots\\
  &&&&&&\ddots&\ddots
  \end{bmatrix}
  \begin{bmatrix}
  \bm z_0\\ \bm z_1\\ \bm z_2\\ \vdots\\ \bm z_K\\ \bm z_{K+1}\\ \vdots
  \end{bmatrix} .
  \label{eq:infinite-carleman-system}
\end{equation}
The first block \(\bm z_0\) is fixed and equal to one.  The physical solution is
\(\bm z_1(t)=\bm u(t)\).  The price of exact linearization is that it lives in
an infinite tensor-power space.

It is worth pausing over the direction of the couplings in
\eqref{eq:infinite-carleman-system}.  The matrix is block tridiagonal, and the
two off-diagonals play different roles: the forcing blocks \(A_{k,k-1}\) push
information \emph{up} the ladder from lower levels, while the quadratic blocks
\(A_{k,k+1}\) pull information \emph{down} from higher levels.  In particular,
the physical level \(\bm z_1\) feels the infinitely many levels above it only
through the superdiagonal blocks generated by \(F_2\), and in the homogeneous
case \(\bm f_0=0\) the flow of information is strictly one-directional, from
high levels to low.  This directional structure is what makes the truncation
analysis tractable: an error committed at the top of the retained ladder must
climb down one rung at a time before it can contaminate \(\bm z_1\).

A level-\(K\) Carleman truncation keeps \(\bm z_0,\ldots,\bm z_K\) and drops the
first omitted coupling
\begin{equation}
  A_{K,K+1}(t)\bm z_{K+1}(t)
  \label{eq:carleman-dropped-tail}
\end{equation}
from the \(K\)th retained equation.  Let
\begin{equation}
  \bm y_K^{\rm ex}(t)
  =
  \begin{bmatrix}
    \bm z_0(t)\\ \bm z_1(t)\\ \vdots\\ \bm z_K(t)
  \end{bmatrix},
  \qquad
  J_K\bm v=
  \begin{bmatrix}
    \bm 0\\ \vdots\\ \bm 0\\ \bm v
  \end{bmatrix},
\end{equation}
where \(J_K\bm v\) places \(\bm v\) in the last retained level.  The exact
retained blocks satisfy
\begin{equation}
  \frac{d}{dt}\bm y_K^{\rm ex}(t)
  =C_K(t)\bm y_K^{\rm ex}(t)
  +J_K A_{K,K+1}(t)\bm z_{K+1}(t),
  \label{eq:projected-exact-carleman-system}
\end{equation}
where \(C_K(t)\) is the finite block matrix formed from the retained levels.
The computed truncated Carleman system sets the tail forcing to zero:
\begin{equation}
  \frac{d}{dt}\widehat{\bm y}_K(t)=C_K(t)\widehat{\bm y}_K(t),
  \qquad
  \widehat{\bm y}_K(0)=\bm y_K^{\rm ex}(0).
  \label{eq:truncated-carleman-system}
\end{equation}
The first physical component of \(\widehat{\bm y}_K(t)\), denoted
\(\widehat{\bm z}_{1,K}(t)\), is the Carleman approximation to \(\bm u(t)\).

\subsection{Truncation and convergence rate for dissipative systems}
\label{subsec:carleman-dissipative-convergence}

We now give a short convergence argument in the homogeneous quadratic case
\(\bm f_0=0\).  This is the cleanest setting and follows the proof strategy of
Theorem~2.2 in \cite{WuWangLi2025NonlinearDynamics}.  The case
\(\bm f_0\ne0\) is discussed afterward.

The argument has three steps.  First, we show that the truncation error obeys
the same one-directional ladder as the solution itself, with a source only at
the top retained level.  Second, we unroll the ladder with Duhamel's formula,
which yields a nested integral representation of the physical error.  Third, we
estimate the nested integral, once under a plain finite-time hypothesis and
once under dissipation; comparing the two resulting bounds is instructive.

When \(\bm f_0=0\), the level \(0\) block decouples and, for
\(j\ge1\), the Carleman chain has only diagonal and upper-neighbor couplings:
\begin{equation}
  \dot{\bm z}_j=A_{j,j}\bm z_j+A_{j,j+1}\bm z_{j+1}.
  \label{eq:homogeneous-carleman-row}
\end{equation}
Let
\begin{equation}
  \bm\eta_j(t)=\bm z_j(t)-\widehat{\bm z}_{j,K}(t),
  \qquad j=1,\ldots,K,
\end{equation}
where \(\widehat{\bm z}_{j,K}\) is the \(j\)th block of the level-\(K\)
truncation.  Then \(\bm\eta_j(0)=0\), and
\begin{align}
  \dot{\bm\eta}_j
  &=A_{j,j}\bm\eta_j+A_{j,j+1}\bm\eta_{j+1},
  \qquad j=1,\ldots,K-1,
  \label{eq:carleman-error-chain-inner}\\
  \dot{\bm\eta}_K
  &=A_{K,K}\bm\eta_K+A_{K,K+1}\bm z_{K+1}.
  \label{eq:carleman-error-chain-last}
\end{align}
Thus the physical error \(\bm\eta_1\) is driven only by the first omitted tensor
level, but this forcing must propagate down through the whole Carleman chain.

Let
\begin{equation}
  E_j(t,s)=\exp\bigl((t-s)A_{j,j}\bigr),
  \qquad t\ge s .
\end{equation}
For notational simplicity, we first present the estimate in the autonomous
case. In the time-dependent case, \(E_j(t,s)\) should be replaced by the
evolution operator generated by \(A_{j,j}(t)\), and \(\|F_2\|\) should be read
as \(\sup_{0\le t\le T}\|F_2(t)\|\).

Repeated use of Duhamel's formula gives the nested representation
\begin{align}
  \bm\eta_1(T)
  &=
  \int_{0\le t_K\le\cdots\le t_1\le T}
  E_1(T,t_1)A_{1,2}E_2(t_1,t_2)A_{2,3}\cdots
  \nonumber\\
  &\hspace{4em}
  \cdots E_K(t_{K-1},t_K)A_{K,K+1}\bm z_{K+1}(t_K)
  \,dt_K\cdots dt_1 .
  \label{eq:carleman-nested-duhamel}
\end{align}
This formula is often the most useful way to understand Carleman convergence:
the truncation error is a \(K\)-fold Duhamel term.  Read from right to left, it
tells a simple story.  The omitted tensor \(\bm z_{K+1}\) is injected at time
\(t_K\) through the block \(A_{K,K+1}\); each propagator \(E_j\) then transports
the error within level \(j\), and each block \(A_{j,j+1}\) drops it down one
rung; the ordering constraint \(t_K\le\cdots\le t_1\) records the fact that
these events must happen in sequence before the error reaches the physical
level at time \(T\).

To convert this representation into a bound, exactly two inputs are needed: a
bound on the trajectory, which controls the injected tail, and a bound on the
diagonal propagators, which controls the transport.  Assume first that the
trajectory is uniformly bounded on \([0,T]\),
\begin{equation}
  \|\bm u(t)\|\le M_u,
  \qquad 0\le t\le T,
  \label{eq:carleman-solution-bound-Mu}
\end{equation}
and that the diagonal propagators satisfy
\begin{equation}
  \|E_j(t,s)\|\le M_{\mathrm{lin}},
  \qquad 1\le j\le K,\quad 0\le s\le t\le T.
  \label{eq:carleman-diagonal-propagator-bound}
\end{equation}
If the one-copy propagator is contractive, one may take
$M_{\mathrm{lin}}=1$. More generally,
$E_j(t,s)=E_1(t,s)^{\otimes j}$, so
$\|E_j(t,s)\|=\|E_1(t,s)\|^j$. A bound uniform in the truncation level
is an additional assumption; a one-copy diagonalization condition number
cannot simply be reused at every tensor level.  The block formula \eqref{eq:carleman-Ak-kplus1}
shows, by the triangle inequality, that
\begin{equation}
  \|A_{j,j+1}\|\le j\|F_2\|.
  \label{eq:carleman-superdiagonal-bound}
\end{equation}
Since \(\|\bm z_{K+1}(t)\|=\|\bm u(t)\|^{K+1}\le M_u^{K+1}\),
\eqref{eq:carleman-nested-duhamel} yields
\begin{align}
  \|\bm\eta_1(T)\|
  &\le
  M_{\mathrm{lin}}^K
  \left(\prod_{j=1}^{K}j\|F_2\|\right)
  M_u^{K+1}
  \int_{0\le t_K\le\cdots\le t_1\le T}dt_K\cdots dt_1
  \nonumber\\
  &=
  M_{\mathrm{lin}}^K K!\|F_2\|^K M_u^{K+1}\frac{T^K}{K!}
  \nonumber\\
  &=
  M_u\left(M_{\mathrm{lin}}M_uT\|F_2\|\right)^K .
  \label{eq:carleman-theorem22-style-bound}
\end{align}
The cancellation of the two factors \(K!\) is the key elementary point: the
simplex volume \(T^K/K!\) cancels the product
\(1\cdot2\cdots K\) coming from the \(K\) off-diagonal Carleman blocks.  This
finite-time bound is useful when \(T\) is fixed or when the nonlinear term is
sufficiently weak.  Indeed, \eqref{eq:carleman-theorem22-style-bound} is
geometric in \(K\) with ratio \(M_{\mathrm{lin}}M_uT\|F_2\|\), so it certifies
convergence precisely when the nonlinear interaction, accumulated over the
whole time window, is small; for a fixed nonlinearity and growing \(T\) it
eventually says nothing.  Dissipation repairs this defect, essentially by
replacing the length \(T\) of the window with the time scale \(1/\mu\) over
which the linear part forgets.

For long times, dissipation therefore gives a sharper and more informative
estimate.  Assume
\begin{equation}
  \frac{F_1+F_1^\dag}{2}\preceq -\mu I,
  \qquad \mu>0.
  \label{eq:carleman-linear-dissipation}
\end{equation}
The diagonal block \(A_{j,j}\) in \eqref{eq:carleman-Akk} is a Kronecker sum of
\(j\) copies of \(F_1\), so each of the \(j\) tensor factors decays at rate
\(\mu\) and the decay rates add across factors.  The \(j\)th diagonal Carleman
block thus dissipates \(j\) tensor factors, and one expects the level-dependent
bound
\begin{equation}
  \|E_j(t,s)\|\le M_{\mathrm{dis}}e^{-j\mu(t-s)}.
  \label{eq:carleman-level-dissipation}
\end{equation}
Under the stated dissipativity inequality one can take $M_{\mathrm{dis}}=1$.
Substituting \eqref{eq:carleman-level-dissipation} into the nested formula and one has,
\begin{align}
  \|\bm\eta_1(T)\|
  &\le
  M_{\mathrm{dis}}^{K}
  \left(\prod_{j=1}^{K}j\|F_2\|\right)
  M_u^{K+1}
  \prod_{j=1}^{K}\frac{1}{j\mu}
  \nonumber\\
  &=
  M_u\left(M_{\mathrm{dis}}\frac{M_u\|F_2\|}{\mu}\right)^K .
  \label{eq:carleman-dissipative-truncation-rate}
\end{align}
This is the intuitive dissipative Carleman estimate.  Comparing with
\eqref{eq:carleman-theorem22-style-bound}, the role of the simplex volume
\(T^K/K!\) is now played by the product \(\prod_{j=1}^K 1/(j\mu)\):
dissipation substitutes for the shortness of the time window, and the resulting
bound is uniform in \(T\).  The natural small
parameter is
\begin{equation}
  R:=M_{\mathrm{dis}}\frac{M_u\|F_2\|}{\mu}.
  \label{eq:carleman-dissipative-ratio}
\end{equation}
The ratio \(R\) should be read as a nondimensional group, in the spirit of a
Reynolds or Damk\"ohler number: the numerator \(M_u\|F_2\|\) is the rate at
which the quadratic term produces amplitude at the scale of the solution, and
\(\mu\) is the rate at which the linear part removes it.  Weak nonlinearity in
the Carleman sense means that removal wins.  When \(R<1\), the error in the
physical level decays geometrically:
\begin{equation}
  \|\bm u(T)-\widehat{\bm z}_{1,K}(T)\|
  =\|\bm\eta_1(T)\|
  \lesssim M_u R^K .
  \label{eq:carleman-dissipative-physical-error}
\end{equation}
Thus a truncation error at most \(\epsilon\) is achieved with
\begin{equation}
  K=
  O\!\left(
  \frac{\log(M_u/\epsilon)}{\log(1/R)}
  \right).
  \label{eq:carleman-K-log-epsilon}
\end{equation}
This is the basic reason Carleman linearization can have logarithmic
truncation depth in the accuracy for weakly nonlinear dissipative systems.

It is sometimes useful to express the same idea directly at the retained-system
level, because this is the form the quantum solver actually sees: the QLSA and
dilation routes of Section~\ref{subsec:quantum-carleman-access} act on the
retained matrix as a whole, and their costs are stated in terms of its
propagator rather than the individual diagonal blocks.  Let
\begin{equation}
  \bm e_K(t)=\bm y_K^{\rm ex}(t)-\widehat{\bm y}_K(t)
\end{equation}
and let \(U_K(t,s)\) be the propagator generated by the retained matrix
\(C_K(t)\).  From \eqref{eq:projected-exact-carleman-system} and
\eqref{eq:truncated-carleman-system},
\begin{equation}
  \bm e_K(t)=
  \int_0^t
  U_K(t,s)J_KA_{K,K+1}(s)\bm z_{K+1}(s)\,ds .
  \label{eq:carleman-tail-duhamel}
\end{equation}
A rescaling $\bm u=s\bm v$ replaces $F_2$ by $sF_2$.
Thus $M_u$ is divided by $s$ while $\|F_2\|$ is multiplied by $s$,
leaving \eqref{eq:carleman-dissipative-ratio} unchanged.
The coefficients and stability constants below refer to the rescaled
system. If the solution has been rescaled so that
\begin{equation}
  \|\bm u(t)\|\le r<1,
  \qquad 0\le t\le T,
  \label{eq:carleman-rescaled-small-solution}
\end{equation}
and if the retained propagator is stable on the invariant subspace
$\mathcal E_K=\{\bm v:v_0=0\}$,
\begin{equation}
  \|U_K(t,s)|_{\mathcal E_K}\|\le M_K e^{-\omega_K(t-s)},
  \qquad \omega_K>0,\quad 0\le s\le t\le T,
  \label{eq:carleman-retained-stability}
\end{equation}
then \(\|A_{K,K+1}\|\le K\|F_2\|\) and
\(\|\bm z_{K+1}\|\le r^{K+1}\) give the tail estimate
\begin{equation}
  \|\bm e_K(t)\|
  \le
  \frac{M_K}{\omega_K}
  K\|F_2\|r^{K+1},
  \qquad 0\le t\le T.
  \label{eq:carleman-geometric-tail-bound}
\end{equation}
The tail and the truncation error lie in $\mathcal E_K$ because their
constant component vanishes. The full propagator cannot satisfy a
strictly decaying bound while $z_0=1$ is preserved. To deduce geometric
convergence, the growth of $M_K/\omega_K$ must also be controlled, for
example polynomially in $K$. With only a nondecaying bound $M_K$ on
this subspace, replace $M_K/\omega_K$ by $M_KT$. The mechanism is that the first omitted tensor block is
small when the solution remains in a small ball, and stability of the retained
linear dynamics prevents that small tail from being amplified.

The inhomogeneous case \(\bm f_0\ne 0\) is conceptually similar but less clean.
The forcing introduces lower-degree couplings through the level \(0\) block and
moves the invariant ball away from the origin.  A scalar comparison argument
has the form
\begin{equation}
  \frac{d}{dt}\|\bm u(t)\|
  \le
  -\mu\|\bm u(t)\|
  +\|F_2\|\,\|\bm u(t)\|^2
  +\|\bm f_0\|_\infty .
\end{equation}
This leads to a modified smallness condition of the form
\begin{equation}
  R
  \sim
  \frac{1}{\mu}
  \left(
    M_u\|F_2\|
    +\frac{\|\bm f_0\|_\infty}{M_u}
  \right),
  \label{eq:carleman-forced-dissipative-ratio}
\end{equation}
with constants depending on the precise invariant-region estimate.  Rigorous
finite-section bounds with computable constants are developed in
\cite{ForetsPouly2018Carleman,AminiZhengSunMotee2025Carleman}.  The quantum
Carleman analysis of Liu et al. uses the same dissipative intuition: linear
decay controls the tensor tail, and logarithmic truncation depth is what makes
the lifted linear-system approach plausible for weakly nonlinear dissipative
systems \cite{LiuKoldenKroviLoureiroTrivisaChilds2021}.

The logarithmic scaling of \(K\) should not be confused with a small lifted
matrix.  If the original semidiscrete nonlinear ODE has dimension \(N\), then
the retained level-\(K\) Carleman vector has dimension
\begin{equation}
  N_K:=1+N+N^2+\cdots+N^K
  =\frac{N^{K+1}-1}{N-1}.
  \label{eq:carleman-lifted-dimension-from-K}
\end{equation}
Combining this dimension count with \eqref{eq:carleman-K-log-epsilon} gives
\begin{equation}
  N_K
  =
  \exp\left(
    O\left(
      \frac{\log N\,\log(1/\epsilon)}{\log(1/R)}
    \right)
  \right)
  =
  (1/\epsilon)^{O(\log N/\log(1/R))} .
  \label{eq:carleman-dimension-epsilon-scaling}
\end{equation}
The two readings of \eqref{eq:carleman-dimension-epsilon-scaling} could hardly
be more different.  Classically, \(N_K\) is the size of a vector that must be
stored and evolved: for fixed $N$ and fixed convergence ratio, it is
polynomial in $1/\epsilon$, with an exponent growing with $\log N$.
This compares classical and quantum representations of the lifted
system: a classical solver may instead evolve the original $N$-dimensional
nonlinear equation.  On a quantum computer, \(N_K\) is only a Hilbert-space
dimension: a register of \(O(K\log N)\) qubits represents it.  This gap is the
formal source of the quantum interest in Carleman lifting---but the block
encoding, conditioning, output weight, and measurement cost still determine how
much of the gap survives in the end-to-end complexity.

One final caution is specific to quantum state preparation.  Dissipation helps
the tensor tail, but it may also make the final physical state small.  If the
algorithm prepares a normalized state proportional to \(\bm u(T)\), then the
output ratio, analogous to \eqref{eq:heat-output-norm-ratio},
\begin{equation}
  g_T:=\frac{\|\bm u_0\|}{\|\bm u(T)\|}
  \label{eq:carleman-output-ratio}
\end{equation}
can enter the success probability.  There is also a level-selection effect in
the lifted state.  If the scaled solution satisfies \(\|\bm u(T)\|\le r<1\), then
the probability of observing the physical level \(k=1\) in the ideal lifted
state obeys the rough bound
\begin{equation}
  p_1(T)
  =
  \frac{\|\bm u(T)\|^2}{\sum_{j=0}^{K}\|\bm u(T)\|^{2j}}
  \ge
  (1-r^2)\|\bm u(T)\|^2 .
  \label{eq:carleman-level-one-prob}
\end{equation}
Thus the constant block \(\bm z_0=1\) can dominate the norm of the lifted
state.  The same damping that improves Carleman convergence can make normalized
trajectory output more expensive.  This is another reason to keep
observable-based outputs in view.

\subsection{How the quantum algorithm sees the lifted system}
\label{subsec:quantum-carleman-access}

The dimension of the lifted space was given in
\eqref{eq:carleman-lifted-dimension-from-K}.  A quantum register can represent
the \(k\)th tensor-power space using \(k\log_2 N\) qubits.  This is why Carleman
lifting is attractive in the quantum setting: tensor powers map naturally to
multiple registers.  The register count is only logarithmic in the classical
lifted dimension, although the block-encoding and conditioning costs still have
to be analyzed for the specific problem. 

Explicit circuit constructions for the lifted operators have been developed
for the Carleman-linearized Burgers equation
\cite{DemirdjianHogancampGunlycke2026Burgers} and for polynomial systems,
including the lattice-Boltzmann equation, using linear combinations of
nonunitary terms and their unitary embeddings
\cite{DemirdjianEtAl2026LCNULatticeBoltzmann}.
These constructions address operator access; truncation error, conditioning,
and physical-output extraction remain separate parts of the analysis.

At the initial time, the exact and truncated lifted vectors agree.  We write
\begin{equation}
  \bm y_K(0)=
  \begin{bmatrix}
  1\\
  \bm u_0\\
  \bm u_0^{\otimes2}\\
  \vdots\\
  \bm u_0^{\otimes K}
  \end{bmatrix}.
  \label{eq:carleman-initial-lift}
\end{equation}
Assume repeatable coherent preparation of a state proportional to
$\bm u_0$ and access to its norm. Tensor powers are then prepared by
repeated calls to that preparation circuit, using up to $K$ copies,
and the norm-dependent level weights are loaded into a level register.
This does not clone a single unknown input state, and its preparation
cost must be included. This preparation step is treated explicitly in the dissipative
Carleman algorithm of \cite{LiuKoldenKroviLoureiroTrivisaChilds2021}.  The more
subtle issues are the size and conditioning of the lifted linear system, the
stability of \(C_K\), and the probability of extracting the physical first
level.

There are several ways to use the linear machinery developed earlier in the
book.  All of them treat \eqref{eq:truncated-carleman-system} as just another
linear ODE; the nonlinearity survives only in the structure of \(C_K\) and in
the interpretation of the levels.

\paragraph{History-state QLSA.}
A time discretization of \eqref{eq:truncated-carleman-system} can be stacked
into a sparse clock-state linear system and solved by a QLSA, as in the linear
ODE algorithms of Berry et al. and Krovi
\cite{BerryChildsOstranderWang2017,Krovi2023ImprovedDE}.  This was the route
used in the first rigorous quantum Carleman algorithms.  The condition number
of the clock system reflects the stability of the retained propagator, while
padding can increase the probability of observing the final-time block, just as
in the heat equation discussion in Chapter~\ref{chap:parabolic-quantum}.

\paragraph{Dilation and Schr\"odingerization.}
The truncated Carleman equation is a linear, generally nonunitary ODE.  It can
therefore be treated by the dilation methods of Chapter~\ref{chap:hyperbolic-quantum}.
For example, in the autonomous case write
\begin{equation}
  C_K=-iH_K+D_K,
  \qquad
  H_K=\frac{C_K^\dag-C_K}{2i},
  \qquad
  D_K=\frac{C_K+C_K^\dag}{2}.
  \label{eq:carleman-HK-KK-decomposition}
\end{equation}
Moment-matching dilation embeds \(e^{tC_K}\) into a unitary evolution on an
enlarged Hilbert space, followed by ancilla projection.  Schr\"odingerization
provides another mapping of general linear nonunitary dynamics into a
Schr\"odinger-type system in one higher dimension
\cite{JinLiuYu2024Schrodingerization,JinLiuMa2025Schrodingerization}.  Since
Carleman has already converted the nonlinear ODE into a linear ODE, these
linear nonunitary simulation methods can be applied after the Carleman
truncation.  This combined viewpoint, often described as
Carleman-linearization plus Schr\"odingerization, has also been proposed
explicitly for nonlinear PDEs \cite{SasakiEndoMuramatsu2025CLS}.

\paragraph{Observable output.}
In any of these routes, the algorithm produces a state proportional to the
lifted vector \(\bm y_K(t)\).  To recover the physical trajectory one must
project onto the first tensor level.  To estimate a scalar quantity such as
\(\bm\ell^\dag\bm u(t)\), it may be better to estimate the corresponding matrix
element of the lifted linear evolution directly, rather than preparing the
normalized full lifted state and then postselecting.  This observable-driven
perspective will reappear in the outlook.

\section{Kolmogorov, Liouville, and Koopman--von Neumann formulations}
\label{sec:kvn-liouville-nonlinear}

Carleman linearization aims at one trajectory.  The Koopman--von Neumann
viewpoint instead describes probability or observable evolution on phase space.
The change of viewpoint is best understood as a change of question.  Carleman
keeps asking where the trajectory is, and pays for the nonlinearity by
enlarging the state.  The Liouville viewpoint asks instead where the
probability mass is---and this question has a linear answer, because
probability mass carried by different members of an ensemble evolves
independently and therefore superposes.  The nonlinearity has not disappeared;
as we will see, it has moved into the characteristics of the resulting
transport equation.

Consider an ODE on a state space \(X\subset\mathbb R^m\),
\begin{equation}
  \dot{\bm x}=\bm b(\bm x).
  \label{eq:nonlinear-flow-x}
\end{equation}
If the initial condition is random with density \(\rho_0(\bm x)\), then the
density at time \(t\) satisfies
\begin{equation}
  \partial_t\rho(t,\bm x)
  +\nabla_{\bm x}\cdot\big(\bm b(\bm x)\rho(t,\bm x)\big)=0.
  \label{eq:liouville-density}
\end{equation}
This is the Liouville equation, or equivalently the deterministic forward
Kolmogorov equation.  It is the zero-diffusion member of the Fokker--Planck
family.  Like the transport systems in Chapter~\ref{chap:hyperbolic-quantum},
it is a first-order hyperbolic equation, now posed on phase space rather than
physical space.  It is linear in \(\rho\), even though its characteristics
satisfy the nonlinear equation \eqref{eq:nonlinear-flow-x}.

A useful square-root representation---the reason for the square root is
explained in Section~\ref{subsec:kvn-square-root-encoding}---is obtained by
writing
\begin{equation}
  \rho(t,\bm x)=|\psi(t,\bm x)|^2.
  \label{eq:rho-psi-square}
\end{equation}
If \(\psi\) solves
\begin{equation}
  \partial_t\psi
  =-\bm b\cdot\nabla\psi
  -\frac12(\nabla\cdot\bm b)\psi,
  \label{eq:kvn-real-form}
\end{equation}
then \(|\psi|^2\) solves \eqref{eq:liouville-density}.  The two terms in
\eqref{eq:kvn-real-form} have transparent roles: the first transports \(\psi\)
along the characteristics of the flow, and the second assigns to \(\psi\) half
of the volume-compression factor \(\nabla\cdot\bm b\), so that the density
\(|\psi|^2\) acquires the whole of it.  This even bookkeeping between \(\psi\)
and its conjugate is precisely what symmetrizes the generator.  Equivalently,
\begin{equation}
  i\partial_t\psi=H_{\rm KvN}\psi,
  \qquad
  H_{\rm KvN}
  =-i\left(
    \bm b(\bm x)\cdot\nabla_{\bm x}
    +\frac12\nabla_{\bm x}\cdot\bm b(\bm x)
  \right).
  \label{eq:kvn-hamiltonian}
\end{equation}
Integration by parts shows formal symmetry when the boundary terms
vanish. Unitary evolution requires a self-adjoint realization on a
specified domain. A smooth complete flow, for example on a periodic
domain or a preserved domain with suitable boundary conditions, gives
a unitary transport of half-densities with the appropriate Jacobian
factor. In such a setting the forward density evolution has a unitary
Schr\"odinger representation \cite{Joseph2020KvN}. Although \cref{eq:kvn-hamiltonian} is a first-order hyperbolic PDE similar to that in Chapter~\ref{chap:hyperbolic-quantum},   a structure-preserving scheme is needed to maintain the Hermitian property on the right-hand side, in order for Hamiltonian simulation algorithms to be applicable right away.

At the discretization level this statement should be read with care.  A
centered, skew-adjoint discretization of the symmetrized operator in
\eqref{eq:kvn-hamiltonian} preserves Hermiticity and can be simulated by the
Hamiltonian methods of Chapter~\ref{chap:hyperbolic-quantum}.  If one instead
uses upwinding or artificial viscosity to stabilize the first-order transport
equation, the discrete operator becomes non-Hermitian, exactly as in the upwind
hyperbolic systems of Section~\ref{sec:first-order-hyperbolic-dilation}.  Then
one must again use dilation, Schr\"odingerization, or another nonunitary-flow
simulation method.  Thus the formal unitary KvN representation does not remove
the numerical choice between centered conservative discretization and dissipative
stabilization.

\subsection{Why encode the square root rather than the density?}
\label{subsec:kvn-square-root-encoding}

A probability density is normalized in \(L^1\):
\begin{equation}
  \int_X\rho_0(\bm x)\,d\bm x=1.
  \label{eq:density-L1-normalization}
\end{equation}
Quantum amplitudes are normalized in \(L^2\).  If one directly amplitude-encoded
sampled density values, the probability of observing grid point \(j\) would be
proportional to \(\rho_j^2\), not to \(\rho_j\).  The natural discrete KvN state
instead uses
\begin{equation}
  \psi_j(0)\approx\sqrt{\rho_0(\bm x_j)\,\Delta V_j},
  \qquad
  \sum_j|\psi_j(0)|^2\approx1,
  \label{eq:kvn-discrete-square-root-state}
\end{equation}
where \(\Delta V_j\) is the phase-space cell volume.  Measurement in the
computational basis then samples the desired probability masses.  In other
words, the square root aligns the Born rule with the classical statistics: the
quantum measurement of \(\ket{\psi}\) \emph{is} a draw from the (discretized)
distribution \(\rho\), with no postprocessing.

This square-root encoding has an initialization subtlety.  Preparing
\(\ket{\psi_0}\) may be efficient for product densities or for smooth
distributions whose hierarchical cell masses can be computed, as in the
state-preparation methods of Chapter~\ref{chap:basic-elements}.  It need not
be efficient for a general correlated high-dimensional density.  A deterministic
initial condition corresponds to a Dirac delta; on a fixed grid this is a basis
state, but in the continuum it is singular and the required grid resolution
grows as the distribution is localized more sharply.  One must also choose an
initial phase for \(\psi\).  The real nonnegative choice
\(\psi_0=\sqrt{\rho_0}\) is simplest, but zeros or nonsmooth densities can reduce
the regularity of the square root.

On the other hand, the observable output is natural.  If \(a(\bm x)\) is a classical observable,
then
\begin{equation}
  \mathbb E[a(\bm X_t)]
  =\int_X a(\bm x)\rho(t,\bm x)\,d\bm x
  =\langle\psi(t),M_a\psi(t)\rangle,
  \label{eq:kvn-observable}
\end{equation}
where \(M_a\) denotes multiplication by \(a\).  A quantum algorithm can
therefore simulate \eqref{eq:kvn-hamiltonian} and estimate
\eqref{eq:kvn-observable} using the measurement and amplitude-estimation tools
of Chapter~\ref{chap:basic-elements}.

\subsection{The backward Kolmogorov or Koopman equation}
\label{subsec:koopman-backward-observables}

There is a dual linear equation for observables.  Let \(\Phi_t(\bm x)\) denote
the flow generated by \eqref{eq:nonlinear-flow-x}, and define the Koopman
evolution
\begin{equation}
  (U_ta)(\bm x):=a(\Phi_t(\bm x)).
  \label{eq:koopman-semigroup}
\end{equation}
Then \(a_t=U_ta\) satisfies the backward transport equation
\begin{equation}
  \partial_t a_t(\bm x)
  =\bm b(\bm x)\cdot\nabla_{\bm x}a_t(\bm x),
  \qquad a_0=a.
  \label{eq:koopman-backward-equation}
\end{equation}
For a stochastic differential equation, this becomes the backward Kolmogorov
equation and includes a second-order diffusion term.  The forward and backward
descriptions are dual:
\begin{equation}
  \int_X a(\bm x)\rho(t,\bm x)\,d\bm x
  =\int_X (U_ta)(\bm x)\rho_0(\bm x)\,d\bm x.
  \label{eq:koopman-liouville-duality}
\end{equation}
Thus one may evolve the probability density forward or evolve the observable
backward, depending on which representation and initial state are easier to
access.

\begin{remark}[Schr\"odinger and Heisenberg pictures]
Readers coming from quantum computing will recognize the structure of
\eqref{eq:koopman-liouville-duality}.  Evolving the density forward while
keeping the observable fixed is the classical analogue of the Schr\"odinger
picture; evolving the observable backward while keeping the initial
distribution fixed is the analogue of the Heisenberg picture; and
\eqref{eq:koopman-liouville-duality} states that both pictures compute the same
expectation, just as
\(\langle\psi(t),A\,\psi(t)\rangle=\langle\psi_0,U_t^\dag A\,U_t\,\psi_0\rangle\)
in quantum mechanics.  The choice between them is a modeling decision: the
forward route asks for an efficiently preparable initial density, while the
backward route asks for an efficiently representable observable.
\end{remark}

Koopman spectral discretizations provide one classical way to approximate the
observable equation.  Shi and Yang compare such a spectral lifting with
Carleman linearization for nonlinear autonomous systems
\cite{ShiYang2024KoopmanSpectral}.  Jennings et al. formulate Quantum Koopman Algorithms using approximately
closed sets of observables, with separate algorithms for observable evolution
and spectral analysis
\cite{JenningsKorzekwaLostaglioWang2026QuantumKoopman}.   Quantum proposals
based on Koopman, KvN, Liouville, and related linear representations include
\cite{Joseph2020KvN,JinLiuYu2023LinearRepresentations,JinLiu2024NonlinearPDEObservables}.
As always, the infinite-dimensional linear identity is only the starting point;
efficient finite-dimensional approximations, and the question of quantum
advantage, remain largely open
\cite{LinLowrieAslangilSubasiSornborger2024Challenges}.

\subsection{When can the KvN viewpoint help?}
\label{subsec:kvn-advantage-perspective}

The KvN formulation is not a magic way to solve one nonlinear trajectory.  If a
nonlinear PDE discretization has \(N\) spatial degrees of freedom, then the KvN
equation lives on the \(N\)-dimensional phase space of possible vectors \(\bm u\).
A tensor-product grid in this phase space is enormous.  A quantum computer can
store amplitudes on such a grid using logarithmically many qubits in the number
of phase-space grid points, but it still needs efficient Hamiltonian access,
state preparation, and measurement.  For discretized KvN operators, the access
model is closely tied to the first-order hyperbolic discretization choices
discussed in Chapter~\ref{chap:hyperbolic-quantum}.

The potential advantage is most plausible when the task is already
statistical.  Examples include uncertainty propagation for a distribution of
initial conditions, low-dimensional ensemble observables, or probability of a
rare event defined by a simple indicator.  In such settings the output is not a
full trajectory but an expectation of the form \eqref{eq:kvn-observable}.  The
KvN representation is then aligned with the output model of quantum computing:
prepare a state, evolve it unitarily, and estimate an observable.  By contrast,
if the goal is to reconstruct one high-resolution deterministic solution field,
the phase-space lift is usually the wrong representation.

\section{Two linearizations, two output models}
\label{sec:nonlinear-two-output-models}

The two approaches in this chapter solve different problems.  It is useful to
keep the contrast explicit.

\begin{center}
\begin{tabular}{lll}
\toprule
 & Carleman lifting & Liouville/KvN lifting \\
\midrule
Original object & one trajectory \(\bm u(t)\) & density \(\rho(t,\bm x)\) \\
Linear space & tensor powers \(\bm u^{\otimes k}\) & phase-space wavefunction \(\psi\) \\
Best suited for & polynomial vector fields & ensemble/probability observables \\
Quantum primitive & QLSA / linear ODE / dilation & Hamiltonian simulation \\
Main output & \(\ket{\bm u(T)}\) or \(\bm\ell^\dag\bm u(T)\) & \(\mathbb E[a(\bm X_T)]\) \\
Main caveat & truncation and stability & phase-space dimension \\
\bottomrule
\end{tabular}
\end{center}

Carleman lifting is close to the method-of-lines viewpoint familiar from
numerical PDEs.  One discretizes the PDE, obtains a nonlinear ODE, lifts it to
a larger linear ODE, and then applies linear quantum-algorithm machinery.  The
price is truncation and the growth of tensor-power spaces.

Liouville and Koopman--von Neumann lifting take a statistical viewpoint.  The
nonlinear flow moves probability mass through phase space; the density evolves
linearly.  The price is that the independent variable is now the state vector
itself.  This is powerful for ensemble observables but may be inappropriate if
the goal is a single deterministic solution field.

Neither column of the table dominates the other; they answer different
questions.  In the vocabulary of the earlier chapters, choosing a linearization
is choosing an output model, and the end-to-end accounting that ran through the
linear chapters---preparation, evolution, extraction---must be carried out
separately for each column.

\section{Outlook: nonlinear algorithms are problem dependent}
\label{sec:nonlinear-outlook}

Just as in classical numerical analysis, useful quantum algorithms for
nonlinear problems are likely to be problem dependent.  The correct
representation depends on the solution properties, the stability mechanism, the
desired output, and the structure of the nonlinearity.  A lifting that is
effective for a dissipative reaction--diffusion equation may be inappropriate
for a conservation law with shocks or for a Hamiltonian system with long-time
phase information.

\paragraph{Chaotic dynamics and Lyapunov exponents.}
Chaotic dynamics provide a particularly sharp warning. Lewis et al.
derive an exponential-in-time lower bound for normalized trajectory
preparation using quantum state-discrimination arguments, under their
input-state access model and dynamical hypotheses, including a positive
Lyapunov exponent and subexponential solution-norm growth in the relevant
coordinates \cite{LewisEidenbenzNadigaSubasi2024Limitations}.
The access model is essential: this is not an unconditional lower
bound for every formulation with explicitly supplied classical initial data.  This does not rule out every
quantum task associated with chaos, but it shows that normalized trajectory
preparation cannot generically evade exponential sensitivity to initial data.
A useful open problem is to understand how Lyapunov exponents enter the
complexity of more refined quantum tasks, such as estimating invariant-measure
averages, correlation functions, or coarse observables, where full trajectory
preparation may not be necessary.

\paragraph{Can Carleman be turned into a direct dilation?}
Carleman linearization is itself a dilation: it embeds the nonlinear trajectory
into the first block of a linear infinite-dimensional system.  The quantum
algorithms above first truncate this dilation and then apply a linear solver,
Schr\"odingerization, or a nonunitary-flow dilation.  A natural question is
whether one can design a more direct dilation from the nonlinear problem to a
linear Schr\"odinger equation, avoiding an excessively large tensor tower or
improving the success probability.  Existing Carleman plus Schr\"odingerization
schemes are first steps in this direction
\cite{SasakiEndoMuramatsu2025CLS,JinLiuYu2024Schrodingerization}, but a general
structure-preserving theory remains open.

\paragraph{Observable-driven nonlinear algorithms.}
The output model may be as important as the linearization.  Preparing
\(\ket{\bm u(T)}\) is often too ambitious, especially for dissipative or chaotic
systems.  Many scientific questions ask instead for a scalar observable:
energy, flux, drag, reaction rate, probability of a transition, or expectation
under an uncertain initial condition.  Re-evaluating Carleman, Koopman, and
Liouville algorithms from an observable-driven perspective may lead to better
end-to-end complexity estimates.  This mirrors the lesson from the linear PDE
chapters: block-encoding a propagator is only one step; the decisive quantity
is the cost of extracting the desired physical answer.

\paragraph{Nonlinear Schr\"odinger-type equations.}
Nonlinear Schr\"odinger and Gross--Pitaevskii equations form another important
class, but they require care.  Their evolution is nonlinear in the wavefunction
and therefore cannot be represented by one fixed unitary acting on a single copy
of an arbitrary input state.  Lloyd and Braunstein's foundational
continuous-variable work showed that nonlinear operations are essential for
universal continuous-variable quantum computation \cite{LloydBraunstein1999ContinuousVariables};
it was not, by itself, an algorithm for the nonlinear Schr\"odinger equation.
Later many-copy and mean-field constructions proposed quantum algorithms for
classes of nonlinear differential equations, including nonlinear
Schr\"odinger-type models \cite{LloydDePalmaEtAl2020Nonlinear}.  Understanding
the preparation cost, copy complexity, approximation regime, and physically
meaningful observables remains essential.

For this introductory book, Carleman and Koopman--von Neumann methods are
enough to indicate the landscape.  Other directions include level-set and
Young-measure representations, stochastic dynamics, nonlinear Hamiltonian
systems, structure-preserving hybrid algorithms, and problem-specific analytic
transforms.  The essential lesson is that a nonlinear quantum algorithm begins
not with a circuit, but with a precise classical statement of the solution
concept and the quantity to be computed.

\section{Exercises}

\begin{exercise}[First Carleman levels]
For the quadratic ODE \eqref{eq:quadratic-ode}, derive
\eqref{eq:carleman-block-row} explicitly for $k=1$ and $k=2$, including the
terms generated by $\bm f_0$.
\end{exercise}

\begin{exercise}[The discarded tensor tail]
Show that the exact first $K$ tensor levels satisfy
\eqref{eq:projected-exact-carleman-system}, and identify the term discarded in
\eqref{eq:truncated-carleman-system}.
\end{exercise}

\begin{exercise}[Scalar Riccati example]
Let $\dot u=f_0+au+bu^2$ be a scalar Riccati equation.  Write the first four
equations of its Carleman system for $z_k=u^k$.
\end{exercise}

\begin{exercise}[Dissipative Burgers]
For the semidiscrete Burgers system \eqref{eq:burgers-semidiscrete-quadratic},
identify $F_1$, $F_2$, and the dissipation rate $\mu$ in
\eqref{eq:carleman-linear-dissipation}.  Under what size condition on the
initial data and convection term does the ratio $R$ in
\eqref{eq:carleman-dissipative-ratio} satisfy $R<1$?
\end{exercise}

\begin{exercise}[Carleman tail bound]
Starting from \eqref{eq:carleman-tail-duhamel}, derive
\eqref{eq:carleman-geometric-tail-bound} under the assumptions
\eqref{eq:carleman-rescaled-small-solution} and
\eqref{eq:carleman-retained-stability}.
\end{exercise}

\begin{exercise}[KvN square root]
Show that if $\psi$ satisfies \eqref{eq:kvn-real-form}, then
$\rho=|\psi|^2$ satisfies the Liouville equation
\eqref{eq:liouville-density}.
\end{exercise}

\begin{exercise}[Hermiticity of the KvN Hamiltonian]
For periodic boundary conditions, verify by integration by parts that the
symmetrized operator in \eqref{eq:kvn-hamiltonian} is Hermitian on $L^2$.
What term would be missed by the nonsymmetrized operator
$-i\bm b\cdot\nabla$?
\end{exercise}

\begin{exercise}[Density versus square-root encoding]
On a uniform phase-space grid, explain why amplitude-encoding $\rho_j$ samples
a distribution proportional to $\rho_j^2$, whereas amplitude-encoding
$\sqrt{\rho_j\Delta V}$ samples the probability masses $\rho_j\Delta V$.
\end{exercise}

\begin{exercise}[Forward--backward duality]
Derive the duality identity \eqref{eq:koopman-liouville-duality} from the flow
map $\Phi_t$.
\end{exercise}
\backmatter
\printbibliography

\end{document}